\documentclass[a4paper,reqno]{amsart}

\usepackage[english]{babel} 
\usepackage{csquotes}
\usepackage[colorlinks=true,linkcolor=black,anchorcolor=black,citecolor=black,filecolor=black,menucolor=black,runcolor=black,urlcolor=black]{hyperref} 
\usepackage{graphicx,color,comment}  
\usepackage[dvipsnames]{xcolor} 
\usepackage{enumitem} 
\usepackage{rotating} 
\usepackage[absolute]{textpos}

\usepackage{mdframed}
\usepackage{xfrac}\usepackage{nicefrac}
\usepackage{comment}

\usepackage{amsmath,amssymb,amsfonts,amsthm} 
\usepackage{mathtools,stmaryrd}  
\definecolor{light-gray}{gray}{0.95}
\usepackage[color=light-gray]{todonotes}
\usepackage{tikz-cd}  
\usepackage[all,cmtip]{xy} 
\usepackage[bbgreekl]{mathbbol} 
\usepackage{mathrsfs,slashed}

\usepackage{mathrsfs}
\usepackage[mathcal,mathscr]{euscript} 

\usepackage[backend=biber, giveninits=true, style=alphabetic, sorting=nyt,isbn=false, maxalphanames=5,url=false,doi=false, maxbibnames=99]{biblatex}
\usepackage{thmtools}

\newtheorem{maintheorem}{Theorem}

\declaretheoremstyle[
  spaceabove=5pt, spacebelow=5pt,
  headfont=\bfseries,
  notefont=\normalfont, notebraces={(}{)},
  bodyfont=\normalfont,
  postheadspace=1em,
  qed=$\Diamond$
]{pluto}
    \declaretheorem[style=pluto,name=Definition,    numberwithin=section]{definition}
    
    \declaretheorem[style=pluto,name=Assumption,
    ]{assumption}

\declaretheoremstyle[
  spaceabove=5pt, spacebelow=5pt,
  headfont=\itshape,
  notefont=\normalfont, notebraces={(}{)},
  bodyfont=\normalfont,
  postheadspace=1em,
  qed=$\Diamond$
]{pluto2}
    \declaretheorem[style=pluto2,name=Remark,    sibling=definition]{remark}

\declaretheoremstyle[
  spaceabove=5pt, spacebelow=5pt,
  headfont=\bfseries,
  notefont=\normalfont, notebraces={(}{)},
  bodyfont=\itshape,
  postheadspace=1em,
]{pluto3}
    \declaretheorem[style=pluto3,name=Theorem,    sibling=definition]{theorem}
    \declaretheorem[style=pluto3,name=Lemma,    sibling=definition]{lemma}
    \declaretheorem[style=pluto3,name=Corollary,    sibling=definition]{corollary}
    \declaretheorem[style=pluto3,name=Proposition,    sibling=definition]{proposition}

\newcommand{\pp}{\partial}
\renewcommand{\d}{\mathbb{d}}

\newcommand{\AR}[1]{{\color{ForestGreen}{#1}}}

\newcommand{\calF}{\mathcal{F}}

\DeclareMathAlphabet{\mathbbmsl}{U}{bbm}{m}{sl}

\newcommand{\UpsilonH}{\Upsilon_\mathfrak{H}}

\newcommand{\Upsilono}{\Upsilon_\circ}
\newcommand{\UpsilonoCou}{\Upsilon_{\circ,\mathrm{Cou}}}
\newcommand{\UpsilonCou}{\Upsilon_{\mathrm{Cou}}}
\newcommand{\G}{\mathscr{G}}
\newcommand{\Go}{\G_\circ}
\newcommand{\Gooff}{\G_\circ^{\off}}
\newcommand{\Gred}{\underline{\G}}
\newcommand{\Gredoff}{\underline{\G}^{\off}}
\newcommand{\Gpoff}{\G_\pp} 
\renewcommand{\L}{\mathbb{L}}
\newcommand{\bi}{\mathbb{i}}
\newcommand{\Lie}{\mathrm{Lie}}

\newcommand{\bom}{\boldsymbol{\omega}}
\newcommand{\bHo}{\boldsymbol{H}_\circ}
\newcommand{\bH}{\boldsymbol{H}}
\newcommand{\bh}{\boldsymbol{h}}
\newcommand{\balpha}{{\boldsymbol{\alpha}}}

\newcommand{\bk}{\boldsymbol{k}}
\newcommand{\bE}{{\boldsymbol{E}}}
\newcommand{\be}{{\boldsymbol{e}}}
\newcommand{\bC}{{\boldsymbol{C}}}
\newcommand{\bc}{{\boldsymbol{c}}}
\newcommand{\fN}{\mathfrak{N}}

\newcommand{\bJ}{\boldsymbol{J}}
\newcommand{\tr}{\mathrm{tr}}
\renewcommand{\pp}{\partial}
\newcommand{\wt}[1]{\widetilde{#1}}
\newcommand{\wh}[1]{\widehat{#1}}
\renewcommand{\Im}{\mathrm{Im}}
\renewcommand{\ker}{\mathrm{Ker}}
\newcommand{\wtpi}{\widetilde{\pi}}

\newcommand{\tX}{\widetilde{\X}}
\newcommand{\tA}{\widetilde{\A}}
\newcommand{\calW}{\mathcal{W}}
\newcommand{\calV}{\mathcal{V}}
\newcommand{\calN}{\mathcal{N}}
\newcommand{\Acal}{\mathcal{A}}
\newcommand{\strong}[1]{{#1}'}
\newcommand{\dual}[1]{{#1}^*}
\newcommand{\bidual}[1]{{#1}^{**}}
\newcommand{\rad}{{\text{rad}}}

\newcommand{\bK}{\boldsymbol{K}}
\renewcommand{\H}{{{H}}}
\newcommand{\F}{\mathfrak{F}_{\bullet}}
\newcommand{\Foff}{\mathfrak{F}^\off_\bullet}
\newcommand{\fG}{\mathfrak{G}}
\newcommand{\tfG}{\widetilde{\mathfrak{G}}}
\newcommand{\fg}{\mathfrak{g}}
\newcommand{\fGo}{\mathfrak{G}_\circ}

\newcommand{\fNoff}{\fN^\off}
\newcommand{\fGred}{\underline{\mathfrak{G}}}
\newcommand{\tfGp}{\tfG_\pp} 
\newcommand{\fGredoff}{\fGred^{\off}}
\newcommand{\fGpoff}{\fG_{\partial}} 
\newcommand{\Ho}{{H}_\circ}
\newcommand{\xio}{\xi_\circ}
\newcommand{\h}{{h}}
\renewcommand{\k}{{k}}

\newcommand{\K}{{K}}
\newcommand{\fK}{\mathfrak{K}}
\newcommand{\Jo}{J_{\circ}}
\newcommand{\bJo}{\boldsymbol{J}_{\circ}}

\newcommand{\calPhi}{\mathcal{M}}
\newcommand{\X}{\mathcal{P}}
\newcommand{\C}{\mathcal{C}}

\newcommand{\A}{\mathsf{A}}

\newcommand{\Ann}{\mathrm{Ann}}
\newcommand{\uCo}{\underline{\C}}
\newcommand{\uf}{\underline{\h}{}_\smbullet}
\newcommand{\uh}{\underline{\h}{}_\smbullet}

\newcommand{\uomegao}{\underline{\omega}}
\renewcommand{\S}{\mathcal{S}}

\newcommand{\uS}{\underline{\mathcal{S}}{}}

\newcommand{\uuS}{\underline{\underline{\mathcal{S}}}{}}
\newcommand{\uuC}{\underline{\underline{\mathcal{C}}}{}}

\newcommand{\uuphi}{\underline{\underline{\phi}}{}_\rad}
\newcommand{\uuU}{\underline{\underline{\mathcal{U}}}}

\newcommand{\uuomegao}{\underline{\underline{\omega}}{}}
\newcommand{\oloc}{\Omega_{\loc}}
\newcommand{\iloc}{\Omega_{\int}}
\newcommand{\osrc}{\Omega_{\text{src}}}
\newcommand{\obdry}{\Omega_{\text{bdry}}}

\newcommand{\Xloc}{\mathfrak{X}_{\loc}}
\newcommand{\loc}{\mathrm{loc}}
\newcommand{\off}{\text{off}}

\newcommand{\tbox}[2]{\mbox{\parbox{#1}{\center\sffamily\footnotesize #2}}}

\newcommand{\bd}{\mathbb{d}}
\newcommand{\cbd}{\check{\mathbb{d}}}

\newcommand{\AD}{\mathrm{AD}}
\newcommand{\Ad}{\mathrm{Ad}}
\newcommand{\ad}{\mathrm{ad}}

\renewcommand{\#}{\sharp}

\newcommand{\smbullet}{{\hbox{\tiny$\bullet$}}}

\title[Hamiltonian gauge theory with corners]{Hamiltonian gauge theory with corners: \\constraint reduction and flux superselection}
\author{A. Riello}
\address{Perimeter Institute for Theoretical Physics, 31 Caroline Street North,
Waterloo, Ontario, Canada, N2L 2Y5}
\email{ariello@perimeterinstitute.ca}
\author{M. Schiavina}
\address{Department of Mathematics, University of Pavia, Via Ferrata 5, 27100 Pavia, Italy}
    \address{INFN Sezione di Pavia, via Bassi 6, 27100 Pavia, Italy}
    \email{michele.schiavina@unipv.it}
\date{}

\thanks{A.R.\ acknowledges support from the Université Libre de Bruxelles, the European Union, and the Perimeter Institute.  M.S.\ acknowledges partial support from the NCCR SwissMAP, funded by the Swiss National Science Foundation. M.S.\ is grateful for the hospitality of Perimeter Institute where part of this work was carried out. Research at Perimeter Institute is supported in part by the Government of Canada through the Department of Innovation, Science and Economic Development Canada and by the Province of Ontario through the Ministry of Colleges and Universities. This project has received funding from the European Union’s Horizon 2020 research and innovation programme under the Marie Sk{\l}odowska-Curie grant agreement No 801505.}

\begin{document}

\begin{abstract}
We study the Hamiltonian formulation of gauge theory on spacetime manifolds endowed with a codimension-$1$ submanifold with boundary. The latter is thought of as a corner component for the spacetime manifold.  
We characterise  the reduced phase space of the theory whenever it is described by a local momentum map for the action of the gauge group $\G$,  by adapting Fr\'echet reduction by stages to the case of gauge subgroups. 

The local momentum map decomposes into a bulk term called constraint map defining a coisotropic constraint set, and a boundary term called flux map. The first stage, or \emph{constraint reduction} views the constraint set as the zero locus of a momentum map for a normal subgroup $\mathcal{G}_\circ\subset\mathcal{G}$, called \emph{constraint gauge group}. The second stage, or \emph{flux superselection}, interprets the flux map as the momentum map for the residual action of the \emph{flux gauge group} $\underline{\mathcal{G}}\doteq\mathcal{G}/\mathcal{G}_\circ$. Equivariance is controlled by cocycles of the flux gauge group $\underline{\mathcal{G}}$.

Whereas the only physically admissible value of the constraint map is zero, the flux map is largely unconstrained.
As a result, the reduced phase space of the theory, when smooth, is only a partial Poisson manifold $\underline{\underline{\mathcal{C}}}=\mathcal{C}/\mathcal{G} \simeq \underline{\mathcal{C}}/\underline{\mathcal{G}}$. Its symplectic leaves---defined through the flux map---are called \emph{flux superselection sectors}, for they provide a classical analogue of, and a road map to, the phenomenon of quantum superselection.

To corners, we further assign a natural symplectic Lie algebroid over a Poisson manifold, $\mathsf{A}_{\partial} \to \mathcal{P}_{\partial}$, and show how the submanifold of on-shell configurations $\mathcal{C}_{\partial}\subset\mathcal{P}_{\partial}$ is also Poisson. We interpret $\C_\pp$ as defining the Noether charge algebra. Both $\mathcal{C}_{\partial}$ and $\underline{\underline{\mathcal{C}}}$ fibrate over a common \emph{space of superselections}, labeling the Casimirs of both Poisson structures.

We showcase the formalism by explicitly working out the first and second stage reductions for a broad class of Yang--Mills theories, where $\underline{\underline{\mathcal{C}}}$ is found to be a Weinstein space, and discuss further applications to topological theories.
\end{abstract}

\maketitle

\tableofcontents

\section{Introduction}\label{sec:intro}

The Hamiltonian picture for local gauge theories revolves around the description of the theory as a constrained Hamiltonian system: it is given in terms of a symplectic manifold $(\X,\omega)$ together with a given \emph{constraint} surface $\C$, to which (physical) field configurations and motions are restricted.

Consider the example of classical Hamiltonian mechanics. This can be seen as a field theory of maps from the half-line $\mathbb{R}_{\geq 0}$ (time) to a target manifold $\mathcal{Q}$, known as ``configuration space''. Then, one can think of the ``phase space'' $\X = T^*\mathcal{Q}$ as naturally assigned to the boundary of the line (a single point) or as the symplectic manifold given by the space of solutions of the evolution equations inside $\mathrm{Maps}(\mathbb{R}_{\geq 0},\mathcal{Q})$. Since every initial-condition in $\X$ is uniquely extended to a solution of the equations of motion, the two spaces are in fact isomorphic.

One can generalise this picture to gauge field theory which, for simplicity of exposition, we consider on a half cylinder $M=\Sigma\times \mathbb{R}_{\geq 0}$ with $\Sigma$ a compact manifold. When $\Sigma$ has a boundary we will say that $M$ has corners and we call $\pp\Sigma$ \emph{the corner submanifold}. Then, the restriction of the (off-shell) field theoretic data to $\Sigma$ defines a symplectic manifold $(\X,\omega)$, called the \emph{geometric phase space} of the theory.\footnote{Our construction follows \cite{KT1979}, but it relates to the ``covariant phase space'' construction.  See Appendix \ref{app:covariant} for an explicit discussion and further references, and \cite{CMRCorfu,CSTime,CattaneoSchiavinaPCH} for a modern point of view pertinent to this work.}
However, in a gauge theory, not all field configurations on $\Sigma$ can be extended to a solution in a (small-enough) cylinder, but suitable \emph{constraints} on the configurations over $\Sigma$ need to be imposed. Denote by $\C\subset \X$ the set of such configurations.

We are interested in the case where the geometric phase space $\X$ is the space of sections of a fibre bundle over $\Sigma$ endowed with a weakly symplectic structure $\omega$, $\C$ is the vanishing locus of a set of local Hamiltonian functions in involution\footnote{This notion must be taken with a grain of salt. See for example \cite[Remark 3.4]{BSW}.} i.e.\ of a set of first-class constraints, \emph{and} the characteristic foliation of $\C$, denoted $\C^\omega$, is generated by the Hamiltonian vector fields of said functions \cite[Sec. 2.1 and 3.8]{CMRCorfu}. Furthermore, we will assume that the set $\C\subset\X$---given by constrained configurations defined by a set of differential equations---is a smooth submanifold (Assumption \ref{ass:setup}).

This is not generically the case, but at times one has the following additional structure: the coisotropic submanifold $\C$ is the zero-level set of a momentum map associated to the local action of an infinite dimensional Lie group $\Go$, the \emph{constraint gauge group}.\footnote{Notably, this scenario is not realized in general relativity \cite{LeeWald}.} 
This is the setup we will study in this article.

If certain technical conditions are satisfied, one can equivalently look at the coisotropic reduction $\uCo \doteq \C/\C^\omega$ of $\C$ by its characteristic foliation $\C^\omega$, or at symplectic reduction by the corresponding group action, $\uCo \simeq \C/\Go$. In this work we will refer to this space as the \emph{constraint-reduced phase space}. When smooth, this is a symplectic manifold.

Now, if the corner manifold $\pp\Sigma$ is nontrivial, the question arises of the relationship between the constraint gauge group $\Go$ and a larger, ``full'', gauge group $\G$ acting on $\X$---which leaves $\C$ invariant and of which $\Go$ is a normal subgroup (Theorem \ref{thm:fGo}).

\medskip

To understand the distinction between $\Go$ and $\G$, let us decompose a closed manifold $\Sigma$ into two subregions $\Sigma_1$ and $\Sigma_2$ sharing a boundary interface $K$. To each subregion we associate, by restriction of fields, a symplectic space $\X_i$, and introduce the respective constraint-reduced phase spaces through the prescription outlined above: $\uCo_i=\C_i/\C_i^{\omega_i}$ with $(\C_i,\omega_i)$ constructed as follows. 

The constraint set $\C$ is defined by differential equations, i.e. by the kernel of a certain differential operator acting on elements of $\X$. By locality, this operator can be restricted to act on elements of $\X_i$ over $\Sigma_i$, and the ensuing kernel defines $\C_i$. (Here we consider the differential operator with \emph{free} boundary conditions so that $\C_i$ and $\uCo_i$ are defined intrinsically to $\Sigma_i$, with no reference to the ambient manifold $\Sigma$.) Any solution to the constraint equations over $\Sigma$, seen as an element of $\X$, restricts to a solution of those same equations over $\Sigma_i$, thus defining a restriction map $\C\to \C_i$, which however need not be surjective.\footnote{E.g.\ in Chern--Simons (CS) theory, the constraint equation is a flatness condition. Consider now the case of CS theory on a trivial principal $\mathrm{U}(1)$-bundle over a a 2-disk $\Sigma$, with $\Sigma_1$ a ``concentric'' disk and $\Sigma_2$ the remaining annulus. Then flat connections on $\Sigma$ and $\Sigma_1$ are necessarily $d$-exact, whereas this is not necessarily the case on $\Sigma_2$. \label{fnt:CSexample}}
The symplectic form $\omega_i$ is also built through a locality assumption: our definition of a gauge theory requires the assignment of a (suitably local) symplectic form \emph{density} $\bom$ which can then be integrated over either $\Sigma$ or $\Sigma_i \subset \Sigma$ to give $\omega$ and $\omega_i$, respectively.

Hence, for each subregion we have the maps
\[
\xymatrix{
\X   \ar[d]_{\bullet\vert_{\Sigma_i}}&\;\C \ar@{_(->}[l]\ar[r]\ar[d] &\uCo\\
\X_{i}  &\;\C_{i} \ar@{_(->}[l] \ar[r] &\uCo_{i}\\
}
\]

Now, classical observables are defined as functions over constrained field configurations that are invariant under the gauge group action, and since the observables for a subsystem should be embeddable into observables of the larger system, three distinct problems arise. 

The first one is due to the non-local nature of the solution space of the constraint equations: global obstructions might arise that prevent a solution over the (sub)manifold with boundary $\Sigma_i$ from being extendable to the entirety of $\Sigma$; in concrete applications, this problem can relate, for instance, to the cohomology of $\Sigma$ (cf.\ Footnote \ref{fnt:CSexample}). 
This can be thought of as a \emph{globalisation} problem.

The second problem is due to the gauge invariance properties of observables; in other words, even if all solutions to the constraint over $\Sigma_i$ were extendable to $\Sigma$, one would still face the problem that not all functions over $\uCo_i$ would correspond to the restriction of an invariant function over $\C$. 
This is because functions over $\uCo_i$ need not be invariant under elements of the characteristic distribution of $\C$ whose restriction to the interface $K$ is nontrivial. (In field theory, the characteristic distribution $\C^\omega$ is local on $\Sigma$ as it is given in terms of sections of bundles over it; as such it can be restricted to $K\hookrightarrow \Sigma$.) Loosely speaking, this means that functions over $\C_i$ need not be invariant under gauge transformations which are nontrivial at $\pp\Sigma_i =\K$. 
As we will see, this problem turns out to be local in nature and, at least in certain cases, can be addressed from within $\Sigma_i$, thought of as a manifold with boundary. This problem of defining \emph{quasi-local} observables could be named the problem of \emph{cutting}.

Finally, the third problem relates to whether the fully gauge invariant observables over $\Sigma_{1,2}$, built to address the cutting problem, are capable of ``generating'' the entire algebra of gauge invariant observables over $\Sigma$. This problem is referred to as \emph{gluing}. (See \cite[Section 6]{RielloGomes} for results within Yang--Mills theory, and \cite{cattaneo2022note} for an analysis within a cohomological framework compatible with the investigations considered here.) 

\medskip

~
In this work we address ``cutting'' in its simplest---yet common---embodiment. We consider gauge theories for which the (symplectic) space of fields $\X=\Gamma(\Sigma,E)$, the space of sections of a fibre bundle on a manifold with boundary $(\Sigma,\pp\Sigma)$, carries the action of a local Lie group $\G$ that admits a \emph{local} momentum map $H\colon \X \to \dual{\fG}$; the notion of observable introduced above naturally leads to considering functions on the quotient $\uuC\doteq\C/\G$, which we investigate.
(Main definitions in Section \ref{sec:localfieldtheory}, see in particular Definition \ref{def:Ham+mommaps}. The main set of requirements we will need is summarized in Assumption \ref{ass:setup} of Section \ref{sec:constraintsandfluxes}, whereas further technical assumptions are presented throughout the text: Assumptions \ref{ass:groups}, \ref{ass:smoothuCo}, \ref{ass:symplecticclosure}, \ref{ass:smoothsuperselections}, \ref{ass:symplecticclosure-f} and \ref{ass:AnnIm}.)

Thinking of $(\Sigma,\partial\Sigma)$ as a standalone manifold, and assuming it has a \emph{nontrivial} boundary---which we hereafter refer to as \emph{corner}, following the spacetime perspective---the constraint set $\C$ no longer coincides with the zero-level set of $H$ due to the interference of boundary contributions. 
Indeed, while it is true that $H$ is still a momentum map for the action of $\G$, symplectic reduction at $H=0$ fails to describe the correct physical phase space, for $H^{-1}(0)$ imposes additional spurious boundary conditions: i.e.\ we only have $H^{-1}(0)\subsetneq \C$. This is the problem of cutting at its heart.

One consequence of this observation is that $\uuC$ can be thought of as the final result of a two-stage symplectic reduction procedure \cite{marsden2007stages}, adapted to the nontrivial case of local subgroups of gauge groups which naturally emerge from the presence of boundaries.

The organization into two stages is controlled by a unique decomposition of the local momentum map $H$ into bulk and boundary terms---the \emph{constraint and flux maps}, denoted $\Ho$ and $h$ respectively---which in turn uniquely determine, through a maximality condition, a normal subgroup $\Go\subset \G$ (the constraint gauge group) and the associated quotient $\Gred = \G/\Go$ (the flux gauge group). (This requires Assumptions \ref{ass:setup} and \ref{ass:groups}.)

The first-stage reduction---also called constraint reduction---is then symplectic reduction for the action of $\Go$ at the zero-level set of the \emph{constraint} momentum map $J_\circ$, built out of $H_\circ$ and $h$. The resulting quotient, $\C/\Go$, coincides with the constraint-reduced phase space $\uCo=\C/\C^\omega$.  (Section \ref{sec:bulkreduction} and in particular Definition \ref{def:uomegao} and Proposition \ref{prop:symplclosVo}, as well as Assumptions \ref{ass:smoothuCo} and \ref{ass:symplecticclosure})

The second-stage reduction arises then naturally from a residual Hamiltonian action $\Gred \circlearrowright \uCo$ (Propositions \ref{prop:bulkredflux} and \ref{prop:hamactionGp}), which can instead be performed at any value of the \emph{flux} momentum map, compatible with the constraint, i.e.\ for every \emph{on-shell} flux (Definition \ref{def:flux}). The fully-reduced phase space, when smooth, is then the Poisson manifold $\uuC=\C/\G \simeq \uCo/\Gred$, whose symplectic leaves correspond to the level sets of the flux map, and thus to coadjoint/affine orbits of the flux gauge group. (This requires Assumptions \ref{ass:symplecticclosure}, \ref{ass:smoothsuperselections}, \ref{ass:symplecticclosure-f}.)  We call these symplectic leaves the (on-shell) ``flux superselection sectors''.  
In the particular case of Chern--Simons theory with $\Sigma$ a punctured disk, this correctly reproduces the coadjoint orbits of the centrally-extended loop group (Section \ref{sec:ex-ChernSimons}), consistently with the known results of Meinrenken and Woodward \cite{MeinrenkenWoodward}.

We chose the name ``flux superselection sector'' in analogy with the quantum theory, where the Hilbert space of a theory factorises over superselection sectors corresponding to the irreps of some central subalgebra of observables. In our classical framework, it is the Casimirs (of the Poisson structure) characterising the flux orbits that yield gauge invariant functionals which are central in the Poisson algebra.

Section \ref{sec:cornerred} concludes the description of the phenomenon of flux superselection from the point of view of symplectic reduction, and what follows is not strictly necessary from this point of view.

In Section \ref{sec:cornerdata} we obtain a description of flux superselections that is both intrinsic to the boundary $\pp\Sigma$ and ``off-shell'', i.e.\ such that it does not require the knowledge or description of the constraint set $\C$. Both of these features simplify the task of identifying the possible superselection sectors attached, in a given theory, to a manifold with a given corner $\pp\Sigma$. The information about the shell---which will in general strongly depend on global properties of $\Sigma$, as opposed to only $\pp\Sigma$---can be fed at the very end of this analysis, and in practice further limits the physically available superselection sectors (cf.\ Footnote \ref{fnt:CSexample}).

This construction works as follows. First, a space of boundary configurations $\X_\pp$ and a boundary Lie algebra $\fGpoff$ are constructed by means of the flux map $h$. Then, boundary configurations inherit a Lie algebra action $\rho_\pp\colon\X_\pp\times\fGpoff\to T\X_\pp$
(Theorem \ref{thm:corneralgd}, requiring the simplifying assumption \ref{ass:AnnIm}.) This allows us to define a Poisson manifold $(\X_\pp,\Pi_\pp)$ by setting $\Pi_\pp^\sharp = \rho_\pp$ (Theorem \ref{thm:Poissontheorem}; or the summary Theorem \ref{mainthm:Poisson}(\textit{i})).
Those boundary configurations in $\X_\pp$ that come from the restriction to the boundary of constrained configurations in $\C$ define a Poisson submanifold $\C_\pp\subset\X_\pp$.
One thus finds those symplectic leaves of $\Pi_\pp$ that are compatible with the prescribed cobordism $\Sigma$ (Definition \ref{def:offshellSSS} and Proposition \ref{prop:onshellsuperselections}). 

In most cases considered in this paper, one retains a corner-\emph{local} description of the above data (Proposition \ref{prop:localGredoff}), which practically simplifies the discussion of nonlocal intricacies associated with the (solutions of the) constraint equation.

We close the circle in Section \ref{sec:onshellequivalence} by showing that both Poisson spaces $\C_\pp$ and $\uuC$ fibrate over a common base, the space of superselections $\mathcal{B}\doteq \C_\pp / \rho_\pp(\fG_\pp)$ (Theorem \ref{thm:spaceofsuperselection}; or the summary Theorem \ref{mainthm:Poisson}(\textit{ii})). 

To understand the physical meaning of this result, consider the set of functions on $\uuC$ that Poisson commute with any other function. This is the center of the Poisson algebra over the space of Hamiltonian functions $C^\infty_{\mathrm{Ham}}(\uuC)$ on $\uuC=\uCo/\Gred$, and its elements are generally called Casimir functions of the Poisson structure. 
Note that, in particular, any function on $\uuC$ that is constant over a given superselection sector $\uuS_{[f]}\subset \uuC$ is a Casimir.
Now, Casimir functions on $\uuC=\uCo/\Gred$ lift to a set of basic (i.e.\ $\Gred$-invariant) functions on $\uCo$, and every other $\Gred$-invariant function on $\uCo$ Poisson commutes with the ``lifted'' Casimirs.\footnote{Recall that $\underline{\pi}\colon \uCo \to \uuC$ is a Poisson map, Theorem \ref{mainthm:Poisson} item \ref{thmitem:main2-algebroids}.}  Since $\Gred$-invariant functions on $\uCo$ are the observables of the theory, any sensible quantization of this space of observables will preserve the centrality of the Casimirs, which are then quantized to the center of the quantum algebra of observables. This provides a road map to the quantization of superselection sectors, where one expects the \emph{quantum} sectors to be irreducible blocks of the algebra of quantum observables (given a choice of module that represents the algebra). In other words, we expect these blocks to correspond to the quantizations of each (classical) superselection sector $(\uuS_{[f]},\uuomegao_{[f]})$.

However, studying $\uuC$ and its Casimirs might be challenging and it is therefore useful to have a simpler proxy that captures the same superselection structure. This is where $(\C_\pp,\Pi_\pp)$ becomes important. This is a Poisson space that makes precise the notion of the on-shell Noether corner algebra used in the physics literature. It has the property that its Casimirs are $\fGpoff$-invariant function over $\C_\pp$, i.e. functions on the space of superselections $\mathcal{B}$. Theorem \ref{thm:spaceofsuperselection}---or the summary Theorem \ref{mainthm:Poisson}(\textit{ii})---then guarantees that functions on $\mathcal{B}$ are in turn also Casimirs of the fully reduced phase space $\uuC$ and therefore capture (at least part) of the superselection structure of the original gauge theory under study.

\medskip

The Hamiltonian methods we just summarized apply to a large class of Lagrangian gauge theories---including Yang--Mills, Chern--Simons, and $BF$ theories---whenever $\Sigma$ is a boundary component of a manifold $M$ \emph{with corners}. Corners naturally emerge in Yang--Mills theory when a spacelike and a timelike boundary intersect, or if we are considering a region bounded by a lightlike boundary. In Chern--Simons and $BF$ they are usually associated to punctures, or surgery procedures.
Yang--Mills theory, for $\Sigma\hookrightarrow M$ a spacelike submanifold with boundary, is treated in great detail in the running example Sections \ref{sec:runex-setup} throughout \ref{sec:runex-cornerdata}. In Section \ref{sec:thetaQCD} we analyze the consequences of the introduction of an extra corner term in its symplectic structure. (Yang--Mills theory with a lightlike boundary will be the subject of a separate publication.) Chern--Simons and $BF$ theories are briefly discussed in Sections \ref{sec:ex-ChernSimons} and \ref{sec:BFtheory}, respectively.

\medskip

The approach and results presented in this work are relevant for the problem of quantization of local field theories in the presence of higher codimension strata, for example in phrasing quantum field theory axiomatically as a functor between some category of (geometrized) cobordisms to some linear category. Classically, field theory is expected to assign symplectic manifolds to codimension-one submanifolds, and Poisson manifolds (or a suitable generalization thereof, see below) to codimension-two submanifolds. 
We show that (when smooth) symplectic reduction by stages indeed produces a symplectic manifold $\uCo$ and a (partial) Poisson manifold $\uuC$ associated respectively to first-stage reduction, a.k.a.\ constraint-reduction, and to second-stage reduction, a.k.a.\ flux superselection. The latter is non-trivial \emph{only} in the presence of corners. Additional distinguished objects that emerge in our construction are: the symplectic leaves of the Poisson manifold $\uuC$, which define the superselection sectors; and an off-shell Poisson manifold $\X_\pp$ intrinsic to the corner (independently, that is, of the corresponding cobordism $\Sigma$).

The assignment of symplectic and Poisson data at different codimensions can be embedded into the world of homological algebra by means of the Batalin--Fradkin--Vilkovisky framework \cite{BV1,BV2,BV3}, and the relation between boundary and corner data is well phrased within the BV-BFV formalism \cite{CMR1,CMR2}. While the boundary homological data (BFV) produces a resolution of the reduced phase space $\uCo$ (see e.g.\ \cite{CMRCorfu}), in the form of $L_\infty$ algebroid data \cite{StasheffConstraints88,SchaetzBFV}, it is expected that corner homological data be associated to Poisson-$\infty$ data \cite[Section 2.1.6]{CanepaCattaneo}.

In Section \ref{sec:cornerdata} we characterise the underlying Poisson manifold resolved by the corner homological data induced by the BV-BFV procedure\footnote{More precisely, the corner dg-manifold produced by the BV-BFV procedure for the cases we consider is the Chevalley--Eilenberg--Lichnerowicz complex associated to the Poisson manifold.} \cite[Section 2.2]{CanepaCattaneo}, in the simple but important scenario of gauge theories with a (local) Hamiltonian momentum map, and relate it to the second stage reduction of a Hamiltonian action on a space of fields on the manifold with boundary $(\Sigma,\pp\Sigma)$. 
However, we find that unless certain additional locality properties are satisfied, it is not guaranteed that our corner data matches those obtained via the BV-BFV framework (Section \ref{sec:ultralocality}).

The generalization of these results to field theories that do \emph{not} admit a momentum map description for a gauge group action is obviously of prime importance, but will be tackled elsewhere. Let us stress that some of the complications that might arise involve more general distributions than Lie algebra actions. For example, it is known that the group of diffeomorphisms that preserve a hypersurface is not a normal subgroup of the total diffeomorphism group (see e.g.\ \cite[Example 8.13]{blohmann2018hamiltonian}) and the quotient is generally a more general object such as a groupoid.

We believe that an extension of the methods and ideas contained here to more general field theories is possible. The above mentioned complication arising from boundary-preserving distributions not being ``normal'' is naturally phrased in an $L_\infty$ setting within the Batalin--Vilkovisky setting as the emergence of Poisson-$\infty$ structures at the corner (see, e.g.\ \cite{CanepaCattaneo}). A particularly interesting, natural generalization of our method that can shed light on some of these issues, is one in terms of Hamiltonian Lie  algebroids \cite{blohmann2018hamiltonian}. We will turn to that in a subsequent work.

\subsection*{Summary of results}

In this work we study the Hamiltonian formulation of local, Lagrangian, gauge field theories admitting a local momentum map. We summarise here our main results, which will depend on a number of assumptions (\ref{ass:setup}, \ref{ass:groups}, \ref{ass:smoothuCo}, \ref{ass:symplecticclosure}, \ref{ass:smoothsuperselections}, \ref{ass:symplecticclosure-f} and \ref{ass:AnnIm}), reported below.

\begin{maintheorem}[Reduction in the presence of corners]\label{mainthm:redsummary}
Let $(\X,\bom,\bH)$ be a locally Hamiltonian $\G$-space satisfying Assumption \ref{ass:setup}. Its local momentum form $\bH$ splits into the order-0 constraint form $\bHo$ and the $d$-exact  flux form $d\bh$, i.e.
\[
\bi_{\rho(\xi)} \bom = \bd \langle \bH,\xi\rangle, \quad\text{with}\quad \bH = \bHo + d\bh,
\]
for all $\xi\in\fG = \mathrm{Lie}(\G)$. Denote by $H,\Ho,\h$ the respective integrated local forms.

Let $\C \doteq \Ho^{-1}(0) \subset \X$ be the constraint set and, given a reference configuration $\phi_\smbullet\in\C$, let $\F \doteq \Im(\iota_\C^*\h_\smbullet)$ be the space of on-shell fluxes, where $\h_\smbullet = \h - \h(\phi_\smbullet)$ is the adjusted flux map.

Also, the integrated local form $\k_\smbullet\in \iloc^0(\fG\wedge\fG)$, defined as
    \[
    k_\smbullet(\xi,\eta) \doteq \langle\L_{\rho(\xi)} h_\smbullet,\eta\rangle  - \langle h_\smbullet , [\xi,\eta] \rangle,
\]
is a Chevalley--Eilenberg cocycle and defines an affine action of $\fG$ on $\F$, via \[
\langle\ad^*_{k_\smbullet}(\xi) \cdot f,\eta\rangle \doteq \langle f, [\xi,\eta]\rangle + k_\smbullet(\xi,\eta).
\]
Denote by $(\mathcal{O}_f,\Omega_{[f]})$ the orbit of $f\in\F$ under this action together with its Kirillov--Kostant--Souriau symplectic form.

Then,  
\begin{enumerate}[label=(\roman*)]
    \item\label{thmitem:main-Go} there exists a maximal, constraining, just, Lie ideal $\fGo\subset\fG$, i.e.\ admitting a unique, local, equivariant momentum map $\Jo$ such that $\C=\Jo^{-1}(0)$ (Definition \ref{def:constraintideal}). It is given by $\fGo = \Ann(\F,\fG)$;
    \item\label{thmitem:main-constraint red} under Assumptions \ref{ass:groups} and \ref{ass:symplecticclosure}, $\C\subset \X$ is coisotropic and
    the \emph{constraint-reduced phase space} $(\uCo,\uomegao)$, defined as the coisotropic reduction $\uCo\doteq\C/\C^\omega$ for $\omega=\int_\Sigma \bom$, is given by the symplectic reduction of $(\X,\omega)$ w.r.t.\ the constraint gauge group $\Go$ generated by $\fGo$, i.e.\ $\uCo = \C/\Go$.
    Assumption \ref{ass:smoothuCo} further ensures it is smooth; 
    \item\label{thmitem:main-resHamaction} additionally, $h_\smbullet$ descends to a momentum map $\uh : \uCo \to \fGred^*$ for the action of $\Gred \doteq \G/\Go$ on $(\uCo,\uomegao)$:
    \[
    \bi_{\rho(\underline{\xi})}\uomegao = \bd \langle \uh,\underline{\xi}\rangle,
    \]
    and $\uh$ is equivariant up to the cocycle $k_\smbullet$, seen here as a CE cocycle of $\fGred\doteq\fG/\fGo$.
\end{enumerate}

\medskip

Moreover, denoting by $\S_{[f],q}\subset\C$ the $q$-th connected component of the preimage  $(\iota_\C^*\h_\smbullet)^{-1}(\mathcal{O}_f)$ of a co-adjoint orbit $(\mathcal{O}_{f},\Omega_{[f]})$ for $f\in\F\subset \fG^*$, if Assumption \ref{ass:symplecticclosure}, \ref{ass:smoothsuperselections} and \ref{ass:symplecticclosure-f} also hold, then
\begin{enumerate}[label=(\roman*), resume]
    \item\label{thmitem:main-cornerred} the \emph{fully-reduced phase space} $\uuC \doteq \C / \G\simeq \uCo/\Gred$ is the disjoint union of a collection of symplectic spaces $(\uuS_{[f],q}, \uuomegao_{[f],q})$, called flux superselection sectors:
    \begin{equation*}
        \uuC = \bigsqcup_{\mathcal{O}_f,q} \uuS_{[f],q}, 
        \qquad 
        \uuS_{[f],q} \doteq \S_{[f],q}/\G ,
        \qquad
        \pi^*_{[f],q}\uuomegao_{[f],q} = \iota_{[f],q}^* \omega - (h_\smbullet\circ \iota_{[f],q})^*\Omega_{[f]},
    \end{equation*}
where $\pi_{[f],q}:\S_{[f],q} \to \uuS_{[f],q}$, and $\iota_{[f],q}:\S_{[f],q} \hookrightarrow \X$. \item\label{thmitem:main-Poisson} the \emph{fully-reduced phase space} $\uuC$ is a (partial) Poisson manifold, of which the superselection sectors $(\uuS_{[f],q},\uuomegao_{[f],q})$ are the symplectic leaves.
\end{enumerate}
\end{maintheorem}

In the diagram we summarise the two-stage reduction with wavy lines, and emphasise the unified construction presented in Theorem \ref{mainthm:redsummary}, above, with dashed lines (we omit the $q$-label):

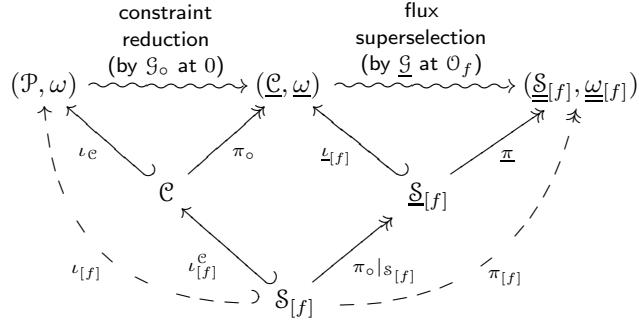
\[
\xymatrix@C=.75cm{
(\X,\omega)
	\ar@{~>}[rr]^-{\tbox{2.2cm}{constraint reduction \\(by $\Go$ at $0$)}}
&&(\uCo,\uomegao)
	\ar@{~>}[rr]^-{\tbox{2.2cm}{flux superselection (by $\Gred$ at $\mathcal{O}_f$)}}
&&(\uuS_{[f]},\uuomegao_{[f]})\\
&{
    \;\C\;
    \ar@{_(->}[ul]^-{\iota_\C}
	\ar@{->>}[ur]_-{\pi_\circ}
	}
&&{
    \;\;\uS_{[f]}
    \ar@{_(->}[ul]^-{\underline{\iota}_{[f]}}
	\ar@{->>}[ur]_-{\underline{\pi}}
	}\\
&&{
    \;\;\S_{[f]}
    \ar@{_(->}[ul]^-{{\iota}^\C_{[f]}}
	\ar@{->>}[ur]_-{\pi_\circ\vert_{\S_{[f]}}}
	\ar@{_(-->}@/^2.7pc/[uull]^{\iota_{[f]}}
	\ar@{-->>}@/_2.7pc/[uurr]_{\pi_{[f]}}
	}\\
}\]

~

We specialise Theorem \ref{mainthm:redsummary} to the case of Yang--Mills theory for $\Sigma$ a non-null\footnote{The null case is discussed in \cite{RielloSchiavinanull}.} manifold with boundary, endowed with a trivial\footnote{This assumption is made for simplicity of exposition. More general statements are found in the sections about Yang--Mills theory.} principal $G$-bundle $P\to \Sigma$, for $G$ either semisimple or Abelian. We further denote $\fg \doteq \mathrm{Lie}(G)$ and $\tr( \cdot \ \cdot )$ an $\Ad$-invariant, nondegenerate, bilinar form on $\fg$:

\begin{maintheorem}[Reduction with corners for non-null YM theory]\label{mainthm:YMon}
    To a smooth non-null manifold with boundary $(\Sigma,\pp\Sigma)$, Yang--Mills theory with structure group $G$ (taken to be either semisimple or Abelian) assigns the locally Hamiltonian $\G$-space $(\X,\bom,\G,\bH)$ where
    \begin{enumerate}
        \item $\X=T^\vee\Acal \simeq \Acal\times \mathcal{E}\ni(A,E)$, with $\Acal = \mathrm{Conn}(P\to\Sigma)\simeq\Omega^1(\Sigma,\fg)$ and $\mathcal{E}\simeq \Omega^{\mathrm{top}-1}(\Sigma,\fg)$, is the geometric phase space,
        \item $\bom = \tr(\bd E \wedge \bd A)$ is the (canonical) symplectic density on $\X$,
        \item $\G= \mathrm{VAut}(P\to\Sigma) \simeq C^\infty_0(\Sigma,G)$ is the gauge group, with action $(A,E)\triangleleft g = ( g^{-1} Ag + g^{-1}dg , g^{-1} E g)$,
        \item $\langle \bH ,\xi\rangle = (-1)^{\mathrm{dim}(\Sigma)}\tr(E \wedge d_A\xi)$ is the (equivariant) local momentum form ($\bK=0$).
    \end{enumerate}
    
    Assumption \ref{ass:setup} is satisfied by $(\X,\bom,\G,\bH)$ when $G$ is Abelian. When $G$ is semisimple, Assumption \ref{ass:setup} is satisfied if $\Acal$ is henceforth reinterpreted as denoting the (dense) subset of \emph{irreducible} principal connections, and $\X$ as the corresponding sub-cotangent bundle, etc. This change in notation will be left implicit in the rest of the theorem.
    
    The constraint- and flux- forms read, respectively,
    \[
    \langle\bHo,\xi\rangle = \tr(d_AE \xi) \qquad \text{and}\qquad d\langle\bh,\xi\rangle = d\tr(E\xi).
    \]
    The constraint set is then the space of configurations satisfying the Gauss constraint $\C = \{(A,E)\in\X \ \vert \ d_A E =0\}$. Denoting $\fGpoff = C^\infty(\pp\Sigma,\fg)$ and $\fg\hookrightarrow\fGpoff$ the constant functions, the flux space is
    \[
    \F \simeq 
        \begin{cases}
             \fGpoff^\vee & \text{$G$ semisimple}\\
             \{f\in\fGpoff^\vee \ | \ \langle f, \chi\rangle = 0 \ \forall \chi \in \fg\} & \text{$G$ Abelian}
        \end{cases}
    \]

    The constraint gauge group is the locally exponential, connected, Lie subgroup of $\G$ integrating, respectively,
    \[
    \fGo =\Ann(\F,\fG) =
    \begin{cases}
    \{\xi\in \fG\ |\ \xi\vert_{\pp\Sigma} = 0\} & \text{$G$ semisimple},\\
    \{\xi\in \fG\ |\ \exists\chi\in \fg, \ \xi\vert_{\pp\Sigma} = \chi\}& \text{$G$ Abelian}
    \end{cases}
    \]
    
    Furthermore, the flux gauge group $\Gred = \G/\Go$ is a central extension by a discrete group $\mathcal{K}$ of $C_0^\infty(\pp\Sigma,G)$ in the semisimple case and of $C_0^\infty(\pp\Sigma,G)/G$ in the Abelian case. In the particular case $\Sigma \simeq D^n$ (the $n$-disk) $\mathcal{K}$ is isomorphic to $\pi_n(G)$.

    The first-stage, constraint-reduction is locally symplectomorphic to
    \[
    \uCo = \C/\Go \simeq T(\Acal/\Go)\simeq_\loc T(\Acal/\G) \times \Gred \times \F.
    \]
    The fully-reduced phase space is a Weinstein space,\footnote{See \cite{marsden2000orbit} and \cite[Sections 6.6.12 and 6.6.18]{RatiuOrtega03}.} locally symplectomorphic to 
    \[
    \uuC \simeq_\loc T(\Acal/ \G) \times \F 
    \]
    and decomposes in the disjoint union of flux superselection sectors as:
    \[
    \uuC \simeq \bigsqcup_{\mathcal{O}_f\subset \F} \uuS_{[f]} 
    \qquad \text{with} \qquad \uuS_{[f]} \simeq_\loc T(\Acal/\G) \times \mathcal{O}_f.
    \]
    In particular, the second-stage reduction map locally reads:
    \[
    \underline{\pi}:\uCo \to \uuC , \qquad ([A], E_\rad, \underline{k}, f) \mapsto (A_\rad, E_\rad, \Ad^*(\underline{k})\cdot f),
    \]
    where $([A],E_\rad)$ denote variables on $T(\Acal/\G)$, with $E_\rad$ identified with a $E\in T^\vee_A\Acal$ such that $d_AE=0$ and $\mathsf{n}E=0$.
\end{maintheorem}

\begin{remark}[Cotangent reduction]
    The fact that in Yang--Mills theory one has $\uCo \simeq T(\Acal/\Go)$ and $\uuC \simeq_\loc T(\Acal/\G) \times \F$ is expected from the general theory of symplectic reduction for lifted actions. In finite dimensions, lifting $\G\circlearrowright \calPhi$ to a Hamiltonian action $\G\circlearrowright T^*\calPhi$ one has that, reduction at the zero level set of the momentum map yields (when smooth) $T^*\calPhi//_0\G \simeq T^*(\calPhi/\G)$. A similar result holds when taking symplectic reduction via orbits at different regular values of the momentum map, i.e. 
    \[
    T^*\calPhi//_{\mathcal{O}_f}\G \simeq_\loc T^*(\calPhi/\G)\times \mathcal{O}_f, \qquad T^*\calPhi//_{\mathcal{O}_f}\G\simeq T^*(\calPhi/\G)\times_{\calPhi/\G}\wt{\mathcal{O}}_f
    \]
    where $\wt{\mathcal{O}}_f \doteq (\calPhi\times \mathcal{O}_f)/\G \to \calPhi/\G$ \cite{marsden2000orbit,RatiuOrtega03}.
    Here, we find that one can naturally describe reduction via the tangent bundle, thus avoiding issues related to dualization on a (nonlocal) infinite dimensional space.
\end{remark}

\begin{remark}[Reducible configurations of YM theory]
    When $G$ is semisimple, the restriction to irreducible connections is required to satisfy the technical Assumption \ref{assA:isotropy}: i.e. regularity of the isotropy locus $\mathsf{I}_\rho =\{ (\phi,\xi)\in \X \times \fG \ \vert \ \rho(\phi,\xi) = 0\}$. In fact, Assumption \ref{assA:isotropy} is satisfied on every subset of connections with the same orbit type (that is with conjugate stabilizers), and not just by the irreducible ones (those with trivial stabilizer); however, not restricting to a given orbit type would yield to singular reduction, with $\Acal/\G$ (and arguably $\uuC$) a stratified manifold (see e.g.\ \cite{DiezHuebschmann-YMred,DiezPhD}).
    In the absence of corners ($\pp\Sigma=\emptyset$), this is a well-known fact since \cite{Arms1979, Arms1981, ArmsMarsdenMoncrief1981, ArmsMarsdenMoncrief:1982}.  Indeed, Theorem 1 of \cite{ArmsMarsdenMoncrief1981} shows that $\C$ is regular at all irreducible configurations, while the rest of the paper focuses on the nature of the singularities at the reducible configurations. Therefore, the fact that  $\C$ is the zero-level set of a momentum map for the action of $\Go$ (Theorem \ref{thm:fGo})---which in the presence of corners contains no nontrivial stabilizer due to the boundary conditions imposed on $\xio\in\fGo$---together with Theorem 1 of \cite{ArmsMarsdenMoncrief1981} suggests that \emph{in presence of corners} $\C$ is always smooth, and so should be $\uCo = \C/\Go$. (Arms et al.\ in \cite{ArmsMarsdenMoncrief1981} work within a Banach setting; reducible configurations and singular reduction are studied in the Fr\'echet setting, but in the absence of boundaries, in \cite{DiezHuebschmann-YMred, DiezRudolph2020, DiezPhD}.) To prove these facts in the Fr\'echet setting one can adapt the arguments of Section \ref{sec:runex-firststage}. This can technically be achieved by employing the Hodge--DeRahm decomposition with Dirichlet, rather than Neumann, boundary conditions (cf.\ Appendix \ref{app:Hodge}). We will leave this aspect of the application of our general framework to YM theory to subsequent work.
\end{remark}

Looking at corner configurations, we show the following:
\begin{maintheorem}[On- and off-shell corner data]\label{mainthm:Poisson}
Under Assumptions \ref{ass:setup} and \ref{ass:AnnIm},
let $\A=\X\times \fG$ be the action Lie algebroid for the Hamiltonian $\G$ space $(\X,\bom)$. Then,
\begin{enumerate}[label=(\roman*)]
    \item \label{thmitem:main2-algebroids} there is a Lie algebroid morphism
\[
\xymatrix{
\A \ar[d]_{\pi_\pp} \ar[r]^{\rho} & T\X \ar[d]^{\bd\pi_{\pp,\X}} \\
\A_\pp \ar[r]^{\rho_\pp} & T\X_\pp,
}
\]
where $\pi_\pp = \pi_{\pp,\X} \times\pi_{\pp,\fG}$ is a surjective submersion, and $\A_\pp=\X_\pp \times \fGpoff$ is a symplectic (action) Lie algebroid supported on $\pp\Sigma$ with anchor $\rho_\pp$, whose base $(\X_\pp,\Pi_\pp)$ is a (partial) Poisson manifold, such that $\Pi_\pp^\sharp = \rho_\pp$.
\item \label{thmitem:main2-diagram} there exists a commuting diagram of Poisson manifolds
\[
\xymatrix{
 \uCo \ar[d]_{\pi_{\pp,\uCo}}\ar[r]^{\underline{\pi}} & \uuC \ar[d]^{\underline{q}} \\
 \C_\pp \ar[r]^-{q_\pp} & \mathcal{B}
}
\]
where $\C_\pp \doteq \pi_{\pp,\X}(\C)$ is a Poisson submanifold of $(\X_\pp,\Pi_\pp)$, $\mathcal{B}\doteq\C_\pp/\Gpoff$ is the space of leaves of the symplectic foliation of $\C_\pp$ ($\mathcal{B}$ is equipped with trivial Poisson structure), and the maps $\underline{\pi}\colon \uCo \to \uuC$ and $q_\pp \colon \C_\pp \to \mathcal{B}$ are the respective group-action quotients by $\Gred$ and $\Gpoff$ respectively. 
We call $\mathcal{B}$ the space of on-shell superselections.
\item \label{thmitem:main2-ultralocal} If $\bom$ is ultralocal and $\rho$ is order 1 as a local map, then $\fGpoff$ is a local Lie algebra on $\pp\Sigma$.
\end{enumerate}
\end{maintheorem}

\begin{maintheorem}[Corner data for non-null YM theory]\label{mainthm:YMPoisson}
    The locally Hamiltonian data for YM theory (under the assumptions above) defines the symplectic Lie algebroid
    \[
    \A_\pp=\X_\pp\times\fGpoff \simeq \Omega^\mathrm{top}(\pp\Sigma, \fg) \times C^\infty(\pp\Sigma, \fg) \ni(E_\pp , \xi_\pp) = \iota^*_{\pp\Sigma}(E,\xi)
    \]
    with $\pi_\pp:\A \to \A_\pp, \ (A,E,\xi) \mapsto \iota^*_{\pp\Sigma}(E,\xi)$, and anchor map $\rho_\pp\colon \A_\pp\to T\X_\pp$ 
    \[
    \rho_\pp(E_\pp,\xi_\pp) = ([E_\pp,\xi_\pp],E_\pp)
    \]
    over the Poisson manifold $(\X_\pp,\Pi_\pp)$, with
    \[
    \Pi_\pp = \tr\left( E_\pp \left[ \frac{\delta}{\delta E_\pp} \stackrel{\wedge}{,} \frac{\delta}{\delta E_\pp} \right]\right),
    \]
    and $\fGpoff$ is a local Lie algebra (i.e.\ YM satisfies the assumptions of Theorem \ref{mainthm:Poisson} item \ref{thmitem:main2-ultralocal}).
    Moreover, the Poisson submanifold of on-shell corner data $\C_\pp$ is
    \[
    \C_\pp = 
        \begin{cases}
             \X_\pp & \text{$G$ semisimple}\\
             \{ E_\pp \in \X_\pp \ \vert \ \int_{\pp\Sigma} \tr(\chi E_\pp) = 0 \ \forall \chi \in \fg\} & \text{$G$ Abelian}
        \end{cases}
    \]
    and it is isomorphic to the flux space, $\C_\pp\simeq\F$.
    Finally, restricting to irreducible configurations and hence using the parametrizations 
    \[
    \uCo \simeq_\loc T^\sigma(\Acal/\G) \times \Gred \times \F \ni (A_\rad,E_\rad, \underline{k}, f), \qquad \uuC \simeq_\loc T^\sigma(\Acal/\G) \times \F,
    \]
    we have $\underline{\pi}(A_\rad, E_\rad, \underline{k}, f) = (A_\rad, E_\rad, \Ad^*(\underline{k})\cdot f)$, and $\pi_{\pp,\uCo}(A_\rad, E_\rad, \underline{k}, f) = (A_\rad, E_\rad, f)$ where we also used the identifications $\C_\pp \simeq \F$, and
    \[
    q_\pp( A_\rad, E_\rad,f) = \mathcal{O}_f = \mathcal{O}_{\Ad^*(k)\cdot f} = \underline{q}(A_\rad,E_\rad, \Ad^*(k)\cdot f) \in \mathcal{B}.
    \]
\end{maintheorem}

In Section \ref{sec:prelim} we outline preliminary material to define local field theories, while Section \ref{sec:localfieldtheory} is devoted to setting up the main definitions we will use throughout.

The core of our paper starts in Section \ref{sec:constraintsandfluxes}, where the constraint and flux gauge groups are constructed, and then Section \ref{sec:reduction}, where symplectic reduction in stages w.r.t.\ the constraint and flux groups is performed, leading to Theorem \ref{mainthm:redsummary} and Theorem \ref{mainthm:YMon}.

Finally, in Section \ref{sec:cornerdata} we analyze the off-shell corner structure of gauge field theory, leading to Theorem \ref{mainthm:Poisson} and Theorem \ref{mainthm:YMPoisson}. 

We analyze concrete examples in the sections devoted to Yang--Mills theory used as a running example (Sections \ref{sec:runex-setup}, \ref{sec:runex-decomposition}, \ref{sec:runex-fluxannihilators}, \ref{sec:runex-fluxgaugegroup}, \ref{sec:runex-firststage}, \ref{sec:runex-secondstage} and \ref{sec:runex-cornerdata}; as well as \ref{sec:thetaQCD}), and in those devoted to Chern--Simons and $BF$ theory (Sections \ref{sec:ex-ChernSimons} and \ref{sec:BFtheory}).

\subsection*{Literature comparison}

Although the problem of symplectic reduction for gauge field theory has gathered significant attention in the past, the philosophy and overall approach presented here are different.

Our aim is to outline a unified, model-independent framework, that can shed light on the interplay between corners, gauge symmetry and symplectic reduction.
In particular, we do so without imposing any conditions on the field configurations at $\pp\Sigma$. Of course, in the explicit application of our framework to specific model theories, such as Yang--Mills, Chern--Simons, or $BF$ theory, we will find a significant overlap of the techniques we employ with the existing literature---of which we now attempt a brief, and necessarily partial, survey.

Electrogmagnetism and de Rham--Hodge theory have profound ties that date back to the origins of the latter. So it is not surprising that a non-Abelian generalization of the de Rham--Hodge decomposition (see Appendix \ref{app:Hodge}) was employed to study symplectic reduction of Yang--Mills theory since the first works on the topic \cite{Singer1978,Singer1981,NarasimhanRamadas79,Arms1981,ArmsMarsdenMoncrief1981,MitterViallet1981,BabelonViallet}.

The analysis of singularities of reduction at YM configurations with nontrivial stabilisers appeared already in \cite{Arms1981,ArmsMarsdenMoncrief1981}.\footnote{For the general theory of singular reduction in finite dimension, see e.g. \cite{ArmsCushmanGotay, sjamaar1991, marsden2007stages}. For application to General Relativity, which in fact motivated most of this research, we recall in particular the names of Arms, Ebin, Fischer, Isenberg, Marsden and Moncrief. Much of these works relies on seminal results by Palais \cite{palais1957,palais1961}, and also \cite{palaisterng1988critical}.} (We will in fact avoid discussing this issue and refer the reader to these references, and to \cite{Rudolph_2002} for a review.) More recently, reduction in the smooth Fr\'echet setting for Yang--Mills theory has been studied in \cite{DiezHuebschmann-YMred, DiezRudolph2020, DiezRudolph2022} and summarized in the doctoral dissertation \cite{DiezPhD} which we will use as our main reference on the topic.\footnote{But see also \cite{AbbatiCirelliMania, AbbatiCirelliManiaMichor}.} We adopt these techniques to present a concrete example of our general construction.

In \cite{Marini92} as well as \cite{SniatyckiSchwarz94,SniatyckiSchwarzBates1996} the Hamiltonian formulation of Yang--Mills dynamics and symplectic reduction were studied in the presence of boundaries. However, non-generic boundary conditions were also imposed so to have an autonomous/well-posed dynamical problem. In our language, these boundary conditions translate to $h=0$, which clearly obscures the corner structure.

During the concluding stages of this work, we learned about \cite{MeinrenkenWoodward} where a symplectic analysis of the moduli space of flat connections on Riemann surfaces with boundaries (Chern--Simons theory), and their gluing, is carried out in great detail. That work follows a 2-stage reduction approach in line with the one investigated in this work. Among other results, it shows that Verlinde formula \cite{VerlindeFormula} can be derived as an application of reduction by stages. Earlier results on the moduli space of flat connections on manifolds with boundary are due to Donaldson \cite{Donaldsonboundary}.

Our key construction of on-shell flux annihilators---which leads to the constraint gauge group (Definition \ref{def:fGo})---is akin to the annihilating sheaf construction of \cite{LinSjamaar}. Their construction naturally applies to reduction by stages, although it is not an explicit application considered in their work, and one recovers the annihilator bundle we discuss in this work.

In \cite{LysovReduction}, similar ideas aiming at the characterization of reduced phase spaces as two-stage symplectic reduction are studied mainly in the context of $d$-dimensional $BF$ theory (which we study in Section \ref{sec:BFtheory}). The way that symplectic reduction by stages is employed there is---to the best of our understanding---different from ours.

The relationship between bulk and boundary data for local field theories that we employ here is inspired by the constructions  proposed in \cite{CMR1,CMR2}, and the local version described in \cite{MSW}. In fact, the construction of corner data we outline in section 6 is a ``degree-0 version'' of the BV-BFV procedure of \cite{CMRCorfu,CMR1}, itself an extension of the procedure outlined in \cite{KT1979}. 

It was shown in \cite{CanepaCattaneo} that the BV data  for a local field theory on a manifold with corners is expected to associate a Poisson-$\infty$ structure  to its codimension 2 stratum. This is an embodiment of the general procedure outlined in \cite{CMR1}.
In this work we give a ``classical'' explanation for that, and show that this indeed is expected to be the consequence of the second stage symplectic reduction.

Finally, the application of our framework to Yang--Mills theory recovers some of the results obtained through more ad-hoc techniques in a series of works by one of the authors and H.\ Gomes \cite{RielloGomesHopfmuller, RielloGomes, RielloSciPost}, where a gluing theorem was also discussed.\footnote{See also \cite{RielloEdge} and \cite{GomesPhD}.}

Here, in addition to that we provide a general, theory-independent, geometric picture for Hamiltonian reduction by stages for field theory on manifolds with corners in the Fr\'echet setting.

\section*{Acknowledgements}
We would like to thank C.\ Blohmann for sharing an early version of \cite{BlohmannLFT} and A.S.\ Cattaneo for helpful exchange during the development of this project. We also thank T. Diez for clarifying to us an argument related to the symplectic reduction for Yang--Mills theory, and for providing feedback on a previous version of this manuscript. Finally, we thank P.\ Mnev and A.\ Alekseev for pointing out relevant literature to us.

\section{Preliminaries}\label{sec:prelim}

Our investigation of the symplectic geometry of gauge theories in the presence of corners relies on the \emph{locality} properties of the underlying physical theory to single out the behavior of fields at the corner. This is mathematically codified in the formalism of local forms on a space of fields $\X$, defined as the space of smooth sections of some (finite dimensional) fibre bundle $E\to \Sigma$, i.e.\ $\X=\Gamma(\Sigma,E)$. 
To feed the notion of corner locality into the symplectic reduction procedure, we introduce \emph{local} notions of Hamiltonian action and momentum maps. 

However, since the constraint set $\C\subset\X$ is given by the space of solutions of some differential equations, in general it does not admit a local description itself, i.e.\ there is no bundle on $\Sigma$ of which $\C$ is the space of sections. A further source of nonlocality (and possibly singularity) emerges when quotienting out the level set of the momentum map by the group action.

Because locality cannot be preserved throughout the symplectic reduction procedure, we need to view local forms within a more general framework for infinite-dimensional geometry.
We choose to work with infinite-dimensional manifolds modelled on locally convex vector spaces, tipically Fr\'echet and their duals.

In this section, we first lay out a series of considerations on the smooth geometry of the infinite-dimensional spaces underlying this work, mostly following the convenient setting of \cite{kriegl1997convenient} (see also \cite{FrolicherKriegl}).
We will then turn to a set of necessary preliminaries by reporting some useful notions in the context of local forms and local vector fields on the space of sections of a fibre bundle, which will be regularly used in the main text. 
For background material, see Appendix \ref{app:localforms} and also \cite{Takens,Zuckerman, Anderson:1989,BlohmannLFT}.

\subsection{Infinite-dimensional geometry}\label{sec:infinitedim}
Field theory unavoidably requires working with infinite-dimensional vector spaces, and smooth manifolds modeled thereon. There are several ways to generalise finite dimensional calculus and differential geometry to a field-theoretic scenario, and the ideal setup for our purposes is that of smooth manifolds modelled on locally convex vector spaces, and in particular Fr\'echet spaces.
\footnote{Another interesting way of looking at smoothness for space of sections of fibre bundles is through diffeologies \cite{Iglesias-Zemmour:Diffeology}. The categories of Fr\'echet manifolds and of diffeological spaces are related by \cite[Theorem 3.1.1]{Losik-DiffeologicalFrechet} (see also \cite[Lem. A.1.7]{waldorfI}).} These smooth manifolds can be described within the ``convenient setting'' of \cite{kriegl1997convenient}, where the $c^\infty$ topology is used (see in particular Sections 1, 2, 33 and 48 of \cite{kriegl1997convenient}). This is defined as the final topology with respect to all smooth \emph{curves}, and therefore all smooth mapping are continuous for it.\footnote{A mapping is said to be smooth if it maps smooth \emph{curves} into smooth curves. 
Depending on the topology of the underlying (infinite-dimensional) spaces, a smooth mapping can be continuous or not.} We expand on the details of this construction in Appendix \ref{app:infinitegeometry}, where we also provide the notion of differential forms and vector fields used throughout.

Since we will be working with group actions on smooth Fr\'echet manifolds, these considerations will also apply to the infinite-dimensional Lie groups and algebras appearing in this work. 

One central question in infinite-dimensional Lie theory concerns whether a Lie algebra encodes the infinitesimal structure around the identity of some Lie group---a property called ``enlargeability'' (see Definition \ref{def:locexpenlargeable})---which becomes relevant when manipulations of infinite-dimensional Lie algebras generate new ones that, a priori, need not be enlargeable. 
We refer to \cite{Neeb-locallyconvexgroups}  for a survey of results on this topic.\footnote{Note that many of those results, however, do not directly apply to the group of diffeomorphisms of a compact surface; this is because of the well-known fact that their exponential map fails to be surjective, even arbitrarily close to the identity \cite{Omori78}. Diffeomorphisms however lie outside of the scope of this work.}

In the remainder of this paper, we will assume that all relevant enlargements of Lie algebras are unhindered and will not look for general conditions under which this happens.
This said, using results of \cite{WockelPhD}, one can show that this assumption is justified in the gauge theories we will consider: Yang--Mills, Chern--Simons, and $BF$ theory, which are defined on a principal $G$-bundle, with $G$ finite dimensional (but not necessarily compact).

Another important question relates to the nature of ``reduced spaces'', obtained as quotients of the level set of a momentum map, modulo the action of an infinite-dimensional Lie group.
Once again, we will not look for sufficient conditions guaranteeing their smoothness and we will simply assume it in general---see the end of Section \ref{sec:intro} for references on this topic.

\begin{remark}[Concrete models]\label{rmk:nuclearFrechet}
When looking at Lie algebras as well as at local models of spaces of fields, we will be chiefly interested in locally convex vector  spaces $\calW$ that have the additional useful structure of a \emph{nuclear, Fr\'echet} vector space (NF). 
In particular, let $\calW\doteq\Gamma(\Sigma,W)$ for $W\to\Sigma$ a finite dimensional vector bundle over a manifold $\Sigma$. Then $\calW$ admits the structure of an NF vector space \cite[Section 30.1, unnumbered Proposition]{kriegl1997convenient}.\footnote{Observe that this is true even when $\Sigma$ and $W$ themselves are infinite-dimensional but (modeled on) nuclear Fr\'echet vector spaces. For us this is relevant e.g. for the construction of the space we will call $\A$.} A summary of properties of NF vector spaces is given in Remark \ref{rmk:NFVS}.

Additionally, spaces of sections of vector bundles over a \emph{compact} manifold are also ``tame'' as Fr\'echet vector spaces (see \cite{Hamilton-NashMoser} for a discussion), and spaces of sections of fibre bundles with tame local models will give rise to tame Fr\'echet manifolds. The main interest of the tame property is the validity of a version of the inverse function theorem due to Nash and Moser \cite{Nash-Imbedding,Moser-NMThm}. This can be used e.g.\ to establish the existence of slices for Hamiltonian actions and thus the smoothness of $\C$ and $\uuC$ (see \cite{DiezPhD}).
\end{remark}

We now turn to infinite-dimensional manifolds which are symplectic or Poisson. 

We denote by $\mathcal{W}^*$ the strong dual of a locally convex vector space $\mathcal{W}$.
A discussion (and definition) of dual spaces in the relevant infinite-dimensional setting is given in Remark \ref{rmk:strong/bornduals}, in particular the topology of the cotangent bundle $T^*\calPhi$ is discussed in Remark \ref{rmk:cotangent}. See also Remark \ref{rmk:cotangentreflexive}.

\begin{definition}[Symplectic manifold]\label{def:symplecticmanifolds}
A skew-symmetric, bilinear form $\omega$ on a vector space $\calW$ is called
\begin{enumerate}
    \item (\emph{weakly}) \emph{symplectic} if $\omega^\flat\colon \calW\to \dual{\calW}$ is injective
    \item \emph{strongly symplectic} if $\omega^\flat$ is bijective.
\end{enumerate}
Similarly, a  (weakly/strongly) symplectic form on a manifold $\calPhi$ is a closed $2$-form $\omega\in\Omega^2(\calPhi)$, $\bd \omega = 0$, such that the induced morphism $\omega^\flat\colon T\calPhi \to T^*\calPhi$ is injective/bijective.
\end{definition}

\begin{remark}
We generally do not consider strongly symplectic forms, and thus drop the ``weak'' characterization: henceforth a symplectic structure is weak if not otherwise stated.
\end{remark}

Let $T^{\mathsf{p}}\calPhi \hookrightarrow {T^*\calPhi}$ be any subbundle of the cotangent bundle of a given smooth manifold $\calPhi$ (here, $\bullet^\mathsf{p}$ is for ``partial''), and look at those functions $f\in C^\infty(\calPhi)$ such that there exists a section $s_f$ of $T^{\mathsf{p}}\calPhi$ for which there exists a commuting diagram of the form:
\begin{equation}\label{eq:partialdiagram}
    \xymatrix{
    T^{\mathsf{p}}\calPhi \;\ar@{^(->}[rr] && {T^*\calPhi}\\
    & \ar[ul]^{s_f} \calPhi \ar[ur]_{\bd f}
    }
\end{equation}
This suggests the following definition:

\begin{definition}[Partial Poisson structure \cite{PelletierCabau}]\label{def:partialpoisson}
Let $T^\mathsf{p}\calPhi\hookrightarrow T^*\calPhi$ be a partial cotangent bundle. The space of \emph{$\mathsf{p}$-functions} $C_{\mathsf{p}}^\infty(\calPhi)$ is the subset of $C^\infty(\calPhi)$ for which the diagram \eqref{eq:partialdiagram} commutes.
A \emph{partial Poisson structure} associated to $T^\mathsf{p}\calPhi$ is a Poisson algebra on $C_\mathsf{p}(\calPhi)$ encoded in a skew-symmetric map, called bracket,
\[
    \{\cdot,\cdot\}^{\mathsf{p}} \colon C_{\mathsf{p}}^\infty(\calPhi)\times C_{\mathsf{p}}^\infty(\calPhi) \to C_{\mathsf{p}}^\infty(\calPhi)
\]
satisfying the Leibniz and Jacobi identities.
\end{definition}

\begin{definition}[Hamiltonian functions]\label{def:Hamfunctions}
Let $(\calPhi,\omega)$ be a  symplectic manifold. If $T^{\mathsf{p}}\calPhi\equiv T^\omega\calPhi \doteq \Im(\omega^\flat)\subset T^*\calPhi$ is a subbundle\footnote{Note that the image of $\omega^\flat$ might not be a subbbundle, e.g. \cite{MarsdenDarboux}.}, we call it the symplectic partial cotangent bundle of $(\calPhi,\omega)$. The \emph{space of Hamiltonian functions} $C_{\omega}^\infty(\calPhi)$ is the space of $\mathsf{p}$-functions associated to $T^\omega\calPhi$.
\end{definition}

\begin{remark}
When $(\calPhi,\omega)$ is strongly symplectic one finds $T^*\calPhi=T^\omega\calPhi$ and all smooth functions are $\mathsf{p}$-functions, and thus Hamiltonian.
\end{remark}

\begin{lemma}\label{lem:weaksymplecticpoisson}
Let $(\calPhi,\omega)$ be a symplectic manifold. The space of Hamiltonian functions $C_\omega^\infty(\calPhi)$ admits a unique Poisson algebra structure compatible with $\omega$.
\end{lemma}
\begin{proof}
Standard (see \cite[Example 2.2.4]{PelletierCabau} and \cite[Lemma 48.6]{kriegl1997convenient}).
\end{proof}

\begin{remark}
The notion of partial Poisson structure is related to that of  \emph{weak Poisson structure} as defined in \cite{NeebSahlmannThiemann} (see \cite[Example 2.2.5]{PelletierCabau}), or in\footnote{Both the authors of \cite{NeebSahlmannThiemann} and \cite{CabreraGualiteriMeinrenken} use the wording ``weak Poisson structure'' albeit with different meaning. We do not know of a general relation between the two. In good cases, they both overlap with the notion of partial Poisson of \cite{PelletierCabau}.} \cite{CabreraGualiteriMeinrenken} and, in the linear case, to that of locally convex Poisson vector space \cite{glockner2007applications}. Important examples of weak Poisson (locally convex) manifolds are given by topological duals\footnote{The ``size'' of the algebra of $\mathsf{p}$-functions, i.e.\ where the weak Poisson structure is defined, depends on the choice of topology on the topological dual. See \cite[Example 2.14]{NeebSahlmannThiemann}.} of locally convex Lie algebras, as discussed in \cite{NeebSahlmannThiemann}. We will not adopt the same terminology, in favour of the subbundle description outlined above. However, in the examples of Lie algebras we will consider, the three definitions will largely overlap. We discuss the Poisson structure on duals of Lie algebras in Section \ref{sec:dualPoisson}.
\end{remark}

For field theoretical purposes, $\calPhi=T^*\calN$ is not always the appropriate space to consider as a phase space. 
For example, distributional momenta for smooth configurations would then arise when taking the strong dual of a space of functions. A natural alternative is available when $\calN$ is modelled on the space of sections of a finite dimensional vector bundle over $\Sigma$: 

\begin{definition}[Densitized cotangent space]\label{def:densitizedDual}
Let $V\to\Sigma$ be a finite dimensional vector bundle, and $\calV \doteq \Gamma(\Sigma,V)$. The \emph{densitized dual} $\calV^\vee$ of $\calV$ is the vector space
\begin{equation}\label{eq:fibrewisedual}
    \calV^\vee \doteq \Gamma(\Sigma, \mathrm{Dens}(\Sigma)\otimes_\Sigma V^*),
\end{equation}
equipped with its nuclear Fr\'echet topology. The \emph{densitized cotangent space} is the (nuclear) Fr\'echet vector space
\[
\calW = T^\vee \calV \doteq \calV \oplus \calV^\vee,
\]
equipped with the following (canonical) skew-symmetric bilinear form $\omega$
\[
    \omega\big( (v_1,v_1^\vee),(v_2,v_2^\vee)\big)\doteq 
    \int_\Sigma  v_1^\vee(x)(v_2(x)) - v_2^\vee(x)(v_1(x)),
\]
where $v^\vee(x)(v(x))$ denotes the fibrewise pairing of sections of $V$ and $V^*$. 
\end{definition}
Notice that $\omega^\flat : \calW \to \dual{\calW}$ is manifestly injective, but it is not surjective, even when $\calV$ is reflexive, since $\calV^\vee \subsetneq \dual{\calV}$. An obvious generalization of this construction to the case of a smooth (Fr\'echet) manifold $\calN$ locally modeled on $\calV$ leads to the definition of the (Fr\'echet) densitized cotangent bundle $(\calPhi = T^\vee \calN,\omega)$ as a symplectic space. Finally, let us introduce the notion of annihilator of a subset of a vector space:

\begin{definition}[Annihilators]\label{def:annihilators}
        Let $\mathcal{W}$ be a nuclear Fr\'echet vector space, or the strong dual of a nuclear Fréchet space. Moreover, let $\mathcal{X}\subset \mathcal{W}$ and $\mathcal{Y}\subset \mathcal{W}^*$ be subsets. The \emph{annihilator of $\mathcal{X}$ in $\mathcal{Y}$} is the set 
        \begin{align*}
        \Ann(\mathcal{X},\mathcal{Y}) &= \{y\in \mathcal{Y}\, : \, \langle y,x\rangle = 0\  \forall x\in \mathcal{X} \} \subset \mathcal{W}^*.
        \qedhere
        \end{align*}
\end{definition}

\begin{lemma}\label{lem:doubleannihilator}
    Let $\mathcal{W}$ be as above, and $\mathcal{X}\subset \mathcal{W}$ be a closed vector subspace, then 
    \[
    \Ann(\Ann(\mathcal{X},\mathcal{W}^*),\mathcal{W}) =\mathcal{X}.
    \]
\end{lemma}

\begin{proof}
    Let $\mathcal{X}\subset \mathcal{W}$ and $\mathcal{Y}=\Ann(\mathcal{X},\mathcal{W}^*)$. From
    \[
    \Ann(\mathcal{X},\mathcal{W}^*) \simeq \left(\mathcal{W}/\mathcal{X}\right)^*, \qquad \Ann(\mathcal{Y},\mathcal{W}) \simeq \left(\mathcal{W}^*/\mathcal{Y}\right)^*
    \]
    we obtain 
    \[
    \Ann(\Ann(\mathcal{X},\mathcal{W}^*),\mathcal{W}) \simeq \left( \mathcal{W}^* / \left(\mathcal{W}/\mathcal{X}\right)^*\right)^*
    \]
    Dualising the short exact sequence 
    \[
    \mathcal{X} \to \mathcal{W} \to \mathcal{W}/\mathcal{X} 
    \]
    we conclude that $\mathcal{X}^* \simeq \mathcal{W}^* / \left(\mathcal{W}/\mathcal{X}\right)^*$. 
    Hence, since both nuclear Fre\'echet vector spaces and their strong duals, as well as their closed subspaces, are reflexive (\cite[Rmk. 6.5]{kriegl1997convenient}), one finds
    \[
    \Ann(\Ann(\mathcal{X},\mathcal{W}^*),\mathcal{W}) \simeq \left(\mathcal{X}^*\right)^* \simeq \mathcal{X}^{**} = \mathcal{X}.\qedhere
    \]
\end{proof}

\subsection{Local forms}
This section builds on standard definitions, recalled in Appendix \ref{app:localforms}. Throughout, $\calPhi=\Gamma(\Sigma,F)$ denotes the space of sections of a given finite dimensional fibre bundle $F\to\Sigma$ over a compact, orientable, manifold $\Sigma$ possibly with boundaries. $\calPhi$ is then a smooth, tame, nuclear Fr\'echet manifold.

With reference to the variational bicomplex associated to $F$, we denote the horizontal differential by $d_{H}$ and the vertical one by $d_V$. 
The infinite jet evaluation map $j^\infty\colon \Sigma \times \calPhi \to J^\infty F$ is surjective on open sets if $\Sigma$ is finite dimensional \cite[Section 41.1]{kriegl1997convenient}, and it is globally surjective iff $F$ has a global section \cite[Prop. 3.1.14]{BlohmannLFT}.
Whenever $j^\infty$ is surjective, the pullback along this map allows one to match the bicomplex structure on $\left(\Omega^{\bullet,\bullet}(J^\infty F),d_H,d_V\right)$ with that on $\left(\Omega^{\bullet,\bullet}(\Sigma\times\calPhi),d, \bd\right)$ (Definition \ref{def:localforms}):
\begin{equation}
\label{eq:localforms-jinfty}
    \left(\oloc^{\bullet,\bullet}(\Sigma \times \calPhi), d, \bd\right) = (j^\infty)^*\left(\Omega^{\bullet,\bullet}(J^\infty F),d_H,d_V\right)
\end{equation}
where $\bd$ denotes the de Rham differential on smooth infinite-dimensional manifolds (here restricted to the image of $(j^\infty)^*$).

Henceforth, we will focus solely on the $(d,\bd)$-bicomplex structure over $\Sigma \times \calPhi$.

Vector fields on $\Sigma \times \calPhi$ are distinguished in notation: we have $X\in\mathfrak{X}(\Sigma)$, with $i_X$ and $L_X$ standing for the contraction and Lie derivative of local forms by $X$, and we use double-struck symbols like $\mathbb{X}\in\Xloc(\calPhi)$ for local vector fields\footnote{These correspond to ``evolutionary vector fields'' on the infinite jet bundle \cite{Anderson:1989,Saunders:1989,AndersonTorre}.} on $\calPhi=\Gamma(\Sigma,F)$. Similarly we denote $\bi_{\mathbb{X}}$ and $\L_{\mathbb X}$  the contraction and Lie derivative of local forms by the local vector field on $\calPhi$. 

We use bold-face symbols to denote local forms in $\oloc^{\bullet,\bullet}(\Sigma\times\calPhi)$, while forms over $\calPhi$ obtained by integrating $(\text{top},\bullet)$-local forms will be denoted with the corresponding non-bold-face symbol, as in
\[
\boldsymbol{\nu}\in\oloc^{\text{top},\bullet}(\Sigma\times\calPhi),
\qquad
\nu = \int_\Sigma \boldsymbol{\nu} \in \iloc^\bullet(\calPhi),
\]

\begin{definition}[Integrated local forms]\label{def:iloc}
The space of \emph{integrated local forms} on $\calPhi$, denoted $\iloc^\bullet(\calPhi)$, is the image of the integration map $\int_\Sigma: \oloc^{\text{top},\bullet}(\Sigma\times\calPhi)\to \Omega^\bullet(\calPhi)$. 
\end{definition}

\begin{definition}[Ultralocal forms]\label{def:ultralocal}
A local $(p,q)$-form $\boldsymbol{\nu}$ is called \emph{ultralocal} iff it descends to a $(p,q)$-form on $J^0F\equiv F\to\Sigma$, i.e.\ iff it does not depend on jets of sections.
\end{definition}

\begin{remark}\label{rmk:sourcetarget}
Let $\calN=\Gamma(\Sigma,H)$ the space of sections of a fibre bundle $H\to\Sigma$. A useful property of local forms in $\oloc^{r,p,q}(\Sigma\times\calPhi\times\calN)$ is that they can equivalently be seen as local $q$-forms on $\calN$ with values in $\oloc^{r,p}(\Sigma\times \calPhi)$ or as local $p$-forms on $\calPhi$ with values in $\oloc^{r,q}(\Sigma\times \calN)$. 
\end{remark}

\begin{definition}[Dual valued forms and vector fields]\label{def:dualvaluednew}
Let $W\to\Sigma$ be a finite dimensional fibre bundle, and $\calW\doteq\Gamma(\Sigma,W)$. A family of forms and a family of vector fields parametrized by $\calW$ are, respectively, smooth maps
\[
\boldsymbol{\nu}\in C^\infty( \calW , \oloc^{\bullet,\bullet}(\Sigma\times \calPhi)),
\qquad 
\mathbb{X}\in C^\infty( \calW, \mathfrak{X}(\calPhi)).
\]
A family of forms $\boldsymbol{\nu}:\calW\to\oloc^{\bullet,\bullet}(\Sigma\times\calPhi)$ is
\begin{enumerate}[label=(\roman*)]
    \item \emph{local}, iff it can be identified with a local form in $\oloc^{\bullet,\bullet,0}(\Sigma\times\calPhi\times\calW)$, i.e.\ iff there exists $\wt{\boldsymbol\nu}\in\Omega^{\bullet,\bullet}(J^\infty(F\times_\Sigma W))$ such that $\boldsymbol{\nu}=(j^\infty)^*\wt{\boldsymbol\nu}$;
    \item \emph{local dual-valued}, iff $W\to\Sigma$ is a vector bundle and $\boldsymbol{\nu}$ is linear over $\calW$, and local as a family of forms. Then, we write\footnote{Bounded (multi-)linear maps between convenient vector spaces are local if they are local as smooth maps, we denote them by $L_\loc$.}
\[
\boldsymbol{\nu} \in \oloc^{\bullet,\bullet}(\Sigma\times\calPhi,\dual{\calW}) \doteq L_\loc(\calW,\oloc^{\bullet,\bullet}(\Sigma\times\calPhi)) \subset \oloc^{\bullet,\bullet,0}(\Sigma\times\calPhi\times\calW).
\]
\end{enumerate}
A family of vector fields $\mathbb{X}:\calW\to\mathfrak{X}_\loc(\calPhi)$ is
\begin{enumerate}[label=(\roman*)]
    \item \emph{local}, iff it is specified by a map $\Gamma(\Sigma,F\times_\Sigma W) \to T\calPhi$ that is local as a map of sections of fibre bundles;
    \item \emph{dual-valued}, iff $W\to\Sigma$ is a vector bundle and $\mathbb{X}$ is linear over $\calW$, and local as a family of vector fields. We write, $\mathbb{X} \in \Xloc(\calPhi,\dual{\calW})\doteq L_\loc(\calW,\mathfrak{X}(\calPhi))$.
    \qedhere
\end{enumerate}
\end{definition}

\begin{proposition}[Prop 6.3.6 \cite{BlohmannLFT}]
Let $\boldsymbol{\nu}$ and $\mathbb{X}$ be local families of forms and vector fields parametrized by the same $\calW$. Then, the contraction 
\[
\bi_{\mathbb{X}}\boldsymbol{\nu} \colon \calW \to \oloc^{\bullet,\bullet -1}(\Sigma \times \calPhi)
\]
is also a local family of forms parametrized by $\calW$.
\end{proposition}

The element of $\oloc^{\bullet,\bullet}(\Sigma\times\calPhi)$ obtained by evaluation of the dual-valued local form $\boldsymbol{\nu}$ at $w\in\calW$ will be denoted by angled bracket:
\[
\langle \boldsymbol{\nu} , w\rangle \equiv \boldsymbol{\nu} (w).
\qedhere
\]

In the course of this article we will encounter only one dual-valued vector field, but it will be ubiquitous. It will be denoted by $\rho$---as customary for actions and anchor maps (see Section \ref{sec:LieAlgPict}).

\begin{remark}\label{rmk:approximationproperty}
Let $E$ be a vector space that possesses the bornological approximation property, which guarantees that $E^*\otimes_\beta E \subset L(E,E)$ is a dense subset, where $\otimes_\beta$ denotes the bornological tensor product.\footnote{This is the tensor product seen as a topological vector space with the finest locally convex topology such that $E\times F\to E\otimes_\beta F$ is bounded, see \cite[Section 5.7]{kriegl1997convenient}.} For $E,F$ metrizable spaces, in particular Fr\'echet and their strong duals, we have $E\otimes_\pi F = E\otimes_\beta F$ where $\otimes_\pi$ is the ``usual'' projective tensor product for locally convex spaces, which allows us to identify $\otimes_\pi = \otimes_\beta \equiv \otimes$. (See Remark \ref{rmk:NFVS}.) Since spaces of sections of vector bundles have the bornological approximation property \cite[Theorem 6.14]{kriegl1997convenient} and they are nuclear Fr\'echet spaces, we can look at the identities\footnote{We thank S.\ D'Alesio for this argument.} (the hat standing for completion):
\begin{multline*}
C^\infty(E,L(F,G)) \simeq C^\infty(E)\hat\otimes L(F,G) \simeq C^\infty(E)\hat\otimes F^*\hat\otimes G \\
\simeq F^*\hat\otimes C^\infty(E)\hat\otimes G \simeq F^*\hat\otimes C^\infty(E,G) \simeq L(F,C^\infty(E,G))
\end{multline*}
and patching local models together we gather that, in particular:
\[
\Omega^{0,0}(\Sigma \times \calPhi, \calW^*) \simeq L(\calW, \Omega^{0,0}(\Sigma \times \calPhi)),
\]
and similarly for $p,q$ forms.
Hence, the space of dual-valued \emph{local} forms of Definition \ref{def:dualvaluednew} corresponds to the space of ``$\calW^*$-valued forms in $\oloc^{p,q}(\Sigma\times \calPhi)$''.

Below, we will use this fact to justify the equivalence between the momentum and comomentum pictures. 
\end{remark}

\begin{definition}[Order--$k$ dual-valued local forms]\label{def:order-kduals}
Let  $\boldsymbol{\nu} \in \oloc^{\bullet,q}(\Sigma\times\calPhi,\dual{\calW})$ be a local dual-valued form. We say that $\boldsymbol{\nu}$ is  \emph{order--$k$} iff the associated local family of forms 
\[
    \bi_{\mathbb{X}_q} \cdots \bi_{\mathbb{X}_1} \boldsymbol{\nu} \colon \calW \to \Omega^{\bullet}(\Sigma)
\]
is of order $k$ as a local map (Definition \ref{def:localmaps}) for all $\mathbb{X}_{i=1, \cdots,q}\in\Xloc(\calPhi)$. Then, we write
\[
\boldsymbol{\nu}\in\oloc^{\bullet,q}(\Sigma\times\calPhi,\dual{\calW}_{(k)}).
\qedhere
\]
\end{definition}

Note that,  although $\calW$ is a $C^\infty(\Sigma)$-module, dual-valued forms are generally only $\mathbb{R}$-linear in their $\calW$-argument; in fact, only dual-valued forms of order-$0$ are also $C^\infty(\Sigma)$-linear. Indeed, order-0 forms in $\oloc^{\bullet,\bullet}(\Sigma\times\calPhi,\dual{\calW}_{(0)})$ are ultralocal when seen as local forms on $\calW$ (with values in $\oloc^{\bullet,\bullet}(\Sigma\times\calPhi)$). 

Finally, upon choosing a volume element on $\Sigma$, we notice that order-$0$, dual-valued, $(\text{top},p)$-local forms on $\Sigma\times\calPhi$ can be understood as local $p$-forms over $\calPhi$ with value in the fibrewise dual, i.e.
\[
\oloc^\text{top}(\Sigma,\dual{\calW}_{(0)})\simeq \calW^\vee\doteq\Gamma(\Sigma,\mathrm{Dens}(\Sigma)\otimes_\Sigma W^*).
\]

To emphasise the enhanced linearity property of elements of order-0 local forms we introduce the following definitions:

\begin{definition}[Local and fibrewise duals]\label{def:dual*}
The \emph{local dual} $\mathcal W_\mathrm{loc}^*$ of $\calW$ is the space of linear functionals in $\iloc^0(\calW)$,
\[
    \dual{\calW}_\mathrm{loc}=  \{\alpha\in \iloc^0(\calW)\,:\, \alpha \text{ is $\mathbb{R}$-linear}\}\subset\iloc^0(\calW),
\]
and the \emph{fibrewise dual} $\dual{\calW}_\text{fbr}$ of $\calW$, is the space of linear functionals in $\iloc^0(\calW)$ which are order-0 as local maps $\calW\to \mathbb{R}$: 
\[
    \dual{\calW}_\text{fbr} = \{ \alpha \in\dual{\calW}_\mathrm{loc}\, :\, \alpha \text{ is order-0}\} \subset \dual{\calW}_\mathrm{loc}.
\qedhere\]
\end{definition}

In particular, identifying $\calW^\vee\simeq\oloc^\text{top}(\Sigma,\dual{\calW}_{(0)})$, we have
\[
\int_\Sigma \oloc^\text{top}(\Sigma,\dual{\calW}_{(0)})= \int_\Sigma \calW^\vee = \dual{\calW}_\text{fbr}.
\]

The previous definitions naturally generalise to multi-linear local family of forms:
\begin{definition}[Multilinear dual-valued forms]
Let $W\to\Sigma$ be a vector bundle and $\calW=\Gamma(\Sigma,W)$. A multilinear local family of $(p,q)$-forms is a multilinear map
\[
\boldsymbol{\nu}\colon\calW\stackrel{n}{\times  \cdots \times}\calW \to \oloc^{p,q}(\Sigma\times\calPhi)
\]
that is local as a form in $\oloc^{p,q,{0, \cdots, 0}}(\Sigma \times \calPhi \times \calW\stackrel{n}{\times \cdots \times}\calW)$.
We will call such forms \emph{local multilinear dual-valued forms}, and denote them as:\footnote{Recall that for our concrete models, the (completed) projective tensor product is enough.}
\[
\boldsymbol{\nu} \in \oloc^{p,q}(\Sigma\times\calPhi,\dual{\calW}\stackrel{n}{\hat\otimes  \cdots \hat\otimes}\dual{\calW}).
\qedhere
\]
\end{definition}

We conclude this survey on local forms by recalling a key result that will be repeatedly employed throughout:

\begin{theorem}[\cite{Takens,Zuckerman}]  \label{thm:LocFormDec} 
For $\calPhi=\Gamma(\Sigma,F)$ as above, the space of $(\text{top},1)$-local forms decomposes as
\[
\oloc^{\text{top},1}(\Sigma \times \calPhi) = \osrc^{\text{top},1}(\Sigma \times \calPhi) \oplus \obdry^{\text{top},1}(\Sigma\times\calPhi),
\]
where $\osrc^{\text{top},1}(\Sigma \times \calPhi)$ is the space of local forms of \emph{source type}\footnote{Roughly speaking, source-type forms are $(\text{top},1)$-local forms which are $C^\infty(\Sigma)$-linear in $\bd \phi$, i.e.\ they do \emph{not} feature any term of the form $\bd(\partial^k\phi) \equiv \partial^k\bd\phi$ for $k>0$, Source-type forms are sometimes called ``Euler--Lagrange'' forms. See Appendix \ref{app:localforms} for details.} (Definition \ref{def:sourcetype}), and $\obdry^{\text{top},1}(\Sigma\times\calPhi)\doteq d \oloc^{\text{top}-1,1}(\Sigma \times \calPhi)$ are called \emph{boundary forms}. 
We denote the corresponding projections by:
\begin{align*}
\Pi_\text{src}\colon \oloc^{\text{top},1}(\Sigma\times \calPhi)\to \osrc^{\text{top},1}(\Sigma\times \calPhi), &\quad \boldsymbol{\nu} \mapsto \boldsymbol{\nu}_\text{src} ;\\
\Pi_\text{bdry}\colon \oloc^{\text{top},1}(\Sigma\times \calPhi)\to \obdry^{\text{top},1}(\Sigma\times \calPhi), &\quad \boldsymbol{\nu} \mapsto \boldsymbol{\nu}_\text{bdry}. 
\end{align*}
\end{theorem}

\begin{proposition}[Decomposition of dual-valued local forms]
\label{prop:dualvaluedformdecomposition}
The space of dual-valued $(\text{top},0)$-local forms decomposes as
\[
\oloc^{\mathrm{top},0}(\Sigma\times \calPhi,\calW^*) = \oloc^{\mathrm{top},0}(\Sigma\times \calPhi,\calW_{(0)}^{*}) \oplus d\oloc^{\mathrm{top}-1,0}(\Sigma\times \X,\calW^*),
\]
where the first term and second terms involve dual-valued local forms which are order-0 and $d$-exact respectively.
\end{proposition}
\begin{proof}
The proof of this lemma follows from Theorem \ref{thm:LocFormDec}.
Since $\calW$ is a vector space there is an injection:
\[
Y:\oloc^\text{top}(\Sigma,\calW^*) \to \oloc^{\text{top},1}(\Sigma\times \calW), 
\quad 
\balpha \mapsto \langle \boldsymbol{\balpha},\cbd w\rangle
\]
where we denoted $\cbd$ the differential on $\calW$. Moreover, this map sends the subspace $\oloc^\text{top}(\Sigma,\calW^*_{(0)})\subset \oloc^\text{top}(\Sigma,\calW^*)$ of order-0 dual-valued forms (resp.\ the subspace $d\oloc^{\mathrm{top}-1}(\Sigma,\calW^*)$ of $d$-exact dual-valued forms) into source (resp. boundary) forms in $\oloc^{\text{top},1}(\Sigma\times \calW)$, cf.\ Theorem \ref{thm:LocFormDec}. We denote the corresponding restriction $Y_{(0)}$ (resp.\ $Y_d$).
These injections are invertible on the image, which can be characterized as 
\begin{align*}
\Im(Y) & = \{ \boldsymbol{\nu}\in \oloc^{\text{top},1}(\Sigma\times \calW) \ : \ \boldsymbol{\nu} \ \text{is $\mathbb{R}$-linear in} \ \calW \},\\
\Im(Y_{(0)}) & = \{ \boldsymbol{\nu}\in \Im(Y) \ : \ \boldsymbol{\nu} \ \text{is $C^\infty(\Sigma)$-linear in} \ \calW \},\\
\Im(Y_d) & = \{ \boldsymbol{\nu}\in \Im(Y) \ : \ \boldsymbol{\nu} \ \text{is $d$-exact}\}.
\end{align*}

Using  Definition  \ref{def:dualvaluednew}(ii), we identify $\boldsymbol{\Lambda} \in \oloc^{\mathrm{top},0}(\Sigma\times \calPhi,\calW^*)$ with a local form in $\oloc^{\mathrm{top},0}(\Sigma\times\A)$, where $\A \doteq \calPhi\times\calW \simeq \Gamma(\Sigma, F\times_\Sigma W)$. 
Denoting $w\in\calW$, and (with slight abuse of notation) $\cbd$ the differential on $\A$, introduce the following local $(\text{top},1)$ form on $\Sigma\times \A$: $\langle\boldsymbol{\Lambda},\cbd w\rangle$. Then, the decomposition into \emph{source} and \emph{boundary} terms of Theorem \ref{thm:LocFormDec} holds, and since $\calW$ is a vector space and by the previous observations there exists unique $\boldsymbol{\Lambda}\in\oloc^{\mathrm{top},0}(\Sigma\times \calPhi,\calW_{(0)}^{*})$ and $d\boldsymbol{\lambda } \in d\oloc^{\mathrm{top}-1,0}(\Sigma\times \calPhi,\calW^*)$ such that
\[
\begin{cases}
\langle \boldsymbol{\Lambda},\cbd w\rangle_\text{src} = \langle \boldsymbol{\Lambda}_\circ,\cbd w\rangle \\
\langle \boldsymbol{\Lambda},\cbd w\rangle_\text{bdry} = d\langle  \boldsymbol{\lambda},\cbd w\rangle
\end{cases}
\iff 
\quad 
\langle \boldsymbol{\Lambda},w\rangle = \langle \boldsymbol{\Lambda}_\circ,w\rangle + d\langle \boldsymbol{\lambda}, w\rangle. \qedhere
\]
\end{proof}

\begin{corollary}\label{cor:cptsupport}
Denote $\mathring\Sigma$ the open interior of $\Sigma$, and let
\[
\calW_c \doteq \{ w\in\calW \ : \ w \text{ compactly supported in $\mathring\Sigma$ } \}.
\]
Then $\boldsymbol{\Lambda}\in\oloc^\text{top}(\Sigma\times\calPhi,\calW^*)$ is $d$-exact iff $\int_\Sigma\langle\boldsymbol{\Lambda},w_c\rangle = 0$ for all $w_c\in\calW_c$.
\end{corollary}
\begin{proof}
The forward implication is straightforward. We focus on the converse statement.
By Proposition \ref{prop:dualvaluedformdecomposition}, and adopting the notation used in its proof, we write $\boldsymbol{\Lambda} = \boldsymbol{\Lambda}_\circ + d \boldsymbol{\lambda}$ and thus find
\[
 \int_\Sigma\langle\boldsymbol{\Lambda}_\circ,w_c\rangle = 0 \quad \forall w_c\in\calW_c.
\]
Applying the map $Y$ to the above equation, we see that the fundamental theorem of the calculus of variations implies that $\boldsymbol{\Lambda}_\circ(x,\phi)=0$ for all $(x,\phi)\in\mathring\Sigma\times\calPhi$. By continuity over $x\in\Sigma$, we conclude.
\end{proof}

\section{Geometric setup for gauge theories with corners}\label{sec:localfieldtheory}

We start by considering a connected, compact manifold $\Sigma$, thought of as a codimension-1 submanifold in some (spacetime) manifold $M$. In typical applications $M$ is endowed with a Lorentzian metric and $\Sigma$ is spacelike, but our formalism can be applied to examples where $\Sigma$ is null \cite{RielloSchiavinanull} or cases where $M$ is not equipped with a metric at all (Section \ref{sec:ex-ChernSimons}). Physically, $\Sigma$ may represent a region inside an initial-value, or Cauchy, surface. Crucially, we will allow $\Sigma$ to possess nontrivial boundaries, $\partial\Sigma\neq\emptyset$.

We call the embedded sumbanifold $\pp\Sigma\hookrightarrow \Sigma$ the \emph{corner}, a terminology justified by the scenario where $\Sigma$ is a connected component of the boundary of a manifold with corners $M$ such that $\pp\Sigma$ is a component of the (true) corner of $M$.

We define the space of fields of our theory as the space $\X\doteq\Gamma(\Sigma,E)$ of sections of some finite dimensional fibre bundle $E\to \Sigma$. We further assume that $\X$ comes equipped with a local symplectic density $\bom$ (Definition \ref{def:locsymp}), and a compatible local action of a Lie group $\G$ with Lie algebra $\fG$ (Definitions \ref{def:localaction} and \ref{def:LocLieAlg}): in physical parlance, these are the ``gauge transformations'' acting on the fields in $\X$.  

Compatibility of the local symplectic structure $\bom$ and of the $\G$-action on $\X$ in the presence of corners, is encoded in the notion of \emph{locally Hamiltonian $\G$-spaces}. The ``local'' characterization refers to the fact that the Hamiltonian flow equation, which defines what we call a \emph{local momentum form} in Definition \ref{def:Ham+mommaps} will be required to hold between local forms over $\Sigma\times\X$, that is to say \emph{point-wise} over $\Sigma$. 

Locality also allows us to conveniently introduce ``corner cocycles'' encoding the fact that equivariance of the local momentum map can hold ``up to boundary''. While corner cocycles appear in examples, e.g. in $BF$ or Chern--Simons theory, a failure of equivariance in the bulk would instead hinder the closure of the constraint algebra, obstructing reduction.

\subsection{Local symplectic structure} \label{sec:localsymplstr}

Let $(\Sigma,\pp\Sigma)$ be a compact, connected manifold with boundary. Let $E\to\Sigma$ be a finite dimensional vector, or affine, bundle and  $\X\doteq\Gamma(\Sigma,E)$ its space of sections, which we call the \emph{geometric phase space} or, for short, the \emph{space of fields}. 

\begin{definition}[Locally symplectic space]\label{def:locsymp}
A symplectic form $\omega$ on $\X=\Gamma(\Sigma,E)$ is \emph{local} iff there exists a \emph{symplectic density} $\bom\in\oloc^{\text{top},2}(\Sigma \times \X)$ such that\footnote{Observe that $\bd\omega=0$ fails to imply $\bd\bom=0$.} $\bd \bom=0$ and $\omega = \int_\Sigma \bom$. 
$(\X,\bom)$ is called a \emph{locally symplectic space}.
\end{definition}

\begin{lemma}\label{lem:nondegeneracystrength}
Let $\bom\in\oloc^{\text{top},2}(\Sigma\times\X)$, and denote by $\bom^\flat_{\text{src}}$, the map
\[
\bom^\flat_{\text{src}} \doteq \Pi_\text{src}\circ\bom^\flat: \Xloc(\X) \to \Omega_\text{src}^{\text{top},1}(\Sigma \times \X).
\]
Then, $\bom^\flat_\text{src}$ injective implies $\bom^\flat$ injective on $\mathfrak{X}_\loc(\X)$, with the opposite implication holding if $\bom$ is ultralocal (Definition \ref{def:ultralocal}). Moreover, if $\partial\Sigma=\emptyset$,  $\bom^\flat_\text{src}$ is injective iff $\omega^\flat$ is injective. 
\end{lemma}
\begin{proof}
Application of Theorem \ref{thm:LocFormDec} to $\mathrm{Im}(\bom^\flat)$ implies that $\ker(\bom^\flat)\subset\ker(\bom^\flat_\text{src})$, from which the first statement follows. 
The second statement follows from the definition of ultralocality \ref{def:ultralocal}, since then $\bom^\flat = \bom^\flat_\text{src}$. Finally, the last statement for $\pp\Sigma=\emptyset$ is a consequence of the identity $\omega^\flat(\mathbb{X}) = \int_\Sigma \bom^\flat(\mathbb{X}) = \int_\Sigma \bom_\text{src}^\flat(\mathbb{X})$ and the fundamental theorem of the calculus of variations which states that, for a source-type form, the vanishing locus of $\int_\Sigma \bom_\text{src}^\flat(\mathbb{X})$ is the same as that of its integrand. (This is because $\bom^\flat_\text{src}(\mathbb{X}) = \boldsymbol{F}_\mathbb{X}(j^k\phi)\bd \phi$ for some local functional $\boldsymbol{F}_\mathbb{X}$ depending on the $k$-jet $j^k\phi$, and because the variations $\bd\phi$ are free.)
\end{proof}

\begin{remark}\label{rmk:nondegeneracystrength}
Here, we present a counterexample to the converse statement in Lemma \ref{lem:nondegeneracystrength}, i.e.\ we find a $\bom$ such that $\bom^\flat$ is injective while $\bom^\flat_\text{src}$ is not.
Let $\Sigma = [-1,1]$ and $E = \Sigma \times \mathbb R$ the trivial line bundle. Denoting $\phi\in\X=\Gamma(\Sigma,E)$, consider the symplectic density\footnote{$J^1E$ is globally isomorphic to $\Sigma \times \mathbb{R} \times \mathbb{R}$. Let $(x,y,y')$ be (global) coordinates such that $j^1 : \Sigma \times \X \to J^1E$, $(x,\phi) \mapsto (x,y,y')=(x,\phi(x),\pp_x\phi(x))$. Then, denoting by $d_H$ ($d_V$) the horizontal (vertical, resp.) differential on $J^1E$, $\bom$ is the pullback along $j^1$ of $\wt\bom = - d_V y \wedge d_V y' \wedge d_H x \in \Omega^{1,2}(J^1E)$.} $\bom = \bd \phi \wedge d \bd \phi$.
As observed above, from the source/boundary decomposition theorem (Theorem \ref{thm:LocFormDec}) it follows that $\ker(\bom^\flat) \subset \ker(\bom^\flat_\text{src})$. Applying the source/boundary decomposition to $\bom^\flat(\mathbb{X})$, and denoting by $X_\phi\doteq \bi_{\mathbb{X}}\bd \phi$, we obtain $[\bom^\flat(\mathbb X)]_\text{bdry} = d( X_\phi\bd \phi)$ and $\bom^\flat_\text{src}(\mathbb X) \doteq [\bom^\flat(\mathbb X)]_\text{src} = -2 (d X_\phi)  \bd \phi $. 
Hence, $\mathbb{X}\in\ker(\bom^\flat_\text{src})$ iff $X_\phi =c $ is constant over $\Sigma$.
Now, among all constant values $X_\phi = c$, only $c=0$ also guarantees that $0 = [\bom^\flat(\mathbb X)]_\text{bdry} = c d\bd \phi$.
Thus we proved that in this case, although $\bom^\flat$ is injective, $\bom^\flat_\text{src}$ is not. 
\end{remark}

\subsection{Local group action}
\label{sec:localgrpaction}

We now equip $\X$ with the right action of a Fr\'echet Lie group $\G$ with Lie algebra $\fG \doteq\mathrm{Lie}(\G)$. We call $\G$ the \emph{gauge group}, and its action on $\X$ \emph{gauge transformations}: $\phi\mapsto \phi^g = \phi\triangleleft g$.

Let us recall the following notions \cite{Neeb-locallyconvexgroups}:

\begin{definition}[Locally exponential/enlargeable]\label{def:locexpenlargeable}
~
\begin{enumerate}
\item A Lie group $\G$ is said to be \emph{locally exponential} iff it has a smooth exponential map $\exp:\fG\equiv\Lie(\G)\to\G$ and there exists an open 0-neighbourhood $U\subset\fG$ such that $\exp\vert_U:U\to\exp(U)$ is a diffeomorphism onto an open $e$-neighbourhood of $\G$.  

\item A Lie algebra $\fG$ is said to be \emph{enlargeable} (to $\G$) if a locally exponential Lie group $\G$ exists such that $\fG=\Lie(\G)$.\qedhere
\end{enumerate}
\end{definition}

\begin{theorem}[\cite{GloecknerNeeb}]\label{thm:GloecknerNeeb}
If $\G_1$ and $\G_2$ are two locally exponential simply connected Lie groups, then $\fG_1\cong\fG_2$ iff $\G_1\cong\G_2$.
\end{theorem}

\begin{remark}[Mapping groups]\label{rmk:mappinggrp}
Let $P\to\Sigma$ be a principal $G$-bundle, for $G$ a finite dimensional Lie group, and let the equivariant mapping group $\G = C_0^\infty(P,G)^G$ be the group of vertical automorphisms of $P$ which are connected to the identity \cite[Thm. 37.17]{kriegl1997convenient}. Then $\G$ is simply connected, and it is also locally exponential \cite[Thm 3.1.11]{WockelPhD} (see also \cite[Thm. 42.21]{kriegl1997convenient} and \cite[Thm. IV.1.12]{Neeb-locallyconvexgroups}).
\end{remark}

We now demand that the gauge group $\G$ is locally exponential and connected 
\footnote{At least in QCD it seems correct to only demand invariance under ``gauge transformations'' that are connected to the identity. In Yang--Mills theory, consider the group $\tilde\G$ of vertical automorphisms of the principal $G$-bundle, and its identity component $\G$. The discrete group $\mathcal{T}\doteq\tilde\G/\G$ plays a role in the (axial) $U(1)$ problem of QCD: the choice of an irrep of $\mathcal{T}$ defines the choice of a $\theta$-vacuum. This suggests that $\mathcal{T}$ should not be considered as a local gauge redundancy, but rather as a global symmetry. See \cite[Chap.8]{Strocchi13} and \cite[Chap.3]{strocchi19} for details.}
---and that it acts on $\X$ locally:

\begin{definition}[Local Lie algebras]\label{def:LocLieAlg}
A $\Sigma$-\emph{local Lie algebra} is a Lie algebra $(\fG,[\cdot,\cdot])$ realized as the space of sections of a vector bundle $\Xi\to\Sigma$, and such that its Lie bracket $[\cdot,\cdot] : \fG \times \fG \to \fG$ is local as a map of sections of fibre bundles (Definition \ref{def:localmaps}).

If a Lie algebra $\mathfrak{H}$ is $\Sigma$-local or the quotient of a $\Sigma$-local Lie algebra $\fG$ by a (not necessarily $\Sigma$-local) ideal $\mathfrak{I}$, i.e.\ $\mathfrak{H}\simeq\fG/\mathfrak{I}$, we say that $\mathfrak{H}$ is \emph{supported on $\Sigma$}.
\end{definition}

\begin{remark}[Nuclear Fr\'echet Lie algebras]\label{rmk:FrechetLieAlgebras}
As a consequence of Remarks  \ref{rmk:nuclearFrechet} and \ref{rmk:NFVS}, all Lie algebras supported on $\Sigma$ (in particular local ones) are nuclear Fr\'echet spaces, and hence reflexive, i.e.\ $\bidual{\fG} \simeq \fG$.
\end{remark}

\begin{definition}[Local action, $\G$-space]\label{def:localaction}
A  manifold $\X$ equipped with a right action of $\G$ (resp. $\fG$) is called a $\G$-space (resp.\ $\fG$-).

The action of $\G$ on $\X$ is said to be \emph{local} iff the corresponding \emph{infinitesimal action} $\rho$ of $\fG$ on $\X$ is local, i.e\ iff (\emph{i}) the Lie algebra $\fG\doteq\mathrm{Lie}(\G)$ is local, and (\emph{ii}) the action
\begin{equation}\label{eq:rhoaction}
    \rho:\X\times\fG\to T\X,  
    \quad \xi \mapsto \rho_\phi(\xi)\doteq \left.\frac{d}{d t}\right\vert_{t=0}\phi\triangleleft \exp(t\xi)
\end{equation}
is local as a bundle morphism. 
\end{definition}

The vector fields $\rho(\xi)$ are called \emph{fundamental vector fields}.
By construction they are equivariant, i.e.\ $g_*\rho(\xi) = \rho(\Ad(g)^{-1}\cdot\xi)$, which infinitesimally translates into the homomorphism of Lie algebras: $[\rho(\xi),\rho(\eta)] = \rho([\xi,\eta])$. We will elaborate further on this in the next section. 

Note, however, that $\rho$ can also be understood as a dual-valued vector field, $\rho\in\Xloc(\X,\dual{\fG})$, i.e.\ as a local linear map:\footnote{Cf. Definition \ref{def:dualvaluednew} and the paragraph above Remark \ref{rmk:approximationproperty}.}
\begin{align}\label{eq:rhoaction2}
    \rho\;\colon\fG\to\Xloc(\X), \quad
    \xi\mapsto\rho(\xi) \notag.
\end{align}

An important question at this point is whether the action has stabilisers. The most restrictive hypothesis is to require that the action is free\footnote{Recall: free means that, at all $\phi$, the action has no nontrivial stabilisers: $\G_\phi \doteq \mathrm{Stab}(\phi)= \{e\}$.} which implies that $\rho(\xi)$  is a nowhere-vanishing vector field for all $\xi\neq0$. However, as we will see later on, this is too strong a restriction.

\subsection{Local momentum forms}\label{sec:momentum maps}
We have so far introduced a symplectic density and a gauge group action on $\X$. The following definitions ensure that the two are suitably compatible.

\begin{definition}[Locally symplectic $\G$-space]
A \emph{locally symplectic $\G$-space} $(\X,\bom,\G)$ is a locally symplectic space which is also a $\G$-space, and for all $g\in\G$, $g^*\bom = \bom$.
\end{definition}

Infinitesimally, the  locally symplectic $\G$-space property requires that $\bom$ be invariant under the flow generated by all elements of $\fG$: for all $\xi\in\fG$, $\L_{\rho(\xi)}\bom = 0$. This formula, together with Cartan's formula for the Lie derivative (Proposition \ref{prop:Cartancalc}), and the fact that $\bom$ is $\bd$-closed, imply that $\bi_{\rho(\xi)}\bom$ is $\bd$-closed. We will now demand that it be $\bd$-exact.

\begin{definition}[Local momentum forms and maps]\label{def:Ham+mommaps}
Let $(\X,\bom,\G)$ be a locally symplectic $\G$-space. Then, the dual-valued local $(\text{top},0)$-form
\[
\bH\in \oloc^{\text{top},0}(\Sigma \times \X,\dual{\fG})
\]
is  a \emph{local momentum form} iff, for all $\xi\in\fG$,
        \begin{equation}
        \bi_{\rho(\xi)}\bom = \bd \langle \bH, \xi\rangle.
        \label{eq:weakHam}
        \end{equation}
If a locally symplectic $\G$-space admits a local momentum form $\bH$, it is said to be \emph{locally Hamiltonian} and will be denoted by $(\X,\bom,\G,\bH)$. The integral over $\Sigma$ of a local momentum form, $\H \doteq \int_\Sigma \bH \in \iloc^0(\X,\dual{\fG})$ is called a \emph{local momentum map}. 
\end{definition}

\begin{remark}[Perspectives on local momentum forms]\label{rmk:co-}
To describe local momentum forms as valued in the dual $\fG^*$, we have adopted the perspective provided by Definition \ref{def:dualvaluednew}. Since $\fG$ is a nuclear Fr\'echet space, see Remarks \ref{rmk:sourcetarget} and \ref{rmk:approximationproperty}) allow us to think of a dual-valued density such as $\bH$ (and its integrated counterpart $H$) in three equivalent ways: 
\begin{enumerate}
    \item as a $(\text{top},0,0)$-local form over $\Sigma\times\X\times\fG $ which is linear in the $\fG$-argument. This is the ``neutral'' local form perspective, in which $\X$ and $\fG$ are on the same footing (see Section \ref{sec:LieAlgPict});
    \item as a (local) map $\bH : \X \to \oloc^{\text{top}} (\Sigma,\dual{\fG})$, which integrates to a map $H:\X \to \dual{\fG}_\loc\subset\dual{\fG}$ (Definition \ref{def:dual*}). This is the local \emph{momentum} form/map perspective;
    \item as a linear local map $\bH:\fG \to \oloc^{\text{top},0}(\Sigma\times\X)$. This is the (local) \emph{comomentum} form perspective.\qedhere
\end{enumerate}
\end{remark}

\begin{remark}[Shifting the momentum form]
\label{rmk:HamShift1}\label{rmk:HamShift2}
In general, $\X$ lacks a natural origin with respect to which $\bH$ can be naturally defined, i.e.\ a prescription that allows us to integrate the (by assumption) exact 1-form $\bi_{\rho(\cdot)}\bom$ to obtain a \emph{specific} 0-form $\bH$. This is manifest in Chern--Simons theory where $\X$ is an affine space of connections lacking a canonical choice of origin (Section \ref{sec:ex-ChernSimons}).

We will assume that the spaces of fields we are working with, , have been endowed with an arbitrary choice of reference configuration to make this integration possible (See \cite{iglesias2010moment} for a more general perspective on the need of a reference configuration to define momentum maps in a diffeological treatment of symplectic geometry). However, we will not keep track of this initial choice of reference, since it is only provisional and might have to be adjusted later on (Section \ref{sec:flux}).

Provided $\bi_{\rho(\cdot)}\bom$ is exact, changing reference configuration  means shifting $\bH \mapsto \bH + \balpha$ by a field-space constant $\balpha\in\oloc^{\text{top},0}(\Sigma \times \X,\dual{\fG})$, $\bd \balpha = 0$. In the following, such $\balpha$'s will be more succinctly thought of as elements $\balpha\in\oloc^{\text{top}}(\Sigma,\dual{\fG})$. 
\end{remark}

\begin{remark}
The Hamiltonian flow equation links the orders of the momentum map and of the local action $\rho$, to that of the symplectic densities $\bom$ (Definition \ref{def:order-kduals}).
All cases of interest (that we know of) deal with a local action $\rho$ of order (at most) 1. Typically---but not always---one works with ultralocal symplectic densities, which are order-0 (Definition \ref{def:ultralocal}). If both conditions are met the ensuing \emph{local} momentum map are necessarily of order at most $1$. In physics, symplectic densities of order-1 naturally appear over light-like hypersurfaces \cite{AshtekarStreubel}. These are studied in the case of Yang--Mills theory in \cite{RielloSchiavinanull}.
\end{remark}

\subsection{Equivariance and cocycles}

We now briefly discuss the equivariance properties of the momentum form $\bH$ upon the action of $\G$. We refer the reader to Appendix \ref{app:CEcoh} for further details as well as the definitions of some of the relevant, but standard, notions.

Given a dual-valued local form (Definition \ref{def:dualvaluednew}), we can compose it with the coadjoint action of an element $g\in\G$, so that
\[
\Ad^* : \G \times \oloc^{\text{top},0}(\Sigma\times\X,\dual{\fG}) \to \oloc^{\text{top},0}(\Sigma\times\X,\dual{\fG}),\quad (g,\bH) \mapsto \Ad^*(g)\cdot\bH
\]
is given by
\begin{subequations}
\begin{equation}
\label{eq:coadj1}
 \langle(\Ad^*(g)\cdot \bH)(\phi),\xi\rangle \doteq \langle \bH(\phi) , \Ad(g)\cdot\xi\rangle.   
\end{equation}
Using the \emph{co}-momentum map viewpoint (see Remark \ref{rmk:co-}), which we emphasise with the notation $\bH^\text{co}: \fG \to \oloc^{\text{top},0}(\Sigma\times\X)$, Equation \eqref{eq:coadj1} becomes:
\begin{equation}
\label{eq:coadj2}
(\Ad^*(g)\cdot\bH)^\text{co} \doteq  \bH^\text{co}\circ\Ad(g).
\end{equation}
\label{eq:coadj}
\end{subequations}

The following notion measures the lack of equivariance of the local momentum form $\bH$ seen as intertwining between the $\G$-actions  on $\X$ and on $\fG^*$. 

\begin{definition}[CE cocycles forms]\label{def:cocycleforms}
A group and algebra \emph{Chevalley--Eilenberg (CE) cocycle forms} are, respectively, maps $\bC:\G \to \oloc^{\text{top}}(\Sigma,\dual{\fG})$ and $\bK: \fG^{\wedge 2}\to\oloc^{\text{top}}(\Sigma)$ such that for all $g,h\in\G$ and all $\xi_{1,2,3}\in\fG$,
\begin{equation}
\bC(gh)=  \Ad^*(h)\cdot\bC(g)+ \bC(h),
\qquad\text{and}\qquad
\begin{cases}
\bK(\xi_1,\xi_2) + \bK(\xi_2,\xi_1) = 0,\\
\mathrm{Cycl}_{123}  \bK([\xi_1,\xi_2],\xi_3)= 0.
\end{cases}
\end{equation}
\end{definition}

\begin{proposition}\label{prop:CKcocycles}
Let $(\X,\bom,\G,\bH)$ be as in Definition \ref{def:Ham+mommaps}, with $\X$ path-connected. Denoting $(g^*\bH)(\phi) \doteq \bH(\phi\triangleleft g)$, consider the map $\bC:\G \to \oloc^{\text{top},0}(\Sigma\times\X,\dual{\fG})$ given by
\[
\bC(g) \doteq g^*\bH - \Ad^*(g)\cdot \bH,
\]
and its tangent map at the identity, i.e.\ the smooth, linear, map $\bK: \fG\to\oloc^{\text{top},0}(\Sigma\times\X,\dual{\fG})$ 
\[
\bK \doteq T_e\bC, \quad \bK(\xi) = \L_{\rho(\xi)}\bH - \ad^*(\xi)\cdot \bH.
\]
Then $C\doteq\int_\Sigma\bC$ and $K\doteq\int_\Sigma\bK$ do not depend on $\phi$, $\bd \bC = 0 = \bd \bK =0$, and are a CE group and algebra cocycle, respectively.
\end{proposition}
\begin{proof}
Standard. We collect here two observations. First, to translate the infinitesimal equation $\bd\bC=0$ to a statement about the independence of $\bC$ from $\phi$, we need to assume that $\X$ is path-connected (cf.\ Assumption \ref{assA:item1}): if it is not, there could be, a priori, distinct cocycles associated to each connected component of $\X$. And second, for $\bK: \fG\to\oloc^{\text{top},0}(\Sigma\times\X,\dual{\fG})$ to be identified with a CE algebra cocycle (after integration), one first needs to (canonically) identify it with a map $\fG^{\times 2} \to \Omega^\text{top}(\Sigma)$ which one can prove, in particular, to be skew-symmetric in its arguments. 
\end{proof}

We now introduce three different notions of equivariance for the local momentum \emph{form} $\bH$. The notion that will be most relevant in the following is that of weak equivariance:

\begin{definition}[Equivariance, corner cocycles]\label{def:equivariance}
On a locally Hamiltonian $\G$-space $(\X,\bom,\G,\bH)$, the local momentum form $\bH$ is said to be:
\begin{enumerate}[label=(\arabic*)]
    \item \emph{strongly equivariant}, iff $\bC = 0$,
    \item \emph{$\Sigma$-equivariant}, iff $\int_\Sigma \bC =0$.
    \item \emph{weakly equivariant}, iff $\bC = d\bc$ for some $\bc:\G\to \Omega^\text{top}(\Sigma,\dual{\fG})$,
\end{enumerate}
Moreover, we say that the integrated local functional $\H = \int_\Sigma \bH\in\iloc^{0}(\calF)$ is \emph{equivariant} iff $\bH$ is $\Sigma$-equivariant. Finally, if there exists a $\bc$ as in (3) above, we denote $\bk =  T_e\bc$, and say that $\bC=d\bc$ an $\bK=d \bk$ are \emph{(CE) corner cocycles}.
\end{definition}

\begin{proposition}
With reference to Definition \ref{def:equivariance}, $(1)\implies(2)\implies(3)$, and if moreover $\pp\Sigma=\emptyset$ then one also has $(3)\implies(2)$.
\end{proposition}
\begin{proof}
The only nontrivial implication is $(2)\implies (3)$. Note that $\bC$ can be thought of as a map $\G\to\oloc^{\text{top},0}(\Sigma\times\X,\fG^*)$. Therefore, for every $g\in\G$ we can apply Proposition \ref{prop:dualvaluedformdecomposition} with $\calPhi = \X$ and $\calW = \fG$, and thus uniquely decompose $\bC(g) = \bC_\circ(g) + d\bc(g)$ with $\bC_\circ$ order-0. Since $\int_\Sigma \bC= 0$, one has in particular $\int_\Sigma \langle \bC,\xi_c\rangle = \int_\Sigma \langle\bC_\circ, \xi_c\rangle$ for all $\xi_c \in \fG_c \subset \fG$ a gauge transformation of compact support in $\mathring{\Sigma}$ (the open interior of $\Sigma$), whence we conclude using Corollary \ref{cor:cptsupport}.
\end{proof}

\begin{lemma}
Let $(\X,\bom)$ be a locally Hamiltonian $\G$-space, and let $\bH$ be any local momentum form for the action of $\G$ with group cocycle form $\bC$. Then, a  momentum map for the action of $\G$ on $(\X,\bom)$ is said (1) strongly-, (2) weakly-, or (3) $\Sigma$-equivariant iff respectively 
\begin{enumerate}[label=(\arabic*)]
    \item $\bC$ is a $\delta$-coboundary,  $[\bC]_{\text{CE}}=0$,
    \item $\bC$ is a $\delta$-coboundary up to a $d$-coboundary, $[\bC]_{\text{CE}}=[d\bc]_\text{CE}$,
    \item $\bC$ integrates to a $\delta$-coboundary,  $[\int_\Sigma\bC]_{\text{CE}}=0$,
\end{enumerate}
where $[\bC]_{CE}$ denotes the CE-cohomology class of $\bC$.
\end{lemma}

We saw that a momentum form $\bH$ is necessarily equivariant at most up to a CE cocycle form. A (local) CE cocycle is said exact if it can be trivialized by shifting the momentum map by a field-space constant $\balpha \in \oloc^\mathrm{top}(\Sigma,\fG^*)$ (Remark \ref{rmk:HamShift1}). If this is not possible, the cocycle is  said to be in the CE cohomology (see Appendix \ref{app:CEcoh}). There are two essentially equivalent ways to deal with this scenario. 

The first one is to introduce an affine action of $\G$ on $\fG^*$ by ``translating'' the coadjoint action of $\G$ ($\fG$) on $\fG^*$ by $C$ ($K$, respectively). This affine action then provides the correct notion of equivariance in the presence of a cocycle and can be used to define the relevant generalization of coadjoint orbits and other relevant objects that we will use in the following. 

The second way is to introduce a central extension $\wh\fG$ of $\fG$ which is governed by $K$, as well as appropriate extensions $\wh\rho$ and $\wh{H}$ of the action $\rho : \fG\times \X \to T\X$ and momentum map $H :\X \to \fG^*$, respectively. This allows one to reinterpret the affine orbits above in terms of genuine coadjoint orbits in the centrally extended algebra, and a non-$\Ad$-equivariant momentum map $H$ in terms of an equivariant one. 

Details on these last considerations, as well as on the Kirillov--Konstant--Souriau (KKS) Poisson structure on $\fG^*$ associated to a CE cocycle $K$, and the related symplectic structure on the affine/coadjoint orbits in $\fG$, can be found in Appendix \ref{app:CEcoh} (see in particular Lemma \ref{lemma:KKS} and Corollary \ref{cor:KKS}).

\begin{definition}[Equivariant locally Hamiltonian $\G$-spaces]\label{def:equivarianceX}
Let $(\X,\bom,\G)$ be as in Definition \ref{def:Ham+mommaps}. Then $(\X,\bom,\G)$ is a \emph{strongly-, $\Sigma$-, weakly-equivariant locally Hamiltonian $\G$-space} iff the action of $\G$ is locally symplectic and it admits a strongly-, $\Sigma$-, weakly-equivariant local momentum form $\bH$, respectively.
\end{definition}


\subsection{Yang--Mills theory:\ general setup}\label{sec:runex-setup}
This is the first of several sections dedicated to a ``running'' example, given by Yang--Mills theory (YM). Throughout this example $\Sigma$ is a compact, oriented manifold with positive-definite\footnote{Generalizations to a Lorentzian (nondegenerate) signatures of $\Sigma$ are possible in principle, but the solution space of the (Gauss) constraint must be re-analyzed. The case of YM on a null $\Sigma$ is studied in \cite{RielloSchiavinanull}.} metric $\gamma$; in other words, $\Sigma$ can be thought of as a (region within a) Cauchy surface in a globally hyperbolic spacetime $M=\Sigma\times\mathbb R$.

We consider two cases: $G$ a real semisimple Lie group with Lie algebra $\mathfrak{g}$, or $\fg$ Abelian.\footnote{The semisimple and Abelian cases are not exhaustive of all possible scenarios, but they capture the basic phenomena that are to be expected in general.}
On $\fg$ we consider a nondegenerate, $\Ad$-invariant, bilinear form, denoted
\[
\tr \colon \mathfrak{g}\times\mathfrak{g} \to \mathbb{R},
\]
and used to (implicitly) identify $\dual{\fg}$ and $\fg$. In the semisimple case we can use the Killing form. This convention allows us to write formulas in a more familiar form. Under this identification $\Ad^*(g)$ is mapped to $\Ad(g^{-1})$ and $\ad^*(\xi)$ to $-\ad(\xi)$.

Consider a principal $G$-bundle $P\to \Sigma$, and denote by\footnote{Note that it is possible to see $\Acal$ as sections of the fibre bundle $(J^1P)/G \to \Sigma$, thus fitting with our general framework (see e.g. \cite{MangiarottiSardanashvily2000, Lopez}).} 
\[
\Acal \doteq \mathrm{Conn}(P\to\Sigma)
\]
the space of principal connections. This is an affine space (locally) modelled over the space of $\fg$-valued 1-forms. Denoting by $\sim$ the local model for infinite-dimensional manifolds, we write $\Acal \sim \Omega^1(\Sigma,\fg)$. Often, the elements $A\in \Acal$ are called \emph{gauge potentials}, or gauge connections.

The \emph{gauge group} $\G$ is then defined to be the group of vertical automorphisms of $P$ which are connected to the identity. This group is isomorphic to the group of {$\AD$}-equivariant $G$-valued functions over $P$ which are homotopic to the identity, that is $\G\simeq C_0^\infty(P,G)^G $; it is also isomorphic to the space of sections of the associated  bundle  $\AD P \doteq P\times_{\AD} G$ which are homotopic to the identity section, $\G \simeq\Gamma_0(\Sigma,\AD P)$. Summarising
\[
\G\simeq C_0^\infty(P,G)^G  \simeq\Gamma_0(\Sigma,\AD P),
\]
which is a smooth nuclear Fr\'echet manifold \cite[37.17, 42.21]{kriegl1997convenient}, which locally over $U\subset \Sigma$ is given by $\G\sim_U C_0^\infty(U,G)$. Its Lie algebra $\fG=\mathrm{Lie}(\G)$ can similarly be expressed in terms of $\Ad$-equivariant $\fg$-valued functions on $P$, or section of the adjoint bundle $\Ad P \doteq P \times_{\Ad} \fg$:
\[
\fG\simeq C^\infty(P,\fg)^G  \simeq\Gamma(\Sigma,\Ad P) \sim_U C^\infty(U,\fg).
\]
In other words, here $\Xi = \Ad P$. Similarly, we denote the dual vector bundle $\Xi^*$ by $\Ad^* P \doteq P \times_{\Ad^*} \dual{\fg}$.

The finite and infinitesimal actions on $\Acal$ are given by $A \mapsto A^g = {\Ad(g^{-1})\cdot} A + g^{-1}d g$ and $\rho(\xi)(A) = d_A \xi \doteq d \xi + [A,\xi]$. This action is proper \cite{Rudolph_2002}, and it admits slices whenever $\Sigma$ is compact \cite{AbbatiCirelliMania, DiezPhD}. 

\begin{remark}[Notation]\label{rmk:dualpairingYM}
The canonical pairing $\tr(\cdot\cdot)$ between elements of $\dual{\fg}$ and $\fg$ \emph{should not be confused} with the related, but distinct, pairing $\langle\cdot,\cdot\rangle$ between dual-valued forms and elements of $\fG$. Indeed, the two can be conflated when the dual-valued form is order-0, i.e.\ it is valued in sections of the dual vector bundle $\Xi^*\to\Sigma$. E.g.\ consider the special case of a \emph{order-0} $(\text{top},\bullet)$-form $\balpha$ and let $\xi\in\fG$:
\[
\text{if } \balpha\in\oloc^{\text{top},\bullet}(\Sigma\times\X,\dual{\fG}_{(0)}) \text{ then }\langle \balpha, \xi\rangle(x) = \tr(\balpha(x) \xi(x))\in\oloc^{\text{top},\bullet}(\Sigma\times\X).
\]
But this equality \emph{fails} for a general $\balpha $ which is \emph{not} order-0, i.e.\ for a general element of $\oloc^{\text{top},\bullet}(\Sigma\times\X,\dual{\fG})$, since then $\langle \balpha, \xi\rangle(x)$ can involve derivatives of $\xi(x)$. (Cf. Definitions \ref{def:dualvaluednew}, \ref{def:order-kduals}, and \ref{def:dual*}). 
\end{remark}

The geometric phase space of Yang--Mills theory is the fibrewise cotangent bundle of the space of principal connections $\Acal$ 
\[
\X \doteq T^\vee \Acal,
\] 
which in this case is a trivial bundle with fibre
\[
T^\vee_A\Acal = \Gamma(\mathrm{Dens}(\Sigma)\times_\Sigma T\Sigma\times_\Sigma \Ad^*P).
\]
Whenever $\Sigma$ is non-null, we have that the composition $\sharp\circ \star$ of the Hodge-star and musical-isomorphism operators associated to the positive-definite metric $\gamma$ over $\Sigma$, yields the isomorphism\footnote{This characterization of the electric fields in terms of $(\text{top}-1)$-forms is what one naturally gets when inducing symplectic data for YM theory using the covariant phase space method (Appendix \ref{app:covariant}).}
\begin{equation}
    \label{eq:mathcalE}
T^\vee \Acal \simeq
\mathcal{E}\doteq \Gamma(\wedge^{\text{top}-1}T^*\Sigma \times_\Sigma \Ad^*P) \sim_U \Omega^{\text{top}-1}(U,\fg^*).
\end{equation}
We denote $E$ the elements of $\mathcal{E}$ and call them \emph{electric fields}. This allows us to locally think of $E$ either as a $\fg^*$-valued vector density of weight $1$ or as $\fg^*$-valued $(\text{top}-1)$-forms on $\Sigma$. In our formulas we adopt the differential-form point of view unless otherwise stated. The fibrewise cotangent bundle is then identified with the Cartesian product
\[
\X \simeq \mathcal{E} \times \Acal.
\]

From the definition of $\X$ and the action of the gauge  group on $\Acal$, one deduces $E \mapsto E^g = \Ad^*(g)\cdot E$, i.e.\ $\L_{\rho(\xi)} E = \ad^*(\xi)\cdot E$. In other words, locally over $\Sigma$, the fundamental vector field $\rho(\xi)\in\Xloc(\X)$ reads
\begin{equation}
  \label{eq:rhoYM}  
\rho(\xi)(A,E) = (d_A \xi,\ad^*(\xi)\cdot E).
\end{equation}

\begin{remark}[Irreducible configurations]\label{rmk:YMirreducibleconfig}
In the case of $G$ non-Abelian semisimple, unless explicitly stated, we will hereafter ignore the existence of reducible configurations at which $\rho$ has a nontrivial kernel, and thus restrict our attention to the dense subset of $\X$ given by the cotangent of the space of irreducible connections---which we keep denoting $\Acal$ and $\X=T^\vee\Acal$.\footnote{This is not the same as the subspace of irreducible configurations of $\X$, which fails to be a cotangent space---in fact in the latter space one can consider pairs $(A,E)$ where both $A$ and $E$ are reducible but by different parameters.}
This is done for later convenience: see Assumption \ref{assA:isotropy} and Remarks \ref{rmk:isotropybundle1} and \ref{rmk:isotropybundle2}.
More generally, we expect that our procedure can be combined with the stratified reduction picture described in \cite{DiezHuebschmann-YMred} to treat reducible configurations as well. One could in fact consider each ``stratum'' given by those configurations whose stabilisers are conjugate to each other. Since the \emph{actual} stabilizer still depends on the point $(A,E)$, one gets a bundle version of our construction. We discuss this generalization in Section \ref{sec:LieAlgPict}.
\end{remark}

The local symplectic density on $\X$ is given by the local 2-form 
\begin{equation}\label{eq:bomegaYM}
\bom = \langle \bd E , \bd A \rangle = \tr( \bd E\bd A).
\end{equation}
Notice that this symplectic density $\bom$ is ultralocal (Definition \ref{def:ultralocal}).

From Equations \eqref{eq:rhoYM} and \eqref{eq:bomegaYM} it is easy to check that the action of $\G$ on $(\X,\bom)$ is locally Hamiltonian with local momentum form
\[
\langle \bH, \xi \rangle = (-1)^{\dim\Sigma}\ \tr( E d_A\xi).
\]
The local momentum form $\bH$ is order 1, since it depends on the first derivative of $\xi$ but no higher jets thereof.
Furthermore, we observe that the momentum form satisfies the following:
\[
\bK(\xi) \doteq \L_{\rho(\xi)}\bH - \ad^*(\xi)\cdot \bH=0
\]
a property that will be called ``strong equivariance'', Definition \ref{def:equivariance}.

We conclude this section by noticing that the phase space $(\X,\bom)$ over the spacelike surface $\Sigma$ can be obtained by following the procedure outlined in Appendix \ref{app:covariant} starting from the Lagrangian density $\boldsymbol{L} = \tfrac14 \tr(F_{\bar A}\star_M F_{\bar A})$, where $\star_M$ is the Hodge operator over $M$ (possibly Lorentzian) and $F_{\bar A}$ is the principal curvature 2-form of the principal connection $\bar A$ on a $G$-principal bundle $\bar P \to M$ which pullbacks to $P\to\Sigma$ along $\iota_{\Sigma}:\Sigma\hookrightarrow M$.

\subsection{The Lie algebroid picture} \label{sec:LieAlgPict}
We conclude Section \ref{sec:localfieldtheory} with the introduction of a different take on the action of $\G$ on $\X$ in terms of an (action) Lie algebroid $\A = \X\times\fG$, where $\G$ and $\X$ are put on a more equal footing. Not only this change of perspective points to mathematically natural generalizations, but it also allows us to perform constructions that are crucial even in the \emph{current} setup---notably Proposition \ref{prop:dualvaluedformdecomposition}, as well as the entirety of Section \ref{sec:cornerdata}.

We obtain a Lie algebroid by looking at the Lie-algebra action as the morphism of vector bundles
\begin{equation}\label{eq:A} 
\xymatrix@C=.5cm{
\A \doteq \X \times\fG \ar[drr]^-{p} \ar[rr]^-{\rho} & &  T\X \ar[d]\\
&&\X
}
\end{equation}
Here, the bundle $p:\A\to \X$ is trivial, and the space of sections of $p:\A\to\X$ can be equipped with a natural Lie bracket $\llbracket\cdot,\cdot\rrbracket$: 
\[
\llbracket\cdot,\cdot\rrbracket : \Gamma(\A)\times\Gamma(\A) \to\Gamma(\A),
\quad 
(\sigma,\tau) \mapsto \llbracket\sigma,\tau\rrbracket \doteq  [\sigma,\tau] + \L_{\rho(\sigma)}\tau - \L_{\rho(\tau)}\sigma.
\]
Then, seen as a map between spaces of sections, the action \eqref{eq:rhoaction} defines a homomorphism of Lie algebras:
\[
\rho: \Gamma(\A) \to \mathfrak{X}(\X), 
\qquad
[\rho(\sigma), \rho(\tau) ] = \rho(\llbracket \sigma, \tau \rrbracket).
\]
$(\A,p,\rho,\llbracket\cdot,\cdot\rrbracket)$ is called the \emph{action Lie algebroid} associated to $(\X,\G)$ with anchor $\rho$.
(Note: if the action $\rho$ is local, a \emph{local} version of this action Lie algebroid can be defined by allowing only local sections $\sigma\in\Gamma_\loc(\X,\A)$, i.e.\ only those sections which are local as maps between fibre bundles over $\Sigma$, $\sigma:\X \to \fG$.)

\medskip

The Lie algebroid $\A$ is itself the space of sections of a fibre bundle,
\[
\A=\X\times\fG=\Gamma(E\times_\Sigma \Xi).
\]
Hence, there is a space of local forms $\oloc^{\bullet,\bullet}(\Sigma\times\A)$, whose horizontal and vertical differentials we denote $d$ and $\cbd$ respectively:
\[
\left( \oloc^{\bullet,\bullet}(\Sigma\times\A),d,\cbd\right)
\]
where the $\check\bullet$ notation allows us to distinguish the differential $\cbd$ over $\A$ from the differential $\bd$ over $\X$. However, since $E\times_\Sigma\Xi$ is a product bundle, it is possible to decompose the local vertical differential into anticommuting differentials over each factor, respectively
\[
\cbd = \cbd_H + \cbd_V,
\]
with $\cbd_H$ essentially identified with $\bd$.
This in turn leads to a $3$-fold grading for the local forms over $\Sigma\times\A$, which is associated to the one horizontal ($d$) and the two vertical ($\cbd_H,\cbd_V$) differentials:
\begin{equation}
    \label{eq:p-horz-vert}
 \left(\oloc^{\bullet,\bullet}(\Sigma\times\A),d,\cbd\right)
\equiv
\left(\oloc^{\bullet,\bullet,\bullet}(\Sigma\times\X\times\fG), d,\cbd_H,\cbd_V\right).
\end{equation}
With reference to the bundle $p:\A\to \X$, we will henceforth refer to $(\cbd_H,\cbd_V)$ as to the \emph{$p$-horizontal} and \emph{$p$-vertical differentials}, respectively. Finally, whenever unambiguous, we identify
\[
\cbd_H\equiv\bd.
\]

\begin{remark}[Beyond action Lie algebroids]\label{rem:LieAlgbds}

In the action Lie algebroid setup, the fundamental vector fields $\rho(\xi)$ correspond to the image of constant sections of $p:\A\to\X$ under the \emph{anchor} map $\rho$. 
In physics parlance, non-constant sections of $\A\to\X$ are sometimes referred to as ``field-dependent gauge transformations'' (see e.g.\ \cite{RielloGomes,RielloSciPost} for an account).

Generalising our constructions to non-action Lie algebroids is possible in principle, at least when $\A$ is equipped with a (flat) connection. 
This allows to directly generalise the notions of $p$-horizontal and $p$-vertical forms delineated here as well as the associated 3-fold grading---the nilpotency of $\cbd_H$ as well as the vanishing of $[\cbd_H,\cbd_V]=0$ being consequences of flatness.
With such a connection at hand, one can replace the notion of constant sections, with that of sections which are constant w.r.t.\ $\cbd_H$. This generalization is in the spirit of \cite{blohmann2018hamiltonian} (see also \cite{KotovStrobl19}).

We note that (flat) connections are also relevant when considering ``infinitesimal ideal systems'' for Lie algebroids \cite{JotzLeanOrtiz}; these generalizations of the notion of a Lie algebra ideal to Lie algebroids are likely to play a role in the generalization of the bulk/corner split of the gauge group discussed in Section \ref{sec:BulkCornerGrps}. 
\end{remark}

\begin{remark}[Hamiltonian Lie algebroids and momentum sections]\label{rmk:HamLieAlg}
In this more general setting, if the fibres $p^{-1}(\phi)\subset \A$ are isomorphic to the space of sections of a vector bundle, $\calW=\Gamma(\Sigma,W)$, local forms in $\iloc^0(\A)$ provide the infinite-dimensional, local, analogue of the momentum sections introduced in \cite{blohmann2018hamiltonian} as generalization of the standard momentum maps for finite-dimensional symplectic spaces with a Lie-algebroid symmetry.
\end{remark}

%

\section{Constraint and Flux momentum maps}\label{sec:constraintsandfluxes}

\subsection{Basic definitions and main assumptions}

So far we have discussed locally Hamiltonian (phase) spaces quite abstractly, without any reference to an underlying Lagrangian gauge theory from which they could be derived, and only paying little attention to the presence of corners. In this section we input additional requirements coming from the treatment of physical models.

Key to the entire construction is the source/boundary decomposition of the local momentum form $\bH$ of Definition \ref{def:Ham+mommaps}:

\begin{definition}[Constraint and flux forms]\label{def:momentumformdecomposition}
Let $(\X,\bom,\G,\bH)$ be a locally Hamiltonian $\G$ space, with $\bH\in\oloc^{\mathrm{top},0}(\Sigma\times \X,\fG^*)$ the local momentum form. Following Proposition \ref{prop:dualvaluedformdecomposition}, we decompose $\bH = \bHo + d\bh$ into the sum of a order-0 and a $d$-exact local form. We call $\bHo$ the \emph{constraint form} and $d\bh$ the \emph{flux form}. Moreover, we call $\Ho\doteq\int_{\Sigma}\bHo$ the \emph{constraint map} and $\h\doteq\int_\Sigma d\bh$ the \emph{flux map}.
\end{definition}

\begin{definition}[Constraint set]
The vanishing locus $\C \subset \X$ of the constraint form is said the \emph{constraint set}, i.e.\footnote{We are using here characterization (2) of Remark \ref{rmk:co-}, that is we are thinking of $\bHo$ as a map $\X \to \osrc^\text{top}(\Sigma,\fG^*)$.} $\C\doteq \bHo^{-1}(0)$.
\end{definition}

\begin{remark}[Off- vs.\ On-shell]\label{rmk:theshell}
To distinguish between quantities that are supported on $\X$ or only at the constraint set $\C\subset\X$, we introduce the (informal) terminology of, respectively, \emph{off-} and \emph{on-shell}. 
\end{remark}

\begin{definition}[On-shell isotropy locus]\label{def:isotropylocus}
Let $\iota_\C^*\A=\C\times\fG \to \C$ be the pullback vector bundle over $\C\hookrightarrow\X$. Then, its subset
\[
\mathsf{I}_\rho \doteq \{ (\phi,\xi)\in \C \times \fG \ | \ \rho(\phi,\xi) = 0 \} \subset \iota_\C^*\A
\]
is said the \emph{on-shell isotropy locus} of the action $\rho$.
\end{definition}

\begin{remark}[Co-flux map]\label{rmk:coflux}
In the spirit of Remark \ref{rmk:co-}, the flux map $\h$ can either be viewed as an integrated local form $h \in\iloc^{0}(\A)$, or as a dual-valued functional $\X\to \dual{\fG}_\loc$ (Definition \ref{def:dualvaluednew}). Additionally, it makes sense to think of $h^\text{co}$ as a \emph{co-flux} map, i.e as a linear map:
\[
\h^\text{co} : \fG \to \iloc^{0}(\X), \qquad \bd\h^{\text{co}}\colon \fG \to \iloc^{1}(\X).
\qedhere
\]
\end{remark}

We can now formulate the main set of assumptions used in this work, which will be \emph{left implicit unless stated otherwise}. 

\begin{assumption}[Main assumptions]\label{ass:setup}
$(\Sigma,\pp\Sigma)$ is a compact, oriented, connected, manifold with boundary. $(\X,\bom)$ is a locally symplectic space (Definition \ref{def:locsymp})
modeled on sections of some fibre bundle over $\Sigma$ that admits a global section\footnote{This condition ensures surjectivity of $j^\infty$, cf.\ Equation \eqref{eq:localforms-jinfty} and \cite[Proposition 3.1.14]{BlohmannLFT}.}, with the following properties:
\begin{enumerate}[label=(\ref{ass:setup}.\arabic*)]
    \item\label{assA:item1} $(\X,\bom)$ is path-connected,\footnote{\label{fn:pathconnected1}This hypothesis is needed to deduce that if a $(\bullet,0)$-local form $\boldsymbol{\nu}$ on $\X$ is such that $\bd \boldsymbol{\nu} = 0$, then $\boldsymbol{\nu}(\phi)=\boldsymbol{\nu}(\phi_\smbullet)$ for an arbitrary $\phi_\smbullet\in\X$. It is straightforward to generalise our results to the union of path-connected components.} and \emph{locally-Hamiltonian} w.r.t.\ the 
    action of a local Lie group $\G$ with $\Sigma$-local Lie algebra $\fG = \Gamma(\Sigma,\Xi)$ (Definitions \ref{def:Ham+mommaps} and \ref{def:LocLieAlg}), and local momentum form $\bH\in\oloc^{\text{top},0}(\Sigma\times\X,\dual{\fG})$:
    \[
        \bi_{\rho(\xi)}\bom = \bd\langle \bH,\xi\rangle;
    \]
    \item\label{assA:equiv} the local momentum form $\bH$ is weakly equivariant (Definition \ref{def:equivariance}), i.e.
    \[
    \bC = d\bc;
    \]
    \item\label{assA:constr} the constraint set $\C\subset \X$ is a smooth and path-connected submanifold;
    \item\label{assA:isotropy} 
    the on-shell isotropy locus $\mathsf{I}_\rho$ is a vector subbundle $\mathsf{I}_\rho \to \C$ of $\A\to\X$; we thus call it the \emph{on-shell isotropy bundle} and denote its (generic) fibre by $\mathfrak{I}\subset\fG$ and call it the \emph{on-shell isotropy algebra}.\footnote{To see that the fibre $\mathfrak{I}$ is a subalgebra of $\fG$, note that a constant section $\xi$ of $\iota_\C^*\A$ is by definition also a (constant) section of the on-shell isotropy bundle $\mathsf{I}_\rho$ iff $\rho(\xi)\vert_\C=0$. Then, if $\chi_{1,2}$ are two constant sections of $\mathsf{I}_\rho$, on $\C$ one has $\rho([\chi_1,\chi_2])\vert_\C = [\rho(\chi_1),\rho(\chi_2)]\vert_\C= 0$ and thus one concludes.}
    \qedhere
\end{enumerate}
\end{assumption}

\subsection{Remarks on Assumption \ref{ass:setup}}

We collect here a few observations on the general assumptions we just outlined, and on their scope.

\begin{remark}[Constraints and the Noether current]\label{rmk:Noethercurrent}
Our momentum form decoposition is a straightforward application of the theory of local forms, but it has an important physical interpretation. Throughout, we will work with $\C$ being \emph{defined} as the vanishing locus of $\bHo$ or, equivalently, of $\Ho$, while forgoing all reference to the Lagrangian origin of the locally Hamiltonian field theory. Noether's theorem is the key link that relates the equation of motion for the associated Lagrangian field theory (say on a cylinder $\Sigma\times [0,1]$) to $\C$. In that case, $\C = \Ho^{-1}(0)$ is the statement that the vanishing locus of $\Ho$ coincides with the canonical constraints of the Lagrangian field theory. This is a consequence of the classical analysis of \cite{LeeWald}. This analysis proves that the momentum form $\bH$ corresponds to the (pullback of) the Noether current on $\Sigma$, and thus---in virtue of Noether's theorem---that $\bH$ is $d$-cohomologous to zero on shell: $[\bH]_d\approx 0$. Hence, physical configurations are determined by the vanishing of $\bH$ ``up to boundary terms'', a statement concisely summarized by Definition \ref{def:momentumformdecomposition}. Finally, note that not all Lagrangian theories with local symmetries give rise to Hamiltonian $\G$-spaces, but \emph{if} one does, \emph{then} the picture described here holds. (An important counterexample is provided by General Relativity and diffeomorphism symmetry, cf.\ \cite{LeeWald}.)
\end{remark}

\begin{remark}[Smoothness: geometric vs.\ cohomological description]\label{rmk:C-smoothness} 
This paper focuses primarily on describing the structure of (symplectic) reduction for locally Hamiltonian gauge field theories \emph{when smooth}: in other words, we are only marginally interested in looking for sufficient conditions for smoothness of $\C$ and its quotients.
Indeed, we are more generally interested in developing a framework that---although able to touch on these topics---looks at cohomological resolutions of quotients as a preferred alternative to their geometric description, especially when they are not smooth, e.g.\ within the Batalin--Fradkin--Vilkovisky method (see \cite{BV1,BV2,BV3,StasheffConstraints88,SchaetzBFV}).
Nonetheless, regarding smoothness, let us note that over a general Fr\'echet manifolds the inverse function theorem is not available, and one needs to rely on other methods to prove the smoothness of $\C = \Ho^{-1}(0)$ (on \emph{tame} Fr\'echet manifolds one has the Nash--Moser theorem). 
Following \cite{ArmsMarsdenMoncrief1981} we know that, besides infinite-dimensional issues, the singularities in $\C$ are linked to the stabilizers of the infinitesimal action. This links Assumption \ref{assA:constr} with Assumption \ref{assA:isotropy}.
This said, in the examples we will consider in this paper, smoothness of $\C$ is granted. (See also Remarks \ref{rmk:smooth-C-revisited}, \ref{rmk:smooth-uCo} for related comments on the smoothness of $\C$ and the corresponding first- and second-stage reduced spaces. As for connectedness, we only need it as a simplifying assumption, but our results can be extended to non-connected $\C$.)
\end{remark}

\begin{remark}[Isotropy locus: the general case]\label{rmk:isotropybundle1}
The on-shell isotropy locus encodes the configurations with non-trivial reducibility parameters. Physically, these can be understood as ``global symmetries'' which (by the $\fG$-equivariance of $\rho$) are well-defined even up to gauge. When configurations with non-trivial reducibility parameters are present, the gauge action is not free. Moreover, if the ``size'' of the reducibility algebra\footnote{In many gauge theories of physical interest, such as YM theory, the reducibility algebra of any given configuration is finite dimensional and its ``size'' is easily codified by its dimension.} is not constant over $\C$, symplectic reduction is singular \cite{ArmsMarsdenMoncrief1981,sjamaar1991,DiezHuebschmann-YMred}.
Assumption \ref{assA:isotropy}, i.e.\ the requirement that all on-shell configurations $\phi\in\C$ have \emph{isomorphic} stabilisers $\ker(\rho_\phi) \simeq \mathfrak{I}$ (also known as reducibility algebras), guarantees that the ``order'' of the induced anchor $\rho\vert_\C$ is constant.
For example, in theories of a $\fg$-valued gauge connection $A$, such as Yang--Mills, Chern--Simons or $BF$, $\X$ decomposes into \emph{subsets of orbit type,} given by configurations whose isotropy subgroups\footnote{In YM and Chern--Simons theory, these are finite-dimensional subgroups of $\G$ isomorphic to subgroups of $G$. In $BF$ the same statement applies with $G$ replaced by $T^*G$.} are conjugate to each other in $\G$ (e.g. \cite[App. A.2]{DiezPhD}).
Assumption \ref{assA:isotropy} is then satisfied if we look at the intersection of $\C$ with each subset of orbit type.
\end{remark}

\begin{remark}[Isotropy bundle: the trivial case]\label{rmk:isotropybundle2}
There are a few cases of interests in which the isotropy bundle is a \emph{trivial} subbundle, i.e.\ all configurations $\phi\in\C$ have the \emph{same} isotropy, $\ker(\rho_\phi) = \mathfrak{I}\subset\fG$, so that one further has
\[
\mathsf{I}_\rho = \C \times \mathfrak{I}.
\]

In YM theory, this happens e.g.\ in the following cases:
\begin{enumerate}
    \item\label{rmk:iso2-item1} if $G$ is Abelian, for then $\mathfrak{I} \simeq \fg$ is the subalgebra of $d$-closed gauge parameters (which is a $\phi$-independent condition);
    \item in non-Abelian YM, provided that we restrict to the (fibrewise cotangent bundle of the) space of irreducible configurations, for then $\{0\}=\ker(\rho_\phi) \simeq \mathfrak{I}$. If $G$ is semisimple, irreducible configurations constitute a dense subset. 
\end{enumerate}

All our theorems work regardless of the triviality of $\mathsf{I}_\rho\to\C$. However, in the non-trivial case some results will be ``sub-optimal'' (cf.\ the discussion of non-Abelian $BF$ theory in Section \ref{sec:BFtheory}). In fact, this scenario naturally calls for a generalization of our construction to more general (i.e.\ non-action) Lie algebroids.
Indeed, these can be accommodated by working  with the notion of Hamiltonian Lie algebroids developed in \cite{blohmann2018hamiltonian}.
We leave this line of investigation for future work.
\end{remark}

\subsection{Weak-equivariance of the local momentum form}

We now show that Assumption \ref{assA:equiv} of weak equivariance of $\bH$ is equivalent to the strong equivariance of the constraint form $\bHo$, which in turn implies that the constraint set $\C\subset\X$ is preserved by the action of $\G$---i.e.\ that \emph{all} gauge transformations are tangent to the constraint set $\C$ (here assumed smooth). This guarantees that the space of on-shell configurations modulo gauge transformations, i.e. $\uuC \doteq \C/\G$, is well-defined as a set of equivalence classes. As observed in Remark \ref{rmk:Noethercurrent}, in physical theories where a momentum form exists, this is a direct consequence of Noether's theorem. In the language of Dirac and Bergmann, the strong equivariance of $\bHo$ is related to the constraints being first-class \cite{henneaux1992quantization}.

\begin{proposition}[Bulk and corner equivariance]\label{prop:equi}
Assumption \ref{assA:equiv} holds if and only if the constraint local form $\bHo$ is strongly $\G$-equivariant while the flux form $d\bh$ is only weakly equivariant, i.e.\ up to the $d$-exact CE cocycle $d\bc$. That is, for all $g\in\G$,
\[
 g^*\bHo - \Ad^*(g)\cdot\bHo =0
\quad\text{and}\quad
 g^*\bh - \Ad^*(g)\cdot\bh = d\bc(g).
\]
Infinitesimally, this yields:
\[
\L_{\rho(\xi)}\bHo - \mathrm{ad}^*(\xi)\cdot\bHo = 0
\quad\text{and}\quad
\L_{\rho(\xi)}d\bh - \mathrm{ad}^*(\xi)\cdot d\bh = d\bk(\xi).
\qedhere
\]
\end{proposition}
\begin{proof}
Recall the decomposition of $\bH = \bHo+d\bh$ (Definition \ref{def:momentumformdecomposition}), the definition of the CE cocycle $\bC(g) = g^*\bH - \Ad^*(g)\cdot\bH$ (Proposition \ref{prop:CKcocycles}), as well as Assumption \ref{assA:equiv} ($\bC=d\bc$). Hence, introduce for each $g\in\fG$ the identically-vanishing local form 
$\boldsymbol{S}_g\in \oloc^{\text{top},0}(\Sigma\times\A)$, defined by
\[
0 \equiv \boldsymbol{S}_g \doteq g^*\bHo - \Ad^*(g)\cdot\bHo 
+ d \big(  g^*\bh - \Ad^*(g)\cdot\bh - \bc(g)   \big),
\]
where e.g.\ $(g^*\bHo - \Ad^*(g)\cdot\bHo)(\phi,\xi) = \langle\bHo(\phi\triangleleft g),\xi\rangle - \langle\bHo(\phi), \Ad(g)\cdot\xi\rangle$.

Consider its $p$-vertical differential $0\equiv\cbd_V\boldsymbol{S}_g\in\oloc^\text{\text{top},1}(\Sigma\times\A)$, and observe that since $\bHo$ is order-0, $\cbd_V (g^*\bHo - \Ad^*(g)\cdot\bHo)$ is of source-type, whereas the remainder of $\cbd_V\boldsymbol{S}_g$ is boundary-type, i.e.\ $d$-exact. The result readily follows from application of Theorem \ref{thm:LocFormDec} on the uniqueness of the source/boundary decomposition of $(\text{top},1)$-local forms to the expression $\cbd_V\boldsymbol{S}_g\equiv0$. The converse implication is immediate.
\end{proof}

\begin{corollary}\label{cor:tangent}
For all $g\in\G$, $g^*\C \subset\C$. 
Infinitesimally, denote $\rho\vert_\C : \fG \times \C \to T_\C \X$ the restriction of $\rho$ to the constraint set; then, $\Im(\rho\vert_\C) \subset T\C$ and thus defines the on-shell anchor $\rho_\C :  \iota_\C^*\A \to T\C$ making $\iota_\C^*\A = \C \times \fG$ an action Lie algebroid.
\end{corollary}

\begin{proof}
Recall that $\C=\bHo^{-1}(0)$. Thus, for the first statement it is enough to prove that if $\bHo(\phi)=0$ then also $g^*\bHo(\phi) \equiv \bHo(\phi\triangleleft g) =0 $---which follows from the first equation of the previous proposition. The second statement similarly follows by considering infinitesimal, rather than finite, actions.
\end{proof}

\subsection{Isotropy locus revisited}

We conclude this discussion with a lemma meant to showcase the utility of thinking of the local momentum form $\bH$ as an auxiliary Lagrangian over $\A= \X\times\fG$,
\[
\bH \in \oloc^{\text{top},0}(\Sigma\times\X,\dual{\fG}) \equiv \oloc^{\text{top},0}(\Sigma\times\A),
\]
a viewpoint that will turn useful in Section \ref{sec:cornerdata}. We start with a definition:

\begin{definition}\label{def:momentmapSourcedecomp}
Let $\bH \in \oloc^{\text{top},0}(\Sigma\times\A)$. Then, by means of Theorem \ref{thm:LocFormDec}, we define the source form $\bE\in\osrc^{\text{top},1}(\Sigma\times\A)$ and the boundary form $d\be\in d\oloc^{\text{top}-1,1}(\Sigma\times\A)$ through the equation
\[
\cbd \bH \doteq \bE + d\be .
\qedhere
\]
\end{definition}

\begin{lemma}[Isotropy locus as Euler--Lagrange locus of $\bH$]\label{lemma:checkC}
Let $\A$, $\bH$ and $\bE$ be as in Definition \ref{def:momentmapSourcedecomp}. Then, if $\bom_\text{src}^\flat$ is injective (cf.\ Lemma \ref{lem:nondegeneracystrength}), the on-shell isotropy locus $\mathsf{I}_\rho\hookrightarrow\iota_\C^*\A\subset \A$ coincides with the Euler--Lagrange locus $\{ \bE = 0 \}\subset\A$.
\end{lemma}

\begin{proof}
Recall the notion of $p$-horizontal/vertical differentials over $\A$, $\cbd = \cbd_H + \cbd_V$ (Equation \eqref{eq:p-horz-vert}).
Splitting the source and boundary forms $\bE$ and $d\be$ into their $p$-horizontal and $p$-vertical components, as in $\bE = \bE_H + \bE_V$,  we obtain $\cbd_H \bH = \bE_H + d \be_H$ and $\cbd_V \bH = \bE_V + d\be_V$, respectively. Notice that $\bE_{H/V}$ are both of source type. Since $\bE_{H/V}$ are linearly independent as forms, the vanishing locus of $\bE$ is the intersection of the vanishing loci of $\bE_H$ and $\bE_V$.

First, consider the $p$-horizontal term $\bE_H$. Using the locally-Hamiltonian property (Definition \ref{def:Ham+mommaps}), we find 
\[
\bE_{H}  \equiv \Pi_\text{src}( \cbd_H \bH  )
= \Pi_\text{src} (\bi_{\rho(\cdot)}\bom ) \equiv \bom^\flat_\text{src}(\rho(\cdot)).
\]

Since $\bom^\flat_\text{src}$ is injective by assumption,
\begin{equation}\label{eq:Ephi-nonfree}
\bE_{H}(\phi,\xi)=0 \; \iff \; \xi \in \ker(\rho_\phi).
\end{equation}

Consider now the $p$-vertical component of $\cbd\bH = \bE+ d\be$.
Since, by Proposition \ref{prop:dualvaluedformdecomposition} and Definition \ref{def:momentumformdecomposition}, $\bH = \bHo + d\bh$ is a dual-valued form that decomposes into a ``bulk'' constraint form $\bHo$ \emph{of order-0}, and a $d$-exact flux form $d\bh$, it is immediate to see that $\cbd_V\bH = \langle \bH , \cbd_V\xi\rangle$ and therefore the $p$-vertical forms $\bE_{V}$ and $d\be_{V}$  are given by
\[
\bE_{V}= \langle\bHo, \cbd_V\xi\rangle,
\qquad
d\be_{V} = d\langle\bh,\cbd_V\xi\rangle.
\]

From this we readily conclude that 
\begin{equation}\label{eq:Exi}
\bE_{V}(\phi,\xi)=0 \; \iff \;
\bHo(\phi)=0
\; \iff \;
\phi \in \C.
\end{equation}
Together, Equations \eqref{eq:Ephi-nonfree} and \eqref{eq:Exi} return the definition of $\mathsf{I}_\rho$ (Definition \ref{def:isotropylocus}).
\end{proof}

\subsection{Yang--Mills theory: constraint/flux decomposition}\label{sec:runex-decomposition}

We continue here our analysis of YM theory we initiated in Section \ref{sec:runex-setup}, where we have already shown that this theory satisfies Assumption \ref{assA:item1} and \ref{assA:equiv}. 

Now, from the formula $\langle \bH,\xi\rangle = (-1)^{\Sigma}\tr(E d_A\xi)$, one readily sees that the local momentum form $\bH$ is strongly equivariant and decomposes into the sum of the (Gauss) constraint form $\bHo$ and the flux form $d\bh$:
\[
\langle \bHo,\xi\rangle = \tr((d_AE)\xi),
\qquad
\langle d \bh,\xi\rangle = -d \tr(E\xi).
\]
Physically viable YM configurations must satisfy the Gauss constraint, i.e.\ they must belong to the Gauss constraint set (Assumption \ref{assA:constr})
\[
\C = \{ (E,A)\in\X\ |\ d_A E = 0\} \subset \X.
\]

The dual-valued form $\bHo$ is of order-$0$, and therefore $C^\infty(\Sigma)$-linear in $\xi$.
Conversely, neither $\bH$ nor $d\bh$ are $C^\infty(\Sigma)$-linear in $\xi$, but only $\mathbb{R}$-linear and, indeed, order-1 as they involve one horizontal ($d$) derivative of $\xi$.
In particular $d\bh$ should be thought of as a linear map $\fG\to d\oloc^{\text{top}-1,0}(\Sigma\times \X)$ meaning that---when evaluated on $\xi\in\fG$---it returns a $d$-exact top-form on $\Sigma$. 

The flux \emph{map}, evaluated at $\xi\in\fG$, is given by the smearing against $\xi$ of the electric flux through $\pp\Sigma$, hence its name:
\[
h(\xi) = \int_{\pp\Sigma} \tr(\xi E).
\]

An explicit description of the constraint set $\C=\bHo^{-1}(0)$ in terms of a Hodge-like decomposition of $E\in\mathcal{E}_A$ will be provided in Section \ref{sec:runex-fluxannihilators}.

The isotropy locus is the set
\[
\mathsf{I}_\rho = \{ (A,E, \xi) \in \C\times \fG \ : \ d_A \xi = 0 = [E,\xi] \}.
\]
For a fixed $A$, the equation $d_A \xi =0$ is an over-determined, first-order, elliptic equation, whose solution throughout $\Sigma$---when it exists---is fully determined by the value of $\xi$ at any one point $x\in\Sigma$. This implies that  the dimension of the reducibility algebra of any given $(A,E)$ is at most $\dim(\fg)$.

If $G$ is Abelian, then $\mathsf{I}_\rho = \C \times \mathfrak{I}$ with $\mathfrak{I} = \fg \hookrightarrow \fG$ as the subalgebra of constant gauge transformations (note that, in the Abelian case, this notion is well-defined even over a non-trivial principal bundle). Assumption \ref{assA:isotropy} is satisfied.

If $G$ is semisimple, on the other hand, the equation $d_A\xi=0$ has nontrivial solutions $\xi\in\fG$ only if $A$ is reducibile. To satisfy Assumption \ref{assA:isotropy}, in Remark \ref{rmk:YMirreducibleconfig} we introduced $\X = T^\vee\Acal$ as the fibrewise cotangent bundle of $\Acal$ the subset of \emph{irreducible} configurations, which are dense in $\mathrm{Conn}(P\to\Sigma)$ (see also Remarks\ref{rmk:isotropybundle1} and \ref{rmk:isotropybundle2}).

\begin{remark}[Abelian gauge theory with matter fields]
Notice that point \ref{rmk:iso2-item1} of the previous remark is affected by the presence of $G$-modules, which one can interpret as charged matter fields. In this case reducibility is generally lost and the corresponding transformations, now nontrivially generated by the charge density, are best understood as global symmetries. E.g. \cite{AbbottDeser82,RielloGomes} and \cite[Chap.7]{Strocchi13}.
\end{remark}

\subsection{The flux map}\label{sec:flux}

In the absence of corners, $\partial\Sigma=\emptyset$, the flux map $h= \int_{\Sigma}d\bh$ vanishes, and the constraints---encoded in the constraint form $\bHo$---are the only information one needs: all functions on $\uCo \doteq \C / \C^\omega$  are observables, i.e. gauge invariant, since then $\uCo \simeq \C/\G$ \cite{DiezPhD}.

In the context of \emph{locally} Hamiltonian systems with nontrivial corners, $\partial\Sigma\neq\emptyset$, the role of the (corner) flux map is instead front and center for it controls the two-stage reduction procedure, as we will see. 
For technical reasons, however, the  flux form $d\bh$ needs the following ``adjustments'':

\begin{definition}[Adjusted flux form]\label{def:adjpreflux}
Let $\phi_\smbullet\in\C$ be a fixed (albeit arbitrary) reference configuration. Then, the \emph{adjusted flux form and map} (with respect to $\phi_\smbullet$) are 
\[
d\bh_\smbullet \doteq d\bh - d\bh(\phi_\smbullet) , \qquad h_\smbullet \doteq \int_\Sigma d\bh_\smbullet.\qedhere
\]
\end{definition}

\begin{remark}
    Note that in Assumption \ref{ass:setup} we required connectedness of $\C$. If this assumption is relaxed, an adjusted flux map $d\bh_\smbullet^i$ has to be defined with respect to every connected component $\C^i$ of $\C$. Consider $\iota^*_\C d\bh^i_\smbullet \doteq \iota^*_\C d\bh^i - \iota^*_\C d\bh(\phi_\smbullet^i)$, where $\phi_\smbullet^i\in\C^i$ and $\iota^*_\C d\bh^i=\iota^*_\C d\bh$ on $\C^i$ and zero elsewhere. 
\end{remark}

\begin{definition}[On-shell flux map]\label{def:flux}
Let $h_\smbullet\in\iloc^0(\X,\dual{\fG})$ be the adjusted flux map. The \emph{on-shell adjusted flux map} (with respect to $\phi_\smbullet\in\C$) is the pullback $\iota_\C^*h_\smbullet$.
\end{definition}

\begin{remark}[About on- and off-shell fluxes]\label{rmk:flux}
In certain instances, it will be necessary to distinguish the ``on-shell flux map''  $\iota_\C^*\h_\smbullet$  from the flux $\h$ \emph{tout court}, which we call ``off-shell flux map.''
A few observations:
\begin{enumerate}
    \item All flux maps and forms, either on- or off-shell, are functions of the restriction of fields at $\partial\Sigma$ (possibly through transverse jets). Hence, if $\partial\Sigma=\emptyset$ they all vanish.
    \item Adjusted flux maps (and forms) vanish at the reference configuration with respect to which they are defined, $h_\smbullet(\phi_\smbullet)=0$.
    \item Since $d\bh(\phi_\smbullet)$ is constant over $\X$, the (local) Hamiltonian flow equations can be rewritten as
    \begin{equation}\label{eq:intHamfloweq}
        \bi_{\rho(\cdot)}\bom = \bd\bHo + \bd d\bh_\smbullet 
        \quad\text{and}\quad
        \bi_{\rho(\cdot)}\omega = \bd\Ho + \bd\h_\smbullet.
    \qedhere
    \end{equation}
\end{enumerate}
\end{remark}

\begin{proposition}[Flux equivariance]\label{prop:adjequi}
Let $d\bc$ and $d\bk$ be as in Proposition \ref{prop:equi}. 
Define
\[c_\smbullet \doteq \int_\Sigma d\bc + \delta d\bh(\phi_\smbullet), \qquad k_\smbullet \doteq \int_\Sigma d\bk - \delta'd\bh(\phi_\smbullet)
\] 
with $\delta$ (resp. $\delta'$) the CE group (resp. algebra) differentials (Appendix \ref{app:CEcoh}).
Then, for all $g\in\G$ and $\xi\in\fG$ 
\begin{align}\label{eq:starredcocycle}
    c_\smbullet(g)  = g^*h_\smbullet - \Ad^*(g)\cdot h_\smbullet, \qquad 
 \k_\smbullet(\xi)   =\L_{\rho(\xi)}h_\smbullet - \ad^*(\xi)\cdot h_\smbullet, 
\end{align}
with $c_\smbullet(g)=\h_\smbullet(\phi_\smbullet\triangleleft g)$. 
These equations hold for the \emph{on-shell} flux map as well, i.e.\ upon the replacement of $\h_\smbullet$ by $\iota_\C^*\h_\smbullet$.
We call $c_\smbullet$ and $\k_\smbullet$ the \emph{corner (CE) group and algebra cocycles}.
\end{proposition}

\begin{proof}
Equation \eqref{eq:starredcocycle} is an immediate consequence of the definition of the adjusted flux form as a shifted version of the flux form, together with Proposition \ref{prop:equi} (cf.\ Remarks \ref{rmk:HamShift2}).
To prove that $c_\smbullet(g) = h_\smbullet(\phi_\smbullet\triangleleft g)$, it is enough to evaluate the first of Equations \eqref{eq:starredcocycle} at $\phi=\phi_\smbullet\in\C$ where $\h_\smbullet(\phi_\smbullet) = 0$. 
Finally, the on-shell statement follows from Corollary \ref{cor:tangent} stating $g^*\C\subset \C$.
\end{proof}

\begin{remark}[Equivariance, revisited]\label{rmk:equipushfwd}
For later convenience we rewrite the equivariance property of the adjusted flux map $\h_\smbullet$ as the push-forward along $\h_\smbullet :\X \to \dual{\fG}_\mathrm{loc}$ of the Hamiltonian action $\rho$ of $\fG$ on $\X$ to the action $\ad^*_K$ of $\fG$ on $\dual{\fG}_\mathrm{loc}$ associated to the corner cocycle $K=k$ (Definition \ref{def:affinecoad}):
\[
(\h_\smbullet)_* \rho(\xi) = \ad^*_{k}(\xi) 
\]
where we identified $T_{\h_\smbullet(\phi)}\dual{\fG}_\mathrm{loc}$ with $\dual{\fG}_\mathrm{loc}$.
\end{remark}

Finally, we introduce the important notions of \emph{off- and on-shell flux spaces:}
\begin{definition}[Flux spaces]\label{def:fluxspaces}
Let $h_\smbullet : \X \to \dual{\fG}_\mathrm{loc}$ be the adjusted flux map, and define the \emph{off}- and \emph{on-shell flux spaces}, respectively, as
\[
\Foff \doteq \Im(\h_\smbullet)
\quad\text{and}\quad
\F \doteq \Im(\iota_\C^*\h_\smbullet)
\]
Then, $\F\subset\Foff\subset\dual{\fG}_\mathrm{loc}$ and we denote their elements, said \emph{off}- and \emph{on-shell fluxes}, by the same symbol $f \in \F^{({\off})}$. 
\end{definition}

\begin{lemma}\label{lemma:FPoisson}
The flux space $\F$ (resp.\ $\Foff$) is a subset of $\fG^*$ invariant under the action of $\ad_K^*$, which is  generated by $\Pi^K_{\fG^*}$ (Lemma \ref{lemma:KKS}). The orbits of $f\in \F$ (resp. $\Foff$) under said (coadjoint/affine) action are symplectic submanifolds of $\fG^*$. We denote the KKS symplectic structure on these orbits by $\Omega_{[f]}$.
\end{lemma}

\begin{proof}
As observed in Lemma \ref{lemma:KKS} and Remark \ref{rmk:PoissonLieAlg}, the Poisson bivector $\Pi_{\dual{\fG}}^{K}$ satisfies $\Pi_{\dual\fG}^{K\sharp} = \ad_K^*$ and thus restricts from $\dual{\fG}$ to $\dual\fG_\loc$, since the coadjoint/affine action of $\G$ on $\dual{\fG}$ is local (Definition \ref{def:localaction}). The off-shell flux space $\Foff \simeq \Im(h_\smbullet)$ is an invariant subset of $\dual{\fG}_\mathrm{loc}$ with respect to the affine or coadjoint action $\ad^*_K$ of $\G$ on $\dual{\fG}_\mathrm{loc}$ (Definition \ref{def:affinecoad}), in virtue of the fact that $\h_\smbullet$ is equivariant (possibly up to a CE cocycle; see Proposition \ref{prop:adjequi}, Remark \ref{rmk:equipushfwd}).
Therefore $\Pi_{\dual{\fG}}^K$ also restricts to $\Foff\subset\fG^*_\loc \subset \dual{\fG}$.
The remaining statements about the symplectic leaves of the (restricted) Poisson structure on $\Foff$ follow since invariance of $\Foff$ implies that the (symplectic) affine/coadjont orbit of $f\in\Foff$ is entirely contained in $\Foff$. Finally, the corresponding statements for the on-shell flux space $\F$ follow in an analogous manner from the observation that $\iota_\C^* \h_\smbullet$ is also equivariant (Proposition \ref{prop:adjequi}).
(Recall that in Assumption \ref{assA:constr} we required $\C$ to be smooth in $\X$.) 
\end{proof}

\begin{remark}[On the choice of $\phi_\smbullet$]
The subsets $\F$ and $\Foff$ of $\fG^*_\loc$ depend on a choice of reference configuration $\phi_\smbullet\in\C$ through $\h_\smbullet\doteq h - h(\phi_\smbullet)$. If we choose another reference configuration $\phi_\smbullet'$ and define $\h_\smbullet' \doteq h - h(\phi_\smbullet')$ we get that $\h'_\smbullet - \h_\smbullet = h(\phi_\smbullet) - h(\phi'_\smbullet)\doteq \Delta_\smbullet\in\fG^*$, which is a constant on $\X$. We can thus define $\F'\doteq \mathrm{Im}(\h_\smbullet')$ and immediately observe that $\F' = \F + \Delta_\smbullet$. From this one can easily conclude that the preimage of an element $f\in\F$ is actually independent of the choice of reference configuration, in the sense that
\[
\phi\in\h_\smbullet^{-1}(f) \iff \phi \in (\h_\smbullet')^{-1}(f + \Delta_\smbullet).
\]
A similar statement holds on-shell. Note, finally, that  $0\in\F$ for any choice of $\phi_\bullet$, by construction.
\end{remark}

\section{Constraint and flux gauge groups}\label{sec:constraintfluxgroups}

\subsection{Flux annihilators}\label{sec:fluxann}
The following definition will play a central role in setting up the symplectic reduction in stages, and in the description of off-shell corner data (Section \ref{sec:cornerdata}).

\begin{definition}[Flux annihilators]\label{def:frakN}
The \emph{on- and off-shell flux annihilators} are, respectively, the following subsets\footnote{Recall: if $\mathcal{X}\subset \mathcal{W}$ and $\mathcal{Y}\subset \calW^*$, then $\Ann(\mathcal{X},\mathcal{Y}) \doteq \{ y\in \mathcal{Y} \ : \ \langle y,x\rangle = 0 \quad \forall x\in\mathcal{X} \}$, see Definition \ref{def:annihilators}. Here, $\mathcal{W} = \fG^*$ and $\fG$ is reflexive, see \cite[Rmk. 6.5]{kriegl1997convenient}.}  of $\fG$:
\[
\fN \doteq \Ann(\Im(\iota^*_\C\bd\h),\fG) \quad\text{and}\quad \fNoff \doteq \Ann(\Im(\bd\h),\fG).
\qedhere
\]
\end{definition}

\begin{remark}\label{rmk:Nhco}
With reference to Remarks \ref{rmk:co-} and \ref{rmk:coflux},
denote $\bd\h^{\text{co}} :\fG \to \iloc^1(\X)$ the differential of the off-shell co-flux map (Remark \ref{rmk:coflux}) and by $\iota_\C^*\bd\h^\text{co} \equiv(\iota_\C^*\bd\h)^\text{co}:\fG \to \iloc^1(\C)$ its on-shell analogue (Definition \ref{def:flux}).
Hence, note that one can alternatively write $\fN = \ker(\iota_\C^*\bd\h^\text{co})$ and $\fNoff = \ker(\bd\h^\text{co})$, and thus conclude $\fNoff\subset \fN$. 
\end{remark}

\begin{theorem}[Flux annihilator ideals]
\label{thm:frakN}
Let $\fN$ and $\fNoff$ be the flux annihilators as per Definition \ref{def:frakN}.
Recall also the definitions of the corner group and algebra cocycles $c_\smbullet$ and $k_\smbullet$ from Proposition \ref{prop:adjequi}.
Then, 
\begin{enumerate}[label=(\roman*)]
    \item $\fNoff = \Ann(\Foff,\fG)$ 
    and 
    $\fN = \Ann(\F,\fG)$;\label{prop:frakNi}
    \item 
    $\fNoff\subset\fN \subset \bigcap_{g\in\G}\ker(c_\smbullet(g))\subset\bigcap_{\xi\in\fG}\ker(\k_\smbullet(\xi))$. \label{prop:frakNii}
    \item $\fN$ and $\fNoff$ are stable\footnote{A \emph{stable} Lie ideal is a Lie ideal which is invariant under the adjoint action of all elements of $\exp(\fG)$. Note that a Lie ideal is guaranteed to be stable only in finite dimensions \cite{Neeb-locallyconvexgroups}.} ideals of $\fG$;\label{prop:frakNiii}
    \item if $\partial\Sigma=\emptyset$, then $\fNoff=\fN = \fG$.\label{prop:frakNiv}
\end{enumerate}
\end{theorem}
\begin{proof}
Recall from Definition \ref{def:fluxspaces} that $\F = \Im(\iota_\C^*\h_\smbullet)$.
Therefore to prove the second statement of \ref{prop:frakNi} we need to show that $\Ann(\Im(\iota_\C^*\h_\smbullet^\text{co},\fG)\subset\Ann(\Im(\iota_\C^*\bd \h),\fG)$ and viceversa---or, equivalently, that $\ker(\iota_\C^* h_\smbullet)\subset \ker(\iota_\C^*\bd h^\text{co})$ (cf.\ Remark \ref{rmk:Nhco}). We start from this inclusion: if $\xi\in\ker(\iota_\C^* h_\smbullet^{\text{co}})$ then $\langle\iota_\C^*\h_\smbullet,\xi\rangle=0$ and therefore, by differentiation and using that $\h(\phi_\smbullet)$ is a constant over $\X$, a fortiori $\langle\iota_\C^*\bd \h, \xi\rangle = 0$.
To show the opposite inclusion, consider $\xi\in\ker(\iota_\C^*\bd\h^{\text{co}})$; then, by definition, one has 
\[
0 = \langle \iota_\C^*\bd \h^{\text{co}},\xi\rangle = \langle \bd\iota_\C^* h, \xi\rangle = \bd\langle \iota_\C^*h, \xi\rangle.
\]
Since $\C$ is path-connected (Assumption \ref{assA:constr}), we conclude that $\langle \iota_\C^*\h, \xi\rangle$ is constant over $\C$ and thus equal to $\langle \h(\phi_\smbullet), \xi\rangle$, for some $\phi_\bullet\in\C$. Since $h_\smbullet(\phi_\smbullet)=0$, it follows that $\langle \iota_\C^* \h_\smbullet,\xi\rangle =0$ for all $\xi\in\ker(\iota_\C^*\bd\h^\text{co})$, and hence $\ker(\iota_\C^*\bd h)\subset\ker(\iota^*_\C\h_\smbullet^{\text{co}})$. A totally analogous argument applies to $\fNoff$, where we use the path-connectedness of $\X$ instead (Assumption \ref{assA:item1}). 

Turning to \ref{prop:frakNii}, the only non-obvious part is showing that $\fN\subset\ker( c_\smbullet(g))$ for all $g\in \G$. The argument relies on Proposition \ref{prop:adjequi}. First,  observe that $c_\smbullet(g=e)=0$. From this it follows that to prove that $\fN\subset\ker( c_\smbullet(g))$ it suffices to show that, for all $\eta\in\fN$,  $\langle c_\smbullet(g),\eta\rangle$ is constant in $g$. Since $\G$ is simply connected, it is enough to show that $\langle c_\smbullet(g),\eta\rangle$ is constant upon all infinitesimal variations of the form $g \mapsto e^{-\epsilon\xi} g$ for $\epsilon\to0$.
Indeed, since
\[
\langle c_\smbullet(g),\eta\rangle = \langle h(\phi_\smbullet\triangleleft g), \eta\rangle =  \langle h(\phi),\eta\rangle\vert_{\phi=\phi_\smbullet\triangleleft g},
\]
the infinitesimal variation $g\mapsto e^{-\epsilon\xi}g$ of the quantity $\langle c_\smbullet(g),\eta\rangle$ equals the evaluation at $\phi=\phi_\smbullet\triangleleft g$ of the infinitesimal variation of $\langle h(\phi),\eta\rangle$ along the gauge transformation $\phi \mapsto \phi\triangleleft e^{-\epsilon \xi}$. 

Therefore, to prove the independence of $\langle c_\smbullet(g),\eta\rangle$ from $g$ when $\eta\in\fN$, it is enough to show that when $\eta\in\fN$ the quantity $\L_{\rho(\xi)}\langle h,\eta\rangle\vert_{\phi_\smbullet\triangleleft g}$ vanishes for all $\xi\in\fG$.
To prove this statement, we observe that since the action of $\G$ preserves $\C$ (Corollary \ref{cor:tangent}) and $\phi_\smbullet\in\C$ by construction, we also have $\phi_\smbullet\triangleleft g\in\C$ as well as 
\[
\L_{\rho(\xi)}\langle h,\eta\rangle\vert_{\phi_\smbullet\triangleleft g} 
= \iota_\C^*\L_{\rho(\xi)}\langle h,\eta\rangle\vert_{\phi_\smbullet\triangleleft g}
= \bi_{\rho(\xi)} \langle \iota_\C^*\bd h,\eta\rangle\vert_{\phi_\smbullet\triangleleft g}  
=0,
\]
where in the last step we used the result of statement \ref{prop:frakNi} that $\eta\in\fN$ iff $\eta\in \ker(\iota_\C^*\bd h^\text{co})$.

To prove \ref{prop:frakNiii} for $\fN$ (and, \emph{mutatis mutandis}, for $\fNoff$), note that it is enough to prove the stability of $\fN\subset\fG$,\footnote{Since $\Ad(g):\fG\to\fG$ is invertible for all $g\in\G$, $\Ad(g)\cdot \mathfrak{N} \subset\mathfrak{N}$ iff $\Ad(g)\cdot \mathfrak{N} = \mathfrak{N}$.} i.e.\ that $\Ad(g)\cdot \mathfrak{N}\subset\mathfrak{N}$ for all $g\in\G$. Indeed, this automatically implies that $\mathfrak{N}$ is a Lie ideal, $[\fG,\mathfrak{N}]\subset\mathfrak{N}$, and in particular a Lie subalgebra, $[\mathfrak{N},\mathfrak{N}]\subset\mathfrak{N}$.
We are thus set to prove that for all $g\in\G$, $\Ad(g)\cdot \mathfrak{N}\subset\mathfrak{N}$.
That is, using \ref{prop:frakNi}, we want to show that for all $g\in\G$ and $\eta\in\mathfrak{N}$, 
$\langle \iota_\C^*\h_\smbullet, \Ad(g)\cdot\eta \rangle = 0$.
But this is indeed a simple consequence of statements \ref{prop:frakNi} and \ref{prop:frakNii} of this proposition, together with Proposition \ref{prop:adjequi}:
\[
\langle
\iota_\C^*\h_\smbullet, \Ad(g)\cdot\eta\rangle 
\equiv\langle \Ad^*(g)\cdot(\iota_\C^*\h_\smbullet), \eta\rangle
=\langle  g^*(\iota_\C^*\h_\smbullet)-c_\smbullet(g),\eta\rangle
=g^*\langle  \iota_\C^*\h_\smbullet, \eta\rangle 
=0.
\]
Finally, if $\partial\Sigma=\emptyset$, the flux vanishes ($h_\smbullet\equiv0$) and $\fNoff=\fN=\fG$, proving \ref{prop:frakNiv}.
\end{proof}

We conclude this section with a generalization of the {{Gauss}} law of Maxwell theory. Recall that Gauss law states that, on-shell of the Gauss constraint, the integral of the electric flux through $\pp\Sigma$ equals the total electric charge in $\Sigma$, and therefore vanishes in the absence of charges (for an explanation of how this proposition recovers this formulation of the Gauss law, see Section \ref{sec:runex-fluxannihilators}; cf.\ also \cite{AbbottDeser82} and \cite[Section 4]{RielloGomes}):

\begin{proposition}[{{Gauss}} law]\label{prop:Gausslaw}
Recall the notation $\rho_\C\colon \Gamma(\iota^*_\C\A) \to \mathfrak{X}(\C)$ (Corollary \ref{cor:tangent}). Whenever the on-shell isotropy locus has the structure of a vector bundle, its space of sections $\Gamma(\C,\mathsf{I}_\rho)= \ker(\rho_\C)$ is an ideal of the Lie algebra of sections of the on-shell algebroid $\iota^*_\C\A$. If, additionally, the bundle is trivial, i.e. $\mathsf{I}_\rho = \C \times \mathfrak{I}$ as per Assumption \ref{assA:isotropy} (cf.\ Remark \ref{rmk:isotropybundle2}), then
\[
\mathfrak{I} \subset \fN,
\] 
that is $\langle \iota_\C^*h_\smbullet , \chi \rangle =0$ for all $\chi \in \mathfrak{I}$.
\end{proposition}
\begin{proof}
The first statement follows from the fact that $\rho(\llbracket \xi_1,\xi_2\rrbracket) = [\rho(\xi_1),\rho(\xi_2)]$ for every $\xi_1,\xi_2\in\Gamma(\iota^*_\C\A)$ since $\iota^*_\C\A$ is a Lie algebroid, and we conclude by specialising to $\xi_1\in \ker(\rho_\C)$.

If the on-shell isotropy locus is a trivial bundle, $\mathsf{I}_\rho = \C \times \mathfrak{I}$, we can view the on-shell anchor as a map $\rho_\C:\fG  \to \mathfrak{X}(\C)$, so that the on-shell isotropy algebra is $\mathfrak{I} = \ker(\rho_\C)$ (Assumption \ref{assA:isotropy}, Remark \ref{rmk:isotropybundle2}).
Assume that $\chi\in \mathfrak{I} =\ker(\rho_\C)$. Then, using Corollary \ref{cor:tangent}:
\[
   \langle \iota^*_\C\bd \h, \chi\rangle =\iota^*_\C(\iota_{\rho(\chi)}\omega - \langle\bd\Ho,\chi\rangle) 
   \equiv 0
\]
and $\chi \in \fN$ (cf.\ Theorem \ref{thm:frakN}, item \ref{prop:frakNi}).
\end{proof}

\begin{remark}[Isotropy bundle: sections and {{Gauss}} law] \label{rmk:isotropybundle3}
In general, the nontriviality of $\mathsf{I}_\rho$ obstructs us from casting $\rho_\C$ as a map from $\fG$ to $\mathfrak{X}(\C)$, and thus viewing $\ker(\rho_\C)$ within $\fG$. This is because, lacking a (flat) connection, sections of $\mathsf{I}_\rho$ cannot be identified with elements of the generic fibre $\mathfrak{I}$. 
Nonetheless, many of our statements admit obvious, and natural, generalizations: e.g.\  one can introduce the \emph{on-shell flux annihilating locus} $\mathsf{N}$,
\[
\mathsf{N} \doteq \{ (\phi,\xi)\in \iota_\C^*\A :  \langle\iota_\C^* \bd h(\phi),\xi\rangle = 0 \},
\]
generalising the concept of on-shell flux annihilator $\mathfrak{N}$ and thus show in analogy to Proposition \ref{prop:Gausslaw} that 
\[
\mathsf{I}_\rho \subset \mathsf{N}.
\]
Then, when $\mathsf{I}_\rho\subset \mathsf{N} \subset \iota_\C^*\A$ are vector subbundles, one is naturally led to consider the theory of Hamiltonian Lie algebroids developed in \cite{blohmann2018hamiltonian} as a way to extend our work. Indeed, non-action Hamiltonian Lie algebroids make a natural appearance in the context of symplectic reduction even when they are not part of the framework from the start (cf.\ Remark \ref{rmk:isotropybundle1}). Further investigations along these lines are deferred to future work.
\end{remark}

\subsection{Yang--Mills theory: flux annihilators}\label{sec:runex-fluxannihilators}

Recall from Section \ref{sec:runex-decomposition} that the flux form of YM theory is $\langle d\bh,\xi\rangle = -d \tr(E\xi)$, thus
\[
\langle h, \xi\rangle = - \int_{\pp\Sigma } \iota_{\pp\Sigma}^* \tr( E \xi).
\]

\begin{remark}[Electric flux]\label{rmk:electricflux}
Denoting $\gamma_{\pp\Sigma}$ the metric induced by $\gamma$ on $\pp\Sigma$ and $n_i$ the conormal of $\pp\Sigma$ in $\Sigma$, one readily sees that 
\[
\langle h ,\xi\rangle = - \int_{\pp\Sigma} \sqrt{\gamma_{\pp\Sigma}} \, n_i \tr( \hat E^i \xi)\vert_{\pp\Sigma}
\]
is given by the smeared integral of (minus) the \emph{flux of the electric vector}  $\hat E^i = \sharp_\gamma \star E$ through the corner $\pp\Sigma$, cf.\ \eqref{eq:mathcalE}. This is the reason why we called $h$ the ``flux'' map.
\end{remark}

Since $E\in \mathcal{E}\sim\Omega^{\text{top}-1}(\Sigma,\fg^*)$, the off-shell flux annihilator is given by those gauge parameters with vanishing restriction on the corner $\pp\Sigma$---this is indeed a consequence of the linearity of $\iota^*_{\partial\Sigma}\tr(E\xi)$ with respect to $E$ and of the fact that, since $\tr(E\xi)(x)\equiv\tr(E(x)\xi(x))$ by definition, the flux map $h$ only involves the restriction of the field at the corner, and none of its higher or normal jets:\footnote{This statement holds also for nontrivial bundles $P\to\Sigma$.}
\[
\fNoff = \ker(\bd h^\text{co}) = \{ \xi \in\fG\ |\ \xi\vert_{\pp\Sigma}=0\}.
\]
This is clearly an ideal of $\fG$, since the ultralocality of $\fG$, 
\[
[\xi,\eta]_\fG(x) = [\xi(x),\eta(x)]_\fg,
\]
implies that $[\xi,\eta](x)=0$ if $\eta(x)=0$ and thus, for $x\in\pp\Sigma$, that $[\xi,\eta]\in\fN$ if $\eta\in\fN$.

We shall now show that in pure YM theory the on- and off-shell flux annihilators coincide:
\[
\fN=\fNoff
\]
Note that in Section \ref{sec:runex-setup} we prescribed that, whenever $G$ is semisimple, we restrict to irreducible connections. However, the above result holds unaltered also in the general case.

To prove the equality between the flux annihilators, a more explicit description of $\C$ is needed. This can in turn be given in terms of the equivariant Hodge--de Rham decomposition detailed in Appendix \ref{app:Hodge}.
According to those results, the cotangent fibre $\mathcal{E}_A\simeq \mathcal{E}$---which is defined as the vector space of $\Ad^*$-equivariant $(\text{top}-1)$-forms over $\Sigma$ as per Equation \eqref{eq:mathcalE}---admits the orthogonal decomposition\footnote{See also \cite{DiezPhD} for an analogous discussion in the corner-less case, as well as \cite{RielloGomesHopfmuller,RielloGomes} for a more physical take on these decomposition in the context with corners.}
\[
\mathcal{E}_A = \mathcal{V}_A \oplus \mathcal{H}_A
\]
where
\begin{subequations}\label{eq:YMcotangentsplit}\begin{align}
\mathcal{V}_A & \doteq \{ E \in \mathcal{E}_A\ |\ \exists \varphi : E = \star d_A\varphi \}, \label{eq:verticalYMcotangentsplit}\\
\mathcal{H}_A & \doteq \{ E \in \mathcal{E}_A\ |\ d_A E =0 \text{ and } \mathsf{t} E =0\}.\label{eq:horizontalYMcotangentsplit}
\end{align}\end{subequations}
Here, we denoted $\varphi$ elements of $\Gamma(\Sigma,\Ad^*P)\sim \Omega^0(\Sigma,\dual{\fg})$, and $\mathsf{t} E = \iota_{\pp\Sigma}^* E$ the tangential part of the electric form $E$ at $\pp\Sigma$---which corresponds to the electric flux $n_i\hat E^i\vert_{\pp\Sigma}$ of the electric vector, see Remark \ref{rmk:electricflux} here above.

In physics parlance, $\mathcal{V}_A$ is the space of ``Coulombic'' electric fields, whereas one could call $\mathcal{H}_A$ the space of ``radiative'' electric fields. We call $\varphi$ the \emph{Coulombic potential}.

Observe that radiative electric fields both satisfy the Gauss constraint and are fluxless (i.e.\ $h(E)=0$ if $E\in \mathcal{H}_A$).
Therefore, since both 
\[
\bHo = \tr(d_A E \cdot)
\quad\text{and}\quad
d\bh = - d\tr(E\cdot)
\]
are linear in $E$, it follows that both $\bHo$ and $d\bh$ effectively depend only on the Coulombic component of $E$ and not on its radiative one.

This is most directly seen by inspecting the equivariant-Neumann elliptic boundary value problem that uniquely determines the Coulombic projection of $E$, i.e.
\begin{equation}
\label{eq:LaplaceEq}
\begin{cases}
\Delta_A \varphi = - \star d_A E &\text{in }\Sigma,\\
 \mathsf{n}  d_A\varphi = (-1)^{\Sigma-1}\mathsf{t}E& \text{at }\partial \Sigma.
\end{cases}
\end{equation}
Since $A$ is by hypothesis irreducible, these equations admit a unique (smooth) solution $\varphi$ for \emph{any} value of $E$. Cf. Appendix \ref{app:Hodge}, especially Remark \ref{rmk:EllipticReducible}.

Now, if $(A,E)\in\C$, the constraint form $\bHo = \tr(d_AE\cdot)$ vanishes and therefore, \emph{leaving the dependence on $A$ implicit},
\begin{equation}
    \label{eq:LaplaceEqOnShell}
\text{on $\C$,} \quad \varphi = \varphi(\mathsf{t}E).
\end{equation}
Let us introduce the notation $\mathcal{E}_\pp$ for the space of tangential electric forms---also called electric fluxes (Remark \ref{rmk:electricflux}),\footnote{Recall: the notation $\mathcal{X}\sim \mathcal{Y}$ is used to signify that $\mathcal{X}$ is \emph{locally modeled on} $\mathcal{Y}$.}
\begin{equation}\label{e:cornerconfYM}
\mathcal{E}_\pp \doteq \Gamma( \mathrm{Dens}(\pp\Sigma) \times_{\pp\Sigma} \iota_{\pp\Sigma}^*\Ad^*P) \sim
\Omega^\text{top}(\pp\Sigma,\fg^*).
\end{equation}
We observe that---owing to the irreducibility of $A\in\Acal$---\emph{any} value $E_\pp \in \mathcal{E}_\pp$ of $\mathsf{t}E$ is compatible with the Gauss constraint, meaning that $\C$ maps \emph{surjectively} onto $\mathcal{E}_\pp$, and moreover---due to the on-shell formula \eqref{eq:LaplaceEqOnShell}---$\C$ can be identified with the product space
\begin{subequations}
    \label{eq:CtoP-YM}
\begin{equation}
\C \simeq \X_\rad \times \mathcal{E}_\pp,  
\end{equation}
where $\X_\rad$ is a fibration over $\Acal$ with fibre $\mathcal{H}_A \subset \mathcal{E}$:
\begin{equation}
\pi_\rad:\X_\rad\to \Acal, \qquad \pi_\rad^{-1}(A) = \mathcal{H}_A.
\end{equation}
The embedding $\C \hookrightarrow \X = \Acal \times \mathcal{E}$ is given by the map:
\begin{equation}
(A, E_{\text{rad},A}, E_\pp)  \ \mapsto \ (A,E) = \big(A, E_{\text{rad},A} + \star d_A \varphi(E_\pp)  \big)
\end{equation}
\end{subequations}
where $E_{\text{rad},A}\in\mathcal{H}_A$ and $E_\pp \in \mathcal{E}_\pp$.

To construct the adjusted flux map $\h_\smbullet$, we observe that $(A_\smbullet,E=0)$, for $A_\smbullet\in\Acal$ a reference connection chosen once and for all  (cf.\ Remark \ref{rmk:HamShift1})\footnote{If $P\to\Sigma$ is trivial, $A_\smbullet$ can be conveniently chosen to vanish in the chosen atlas with trivial transition functions.} is an element of $\C$ and can therefore be used as the reference configuration:
\[
\phi_\smbullet = (A_\smbullet,E=0).
\]
Then, since $h(\phi_\smbullet)=0$, we have that $h_\smbullet = h - h(\phi_\bullet) \equiv h$ is the (smeared) integral of $\mathsf{t}E$ and---using the surjectivity of the mapping of $\C$ onto $\mathcal{E}_\pp$---that $\iota_\C^*h_\smbullet$ is the smeared integral of $E_\pp \in \mathcal{E}_\pp$. In other words, each on-shell flux $f\in\F$ is of the form $f = \int_{\pp\Sigma} \tr( E_\pp \cdot)$ for some $E_\pp\in\mathcal{E}_\pp$---and indeed one has 
\begin{equation}\label{eq:YM-fluxdensity}
\mathcal{E}_\pp \simeq \F=\Foff = \mathrm{Im}(h_\smbullet) \subset \dual{\fG}_\mathrm{loc}.
\end{equation}
Hence, the announced equality $\fN=\fNoff$  \cite[Section 4]{RielloGomes} can be read directly off of Definition \ref{def:frakN}, the fact that $\F=\Foff$, and Theorem \ref{thm:frakN}.

Using the isomorphism $\mathcal{E}_\pp \simeq \F$, the KKS Poisson structure on $\F$ can be written as
\[
\{ f_1, f_2 \} = f_3,\qquad f_3= \int_{\pp\Sigma} \tr( [E_{\pp1}, E_{\pp2}] \,\cdot\, ). 
\]

Finally, the KKS symplectic structure on the symplectic leaf $\mathcal{O}_{f}$  can be explicitly written down by introducing a reference flux $\bar f\in\mathcal{O}_f$ so that $\mathcal{O}_f\simeq \Ad^*(\G)\cdot \bar f$. Thus, using $u\in\G$ as an overcomplete set of coordinates over $\mathcal{O}_f$ according to $f(u)=\Ad^*(u)\cdot \bar f$, the KKS symplectic structure $\Omega_{[f]}$ reads
\begin{equation}
    \label{eq:KKSYM}
\Omega_{[f]} = \tfrac12 \langle \bar f , \left[ \bd uu^{-1} , \bd uu^{-1} \right] \rangle.
\end{equation}
The overcompleteness of the $u$-coordinate follows from the fact that $\bar f$ has a large stabiliser $\mathrm{Stab}(\bar f)\subset\G$ which includes, but is not restricted to, $\Go$. It is a standard exercise to check that $\Omega_{[f]}$ is indeed well-defined over $\mathcal{O}_{f}$ despite this large redundancy.\footnote{This computation is done, among many other places, in \cite{RielloSciPost}. At heart it is nothing else than a symplectic reduction of $\G \circlearrowright T^*\G$ at $f\in\fG^*$ \cite{MarsdenWeinstein}.}
Notice that, since $\Go \subset \mathrm{Stab}(\bar f)$, $u$ can as well be thought of as element of $\Gred$.

\medskip

We now turn to the Abelian case. 
First, note that in this case $\Ad P \to \Sigma$ is trivial, so that $\fG = \Gamma(\Ad P)\simeq C^\infty(\Sigma,\fg)$.
Now, the main difference with the non-Abelian case is that here $E$ is gauge invariant and moreover all gauge connections $A\in\Acal$ are reducible.
For this reason, it is meaningless to restrict to irreducible gauge connections as we did in the non-Abelian case. 
However, to our advantage, in the Abelian case all gauge connections have the same reducibility group $\G_A$, given by the group of (locally\footnote{Generalizations of the following statements to a non-connected $\Sigma$ are straightforward.}) constant group-valued functions, $dg=0$, so that 
\[
\G_A \simeq G,
\]
and $\fG_A\simeq\fg$ is its Lie algebra. This is encoded in the statement that the isotropy locus $\mathsf{I}_\rho$ is a trivial bundle. (See Remark \ref{rmk:isotropybundle2}.)

This fact has a bearing on the structure of the Gauss constraint set $\C$ and introduces a distinction between the on- and off-shall flux annihillators and flux-spaces. 
This is because only electric fluxes $E_\pp$ satisfying the ``integrability'' condition
\begin{equation}\label{eq:abeliangausslaw}
0=\int_{\pp\Sigma} \tr( E_\pp\ \chi\vert_{\pp\Sigma}) \qquad\forall \chi\in\fG_A
\end{equation}
are compatible with the Gauss constraint. Inspired by Maxwell theory, we call this condition \emph{{{Gauss}} law}, and observe that it can be expressed as
\[
0=\langle \iota_\C^* h_\smbullet, \chi \rangle \qquad\forall \chi\in\fG_A
\]
(cf.\ Theorem \ref{thm:frakN}\ref{prop:frakNi} and especially Proposition \ref{prop:Gausslaw}.)

This leads us to introducing $\mathcal{E}_\pp^\text{ab}$, the subspace of $\mathcal{E}_\pp$ which satisfy the {{Gauss}} law \eqref{eq:abeliangausslaw}. If $\Sigma = \bigsqcup_\alpha \Sigma_\alpha$ is the decomposition of $\Sigma$ into its connected components, one has that
\begin{equation}
\label{eq:EppAb}
\mathcal{E}_\pp^\text{ab} = \left\{ E_\pp \in \mathcal{E}_\pp \ \left| \ \textstyle{\int}_{\pp(\Sigma_\alpha)} E_\pp = 0 \;\; \forall \Sigma_\alpha\right.\right\}.
\end{equation}
Hence, observing that in the Abelian case $\mathcal{H}_A = \mathcal{H} \subset \mathcal{E}$ does not depend on $A$, we find that $\C$ factorises as:
\begin{equation}\label{eq:CfibrationAb}
\C = \Acal \times \mathcal{H} \times \mathcal{E}_\pp^\text{ab},
\end{equation}
with embedding $\C \hookrightarrow \X = \Acal\times \mathcal{E}$ completely analogous to the one of Equation \eqref{eq:CtoP-YM}.

From {{Gauss}} law \eqref{eq:abeliangausslaw}, we deduce that the on-shell flux annihilator differs from the off-shell one, and is given by
\[
\fN^{\text{ab}} = \{ \xi \in \fG \ |\ \exists \chi\in\fG_A : \xi\vert_{\pp \Sigma} = \chi\vert_{\pp\Sigma} \} \supset \fNoff.
\]
In the special case where $\Sigma$ and $\pp\Sigma$ have only one connected component, $\fN^\text{ab}$ is given by the $\xi\in\fG$ which take a constant---but not-necessarily zero---value at the corner. Similarly, the on-shell flux space differs from the off-shell one: 
\[
\F^\text{ab} \simeq \mathcal{E}_\pp^\text{ab} \subset \mathcal{E}_\pp \simeq \Foff.
\]
Due to the Abelian nature of $G$, the Poisson structures on $\F$ and $\Foff$ are trivial, and the orbits $\mathcal{O}_f$ are points.

\begin{remark}[Semisimple YM: reducible configurations]
In semisimple YM theory,  integrability conditions analogous to {{Gauss}} law \eqref{eq:abeliangausslaw} exist at \emph{reducible} configurations only. This can be seen as a consequence of the presence of a nontrivial kernel in Equation \eqref{eq:LaplaceEq} when $A$ is reducible. 
Since not all electric fluxes are compatible with the Gauss constraint, and since this compatibility depends on the configuration variable $A\in\Acal$, the structure of $\C$ as a fibration over $\Acal$ takes the more general form:
\[
\pi_\C:\C \to \Acal, \qquad \pi_\C^{-1}(A) \simeq \mathcal{H}_A \times \mathcal{E}_{\pp, A},
\]
where $\mathcal{E}_{\pp,A}$ denotes the space of electric fluxes which are compatible with the Gauss constraint at $A$.

In addition to the {{Gauss}} law generalization of Proposition \ref{prop:Gausslaw} and Remark \ref{rmk:isotropybundle3} (see also the comment below Equation \eqref{eq:LaplaceEq}), which is related to the reducibility of $(E,A)$ jointly, in YM theory additional compatibility conditions emerge from reducibilities of $A$ alone.

As emphasized in Remark \ref{rmk:isotropybundle1}, Assumption \ref{assA:isotropy} forces us to focus on subspaces of the full YM configuration space associated to configurations with fixed orbit types (i.e.\ of fixed isotropy group up to conjugation). 
In Remark \ref{rmk:isotropybundle2} we have also emphasized that irreducible configurations (those of trivial orbit type) are dense in semisimple YM theory, and in that case our previous discussion shows that no analogue of {{Gauss}} law exists, and indeed $\fN = \fNoff$. 

Conversely, in subspaces of a given nontrivial orbit type, the situation is more subtle. This is because in these cases the isotropy bundle $\mathsf{I}_\rho$ (Assumption \ref{assA:isotropy}) is nontrivial and in particular cannot be described in terms of an isotropy subalgebra common to \emph{all} configurations of the fixed orbit type.
This means that although an analogue of Gauss law holds at these reducible configurations, this is not registered in $\fN$, which is still equal to $\fNoff$. As discussed in Remark \ref{rmk:isotropybundle3}, to fully capture this situation from a geometrical standpoint, one would have to generalise the strictly Hamiltonian framework discussed here to include non-action Lie algebroids \cite{blohmann2018hamiltonian}. We leave this task for future work.
\end{remark}

\subsection{The constraint and flux gauge groups}\label{sec:BulkCornerGrps}

The goal of this section is to explicitly characterise the normal subgroup $\Go\subset\G$, which we call the ``constraint gauge group'', defined as the maximal normal subgroup that admits an equivariant momentum map $\Jo$ whose zero-level set coincides with the constraint set $\C=\bHo^{-1}(0)$. The identification of $\Go$ is the first main result of this work, and---as we will now explain---it constitutes the foundations of what will follow. 

The constraint group $\Go$, being by construction normal in $\G$, allows us to identify the residual ``flux gauge group'' defined as the quotient $\Gred\doteq \G/\Go$. The decomposition of $\G$ into the constraint  and flux gauge groups, $\Go$ and $\Gred$, allows us to adopt a reduction-by-stages approach to understand the structure of the space of constrained fields modulo \emph{all} gauge transformations, $\C/\G$.

\medskip

To introduce the normal constraint gauge subgroup $\Go\subset\G$ we start in the following section from an intrinsic definition of its Lie algebra $\fGo$ as the maximal Lie ideal of $\fG$ whose momentum map vanishes at $\C$ (Definition \ref{def:fGo}). We will then show (Theorem \ref{thm:fGo}) that $\fGo$ coincides with  the flux annihilator $\fN$ introduced in the previous section. This result is the first evidence of the centrality of the flux in the analysis of gauge symmetry in the presence of corners. Finally, we conclude the section with the definition of the flux gauge algebra $\fGred$ and some observations about the existence of the corresponding constraint and flux gauge groups.

\subsubsection{The constraint gauge ideal}
We start by stating the main definitions and theorems relevant to the identification of $\fGo$.

\begin{remark}\label{rmk:adjoint-inclusion}
In the following we will define the local momentum forms $\bJ_\mathfrak{H}$ for a subalgebra $\mathfrak{H}$ by means of the following identity, valid for any $\eta \in \mathfrak{H}\subset \fG$:
\[
\langle \iota_\mathfrak{H}^*\bH, \eta \rangle \doteq \langle \bH, \iota_\mathfrak{H} \eta \rangle = \langle \bH,  \eta \rangle.
\]
The symbol $\iota_\mathfrak{H}^*$ on the left-most term is best understood as a ``pullback" in the co-momentum map picture $\bH^\text{co} : \fG \to \oloc^{\text{top},0}(\Sigma\times \X)$ along the embedding $\iota_\mathfrak{H} : \mathfrak{H}\hookrightarrow \fG$ (as above, we will most often leave the $\bullet^\text{co}$ notation implicit). It is important here to recall that the equivalence between the momentum and co-momentum map pictures in our infinite dimensional setting---which underpins the previous formula---is warranted because we are working with \emph{dual-valued local forms} (Definition \ref{def:dualvaluednew} and Remark \ref{rmk:co-}). In other words, the \emph{local} setting provides us with a solution to the difficulty of having to define an adjoint operation for the inclusion $\iota_\mathfrak{H} : \mathfrak{H}\hookrightarrow \fG$, whose existence is indeed subordinate to the type of the dual one works with and is generally not warranted in infinite dimensions. 
\end{remark}

\begin{definition}[Constraining algebras]\label{def:constraintideal}
Let $(\X,\bom,\C)$ be as in Assumption \ref{ass:setup}, and $\mathfrak{H}\subset\fG$ a Lie subalgebra with local momentum form $\bJ_{\mathfrak{H}}\in \oloc^{\text{top},0}(\Sigma\times\X,\mathfrak{H}^*)$, and denote $J_{\mathfrak{H}}\doteq\int_\Sigma \bJ_{\mathfrak{H}}$ the corresponding local momentum map. We say that
\begin{enumerate}[label=(\roman*)]
    \item $J_{\mathfrak{H}}$ (resp. $\bJ_{\mathfrak{H}}$) is a \emph{constraining local momentum map} (resp.\ \emph{form}) iff $\C\subset J_{\mathfrak{H}}^{-1}(0)$.
    
    \item A constraining momentum map (resp. form) such that $\C=J_{\mathfrak{H}}^{-1}(0)$ is said to be \emph{just}.
    
    \item $\mathfrak{H}$ is a (\emph{just}) \emph{constraining algebra} iff it admits a (just) constraining local momentum map.
    \qedhere
\end{enumerate}
\end{definition}

We called $\bHo$, and $\Ho$ the \emph{constraint form and map}, respectively, and now we are introducing the concept of \emph{constraining momentum form and map}. The clash in nomenclature is justified by the fact that $\Ho$ is \emph{almost} the correct object one needs to look at to have a well-defined momentum map, as we will clarify shortly.

\begin{lemma}\label{lem:uniqueJ}
Fix a Lie subalgebra $\mathfrak{H}\subset\fG$. Denote $\bJ_{\mathfrak{H}}$ a constraining local momentum form for the action of $\mathfrak{H}$. Then $\bJ_{\mathfrak{H}}$ is unique up to an additive constant $\balpha\in\oloc^\text{top}(\Sigma,\mathfrak{H}^*)$ such that $\int_\Sigma\balpha=0$, and $J_{\mathfrak{H}}$ is unique.
\end{lemma}
\begin{proof}
Local momentum forms are defined up to additive constants $\balpha\in\oloc^\text{top}(\Sigma,\mathfrak{H}^*)$ (Remark \ref{rmk:HamShift2}). The claim follows from the fact that, by the definition of a local momentum form, $\bJ_{\mathfrak{H}}$ is constraining iff $\iota_\C^*\int_\Sigma \bJ_{\mathfrak{H}} = 0$, and
\[
  \iota_\C^*\int_\Sigma (\bJ_{\mathfrak{H}} + \balpha) = 0 \iff \int_\Sigma\balpha=0. 
    \qedhere
\]
\end{proof}

\begin{definition}[Constraint gauge ideal, $\fGo$]\label{def:fGo}
The \emph{constraint gauge ideal} $\fGo$ is defined as the \emph{maximal} constraining ideal in $\fG$ (when it exists).
Denote by $\iota_{\fGo}:\fGo\hookrightarrow\fG$ the inclusion. The fundamental vector fields associated to its action on $\X$ are called \emph{constraint gauge transformations}.
\end{definition}

\begin{remark}[Existence of maximal ideal]
Existence of the maximal ideal $\fGo$ is not obvious \emph{a priori}. In Proposition \ref{prop:Nsubalgebraprop} we show that all constraining subalgebras (and \emph{a fortiori} ideals) must belong to $\fN$, which is itself a constraining ideal and hence maximal. Nontrivial, non-maximal, examples of constraining ideals will appear in Lemma \ref{lemma:fGc} and Proposition \ref{prop:Noffideal}.
\end{remark}

We can now state the main theorem of this section, which constructs the constraint gauge ideal $\fGo\subset\fG$ and its local momentum form. 
We will present the proof of this theorem in the forthcoming Section \ref{sec:proofofThm}.

\begin{theorem}[Characterization of the constraint gauge ideal]\label{thm:fGo}
The constraint gauge ideal coincides with the on-shell flux annihilator ideal
\[
\fGo = \fN.
\]
Moreover, $\fGo$ admits a $\Sigma$-equivariant (Definition \ref{def:equivariance}), constraining, local, momentum form,
\[
    \bJo{}_\smbullet \doteq \iota_{\fGo}^*(\bHo + d\bh_\smbullet).
\]
Denoting by $\Jo$ the unique (integrated) constraining momentum map associated to $\bJo{}_\smbullet$, we have that $\Jo$ is equivariant and it is just, i.e.\ 
\[
\C = \Jo^{-1}(0).
\qedhere
\]
\end{theorem}

\begin{remark}
The above statement proves Theorem \ref{mainthm:redsummary}, item \ref{thmitem:main-Go}.
\end{remark}

\begin{remark}[Coisotropic sets as zero-level sets]
The problem of whether a given coisotropic submanifold can be thought of as the zero-level set of a momentum map is a fairly general problem in symplectic geometry and has generally a negative answer (see for instance \cite[Section 5]{BSW}). In the particular case of constrained systems, it relates to important questions on the algebraic structure defined by the Cauchy data for initial value problems. For the special example of General Relativity, it is a consequential problem that motivated a number of studies, such as the generalization of Hamiltonian reduction to Lie algebroids \cite{blohmann2018hamiltonian}. 
\end{remark}

\begin{remark}[Smoothness of $\C$, revisited]\label{rmk:smooth-C-revisited}
According to Theorem 1 of \cite{ArmsMarsdenMoncrief1981} (which relies on the inverse function theorem and thus holds as-is in the Banach setting only), the set $\C = \Jo^{-1}(0)$ defines a smooth manifold if at all $\phi\in \C$ the momentum map $\Jo$ is regular, which in turn is the case iff $\phi$ has no stabiliser in $\fGo$.
In the context of semisimple Yang--Mills theory, we showed in Section \ref{sec:runex-fluxannihilators} that $\xio\in\fN = \fGo$ iff $\xio|_{\pp\Sigma}=0$; therefore, since $d_A\xi = 0$ and $\xi|_{\pp\Sigma}=0$ together imply $\xi=0$, it readily follows that no configuration in $\X$, and thus in $\C$, is reducible with respect to the action of $\Go$. Following \cite{ArmsMarsdenMoncrief1981} one therefore expects $\C$ to be smooth in Yang--Mills theory even without excluding configurations which are reducible with respect to $\G$---thus suggesting that the presence of boundaries can improve the smoothness of $\C$. See Remark \ref{rmk:smooth-uCo} for comments on the smoothness of the corresponding first- and second-stage reduced spaces.
\end{remark}

\subsubsection{Proof of Theorem \ref{thm:fGo}} \label{sec:proofofThm}
We will break down the proof of this theorem into various lemmas and propositions.

\begin{proposition}\label{prop:Nsubalgebraprop}
Let $\mathfrak{H}\subset\fG$ be a Lie subalgebra, and denote by $\iota_\mathfrak{H} : \mathfrak{H} \hookrightarrow\fG$ the canonical inclusion. Then, $\mathfrak{H}$ is a constraining subalgebra iff $\mathfrak{H}\subset \fN$.
Moreover, its associated constraining, local, momentum form $\bJ_{\mathfrak{H}\smbullet} \doteq \iota_{\mathfrak{H}}^*(\bHo + d\bh_\smbullet)$ is $\Sigma$-equivariant, and its unique constraining momentum map $J_{\mathfrak{H}}$ is equivariant.\footnote{Recall that for integrated momentum maps we have only one notion of equivariance.}
\end{proposition}

\begin{proof}
We set to prove that $\mathfrak{H}$ is a constraining subalgebra iff $\mathfrak{H}\subset\fN$. We start from the ``forward'' implication ($\Rightarrow$).
For every element $\eta\in\mathfrak{H}$ we have
\[
\langle{J}_{\mathfrak{H}},\eta\rangle =  \langle \Ho, \eta\rangle + \langle \h_\smbullet,\eta\rangle,
\]
where we used the fact that $\langle \iota_{\mathfrak{H}}^*\bH,\eta\rangle = \langle \bH, \iota_{\mathfrak{H}}(\eta)\rangle= \langle \bH, \eta\rangle$ for all $\eta\in \mathfrak{H}\subset\fG$ (cf.\ Remark \ref{rmk:adjoint-inclusion}). Restricting to $\C$, and recalling that by hypothesis $\C\subset J^{-1}_{\mathfrak{H}}(0)$, by evaluating the above equation at $\C$ we get that, for all $\eta\in\mathfrak{H}$,
\[
\eta \in \ker (\iota_\C^* \h_\smbullet)=\fN,
\]
where we used Theorem \ref{thm:frakN}. This proves one implication of the first statement. For the opposite implication ($\Leftarrow$), assume that $\mathfrak{H}\subset \fN $. In virtue of Definition \ref{def:frakN} and Theorem \ref{thm:frakN},
\[
\langle \iota_\C^* \h_\smbullet ,\eta\rangle = 0
\]
for all $\eta\in\mathfrak{H}\subset\fN$. Since $\bHo=0$ on $\C$, and in virtue of $\mathfrak{H}\ni\eta$ we have
\[
\iota_\C^*\int_{\Sigma}\langle \iota_{\mathfrak{H}}^*(\bHo + d\bh_\smbullet),\eta\rangle  = 0.
\]
Therefore, we can define the local and integrated forms
\[
\bJ_{\mathfrak{H}\smbullet}\doteq\iota_{\mathfrak{H}}^*(\bHo + d\bh_\smbullet) \in \oloc^{\text{top}}(\X,\mathfrak{H}^*),
\quad
J_{\mathfrak{H}}\doteq \int_{\Sigma}\bJ_{\mathfrak{H}\smbullet}\in \iloc^{0}(\X,\mathfrak{H}^*)
\]
and check by direct inspection that they are, respectively, a constraining momentum form and map for the action of $\mathfrak{H}$, i.e.\ such that $\C\subset J_{\mathfrak{H}}^{-1}(0)$ as per Definition \ref{def:constraintideal}. Uniqueness of $J_{\mathfrak{H}}$ follows from Lemma \ref{lem:uniqueJ}. 

Finally, to prove equivariance, we use Proposition \ref{prop:equi} and Proposition \ref{prop:adjequi} to find, for an arbitrary $g\in\G$, that
\[
g^* J_{\mathfrak{H}} = g^* \int_\Sigma \bJ_{\mathfrak{H}\smbullet}=\int_\Sigma \iota_{\mathfrak{H}}^*\big( \Ad^*(g)\cdot(\bHo + d\bh_\smbullet) + d\bc_\smbullet(g) \big) = \Ad^*(g)\cdot J_{\mathfrak{H}} +  \iota_{\mathfrak{H}}^*c_\smbullet(g).
\]
Thus, since $\mathfrak{H}\subset\fN$, the desired equivariance properties follow from Proposition \ref{prop:adjequi}, which states that for all $g\in\G$, $\fN\subset \ker(c_\smbullet(g))$, and thus $ \iota_{\mathfrak{H}}^*c_\smbullet(g) = 0$.
\end{proof}

\begin{corollary}[of Proposition \ref{prop:Nsubalgebraprop}]\label{cor:inclusion}
Let $\mathfrak{H}_1\subset\mathfrak{H}_2\subset\fG$ be two constraining subalgebras with constraining local momentum maps $J_{\mathfrak{H}_1}$ and $J_{\mathfrak{H}_2}$. Then
\[
\C\subset{J}_{\mathfrak{H}_2}^{-1}(0)\subset{J}_{\mathfrak{H}_1}^{-1}(0).
\]
\end{corollary}
\begin{proof}
Immediately follows from the previous proposition, using
\[
\xymatrix@C=.5cm{
    \mathfrak{H}_1 \;\ar@{^(->}[rr]^-{\iota_{\mathfrak{H}_1}} \ar@{^(->}[dr]_-{\iota_{21}} && \fG \\
    & \mathfrak{H}_2 \ar@{^(->}[ur]_-{\iota_{\mathfrak{H}_2}}
    }
\]
for $\iota_{21} \colon {\mathfrak{H}_1} {\hookrightarrow} {\mathfrak{H}_2}$ such that $\iota_{\mathfrak{H}_1}=\iota_{\mathfrak{H}_2}\circ \iota_{21}$.
\end{proof}

A useful class of constraining ideals are given by the following:
\begin{lemma}\label{lemma:fGc}
Let $K\subset\mathring\Sigma$ be a compact subset of the open interior of $\Sigma$, and let $\fG_K$ be the set of elements of $\fG$ with support in $K$. Then $\fG_K$ is a constraining ideal. In particular, $\fG_c\doteq\bigcup_K\fG_K$ is a just constraining ideal:
\[
\C=J_{c}^{-1}(0),
\]
where we denoted $J_{c}\equiv J_{\fG_c}$. 
\end{lemma}
\begin{proof}
From the compact support property of the $\xi_K\in\fG_K$ it is easily seen that $\fG_K$ is a Lie ideal, i.e.\ that $[\fG,\fG_K]\subset\fG_K$. From that same property and Stokes theorem, it also readily follows that $\langle h_\smbullet, \xi_K\rangle= \int_\Sigma d\langle\bh_\smbullet,\xi_K\rangle =0$. Hence $\fG_K$ is a subideal of $\fN$ and hence a constraining ideal, in virtue of Proposition \ref{prop:Nsubalgebraprop}.

Since all $\fG_K$ are constraining ideals, $\fG_c$ is a constraining ideal as well. We now prove $\C=J_c^{-1}(0)$.  Since $\fG_c$ is a constraining ideal, it is enough to prove that $J_c^{-1}(0)\subset\C$. 
From Proposition \ref{prop:Nsubalgebraprop} we know that
\[
J_{c} = \int_\Sigma \iota_{\fG_c}^*(\bHo + d\bh_\smbullet)
\]
is a momentum map for $\fG_c$. 
Then $\phi\in J_c^{-1}(0)$ iff for all $\xi_c\in\fG_c$
\[
0 = \langle J_c(\phi),\xi_c\rangle = \int_\Sigma \langle\bHo(\phi),\xi_c\rangle,
\]
where the corner term dropped because by definition there exists a compact $K\subset\mathring\Sigma$ such that $\mathrm{supp}(\xi_c)\subset K$. 
Finally, by extremising over all possible $\xi_c\in\fG_c$ one finds that if $\phi\in J_c^{-1}(0)$ then
\[
\bHo(x,\phi) = 0
\]
at all $x\in \mathring\Sigma$. By continuity, we conclude that $\bHo(\phi) = 0$ everywhere on $\Sigma$, and thus $J_c^{-1}(0)\subset\C$.
Equivariance follows from Proposition \ref{prop:Nsubalgebraprop}.
\end{proof}

\begin{corollary}[of Corollary \ref{cor:inclusion} and Lemma \ref{lemma:fGc}]\label{cor:justconstr}
Any subalgebra $\mathfrak{H}$ of $\fG$ such that $\fG_c \subset \mathfrak{H} \subset\mathfrak{N}$ is a just constraining subalgebra.
\end{corollary}

\begin{remark}\label{rmk:uniquejustfd1}
Notice that this scenario is only interesting in infinite dimensions.
In \emph{finite} dimensions, one can prove that the requirement that a momentum map vanish at a given submanifold, when possible, implies that there exists at most one just, constraining, subalgebra. See Corollary \ref{rmk:uniquejustfd2}.
\end{remark}

We are ready now to combine these results into a proof of Theorem \ref{thm:fGo}.

\begin{proof}[Proof of Theorem \ref{thm:fGo}]
In virtue of Proposition \ref{prop:Nsubalgebraprop}, $\mathfrak{H}\subset\fN$ iff it is a constraining subalgebra. All constraining ideals are then (\emph{not} necessarily proper) subsets of $\fN$, but since $\fN$ is a constraining Lie ideal itself (Theorem \ref{thm:frakN}), $\fN$ must be \emph{maximal} as a constraining ideal:
\[
\fGo = \fN.
\]
One concludes that $\fGo$ is also a stable Lie ideal (Theorem \ref{thm:frakN}).

From Proposition \ref{prop:Nsubalgebraprop}, we also know that $\fGo$ admits a unique, constraining,  momentum map $\Jo$ such that $\C \subset\Jo^{-1}(0)$, as well as a $\Sigma$-equivariant, constraining, local momentum map of the form
\[
\bJ_{\circ\smbullet} \doteq \boldsymbol{J}_{\mathfrak{N}\smbullet} = \iota_{\fN}^*(\bHo + d\bh_\smbullet), \qquad \Jo=\int_\Sigma \bJ_{\circ\smbullet}.
\]

Finally, to show that, truly, $\C=\Jo^{-1}(0)$, we use Corollary \ref{cor:justconstr}: since $\fGo$ is maximal and $\fG_c$ is a constraining ideal, one has $\fG_c\subset\fGo$ and therefore by the corollary, it follows that
\[
\Jo^{-1}(0) \subset J^{-1}_c(0) = \C,
\]
so that $\Jo$ is just. This concludes the proof of Theorem \ref{thm:fGo}.
\end{proof}

\begin{proposition}\label{prop:Noffideal}
The off-shell flux annihilator $\fNoff$ is a just constraining ideal.
\end{proposition}
\begin{proof}
The fact that $\fNoff$ is a constraining ideal follows from the obvious inclusion $\fNoff \subset \fN $ (Proposition \ref{prop:Nsubalgebraprop}), whereas the fact that it is just follows from $\fG_c\subset \fNoff$ (cf.\ Lemma \ref{lemma:fGc}). The latter inclusion can be checked by observing that for any $\xi\in\fG_c$ we have $\langle \bd h,\xi\rangle=0$ since  this expression only depends on restrictions of $\xi$ and its jets at $\pp\Sigma$.
\end{proof}

We conclude this section with two important remarks in which we compare $\fN=\fGo$ with both $\fG_c$ and $\fNoff$.

\begin{remark}[$\fGo$ vs. $\fG_c$]
The compact-support constraining ideal $\fG_c$  satisfies all the desired properties we want for $\fGo$---\emph{except} for maximality. In fact this also shows that \emph{maximality} of the constraining ideal is, in general, a stronger requirement than being ``just'' (see Definition \ref{def:constraintideal}).
Note also that the definition of $\fG_c$ does not depend on the local momentum form $\bH$ and, in particular, it does not depend on the local flux form $d\bh$.
This is in contrast to the definition of $\fN$, and therefore to $\fGo$. It is precisely thanks to its \emph{maximality} that $\fGo$ is ``locked'' to the precise form of the local momentum map $\H=\Ho+d\h$ and thus to the on-shell flux map $\iota_\C^*\h_\smbullet$, so that $\fGo=\fN$ (Theorem \ref{thm:fGo}). To see this, we first note that given a constraining subalgebra $\mathfrak{H}$:(\textit{i}) $\iota_\C^* J_\mathfrak{H} = 0$, and (\textit{ii})  $J_\mathfrak{H}^\text{co} = \iota_\mathfrak{H}^* (\Ho + h_\smbullet)^\text{co}$ (see Lemma \ref{lem:uniqueJ} and Remark \ref{rmk:co-}). From these we deduce that, on $\C$, we have $0=(\iota_\C^*J_\mathfrak{H})^\text{co}  = \iota^*_\mathfrak{H}(\iota_\C^*h_\smbullet)^\text{co}$ and hence one sees that the definition of $\mathfrak{H}=\fGo$ as the \emph{maximal} constraining ideal depends on the properties of the (adjusted) on-shell flux map $\iota_\C^*h_\smbullet$.
The ultimate importance of the maximality condition relates to the characteristic distribution $\C^\omega$, and will be explained in Proposition \ref{prop:symplecticclosurenogo}.
\end{remark}

\subsubsection{The on- and off-shell flux gauge algebras and groups}

Having fully characterized the Lie \emph{ideal} of constraint gauge transformations, and strong of Proposition \ref{prop:Noffideal}, we can introduce the algebra of flux gauge transformations both on- and off-shell, by quotienting out the respective ideals:
\begin{definition}[On/Off-shell flux gauge algebras]\label{def:fGp}
The on-shell and off-shell \emph{flux gauge algebras} are the quotients:
\[
    \fGred\doteq\fG/\fGo \equiv \fG/\fN, \qquad  \fGredoff \doteq \fG/ \fNoff.
\]
We denote the elements of $\fGred$ and $\fGredoff$ respectively by $\underline{\xi} \equiv \xi+\fGo$ and $\underline{\xi}^{\off}= \xi + \fNoff$. 
\end{definition}

It follows that there is a canonical isomorphism of Lie algebras:
\[
\fGred \simeq \frac{\fG/\fNoff}{\fN/\fNoff} = \frac{\fGredoff}{\fN/\fNoff}.
\]

In Theorem \ref{thm:corneralgd}, we will provide an alternative, equivalent, characterization of $\fGredoff$ which will allow us to show that $\fGred$ and $\fGredoff$ are canonically isomorphic to algebras \emph{supported on the corner} $\partial\Sigma$ (Definition \ref{def:LocLieAlg}).

So far we have defined the constraint and flux gauge \emph{algebras}. We are now going to introduce the corresponding \emph{groups}.

\begin{assumption}[Constraint and flux gauge groups]\label{ass:groups}
The on- and off-shell flux annihilators $\fN$ and $\fNoff$ are enlargeable to connected, normal, Fr\'echet, Lie subgroups of $\G$, denoted $\Go$ and $\Gooff$ respectively, which we call the \emph{on- and off-shell constraint gauge groups}. 
Moreover, the quotients 
\[
\Gred \doteq \G/\Go \quad\text{and}\quad \Gredoff=\G/\Gooff
\]
are locally exponential  Fr\'echet Lie groups.
\end{assumption}

\begin{definition}
The Fr\'echet Lie groups $\Gred$ and $\Gredoff$ are called the \emph{on- and off-shell flux gauge groups}, respectively.
\end{definition}

\begin{remark}\label{rmk:split}
The following results are useful to check the validity of Assumption \ref{ass:groups} in concrete cases: \cite[Theorems IV.3.5, IV.4.9, and VI.1.10]{Neeb-locallyconvexgroups}.
They also state that, if the assumption is valid, $\fGred$ and $\fGredoff$ are enlargeable to $\Gred$ and $\Gredoff$, respectively.
Moreover, since we assumed from the start that $\G$ is connected, $\Gred^{(\text{off})}$ is necessarily connected too.
Finally, if $\Go \hookrightarrow \G$ happens to be split, then $\G$ is a principal principal bundle over $\Gred$ with fibre $\Go$ \cite[Definition I.3.5 and Remark IV.4.13]{Neeb-locallyconvexgroups}.
\end{remark}

Observe that we also have
\[
\Gred \simeq \frac{\Gredoff}{\Go / \Gooff}.
\]

\begin{remark}[Higher order theories]\label{rmk:Goingravity}
The fact that in the theories we explicitly analyze the flux gauge group and algebra are simply given by the gauge group and algebra associated to the manifold $\partial\Sigma$, appears to be a rather special occurrence. In particular, it is a consequence of the fact that $\H$ is of order-1 and $\h$ is of order-0.
For theories in which $\bH$ is of order $k>1$, $\fGo$ has a more subtle characterization, which might generically involve transverse jets of $\xi$ at $\partial \Sigma$. Cf. also Section \ref{sec:ultralocality}.
\end{remark}

\subsection{Yang--Mills theory: flux gauge groups}\label{sec:runex-fluxgaugegroup}

As in the previous running example section we treat the  semisimple and Abelian cases separately. 

The semisimple case is straightforward. In Section \ref{sec:runex-fluxannihilators} we have shown that the constraint gauge ideal $\fN\simeq\fGo$ for YM theory with semisimple structure group $G$ coincides with those elements of $\fG$ that vanish when restricted to the corner,
\[
\fN = \fNoff =  \{ \xi\in\fG \,|\, \xi\vert_{\pp\Sigma}=0 \}.
\]
From this, we find on the one hand that the constraint gauge ideal $\fGo=\fN$ (Theorem \ref{thm:fGo}) is simply given by those sections of the induced adjoint bundle $\iota_{\pp\Sigma}^*\Ad P$ that vanish at the corner. On the other hand, the on- and off-shell flux gauge algebras coincide and are isomorphic to the algebra of gauge transformations intrinsically defined on the corner surface (this is a special case of a general result about $\fGredoff$, Theorem \ref{thm:corneralgd}):
\begin{equation}
\label{eq:uGYM}    
\fGred = \fGredoff \simeq \Gamma(\pp\Sigma, \iota_{\pp\Sigma}^* \Ad P).
\end{equation}

\medskip

The Abelian case is more subtle.
This is because, in this case, there exists a nontrivial stabiliser which affects the structure of the constraint set (see Proposition \ref{prop:Gausslaw}).
Recall that in the Abelian case $\Ad P$ is a trivial vector bundle and $\fG \simeq C^\infty(\Sigma,\fg)$.
Then, from the results of Section \ref{sec:runex-fluxannihilators},
\[
\fN = \{ \xi \in \fG \ |\ \exists \chi\in\fG_A : \xi\vert_{\pp \Sigma} = \chi\vert_{\pp\Sigma} \} \supset \fNoff
\qquad \text{(Abelian)}
\]
where both $\fNoff$ and $\fGredoff$ are as above, whereas for all Abelian gauge potential $A$ the stabiliser $\fG_A$ is isomorphic to the space of (locally) constant $\fg$-valued functions in $\fG\simeq C^\infty(\Sigma,\fg)$. Therefore, having assumed $\Sigma$ connected, the on-shell flux gauge group is the quotient
\[
\fGred \simeq \Gamma(\pp\Sigma, \iota_{\pp\Sigma}^*\Ad P)/\fg 
\simeq \fGredoff/\fg 
\qquad \text{(Abelian)},
\]
where $\fg$ is embedded in $\fGredoff$ as the algebra of $\fg$-valued functions on $\pp\Sigma$ which take the same value on the boundary.
For instance, if $G$ is Abelian and $\pp\Sigma$ has one single connected component, then the on-shell flux gauge algebra $\fGred$ is (isomorphic to) the algebra of gauge transformations supported on the corner \emph{modulo} constant gauge transformations.

In both the semisimple and Abelian cases, $\fGredoff$ is the space of sections of the induced adjoint bundle on $\partial \Sigma$, i.e.\ the gauge Lie algebra for principal connections on $\iota_{\pp\Sigma}^*P\to\pp\Sigma$. 
This algebra can also be identified with the equivariant mapping class group $\fGred \simeq C^\infty(\iota_{\partial\Sigma}^*P,\fg)^G$.

Recall that $\Sigma$ is assumed connected. If the bundle $P\to\Sigma$ is trivial, adapting the arguments of \cite[Sections 5 and 6]{RielloSchiavinanull}, one can show that $\Gred$ is a connected discrete central extension by $\mathcal{K}$ of either the identity component of the mapping group $C^\infty(\pp\Sigma,G)$, denoted $C_0^\infty(\pp\Sigma,G)$, in the semisimple case, or of the quotient $C_0^\infty(\pp\Sigma,G)/G$ in the Abelian case, where $G\hookrightarrow C_0^\infty(\pp\Sigma,G)$ is the embedding of $G$ as constant maps. In both cases, $\mathcal{K}$ is the discrete group of components of the ``relative\footnote{``Relative'' refers to $\pp\Sigma$.} subgroup" of the elements of $C_0^\infty(\Sigma,G)$ that are equal to the identity at $\pp\Sigma$, denoted $C_{0,\text{rel}}^\infty(\Sigma,G)$. 

If $G$ is connected and $\Sigma \simeq D^n$ i.e. the $n$-disk, with boundary the $(n-1)$-sphere $S^{n-1}$, then $C_0^\infty(\Sigma,G)\simeq C^\infty(\Sigma,G)$ and $\mathcal{K}\simeq \pi_0(C_{\text{rel}}^\infty(D^n,G))\simeq \pi_0(C^\infty(S^n,G))\simeq \pi_n(G)$. 

Summarising:
\[
\begin{dcases}
\mathcal{K} \to \Gred \to C_0^\infty(\pp\Sigma,G) & \text{(semisimple)}\\
\mathcal{K} \to \Gred \to C_0^\infty(\pp\Sigma,G)/G & \text{(Abelian)}
\end{dcases}
\]
and 
\[
\Sigma\simeq D^n \implies \mathcal{K} \simeq \pi_n(G).
\]

\begin{remark}
We note that this is the group relevant for the study of the $\theta$-vacua and the axial $U(1)$ symmetry breaking. Similarly to \cite{MorchioStrocchiTheta}, this group is here obtained without referring to instantons in an auxiliary Euclidean 4-dimensional theory, but rather by carefully studying the topological properties of the gauge group over $\Sigma$. However, differently from \cite{MorchioStrocchiTheta}, we never invoke the group of gauge transformation over an entire Cauchy surface nor we need to give a meaning to gauge transformations that are not connected to the identity.
Instead, we recover $\pi_3(G)$ by studying reduction by stages with respect to a connected gauge group $\G$ over a topological disk $\Sigma\simeq D^3$ \emph{within} a Cauchy surface.
We defer the study of the relevance for the definition of $\theta$-vacua of our quasi-local derivation of the ``instanton group'' $\pi_3(G)$ to future work, where we will address the quantization of the classical geometrical structures studied in this work.
\end{remark}

Now, since equivariant mapping groups on manifolds with corners, such as $C_0^\infty(\pp\Sigma,G)$, are locally exponential (see \cite{WockelPhD} and Remark \ref{rmk:mappinggrp}) and since $\mathcal{K}$ is discrete, Assumption \ref{ass:groups} is ensured in virtue of \cite[Thm. VI.1.10]{Neeb-locallyconvexgroups} (See also  \cite[Prop. IV.3.4 and Thm. IV.3.5]{Neeb-locallyconvexgroups} and \cite{AbouqatebNeeb,kriegl1997convenient}).

This scenario can be extended to other gauge theories, such as those discussed in Section \ref{sec:applications}.

\section{Locally Hamiltonian symplectic reduction with corners}\label{sec:reduction}
In the previous section we decomposed the locally Hamiltonian space $(\X,\bom,\G,\bH)$ into the data of a constraint form $\bHo$, which induces the constraining, just, equivariant momentum map $\Jo$ of Theorem \ref{thm:fGo} for the Hamiltonian action of $\Go\subset \G$, and the data of a flux form $\d\bh$, and associated flux gauge group $\Gred=\Gred/\Go$. 

In what follows we will apply reduction by stages to this decomposition, as a two-step approach to the description of the quotient $\uuC = \C/\G$, the fully-reduced phase space, the functions over which define the observables of the theory in the presence of boundaries.

\subsection{First stage symplectic reduction: the constraint}\label{sec:bulkreduction}

The first stage of our reduction procedure of the locally Hamiltonian space $(\X,\bom,\G,\bH)$, consists in applying the Marsden--Weinstein--Meyer (MWM) reduction procedure \cite{MarsdenWeinstein,Meyer,RatiuOrtega03} (adapted to the Fr\'echet setting as in \cite{DiezPhD,DiezHuebschmann-YMred}) to the symplectic space $(\X,\omega)$ with respect to the action of the constraint gauge group $\Go\subset\G$, at the \emph{zero}-level set $\C=\bHo^{-1}(0)$ of the \emph{constraint form} $\bHo$.

We assume that reduction can be performed and look at the quotient space $\uCo=\C/\Go$ of the flux gauge group $\Gred = \G/\Go$.

Technically, there are two important steps in the Hamiltonian reduction procedure. 
The first step consists in showing that the pullback of $\omega$  to the constraint set $\C$ is basic, i.e.\ is the pullback of a 2-form $\uomegao$ defined on the quotient space $\uCo$. This is summarized in the diagram below for the reduction $\pi_\circ: \C\to \uCo =\C/\Go$:
\[
\xymatrix@C=.75cm{
(\X,\omega)
	\ar@{~>}[rr]^-{\tbox{2.2cm}{constraint red. (red. by $\Go$ at $0$)}}
&&(\uCo,\uomegao)\\
&\;\C\;
    \ar@{_(->}[ul]^-{\iota_\C}
	\ar@{->>}[ur]_-{\pi_\circ}\\
}
\]

~

The second step consists in showing that the 2-form $\uomegao$ is nondegenerate. 
In infinite dimensions this is not automatic, and we will discuss this point in some generality at the end of the section, and within YM theory in Section \ref{sec:runex-firststage}. 

Thanks to Corollary \ref{cor:tangent}, which states that the fundamental vector fields $\rho(\fG)\vert_\C$ are everywhere tangent to the constraint set $\C\subset\X$, one has that the action of $\G$ on $\X$ induces an action of $\Go$ on $\C$. Therefore, it is meaningful to introduce the following:

\begin{definition}[Constraint-reduced phase space]\label{def:bulkredsp}
The \emph{constraint-reduced phase space}, denoted $\uCo$, is the set of equivalence classes of constrained field configurations $\phi\in\C$ under the action of the constraint gauge group $\Go$
\[
\uCo \doteq \C / \Go.
\]
Denote $\pi_\circ : \C \to \uCo$ the corresponding projection.
\end{definition}

\begin{remark}
In principle, the constraint-reduced phase space should be defined as the coisotropic reduction of the constraint set $\C/\C^\omega$, and then one needs to \emph{prove} that it coincides with the symplectic reduction by the action of $\Go$. For ease of presentation, here we reverse the discussion and introduce the constraint-reduced phase space as the quotient of the constraint set by the action of the constraint gauge ideal. We will then revisit the key steps in the Marsden--Weinstein--Meyer reduction procedure to highlight what extra conditions need to be satisfied in order to have $\C/\Go\simeq \C/\C^\omega$ in the infinite-dimensional Fr\'echet setting.
\end{remark}

\begin{assumption}\label{ass:smoothuCo}
The constraint-reduced phase space $\uCo$ is a smooth manifold and $\pi_\circ : \C \to \uCo$ is a principal $\Go$-bundle.
\end{assumption}

\begin{remark}[Smoothness of quotients]\label{rmk:smooth-uCo}
In order to ensure that the quotient by a group action is smooth, one often requires the action to be smooth and proper. In certain cases this can be established even beyond the traditional Banach setting of \cite{AtiyahBott}. 
However, even if the action is proper, $\C$ may contain configurations with different stabilisers. 
As a consequence, symplectic reduction might be singular, yielding a quotient which is often a \emph{stratified} manifold, each stratum corresponding to configuration  with isomorphic stabilisers---over which reduction is smooth (see \cite{Arms1981, ArmsMarsdenMoncrief1981,  DiezHuebschmann-YMred, DiezRudolph2020, DiezRudolph2022, DiezPhD}).\footnote{For the general theory of singular reduction in finite dimension, see e.g. \cite{ArmsCushmanGotay, sjamaar1991, marsden2007stages}. For application to General Relativity, which in fact motivated most of this research, we recall in particular the names of Arms, Ebin, Fischer, Isenberg, Marsden, Moncrief, and Palais. See also \cite{CattaneoFelderPoisson,Contreras} for the Poisson sigma model, where $\uCo$ is the symplectic groupoid (when smooth).}

We should note that in concrete cases it often happens that there are no $\Go$-reducible configurations in $\C$ and therefore that the action of $\Go$ is free. If the action of $\G$ (and thus that of $\Go$) was also proper, this is enough to show that $\uCo$ is also smooth.
E.g., as discussed in Remark \ref{rmk:smooth-C-revisited}, this is true in any theory of principal connections---such as Yang--Mills theory---if boundaries are present. 
\end{remark}

\begin{remark}[Nonlocality of $\uCo$]\label{rmk:nonlocality}
The constraint-reduced space $\uCo$ in general fails to be local, in the sense that it is not the space of sections of some bundle over $\Sigma$. There are a priori two potential sources of nonlocality:
\begin{enumerate}
    \item in theories where fields transform by a derivative under a gauge transformation, i.e.\ in which $\rho$ is of order $k>0$, one expects that quotienting by $\Go$ be a source of nonlocal behaviour. Theories involving principal connections or diffeomorphism symmetry belong to this class;
    \item even when a set of local gauge-invariant field-space coordinate exists (e.g.\ the electric and magnetic fields in Maxwell theory), the on-shell field configurations must be solutions to some elliptic PDE (e.g.\ the Gauss constraint in Maxwell theory), and therefore $\C$ is ``nonlocal'' in the sense that it cannot be expressed as the space of sections of some bundle over $\Sigma$.
\end{enumerate} 
Therefore, to work on $\uCo$, the local setting does not suffice and a more general infinite-dimensional geometric setting becomes necessary. \qedhere
\end{remark}

\begin{definition}
Denoting $T_\C \X$ the ambient tangent bundle over $\C$,
let $V_\circ \subset  T_\C\X$ be the image of $\C\times\fGo$ under the anchor map $\rho:\A \to T\X$ (Equation \eqref{eq:A}), 
\[
V_\circ \doteq \rho(\C\times\fGo)
\qedhere
\]
\end{definition}

\begin{remark}
Note that $V_\circ$ can be equivalently defined as the union of the tangent spaces (in $\X$) to the orbits under $\G_\circ\subset\G$ of the elements of $\C$, i.e.\ $V_\circ = \bigcup_{\phi\in\C} T\mathcal{O}^\circ_\phi$, where $\mathcal{O}^\circ_\phi = \phi\triangleleft\Go\subset \X$ is the orbit of $\phi$ under the action of $\Go$.
Then, thanks to Corollary \ref{cor:tangent}, it is immediate to see that $V_\circ$ is in fact tangent to $\C$, i.e.\ 
\[
V_\circ \subset  T\C \subset  T_\C\X.
\qedhere
\]
\end{remark}

\begin{lemma}\label{lemma:coisoC} The following two statements hold true:
\begin{enumerate}[label=(\roman*)]
    \item $V_\circ^\omega = T\C $. \label{prop:C2dual}
    \item $\C\subset\X$ is coisotropic,\label{prop:C1}
\end{enumerate}
\end{lemma}
\begin{proof}
Let $\xio \in \fGo$. From Theorem \ref{thm:fGo}, we have  $\C = \Jo^{-1}(0)$ as well as the following Hamiltonian flow equation for $\Go$
\[
\bi_{\mathbb{X}} \bi_{\rho(\xio)}\omega = \langle \L_{\mathbb{X}} J_\circ, \xio \rangle,
\]
which we wrote contracted with an arbitrary vector $\mathbb{X} \in T\X$.
Now, we evaluate the above equation at a $\phi\in\C\subset \X$. If $\mathbb{X}\in T_\phi\C \subset T_\phi\X$, then the right-hand side vanishes because $\C$ is the zero-level set of $\Jo$: this means that $T\C \subset V_\circ^\omega$.
Conversely, if $\mathbb{X}\in V_\circ^\omega$, then the left-hand side vanishes by definition, and thus---using once again that $\C=\Jo^{-1}(0)$---we deduce $V_\circ^\omega \subset T\C$. 
From these two inclusions, we conclude that $V_\circ^\omega = T\C$. Finally, the coisotropy of $\C$, i.e. $T\C^\omega \subset  T\C$, follows from this statement together with the observation that $V_\circ \subset T\C$ (Corollary \ref{cor:tangent}) implies $T\C^\omega \subset V_\circ^\omega$.
\end{proof}

We now construct a symplectic form on $\uCo$. Recall Theorem \ref{thm:fGo} which states that $\C$ is the zero-level set of the constraining $\Go$-equivariant momentum map $J_\circ \doteq \iota^\ast_{\fGo}( \Ho + h_\smbullet)$.
A standard application of Marsden--Weinstein--Meyer reduction theory (see e.g. \cite{RatiuOrtega03}) allows us to define the following induced 2-form on $\uCo$: 

\begin{definition}[Constraint-reduced form]\label{def:uomegao}
Let $\uomegao\in\Omega^{2}(\uCo)$ be the unique 2-form on $\uCo$ such that
\begin{equation}\label{e:uomegao}
\pi_\circ^* \uomegao = \iota_\C^*\omega, \qquad \pi_\circ\colon \C\to\uCo.
\end{equation}
Then, we call $\uomegao$ the \emph{constraint-reduced form}.
\end{definition}

What does not automatically go through in infinite dimensions is the proof of nondegeneracy of the constraint-reduced form $\uomegao$, which requires an additional assumption. Indeed, the two-form $\uomegao$ is nondegenerate if, and only if, the kernel of $\iota^*_\C\omega$ is given solely by the $\pi_\circ$-vertical vector fields corresponding to the infinitesimal action of $\Go$, i.e.\ 
\begin{equation}\label{eq:nondeg}
\uomegao^\flat \text{ is injective iff } \ker(\iota_\C^*\omega^\flat) = V_\circ.
\end{equation}
Since $\ker(\iota_\C^*\omega^\flat) = T\C \cap \ker(\omega^\flat\vert_{T_\C\X})\equiv T\C \cap T\C^\omega$, this condition is often seen as the consequence of the following two (stronger) conditions \cite{MarsdenRatiuPoissonRed}:
\begin{enumerate}[label=(C\arabic{*})]
    \item $\C$ is a coisotropic submanifold of $(\X,\omega)$, i.e.\ $T\C^\omega\subset  T\C$,  \label{C1}
    \item $V_\circ$ is the symplectic complement of $T\C$ in $(\X,\omega)$, i.e.\ $V_\circ = T\C^\omega$. \label{C2}
\end{enumerate} 
Lemma \ref{lemma:coisoC} already subsumes \ref{C1}, however it only implies the ``\emph{symplectic complement}" of condition \ref{C2}.

Although in finite dimension the non-degeneracy of $\omega$ readily implies ``symplectic reflexivity'' of all subspaces $V\subset  W$, i.e.\ $V^{\omega\omega} = V$, in the infinite-dimensional case one can in general only conclude that $V \subset  V^{\omega\omega}$. 
Therefore, Lemma \ref{lemma:coisoC} only implies that $V_\circ \subset  V_\circ^{\omega\omega}=T\C^\omega$. Therefore, condition \ref{C2}, and with it the nondegeneracy condition \eqref{eq:nondeg}, holds iff $V_\circ = V_\circ^{\omega\omega}$.    
We will refer to this condition as the \emph{symplectic closure} of $V_\circ$ (see \cite[Prop. 4.1.7 and Lem. 4.2.14]{DiezPhD} for an explanation of this nomenclature).

We can summarise these considerations in the following:
\begin{proposition}\label{prop:symplclosVo}
The constraint-reduced symplectic form $\uomegao^\flat$ is injective iff $V_\circ = \rho(\C\times\fGo)$ is symplectically closed, i.e.\ iff $V_\circ = V_\circ^{\omega\omega}$.
\end{proposition}

Whether $V_\circ$ is actually symplectically closed depends on the specific functional analytic properties of $(\X,\omega,\G, \rho)$. For more detail on such conditions over Fr\'echet spaces, see \cite[Section  4]{DiezPhD} (as well the end of Section \ref{sec:runex-firststage}). However, it is important to emphasise that the \emph{maximality} requirement defining $\fGo$ is in fact a \emph{necessary} (albeit not sufficient) condition for $V_\circ$ to be symplectically closed and therefore for the constraint-reduced space $\uCo$ to be symplectic. More precisely:

\begin{proposition}[Symplectic closure and the maximality of $\fGo$]\label{prop:symplecticclosurenogo} 
Assume that the on-shell isotropy bundle is trivial $\mathsf{I}_\rho = \C \times \mathfrak{I}$ (Remark \ref{rmk:isotropybundle2}). Let $\mathfrak{K}$ be a just constraining ideal (i.e.\ $\C=J_{\mathfrak{K}}^{-1}(0)$) such that $\mathfrak{I}\subset \mathfrak{K}$, and let $V_\mathfrak{K}\doteq \rho(\C\times\mathfrak{K})$ be the associated vertical distribution over $\C$. Then, $V_{\mathfrak{K}}\doteq \rho(\C\times \mathfrak{K})$ can be symplectically closed only if $\mathfrak{K}$ is equal to $\fGo$, i.e.
\[
\mathfrak{K} \subsetneq\fGo \implies V_\mathfrak{K} \subsetneq V_\mathfrak{K}^{\omega\omega}.
\]
\end{proposition}
\begin{proof}
Using Proposition \ref{prop:Nsubalgebraprop}, the proof of Lemma \ref{lemma:coisoC} can be repurposed to show that since $\mathfrak{K}$ is a \emph{just} constraining ideal it must satisfy
\[
V_\mathfrak{K}^\omega = T\C.
\]
As already observed, in infinite dimensions this equation generally only implies $V_\mathfrak{K} \subset T\C^\omega$, but not their equality. However, from the maximality condition one has that for all $\mathfrak{K}\neq\fGo$, $\mathfrak{K} \subsetneq \fGo$, and hence---thanks to the condition on the isotropy, $\ker(\rho_\phi)=\mathfrak{I}\subset\mathfrak{K}\subset\fGo$---also that $V_\mathfrak{K}\subsetneq V_\circ$.
Summarising:
\[
V_\mathfrak{K}\subsetneq V_\circ \subset T\C^\omega  = V_\circ^{\omega\omega} = V_\mathfrak{K}^{\omega\omega}.
\]
Comparison between the first and the last term concludes the proof.
\end{proof}

\begin{remark}\label{rmk:uniquejustfd2}
Note that it is still possible to have a just constraining ideal $\mathfrak{K} \subsetneq \fGo$ such that $V_\mathfrak{K} = V_\circ$,  
e.g. if $\fGo = \mathfrak{K} + \mathfrak{I}$ with $\mathfrak{K}\simeq \fGo/\mathfrak{I}$.
For an explicit example consider $\mathfrak{K} = \{\xi\in C^\infty(\Sigma,\fg)\ |\ \xi\vert_{\pp\Sigma}=0\}$ in Abelian Yang--Mills theory with $\fg=\mathbb{R}$, where the (configuration-independent) isotropy is given by constant gauge transformations $\mathfrak{I} = \fg\hookrightarrow\fG$, and the constraint gauge algebra by $\fGo =  \mathfrak{K} + \mathfrak{I}$.
Furthermore, we note that from Proposition \ref{prop:symplecticclosurenogo} we can conclude that if $\X$ is finite dimensional then $\mathrm{dim}(V_\mathfrak{K}^{\omega\omega}) = \mathrm{dim}(V_\mathfrak{K})$ and therefore $V_{\mathfrak{K}} \simeq V_\circ$ for every just constraining ideal $\mathfrak{K}$, i.e.\ that just constraining ideals are unique up to elements in the kernel of the action $\rho$ (cf.\ Remark \ref{rmk:uniquejustfd1}).
\end{remark}

\begin{assumption}\label{ass:symplecticclosure}
The vector space $V_\circ = V_\circ^{\omega\omega}$ is symplectically closed.
\end{assumption}

\begin{remark}[Symplectic closure and $\C^\omega$]\label{rmk:symplclosure-Comega}
This section proves Theorem \ref{mainthm:redsummary}, item \ref{thmitem:main-constraint red}. Observe that the symplectic closure of $V_\circ$ (Assumption \ref{ass:symplecticclosure} allows us to say that the characteristic distribution $\C^\omega$ is given by the the image of $\fGo$ through $\rho$, regardless from the smoothness of the quotient (Assumption \ref{ass:smoothuCo}). When, additionally, $\uCo$ is smooth, symplectic closure ensures that $\uomegao^\flat$ is injective and $(\uCo,\uomegao)$ is \emph{symplectic}.
\end{remark}

\subsection{Yang--Mills theory: constraint-reduced phase space}\label{sec:runex-firststage}

In this part of the running example we show how to use a principal connection on the field space bundle $\Go\hookrightarrow \C\to \uCo$ to provide a parametrization of the constraint-reduced phase space $\uCo$ of YM theory. This parametrization can furthermore be made explicit by introducing a particular choice of connection. 

In the semisimple case, recall from Remark \ref{rmk:smooth-C-revisited} that the action of $\Go$ on $\Acal$ is free. Setting aside for a moment questions of smoothness to be addressed later, we can therefore consider $\Acal\to\Acal/\Go$ to be a principal $\Go$-bundle. 

In the Abelian case, $\Go$ does not act freely, because constant gauge transformations are in the kernel of $\rho$. However, gauge transformations that trivialize at $\pp\Sigma$ act freely, and the quotient of $\X$ (resp.\ $\Acal$) with respect to them coincides with $\X/\Go$ (resp.\ $\Acal/\Go$). Therefore, if one replaces $\Go$ by the set of boundary-trivial gauge transformations, the following discussion holds for both the semisimple and Abelian cases (note, however, that $\fGred$ in the Abelian case must still be taken to be $C^\infty(\pp\Sigma,\fg)/\fg$).

In the absence of boundaries, the bundle\footnote{In the cited literature $\pp\Sigma=\emptyset$ so that $\Go = \G$ and instead of excluding reducible configurations focus is placed on the group of pointed gauge transformations, i.e. gauge transformations that are trivial at a marked point of $\Sigma$. Note that pointed gauge transformations on a $n$-sphere are homotopic to gauge transformations on the $n$-disk that vanish at the boundary. Cf. Remark \ref{rmk:smooth-C-revisited}.} $\G\hookrightarrow \Acal\to \Acal/\G$ was first studied by means of a connection form by \cite{Singer1978,NarasimhanRamadas79} in the Banach setting (see also \cite{ArmsMarsdenMoncrief1981}); their work was later extended to the $C^\infty$ setting by \cite{MitterViallet1981}. The tame-Fr\'echet properties of this fibration are discussed in \cite{AbbatiCirelliMania,AbbatiCirelliManiaMichor}. The interplay between with boundaries was first addressed in \cite{GomesRiello-ghost, RielloGomesHopfmuller, RielloGomes, RielloSciPost}

\begin{definition}\label{def:Upsilon}
Let $\mathcal{H}$ be either $\Go$ or $\G$, and $\mathfrak{H}=\mathrm{Lie}(\mathcal{H})$.
A $\mathfrak{H}$-valued 1-form $\Upsilon_{\mathfrak{H}} \in \Omega^1(\Acal,\mathfrak{H})$ is said to be a \emph{principal (Ehresmann) connection form} on $\Acal\to\Acal/\mathcal{H}$ if 
\begin{equation}
    \label{eq:varpi}
\begin{cases}
\bi_{\rho(\xi)} \Upsilon_{\mathfrak{H}} = \xi\\
\L_{\rho(\xi)}\Upsilon_{\mathfrak{H}} = [\Upsilon_{\mathfrak{H}},\xi]
\end{cases}
\qquad\forall\xi\in\mathfrak{H}. 
\end{equation}
We then denote by $\mathbb{F}_{\Upsilon_\mathfrak{H}} =\bd \Upsilon_\mathfrak{H} + \Upsilon_\mathfrak{H} \wedge \Upsilon_\mathfrak{H}\in \Omega^2(\Acal,\mathfrak{H})$ its curvature 2-form.

The distribution $V_\mathfrak{H} \doteq \rho(\Acal\times\mathfrak{H})\subset T\Acal$ generated by the fundamental vector fields is said the \emph{vertical distribution}\footnote{The vertical distribution corresponds to the ``pure gauge'' directions in $\Acal$.} in $T\Acal$ associated to $\mathfrak{H}$.

If $\mathfrak{H}=\fGo$ (resp. $\fG$), we simply write $\Upsilon_{\mathfrak{H}} = \Upsilono$ (resp $\Upsilon$), $V_\circ$ (resp. $V$).
\end{definition}

\begin{remark}[Connection forms]
The kernel of a (principal) Ehresmann connection, $\ker(\Upsilon_\mathfrak{H}) \subset T\Acal$ is called the $\Upsilon_\mathfrak{H}$-\emph{horizontal distribution} of $T\Acal$. Denote $\underline\Upsilon_{\mathfrak{H}}$ the corresponding projector:
\[
\underline\Upsilon_{\mathfrak{H}} : T\Acal \to \ker(\Upsilon_\mathfrak{H}).
\]
Then one has the splitting:
\[
T_{A}\Acal = \rho(\mathfrak{H})\vert_{A} \oplus \ker(\Upsilon_\mathfrak{H}\vert_{A}).
\]
A (principal) Ehresmann connection defines an equivariant horizontal lift of vectors from the base to the bundle:
\[
\check{\Upsilon}_\mathfrak{H}:\Acal\times_{\Acal/\Go} T(\Acal/\Go) \to T\Acal, \quad \Im(\check{\Upsilon}_\mathfrak{H})=\ker(\Upsilon_\mathfrak{H})=\Im(\underline\Upsilon_{\mathfrak{H}}).
\qedhere
\]

Finally, we introduce the following $\Upsilon_\mathfrak{H}$-horizontal differentials (with respect to the action $\mathfrak{H}\circlearrowright\X$) in $\Omega^1(\X)$ (generalizations to higher forms are straightforward):
\[
\bd_{\Upsilon_{\mathfrak{H}}} A \doteq \bd A + d_A \Upsilon_{\mathfrak{H}}, 
\qquad
\bd_{\Upsilon_{\mathfrak{H}}} E \doteq \bd E - \ad^*(\Upsilon_{\mathfrak{H}})\cdot E,
\]
which owe their name to the identity $\bi_\mathbb{X} \bd_{\Upsilon_{\mathfrak{H}}} A  \equiv \bi_{\underline \Upsilon_{\mathfrak{H}}(\mathbb{X})}\bd A$, and similarly for $A\leadsto E$.\footnote{\label{fnt:trivialization}Recall that $\bi_\bullet$ and $d$ anticommute: $\bi_{\rho(\xi)}\bd_{\Upsilon_{\mathfrak{H}}} A = \L_{\rho(\xi)}A - d_A \bi_{\rho(\xi)} \Upsilon_{\mathfrak{H}} = 0$. We write ${\Upsilon_{\mathfrak{H}}}$ to the left of $E$ to ensure that $\bi_{\rho(\xi)}\ad^*({\Upsilon_{\mathfrak{H}}})\cdot E = \ad^*(\xi)\cdot E$, even though $\Upsilon_{\mathfrak{H}}$ and $\bi_\bullet$ are odd-degree objects and $\mathrm{deg}(E) = \mathrm{dim}(\Sigma)-1$. Our sign-conventions differ from \cite{RielloGomes, RielloSciPost} cited below.} The horizontal differentials are  equivariant, i.e.\ $\L_{\rho(\xi)} \bd_{\Upsilon_{\mathfrak{H}}} A = [\bd_{\Upsilon_{\mathfrak{H}}} A,\xi]$ and similarly for $A\leadsto E$. 
\end{remark}

\begin{remark}[Sections and flat connections]\label{rmk:sectionflatconn}
To each section $\sigma_\mathfrak{H}: \Acal/\mathcal{H} \hookrightarrow \Acal$---i.e.\ to each ``gauge fixing'' of the action of $\mathcal{H}$ in $\Acal$---there corresponds a trivialization of $\Acal \stackrel{\sigma_\mathfrak{H}}{\simeq}  (\Acal/\mathcal{H}) \times \mathcal{H} \ni ( [A], k)$ and thus a flat connection $\Upsilon_\mathfrak{H} = k^{-1}\bd k$ such that $\ker(\Upsilon_\mathfrak{H} \vert_{\mathrm{Im}(\sigma_\mathfrak{H})})=T\mathrm{Im}\sigma_\mathfrak{H}$. 
The nonexistence of global sections is often referred to as the ``Gribov problem'' \cite{Gribov, Singer1978, Singer1981, NarasimhanRamadas79}.
For this reason, one can think of $\Upsilon_\mathfrak{H}$ as a generalization of the notion of gauge fixing, which does away with the flatness condition and can be applied even when global sections of $\mathcal{H} \hookrightarrow \Acal \to \Acal/\mathcal{H}$ do not exist. 
\end{remark}

\medskip

On the configuration space of YM theory one has the DeWitt metric (cf.\ \cite[Equation 6.3]{DeWitt1967QG1} and the references in footnote \ref{fnt:Coulconnection}):

\begin{definition}[DeWitt metric]\label{def:deWitt}
Let $\mathbb{X} \in T_A\Acal$ be given by $\mathbb{X}(A_i(x)) = X_i(x)$, and similarly for $\mathbb{Y}$.
The \emph{DeWitt metric} over $\Acal$ is defined by 
\[
\langle\langle \mathbb{X}, \mathbb{Y} \rangle\rangle \doteq \int_\Sigma \gamma^{ij}\tr( X_i Y_j)vol_\Sigma.
\] 
The DeWitt diffeomorphism is the associated isomorphism 
\[
\varphi_{\text{dW}} : T\Acal \to T^\vee \Acal  = \X \simeq \Acal \times \mathcal{E}, \quad (A, \mathbb{X}) \mapsto (A,E) = (A, \star X).
\qedhere
\]
\end{definition}

Since the DeWitt metric is gauge invariant, i.e. $\L_{\rho(\xi)}\langle\langle \mathbb{X},\mathbb{Y}\rangle\rangle = \langle\langle [\rho(\xi),\mathbb{X}],\mathbb{Y}\rangle\rangle + \langle\langle\mathbb{X},[\rho(\xi),\mathbb{Y}]\rangle\rangle$, it can be used to define a principal (Ehresmann) connection by means of an orthogonal splitting in $T\Acal$, cf.\ Definition \ref{def:Upsilon}.  We call the corresponding connection the Coulomb connection.\footnote{\label{fnt:Coulconnection}To the best of our knowledge this connection was first introduced in \cite{Singer1978} and, independently, by \cite{NarasimhanRamadas79}, in the absence of boundaries. See also \cite[Chapter 24]{DeWittBook} or \cite{DeWitt2005FiftyYearsYM}. In the presence of boundaries, the first treatment for $\mathcal{H}=\G$ was given in \cite{RielloGomesHopfmuller} (see also \cite{RielloGomes}). However, closely related gauge conditions had appeared in \cite{Marini92,SniatyckiSchwarz94,SniatyckiSchwarzBates1996}. The connection provided here is adapted to the action of $\Go$.} It is easy to check that for $\mathcal{H}=\Go$ or $\G$ this definition is equivalent to the following:

\begin{definition}[Coulomb connections on $\Acal$]\label{def:Coulombconnection}
The \emph{Coulomb connections} $\UpsilonoCou$ and $\UpsilonCou$ associated to $\Go$ and $\G$ respectively, are given by the unique\footnote{In the Neumann case, this requires $A$ to be irreducible, see Appendix \ref{app:Hodge}, where details on the notation can be found the notation.} solutions to the following elliptic Dirichlet, Neumann resp., boundary value problems:
\[
\begin{cases}
\Delta_A \UpsilonoCou = d_A^\star \bd A & \text{in }\Sigma \\
\UpsilonoCou = 0 & \text{at }\pp\Sigma
\end{cases}
\quad \text{and, resp.,} \quad
\begin{cases}
\Delta_A \UpsilonCou = d_A^\star \bd A & \text{in }\Sigma \\
\mathsf{n}d_A \UpsilonCou = \mathsf{n}\bd A & \text{at }\pp\Sigma
\end{cases}
\]
where $\Delta_A = d_A^\star d_A$ is the gauge-equivariant Laplace-Beltrami operator at $A\in\Acal$ over $(\Sigma,\gamma)$. Denote the associated horizontal derivatives by $\bd_{\circ,\mathrm{Cou}}$ (resp. $\bd_{\mathrm{Cou}}$).
\end{definition}

The Coulomb connections are flat iff $G$ is Abelian. \AR{More specifically their curvature is given by \cite{Singer1978,NarasimhanRamadas79,RielloGomesHopfmuller}:}

\begin{lemma}\label{lemma:curv-Coul-conn}
    The curvature of the Coulomb connections $\UpsilonoCou$ and $\UpsilonCou$, i.e. $\mathbb{F}_{\circ,\mathrm{Cou}}$ and $\mathbb{F}_\mathrm{Cou}$ resp., satisfy the following elliptic boundary value problems:
    \[
    \begin{cases}
    \Delta_A \mathbb{F}_{\circ,\mathrm{Cou}} = \gamma^{ij}[\bd_{\circ,\mathrm{Cou}}A_i, \bd_{\circ,\mathrm{Cou}}A_j] & \text{in }\Sigma \\
    \mathbb{F}_{\circ,\mathrm{Cou}} = 0 & \text{at }\pp\Sigma
    \end{cases}
    \]
    and, resp.,
    \[
    \begin{cases}
    \Delta_A \mathbb{F}_{\mathrm{Cou}} = \gamma^{ij}[\bd_{\mathrm{Cou}}A_i, \bd_{\mathrm{Cou}}A_j] & \text{in }\Sigma \\
    \mathsf{n}d_A \mathbb{F}_{\mathrm{Cou}} = 0 & \text{at }\pp\Sigma
\end{cases}
    \]
\end{lemma}

\begin{lemma}\label{lem:coulombicKernel}
    Let $\mathbb{X}(A_i(x)) = X_i(x)$ be a tangent vector in $T_A\Acal$, as above. Then, $\mathbb{X}$ is Coulomb-horizontal with respect to $\Go$ iff
    \[
     \underline\UpsilonoCou(\mathbb{X}) = \mathbb{X} \iff  \UpsilonoCou(\mathbb{X}) = 0 \iff d_A^\star X = 0.
    \]
\end{lemma}
\begin{proof}
The horizontality condition holds iff $\UpsilonoCou(\mathbb{X}) = 0$. On the other hand, contracting the defining equation of $\UpsilonoCou$ with $\mathbb{X}$, we compute $\Delta_A(\UpsilonoCou(\mathbb{X})) = d_A^*X$.
Using the invertibility of $\Delta_A$ with Dirichlet boundary condition (Proposition \ref{prop:FaddeeevPopov}), we conclude.
\end{proof}

\begin{remark}[Arms--DeWitt K\"ahler structure]
The reason why the Coulomb connection and the radiative-Coulombic decomposition $\mathcal{E}_A=\mathcal{V}_A\oplus\mathcal{H}_A$ (cf.\ Equation \eqref{eq:YMcotangentsplit}) are convenient in dealing with the symplectic reduction of YM theory, is that they are both orthogonal with respect to $\langle\langle\cdot,\cdot\rangle\rangle$ (Remark \ref{rmk:HodgeInYM}), and there exists a complex structure $j:T\X\to T\X$, $j^2=-1$, that intertwines the symplectic 2-form $\omega$ with the DeWitt supermetric  $\langle\langle\cdot,\cdot\rangle\rangle$  on $\X\simeq \Acal \times \mathcal{E}$ \cite{Arms1981}. 
More precisely, if $\mathbb{X} = (\mathbb{X}^{(A)}, \mathbb{X}^{(E)})$, one sets $j\mathbb{X} \equiv -( \star \mathbb{X}^{(E)}  , (-1)^{\Sigma}\star \mathbb{X}^{(A)})$, and thus finds the K\"ahler compatibility
\[
\omega(\mathbb X_1 , \mathbb X_2) = \langle\langle\mathbb{X}_1, j \mathbb{X}_2\rangle\rangle.\qedhere
\]
\end{remark}

\begin{proposition}\label{prop:YMlocaliso}
Assume $\Acal/\Go$ and $\A/\G$ are smooth. 
    \[
    \uCo \simeq T(\Acal/\Go) \simeq_\loc T(\Acal/\G)\times T\Gred.
    \]
\end{proposition}
\begin{proof}
    Using Lemma \ref{lem:coulombicKernel} and the identification $\X = T^\vee\Acal \simeq \Acal \times \mathcal{E}$ (Section \ref{sec:runex-setup}), we have that 
    \[
    \varphi_{\text{dW}}(\mathrm{Ker}({\Upsilon}_{\circ,\mathrm{Cou}})) \simeq \{(A,E)\in\X\ |\ d_AE =0\} = \C.
    \]

    From the $\Go$-equivariance of $\UpsilonoCou$ it follows that $\ker(\UpsilonoCou)/\Go \simeq T(\Acal/\Go)$. Moreover, since the DeWitt metric is gauge-invariant, $\varphi_\mathrm{dW}$ is gauge equivariant isomorphism which descends to the quotients (by $\Go$ or $\G$). Hence, together with the previous equation, we deduce
    \[
    T(\Acal/\Go) \simeq \C/\Go \doteq \uCo.
    \]
    
    Finally, note that since $\Gred = \G/\Go$ and since $\Acal$ is a principal $\G$-bundle (in view of our restriction to irreducible configurations), we have that locally $\Acal/\Go \simeq_\loc \Acal/\G \times \Gred$. Passing to the tangent spaces, we conclude.
\end{proof}

\begin{remark}
    An intrinsic description of the quotient $\Acal/\Go$ is hard to achieve, and so is the exhibition of a viable (local) section $\sigma_\circ:\Acal/\Go\to\Acal$. However, a good model for $\sigma_\circ$ is a slice through a given (irreducible) configuration $A\in\Acal$ \cite{palais1957, palais1961, palaisterng1988critical, ArmsMarsdenMoncrief1981}. A viable, and well-known, choice of slice is described in Proposition \ref{prop:FaddeeevPopov}, which uses a generalized Coulomb gauge fixing and the local invertibility of generalized Faddeev--Popov operators.
\end{remark}

We now show that the DeWitt metric allows us to define a natural symplectic structure on $T(\Acal,\Go)$.

\begin{definition}[Reduced DeWitt metric]
    The \emph{$\Go$-reduced DeWitt metric} is the bilinear form on $T(\Acal/\Go)$ defined by the identity\footnote{Recall, the Coulomb horizontal lift $\check{\Upsilon}_{\circ,\mathrm{Cou}}$ is defined by the prescription that the lifted vector be orthogonal to the fibres of $\Acal\to\Acal/\Go$ with respect to the DeWitt metric.}
    \begin{equation}\label{eq:reduceddewitt}
    \langle\langle\mathbb{X},\mathbb{Y}\rangle\rangle_\circ\vert_{[A]_\circ} \doteq \langle\langle\check{\Upsilon}_{\circ,\mathrm{Cou}}(A,\mathbb{X}) ,\check{\Upsilon}_{\circ,\mathrm{Cou}}(A,\mathbb{Y})\rangle\rangle
    \quad
    \forall \mathbb{X},\mathbb{Y}\in T_{[A]_\circ}(\Acal/\Go),\ A\in [A]_\circ.\qedhere
    \end{equation}
\end{definition}

The above definition is well-posed: simply note  that the gauge invariance of the DeWitt metric implies that Definition \eqref{eq:reduceddewitt} is independent of the point $A$ in the fibre of $[A]_\circ$ at which the lift is performed. The smoothness of the reduced metric is implied by those of the DeWitt metric and of the Coulomb connection.
It is also straightforward to see that the reduced DeWitt metric is (weakly) nondegenerate.

Denote $([A]_\circ, \Pi_\circ)\in T(\Acal/\Go)$. 
Using the linearity of $T_{[A]_\circ}(\Acal/\Go)$, identify $T_{([A]_\circ,\Pi_\circ)}(T(\Acal/\Go)) \simeq T_{[A]_\circ}(\Acal/\Go)$.
Thus a vector $\mathbb{X} \in T_{[A_\circ,\Pi_\circ]}(T(\Acal/\Go))$ can be written as a pair of vectors $(\mathbb{X}^{(A)},\mathbb{X}^{(\Pi)})$ with $\mathbb{X}^{(\bullet)} \in T_{[A]_\circ}(\Acal/\Go)$.

\begin{definition}
    Define $\vartheta \in \Omega^1(\calPhi)$ and $\varpi\in\Omega^2(T(\Acal/\Go))$ by 
    \[
    \bi_{\mathbb{X}}\vartheta|_{([A]_\circ, \Pi_\circ)} \doteq \langle\langle\ \Pi_\circ, \mathbb{X}^{(A)}\ \rangle\rangle_\circ 
    \qquad \forall \mathbb{X} \in TT(\Acal/\Go),
    \]
    and the manifestly nondegenerate form
    \[
    \varpi \doteq \bd \vartheta \in \Omega^2(T(\Acal/\Go)).\qedhere
    \]
\end{definition}

\begin{proposition}[Local symplectic model for $\uCo$]\label{prop:YM-uCo}
Assume $\Go\circlearrowright \C$ is proper. Then $\uCo$ is smooth and $\iota_\C^*\omega$ defines a unique (weakly) symplectic 2-form $\uomegao\in\Omega^2(\uCo)$. Moreover, we have the local symplectomorphism
    \[
    (\uCo,\uomegao) \simeq_\loc (T(\Acal/\Go), \varpi).
    \]
\end{proposition}

\begin{remark}[Proper actions]
    The assumption of properness in Proposition \ref{prop:YM-uCo} can be relaxed when $G$ is compact, as proven in the absence of corners in \cite{NarasimhanRamadas79, Rudolph_2002} (with $\Sigma$ compact). In the presence of corners the properness of the action of $\G$ on $\C$ when $G$ is compact  can be proved in the same way, whereas the action of the normal subgroup $\Go \subset \G$ is also proper since $\Go$ is closed.\footnote{In \emph{finite} dimensions, normal subgroups of connected semisimple groups are closed \cite{Ragozin}.} (Inclusions of closed submanifolds are proper maps, and $\Go=\pi_{\pp}^{-1}(\mathrm{id})$ where the restriction to the boundary $\pi_\pp$ is a morphism of groups). 
\end{remark}

\begin{proof}
Our proof of the first part of the proposition mostly follows the arguments of \cite{DiezHuebschmann-YMred, DiezPhD}, adapting them to the case with boundary.
To show the smoothness of $\uCo$ we argue that (\textit{i}) the action of $\fGo$ over $\C$---which is proper by assumption---is free, and (\textit{ii}) the action of $\fGo$ on $\C$ admits slices; whereas to show that $\uomegao$ is (weakly) symplectic we use the previous lemma and argue that (\textit{iii}) the on-shell vertical distribution associated to $\fGo$, $V_\circ \doteq \rho(\C\times \fGo)\subset T_\C\X$, is symplectically closed (cf.\ Proposition \ref{prop:symplclosVo}).

Regarding (\emph{ii}), let us first recall that a slice through $A_0\in \Acal$ for a free action by $\Go$ is the image of a local section of $\Acal\to \Acal/\G$ through $A_0$ \cite{palais1957, palais1961, palaisterng1988critical, ArmsMarsdenMoncrief1981}.
In the absence of corners, $\fG=\fGo$ and a (candidate) slice is introduced constructively by showing that the orthogonal decomposition $T_{A_0}\Acal = \rho_{A_0}(\fGo) \oplus \ker(\UpsilonoCou)$ can be affinely extended to open neighborhoods of any $A_0\in\Acal$. One then goes on to show that the candidate slice is well-defined by means of the Nash--Moser theorem \cite{AbbatiCirelliMania,DiezMaster} which holds on tame Fr\'echet spaces (i.e.\ for $\Sigma$ compact).
This is related to the invertibility of the Faadeev--Popov operator  \cite{Gribov,NarasimhanRamadas79,BabelonViallet}. The construction of the candidate slice adapted to the presence of corners is described in Proposition \ref{prop:FaddeeevPopov} and Remark \ref{rmk:HodgeInYM}. 

To prove (\emph{iii}) one first observes that since $\omega(\cdot,\cdot) = \langle\langle\cdot,j\cdot\rangle\rangle$ and $\langle\langle\cdot,\cdot\rangle\rangle$ is an $L^2$ metric, symplectic closure of any $W\subset T\X$ is equivalent to its $L^2$ closure. 
One then observes that $\fGo$ is itself closed, and that so are both $\rho_A(\fGo)$ and its $L^2$ orthogonal complement in $T_A\Acal$.
This follows from the fact that $\rho\vert_A = d_A$ is weakly open,\footnote{We thank T.\ Diez for clarifying this point to us.} and Proposition \ref{prop:HodgeEqui} on the equivariant Hodge--de Rham decomposition of $\Omega^1(\Sigma, \fg)$ (see also Remarks \ref{rmk:openmaps} and \ref{rmk:HodgeInYM}).
With this at hand, one can step-by-step adapt to the case with corners the proof of the $L^2$ closure of $V_\circ$ provided in \cite[Lemma 5.6.9]{DiezPhD} for the corner-less case.

Finally, turning to the second part of the proposition, we now show that the local diffeomorphism we constructed in Proposition \ref{prop:YMlocaliso} is a map of symplectic manifolds.
We introduce the following notation: the projection 
\[
\wt\pi_\circ : \ker(\UpsilonoCou) \to T(\Acal/\Go)
\]
and the 1-form $\theta \in\Omega^1(\X)$,
\[
\theta= \int_\Sigma \tr(E \bd A), \qquad \omega = \bd\theta.
\]
We also recall that in light of Proposition \ref{prop:YMlocaliso}, the DeWitt diffeomorphism $\varphi_\mathrm{dW} : T\Acal\to T^\vee\Acal$ restricts to:
\[
\varphi_\mathrm{dW}\vert_{\ker(\UpsilonoCou)} :
\ker(\UpsilonoCou) \to \C.
\]

In the following we will leave the following diffeomorphisms implicit
\[
T\Acal \stackrel{\varphi_\mathrm{dW}}{\simeq} T^\vee\Acal = \X \simeq \Acal \times \mathcal{E}
\]
which restrict to 
\[
\ker(\UpsilonoCou)\stackrel{\varphi_\mathrm{dW}}{\simeq} \C \simeq \{ (A,E) \ | \ d_A E = 0\}, 
\quad(A,\wt{\Pi}_\circ) \mapsto (A,E).
\]

With this notation, and denoting $([A]_\circ,\Pi_\circ)=\wt{\pi}_\circ(A,\wt{\Pi}_\circ)$, we compute for an arbitrary $\mathbb{X}\in T\ker(\UpsilonoCou)\simeq T\C$:
\begin{align*}
\bi_{\mathbb{X}}(\wt\pi_\circ^*\vartheta)\vert_{(A,\wt{\Pi}_\circ)}
&= \langle\langle \Pi_\circ, \bd\wt{\pi}_\circ(\mathbb{X}^{(A)})\rangle\rangle_\circ
= \langle\langle \wt{\Pi}_\circ, \underline{\UpsilonoCou}(\mathbb{X}^{(A)}
) \rangle\rangle \\
&= \iota^*_\C\int_\Sigma \tr( E \ \underline{\UpsilonoCou}(\mathbb{X}^{(A)}
)  )
= \iota^*_\C\bi_{\mathbb{X}} \int_\Sigma \tr( E \ \bd_{\circ,\mathrm{Cou}} A)
\\
&= \bi_{\mathbb{X}} \iota^*_\C\int_\Sigma \tr( E \ \bd A) = \bi_{\mathbb{X}} \iota^*_\C\theta\vert_{(A,E)}
&\end{align*}
where in the second step we used that $\check{\Upsilon}_{\circ,\mathrm{Cou}}(\Pi_\circ) = \wt\Pi_\circ$ and $\check{\Upsilon}_{\circ,\mathrm{Cou}}(\bd \wt{\pi}_\circ(\mathbb{X}^{(A)})) = \underline{\UpsilonoCou}(\mathbb{X}^{(A)})$, and in the last step we used that (\textit{i}) $(A,E)$ satisfy the Gauss constraint $\mathsf{G} = d_AE = 0$, and (\textit{ii}) $\UpsilonoCou$ satisfies the boundary condition $\UpsilonoCou|_{\pp\Sigma}=0$, in order to simplify the following general identity (signs depend on the dimensionality of $\Sigma$):
\[
\int_\Sigma \tr( E \bd_\Upsilon A) 
= \int_\Sigma \tr( E (\bd A + d_A \Upsilon)) 
= \int_\Sigma \tr(E \bd A) \pm \int_\Sigma \tr( \mathsf{G} \Upsilon ) \pm \int_{\pp\Sigma}\tr(E \Upsilon).
\]
Summarizing the computation above, we have found that: 
\[
\wt\pi_\circ^*\vartheta = \iota_\C^*\theta.
\]
And thus, differntiating and using the definition of $\uomegao$, we obtain:
\[
\wt\pi_\circ^*\varpi = \iota_\C^*\omega = \pi_\circ^*\uomegao.
\]
One concludes using the surjectivity of the maps $\pi_\circ$ and $\wt\pi_\circ$, as well as the diffeomorphic nature of their images, $T(\Acal/\Go) \simeq \uCo$.
\end{proof}

\begin{remark}
    This proves the local models for $\uCo$ claimed in Theorem \ref{mainthm:YMon}.
\end{remark}

We conclude by stating a result that uses the Ehresmann connections introduced here, but will be useful only later on. Note that $\omega^H_{\mathrm{Cou}}$ in Equation \eqref{eq:omegaHSdW} can be thought of as another expression for $\uomegao$ on $\uCo$ and, equivalently, $\varpi$ on $T(\Acal/\Go)$.

\begin{proposition}\label{prop:omegadecomp}
Given any principal Ehresmann connection $\Upsilon_\mathfrak{H}\in\Omega^1(\X,\mathfrak{H})$, the on-shell pre-symplectic structure $\iota_\C^*\omega$ can be written as the sum of two $\bd$-exact contributions:\footnote{The last term on the first line cancels the first term onthe second line. These terms are introduced to make both forms $\bd$-exact.}
\begin{equation}
    \label{eq:omegaHpp}
    \iota_\C^*\omega = \omega^H + \omega^\pp, 
    \quad
        \begin{dcases}
        \omega^H \doteq \iota_\C^* \int_\Sigma \tr(\bd_{\Upsilon_\mathfrak{H}} E\ \bd_{\Upsilon_\mathfrak{H}} A) - \iota_\C^*\int_{\pp\Sigma} \tr(E\, \mathbb{F}_{\Upsilon_\mathfrak{H}})\\
        \omega^\pp\doteq \iota_\C^*\int_{\pp\Sigma} \tr\left(E\, \mathbb{F}_{\Upsilon_\mathfrak{H}} + \tfrac12E\, [\Upsilon_\mathfrak{H},\Upsilon_\mathfrak{H}] - (-1)^{\dim\Sigma} \bd_{\Upsilon_\mathfrak{H}} E \wedge\Upsilon_\mathfrak{H}\right)
        \end{dcases}
    \end{equation}
    Both contributions are invariant with respect to the action of $\mathcal{H}$, and $\omega^H$ is basic with respect to the action of $\mathcal{H}$. If $\Go \subset \mathcal{H}$, then both contributions are basic with respect to the action of $\Go$. In the case of $\mathfrak{H}=\fGo$ and $\UpsilonoCou$ the expression simplifies ($\omega^\pp|_{\Upsilono=\UpsilonoCou}=0$) to
    \begin{equation}
    \label{eq:omegaHSdW}
    \omega^H_\mathrm{Cou} \doteq 
    \omega^H\vert_{\Upsilono=\UpsilonoCou} = \int_\Sigma \tr(\bd_{\circ,\mathrm{Cou}} E_{\text{rad},A}\wedge \bd_{\circ,\mathrm{Cou}} A).
    \end{equation}
    where $E_{\rad,A}$ is defined in Equation \eqref{eq:YMcotangentsplit}. This expression is manifestly basic with respect to $\Go$.
\end{proposition}

\begin{proof}
The formulas in this proposition generalize those in \cite{RielloGomes, RielloSciPost} to an arbitrary subgroup $\mathcal{H}\subset \G$.
Note that $\omega = \bd \theta$ and $\theta = \theta^H + \bd \theta^V$, where $\theta^H \doteq \int\tr(E \bd_{\UpsilonH}A)$ and $\theta^V \doteq  - \int_\Sigma \tr(E d_A\UpsilonH)$. 
We then define, and through a lengthy but straightforward computation we obtain:
\[
\bd \theta^H 
= \int_\sigma \tr(\bd_{\UpsilonH}E \ \bd_{\UpsilonH} A ) +  \tr(\mathsf{G}\ \mathbb{F}_{\UpsilonH}) - d\tr(E \ \mathbb{F}_{\UpsilonH}),
\]
and, since $\theta^V = (-1)^{\dim \Sigma}\int_\Sigma \tr(\mathsf{G} \ \UpsilonH) - d\tr(E \ \UpsilonH)$:
\[
\bd \theta^V = (-1)^{\dim \Sigma}\int_\Sigma \tr( (\bd \mathsf{G}) \wedge \UpsilonH) +d \tr((\bd E) \ \UpsilonH + E \ \bd \UpsilonH).
\]
Rearranging terms and pulling back to the constraint surface we find the sought result with
\[
\omega^H = \iota_\C^*\bd\theta^H \qquad \text{and}\qquad \omega^\pp = \iota_\C^*\bd \theta^V.
\]
Note that $\omega^{H}$ is manifestly basic with respect to $\C\to\C/\mathcal{H}$. It is also manifest that $\omega^\pp$ is $\mathcal{H}$-invariant, i.e. $\L_{\rho(\xi)} \omega^\pp = 0$ forall $\xi\in\mathfrak{H}$. To show that, for $\Go\subset \mathfrak{H}$, $\omega^\pp$ is $\fGo$-horizontal and thus $\Go$-basic, we use the parametrization $\C \simeq \X_\rad \times \mathcal{E}_\pp \ni (A,E_\rad, E_\pp)$ of Equation \eqref{eq:CtoP-YM} to compute
\[
\bi_{\rho(\xio)} \omega^\pp = \bi_{\rho(\xio)} \iota_\C^*\omega = \bd \int_{\pp\Sigma} \tr(E_\pp \xio) = 0
\]
which vanishes since $\xio\in\fGo = \{ \xi \in\fG \, | \, \xi\vert_{\pp\Sigma}=0\}$. For the same reason, the expression \eqref{eq:omegaHpp} simplifies to \eqref{eq:omegaHSdW}.
\end{proof}

\subsection{Second stage symplectic reduction: flux superselections}\label{sec:cornerred}

In this section we finally discuss the structure of the space of constrained configurations modulo \emph{all} gauge transformation, i.e.\ the structure of:
\begin{definition}[Fully-reduced phase space]\label{def:uX}
The \emph{fully-reduced phase space} is 
\[
\uuC \doteq \C/\G.\qedhere
\] 
\end{definition}

In Section \ref{sec:bulkreduction} we have analyzed and performed the symplectic reduction of $(\X,\omega)$ with respect to the action of the constraint gauge (sub)group $\Go\subset\G$, the \emph{maximal} normal subgroup for which $\C$ is the vanishing locus of the associated momentum map. 
This produced the symplectic constraint-reduced phase space $(\uCo,\uomegao)$. Now, the constraint-reduced phase space $\uCo$ carries the residual action of the group of flux gauge transformations  $\Gred = \G/\Go$. 

From a general standpoint, $\uuC$ cannot be expected to be symplectic, but rather a \emph{collection} of symplectic spaces---we call them superselections sectors. In regular enough cases, these are organized as the symplectic leaves of a (partial) Poisson manifold (Definition \ref{def:partialpoisson}).

The goal of this section is to  construct $\uuC$ in terms of the symplectic reduction of the constraint-reduced phase space $(\uCo, \uomegao)$ by residual flux gauge group $\Gred$. This is the second-stage in the reduction of $(\X,\omega)$ by $\G=\G/\Go$. 

The first step consists in showing that  $\Gred $ has a \emph{Hamiltonian} action on $(\uCo,\uomegao)$ generated by a momentum map $\uh:\uCo\to\fGred^*$ readily derived from $\iota_\C^*\h_\smbullet$, a fact directly related to the identity $\fGo = \ker(\iota_\C^*\h^\text{co}_\smbullet)$ (Theorems \ref{thm:frakN} and \ref{thm:fGo}). 

It is important to notice that---in contrast to the constraint reduction which was performed at the \emph{zero}-level set of $\bHo$---there isn't a priori any one preferred value of the flux $f\in\F \simeq \Im(\uh)$ at which to perform the corner reduction.  As a result, we must conclude that the reduction of $\uCo$ by $\Gred$ provides a symplectic space for \emph{each} value of the flux $f\in\F$. According to the MWM symplectic reduction, one can construct such symplectic spaces as
\[
\uuS_f \doteq \uS_f / \Gred_f ,\qquad \uS_f \doteq \uh^{-1}(f) \subset \uCo
\]
where $\Gred_f$ is the stabiliser of $f\in \Im(\uh)$ with respect to the appropriate action, dictated by the equivariance properties of $\h_\smbullet$. This reduction is sometimes called ``point reduction'' with reference to the fact that a specific point in $\Im(\uh)$ is singled out \cite{RatiuOrtega03}. 

However, it is more convenient to adopt a manifestly equivariant construction and thus consider the symplectomorphic spaces $\uuS_{[f]} \simeq \uuS_f$ defined in terms of the coadjoint, or affine, orbit $\mathcal{O}_f$ of $f\in \Im(\uh)$ (Definition \ref{def:affinecoad}), as
\[
\uuS_{[f]} \doteq \uS_{[f]} / \Gred, \qquad \uS_{[f]} \doteq \uh^{-1}(\mathcal{O}_f)\subset \uCo.
\]
The ensuing reduction procedure is sometimes called ``orbit reduction''. In view of the equivariance properties of $\uh$, the $\G$-action on $\X$ induces a foliation of $\Im(\uh)$ by the orbits of $\Gred$, and therefore the fully reduced space ends up being the disjoint union of the symplectic spaces $\uuS_{[f]}$:
\[
\uuC \simeq \bigsqcup_{\mathcal{O}_f} \uuS_{[f]}.
\]
In the rest of this section, we will make these considerations precise.

\medskip

As anticipated, we begin by studying the Hamiltonian properties of $(\uCo,\uomegao)$.

First recall that the dual $\dual{U}$ of a quotient vector space $U=V/W$ is (canonically) isomorphic to the annihilator of $W$: $\dual{U}\simeq \Ann(W)\subset \dual{V}$.

\begin{definition}\label{def:Gpdual}
The \emph{local dual of the corner algebra} $\fGred$ is 
\[
\fGred^*_\loc\doteq\Ann_{\loc}(\fGo) \subset \fG^*_\loc,
\] 
the subspace of local dual functionals $\dual{\fG}_{\loc}$ that annihilate $\fGo$ (Definition \ref{def:dual*}).
\end{definition}

\begin{proposition}[Reduced flux map]\label{prop:bulkredflux}
Let the adjusted flux map $\h_\smbullet$ be as in Definition \ref{def:adjpreflux}. Then there exists a unique $\uf :\uCo \to \fGred^*_\loc$ such that, for all $\xi\in\fG$,
\[
\pi_\circ^*\langle\uh, \underline{\xi} \rangle = \iota_{\C}^*\langle \h_\smbullet,\xi\rangle,
\]
where $\underline{\xi}=\xi +\fGo$ denotes the equivalence class of an element $\xi\in\fG$ in $\fGred=\fG/\fGo$. We call $\uh$ the \emph{reduced flux map}.
\end{proposition}
\begin{proof}
First, notice that by Theorem \ref{thm:fGo}, $\fGo  = \fN \doteq \ker(\iota_\C^*\h_\smbullet^\text{co})$,  one has that 
$\iota_\C^*\h_\smbullet^\text{co}(\xi) = \iota_\C^*\h_\smbullet^\text{co}(\xi + \eta_\circ)$ for all $\eta_\circ\in\fGo$.
Therefore, $\iota_\C^*\h_\smbullet$ can be understood as being valued in $\fGred^*_\loc\subset \fG^*_\loc$. 

We are now left to show that $\iota_\C^*\h_\smbullet:\C\to\fGred^*_\loc\subset \dual{\fG}_\loc$ descends unambiguously to a functional on $\uCo$, i.e.\ that it is basic w.r.t.\ to $\C \to \uCo$. In formulas, $\L_{\rho(\eta_\circ)}\iota_\C^*\h_\smbullet = 0$ for all $\eta_\circ\in\fGo$. In Proposition \ref{prop:adjequi} and Proposition \ref{prop:adjequi}, we found that for all $\eta,\xi\in\fG$, 
\[
\langle \L_{\rho(\eta)}\iota_\C^*\h_\smbullet,\xi\rangle 
=  \langle \iota_\C^*\h_\smbullet, [\eta,\xi]\rangle + \k_\smbullet(\eta,\xi) 
= - \L_{\rho(\xi)} \langle\iota_\C^*\h_\smbullet,\eta\rangle,
\]
which vanishes whenever $\eta=\eta_\circ\in\fGo\subset\fG$ thanks to Corollary \ref{cor:tangent} ($\rho(\C\times\fG) \subset T\C$) and Theorem \ref{thm:fGo} ($\fGo = \fN = \Ann \Im( \iota_\C^*\h_\smbullet)$). This concludes the proof.
\end{proof}

\begin{proposition}[Hamiltonian action of $\fGred$]\label{prop:hamactionGp}
Let the constraint-reduced phase space $(\uCo,\uomegao)$ be as in Definitions \ref{def:bulkredsp} and \ref{def:uomegao}, the flux gauge group $\Gred $ be as in Assumption \ref{ass:groups}, and the reduced flux map $\uh$ be as in Proposition \ref{prop:bulkredflux}. Then, locally Hamiltonian action of $\G$ on $(\X,\bom)$ descends to a Hamiltonian action of  $\Gred$ on $(\uCo,\uomegao)$, with the reduced flux map $\uh$ as momentum map:
\[
\bi_{\rho(\underline{\xi})}\uomegao = \bd \langle \uh,\underline{\xi}\rangle,
\]
where we denoted $\underline{\xi} = \xi +\fGo \in \fGred$.
Moreover, the equivariance of the reduced flux map under the action of $\fGred$ is controlled by the corner CE cocycle:
\[
\L_{\rho(\underline{\xi})}\uh + \ad^*(\underline{\xi})\cdot\uh = \k_\smbullet(\underline{\xi}).    
\]
Finally, if the on-shell isotropy locus is a trivial vector bundle, i.e.\ $\mathsf{I}_\rho = \C \times \mathfrak{I}$, the action $\Gred \circlearrowright \uCo$ is free.
\end{proposition}
\begin{proof}
Recall the local Hamiltonian flow Equation \eqref{eq:weakHam} and the split $\bH = \bHo + d\bh$ of Proposition \ref{prop:dualvaluedformdecomposition}: $\bi_{\rho(\xi)}\bom = \bd \langle \bHo + d\bh_\smbullet,\xi\rangle$. Thanks to Corollary \ref{cor:tangent} ($\C \triangleleft \G \subset \C$), one can pull back this equation to $\C = \bHo^{-1}(0)$. Integrating,  one obtains
\[
\bi_{\rho(\xi + \eta_\circ)}\iota_\C^* \omega= \iota_\C^*\bi_{\rho(\xi)\vert_\C}\omega = \bd \langle \iota_\C^*\h_\smbullet,\xi\rangle = \bd \langle \iota_\C^*\h_\smbullet,\xi + \eta'_\circ\rangle \qquad \forall \eta_\circ,\eta'_\circ\in\fGo,
\]
where the first and last identities rely on Lemma \ref{lemma:coisoC} ($\rho(\C\times \fGo) \equiv V_\circ\subset T\C^\omega$) and Theorem \ref{thm:fGo} ($\fGo = \fN = \Ann\Im(\iota_\C^* h_\smbullet)$).
Now, using again Corollary \ref{cor:tangent} and Assumption \ref{ass:groups}, it is clear that the action of $\G$ on $\C$ descends to an action of $\Gred$ on $\uCo$---which we denote by the same symbol. Therefore, one can use Proposition \ref{prop:bulkredflux} to pass to the quotient in the previous equation both along $\C\to\uCo$ and $\fG\to\fGred$, thus proving the first statement of the proposition.

The equivariance properties of $\uh$ immediately follow from Propositions \ref{prop:bulkredflux} and  \ref{prop:adjequi}. Since $\k_\smbullet$ does not depend on $\phi$ and $\k_\smbullet(\fGo,\cdot)= \k_\smbullet(\cdot,\fGo)=0$ (ibidem), there is no need to adopt a new symbol for it in passing to the quotient.

Finally, the freeness of the action in the case of a trivial on-shell isotropy bundle is a direct consequence of Gauss law (Proposition \ref{prop:Gausslaw}) which can be rephrased as saying that $\mathfrak{\fGo} \supset \ker(\rho\vert_\phi)$.
\end{proof}

\begin{remark}
This Proposition proves Theorem \ref{mainthm:redsummary}, item \ref{thmitem:main-resHamaction}.
\end{remark}

Now that we have proved the Hamiltonian structure of $(\uCo,\uomegao, \Gred)$, we need to set the stage for the study of its symplectic reduction. In particular, we need to study the image of the momentum map and its foliation by orbits associated to the action of $\Gred$.

\begin{lemma}\label{lemma:Imh=flux}
The image of the reduced flux map is canonically isomorphic to the space of (on-shell) fluxes (Definition \ref{def:fluxspaces})
\[
\Im(\uh) \simeq \F \doteq \Im(\iota_\C^*\h_\smbullet),
\]
as a subspace of $\fGred^*_\loc \subset \dual{\fG}_\mathrm{loc}$ (Definition \ref{def:Gpdual}).

\end{lemma}
\begin{proof}
Thanks to Theorem \ref{thm:fGo} ($\fGo = \fN$), the space of on-shell fluxes $\F \doteq \Im(\iota_\C^*\h_\smbullet)$ (Definition \ref{def:fluxspaces}) is a subset of $\fGred^*_\loc \doteq \Ann_\mathrm{loc}(\fGo) \subset \dual{\fG}_\mathrm{loc}$. Moreover, since $\iota_\C^*\h_\smbullet$ is basic w.r.t.\ $\pi_\circ \colon \C \to \uCo$ (Proposition \ref{prop:bulkredflux}), one has that the space of on-shell fluxes $\F$ also coincides with $\Im(\uh)$.
\end{proof}

\begin{remark}
In Remark \ref{rmk:equipushfwd} we expressed the action $\ad^*_K$ of $\fG$ on $\dual{\fG}_\mathrm{loc}$ as the pushforward of the action $\rho$ of $\fG$ on $\X$ along $\h_\smbullet:\X \to \dual{\fG}_\mathrm{loc}$---that is: $(\h_\smbullet)_*\rho(\xi) = \ad^*_K(\xi)$. Also, from Proposition \ref{prop:bulkredflux}, we deduce that $f\in\F$ is invariant under the action of $\Go$, i.e.\ $\ad^*_K(\xio) \cdot f =0$ for all $\xio\in\fGo$. 
Therefore, if $\mathcal{O}_f$ is the affine/coadjoint orbit of $f\in\dual{\fG}_{\loc}$ (Definition \ref{def:affinecoad}), one has that
\[
\mathcal{O}_f \doteq f \triangleleft\G \simeq \G_f\backslash \G \simeq \Gred_f\backslash\Gred,
\]
where $\G_f$ (resp.\ $\Gred_f$) are the stabilisers of $f \in \F \subset \fGred^*_\loc \subset \dual{\fG}_\mathrm{loc}$ in $\G$ (resp.\ $\Gred$), and the backslash denotes a \emph{left} quotient. This observation suggests that $f$, its orbit $\mathcal{O}_f$, as well as the KKS symplectic structure $\Omega_{[f]}$ can all be understood either as associated to $\G$, through its action on the local dual $\dual{\fG}_\mathrm{loc}$, or to $\Gred$, through its action on the local dual $\fGred^*_\loc$. This allows us to understand the image of the reduced flux map and the associated coadjoint/affine orbits---which are a priori nonlocal objects associated to the quotient $\fGred=\fG/\fGo$---in a local framework over $\Sigma$.\footnote{In Section \ref{sec:cornerdata} we will show that in many cases of interest, \emph{both} $\fGred$ and $\fGred^*_\loc$ can be naturally identified with a local Lie algebra and its dual \emph{over the corner} $\partial\Sigma$.}
\end{remark}

\begin{definition}[Flux superselection sector]\label{def:fluxSSS}
Let $\mathcal{O}_f\subset\F$ be the orbit of a flux, and let $\uS_{[f],q}$ be the $q$-th connected components of the sets $\uS_{[f]}\doteq\uh^{-1}(\mathcal{O}_f)$, with $q\in I$ some index.
The \emph{$q$-th superselection sector of $f$} is the quotient
\[
\uuS_{[f],q}\doteq \uS_{[f],q}/ \Gred.\qedhere
\] 
\end{definition}

\begin{lemma}
The fully-reduced phase space $\uuC\doteq \C/\G$ is isomorphic to the disjoint union of the superselection sectors:
\[
    \uuC \simeq \bigsqcup_{\mathcal{O}_f,q}\uuS_{[f],q},
\]
and we denote $\underline{\iota}{}_{[f],q}: \uuS_{[f],q} \hookrightarrow \uuC$.
\end{lemma}
\begin{proof}
Since $\uuC\doteq \C/\G \simeq \uCo/\Gred$, the lemma follows from the following observations: (i) $\uCo$ decomposes in orbits under the action of $\Gred$, (ii) $\F$ is foliated by the orbits $\mathcal{O}_f$, and (iii) the map $\uh$ is equivariant under the action of $\fGred$ so that the sets $\uS_{[f]} \doteq \uh^{-1}(\mathcal{O}_f)$, which are unions of $\Gred$-orbits in $\uCo$, are $\Gred$-invariant. Then, the set of orbits $\{\mathcal{O}_f\}_{f\in\F}$ parametrises the orbit spaces $\uS_{[f]}/\Gred$, (possibly) up to multiple connected components labelled by $q\in I$.
\end{proof}

\begin{remark}[Other superselection charges]\label{rmk:topochargesSSS}
In principle, the flux superselection sectors might have multiple connected components labelled by an index $q$. This happens if the same flux $f$ is the image by $\h_\smbullet$ of two (on-shell) configurations that cannot be smoothly deformed into one another without $f$ leaving its orbit.\footnote{Recall that we assume all the relevant groups to be simply connected.} The number $q$ can be understood as an additional ``topological superselection charge''.
\end{remark}

We will now prove that flux superselection sectors are the result of the second-stage reduction, whose smoothness (as elsewhere in this article) we simply assume:

\begin{assumption}\label{ass:smoothsuperselections}
The connected components $\uS_{[f],q}$ of $\uS_{[f]}=\uh^{-1}(\mathcal{O}_f)$ are smooth submanifolds of $\uCo$. Moreover, the quotients $\uuS_{[f],q}\doteq \uS_{[f],q}/\Gred$ are smooth manifolds, with associated surjective submersions  $\underline{\pi}_{[f],q}\colon \uS_{[f],q} \to \uuS_{[f],q}$.
\end{assumption}

\begin{remark}
Observe that if the action of $\G$ on $\X$ is proper, the induced action of $\Gred$ has the same property whenever the map $\C\to \uCo$ is a proper map. 
\end{remark}

\begin{proposition}[Second-stage reduction: flux superselection]\label{prop:secondstage}
Let $\mathcal{O}_f\subset\F$ be the orbit of a flux, and let $\uuS_{[f],q}$ be the $q$-th corresponding superselection sector. Then, under Assumption \ref{ass:smoothsuperselections} there exists, unique, a closed 2-form $\uuomegao_{[f],q}$ on $\uuS_{[f],q}$ such that
\[
\underline{\pi}{}_{[f],q}^*\uuomegao_{[f],q} = \underline{\iota}{}_{[f],q}^*  \uomegao - (\uh\circ\underline{\iota}{}_{[f],q})^* \Omega_{[f]}.
\]
\end{proposition}
\begin{proof}
This follows from the fact that the 2-form on the right-hand side is basic w.r.t.\ the action of $\Gred$, which is in turn a consequence of Proposition \ref{prop:hamactionGp}, Lemma \ref{lemma:Imh=flux}, and Remark \ref{rmk:equipushfwd} ($(\h_\smbullet)_*\rho(\xi) = \ad^*_K(\xi)$). In this regard, we recall Remark \ref{rmk:KKS-Hamiltonian}, stating:
\[
\bi_{\ad^*_K(\xi)}\Omega_{[f]}\vert_{f'} = \bd  \langle i_{[f]}(f'), \xi\rangle,
\]
where $i_{[f]}$ is the embedding $i_{[f]} :\mathcal{O}_f \to \dual{\fG}_\mathrm{loc}$. Closure is ensured by the existence of a surjective submersion, as per Assumption \ref{ass:smoothsuperselections}. This way of constructing $\uuomegao_{[f]}$ is sometimes called the ``shifting trick'' \cite{Arms1996shifttrick}.
\end{proof}

Similarly to what discussed in Section \ref{sec:bulkreduction}, we need to check conditions akin to \ref{C1} and \ref{C2}, and Assumption \ref{ass:symplecticclosure}:

\begin{proposition}[Section 4.1 of \cite{DiezPhD}]\label{prop:uuSf-nondeg}
Denote $\underline{V}_{[f],q}  = \rho(\uS_{[f],q}\times \fGred) \subset T\uCo$. The map $\uuomegao_{[f],q}^\flat : T\uuS_{[f],q} \to T^*\uuS_{[f],q}$ is injective, and therefore symplectic, iff $\underline{V}_{[f],q} = (T \uS_{[f],q})^{\uomegao} $ or, equivalently, iff $\underline V_{[f],q}$ is symplectically closed i.e.\ $\underline{V}_{[f],q}^{\uomegao\uomegao} = \underline{V}_{[f],q}$.
\end{proposition}

Hence, to ensure that the superselections are sympelctic, we assume:
\begin{assumption}\label{ass:symplecticclosure-f}
$\underline V_{[f],q}$ is symplectically closed, i.e.\ $\underline V_{[f],q}^{\omega\omega} = \underline V_{[f],q}$.
\end{assumption}

To conclude our discussion of symplectic reduction in the presence of corners, we observe that under our regularity assumptions the superselections foliate a Poisson space:

\begin{proposition}[Poisson structure of $\uuC$]\label{prop:uXPoisson}
When smooth, the fully-reduced phase space $\uuC = \bigsqcup_{[f],q} \uuS_{[f],q}$ is a (partial) Poisson manifold, of which the superselection sectors $(\uuS_{[f],q},\uuomegao_{[f],q})$ are the symplectic leaves.
\end{proposition}

\begin{remark}
This could also be obtained as an infinite-dimensional adaptation of the construction of Theorem \ref{thm:extendedreduction} (cf.\ Remark \ref{rmk:extendedreduction}), where $\X$ is replaced by $\uCo$ and $H$ by $\uh$, provided that the issues arising from infinite-dimensionality are appropriately handled.
\end{remark}

\begin{remark}
Propositions \ref{prop:secondstage}, \ref{prop:uuSf-nondeg}, and \ref{prop:uXPoisson} prove Theorem \ref{mainthm:redsummary}, items \ref{thmitem:main-cornerred} and \ref{thmitem:main-Poisson}.
In particular, Assumptions \ref{ass:groups}--\ref{ass:symplecticclosure-f} guarantee that the flux superslection sectors $\uuS_{[f]}$ are well-defined, (weakly) symplectic, manifolds.

More precisely, Theorem \ref{mainthm:redsummary} is a reformulation of these results upon their lifting from $\uCo$ to $\C$, as per the diagram below which summarises the two-stage reduction with wavy lines, and emphasises the unified construction presented in Theorem \ref{mainthm:redsummary}, above, with dashed lines. 
(To avoid clutter, in the diagram we omit the label $q$.)
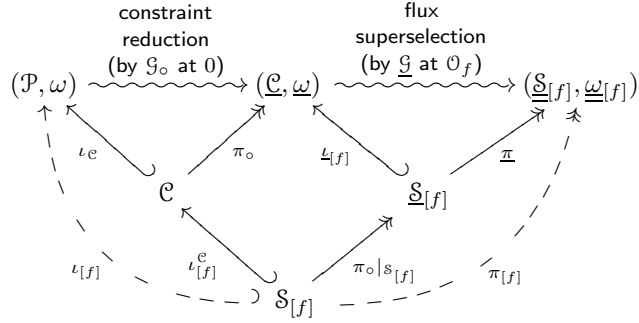
\[
\xymatrix@C=.75cm{
(\X,\omega)
	\ar@{~>}[rr]^-{\tbox{2.2cm}{constraint reduction \\(by $\Go$ at $0$)}}
&&(\uCo,\uomegao)
	\ar@{~>}[rr]^-{\tbox{2.2cm}{flux superselection (by $\Gred$ at $\mathcal{O}_f$)}}
&&(\uuS_{[f]},\uuomegao_{[f]})\\
&{
    \;\C\;
    \ar@{_(->}[ul]^-{\iota_\C}
	\ar@{->>}[ur]_-{\pi_\circ}
	}
&&{
    \;\;\uS_{[f]}
    \ar@{_(->}[ul]^-{\underline{\iota}_{[f]}}
	\ar@{->>}[ur]_-{\underline{\pi}}
	}\\
&&{
    \;\;\S_{[f]}
    \ar@{_(->}[ul]^-{{\iota}^\C_{[f]}}
	\ar@{->>}[ur]_-{\pi_\circ\vert_{\S_{[f]}}}
	\ar@{_(-->}@/^2.7pc/[uull]^{\iota_{[f]}}
	\ar@{-->>}@/_2.7pc/[uurr]_{\pi_{[f]}}
	}\\
}\]
Notice that the lift of $\uS_{[f],q}$ to $\C$ can be directly defined as the $q$-th connected component of 
\[
\uS_{[f]} \doteq h^{-1}_\smbullet(\mathcal{O}_f)\cap \C.
\qedhere
\]
\end{remark}

\subsection{Yang--Mills theory: second stage reduction}\label{sec:runex-secondstage}

Let's start from the semisimple case. In the present discussion we will follow the construction of Theorem \ref{mainthm:redsummary}.

Let us start by recalling that in semisimple YM theory the space of fluxes is isomorphic to the space of electric fields pulled-back to $\pp\Sigma$, i.e.\ $\mathcal{E}_\pp\simeq \F$, where the isomorphism is given by the map $E_\pp \mapsto \int_{\pp\Sigma}\tr(E_\pp\cdot)$.

\begin{proposition}[Local model for second stage reduction]\label{prop:localsecondstagemodel}
    \[
    T(\Acal/\Go)\simeq_\loc  T(\Acal/\G) \times \Gred \times \mathcal{E}_\pp \qquad 
     \uuC=\uCo/\Gred \simeq_\loc T(\Acal/\G) \times \mathcal{E}_\pp. 
    \]
\end{proposition}
\begin{proof}
    From the local model for the first stage reduction (Proposition \ref{prop:YMlocaliso}), since $\uCo \simeq_\loc T(\Acal/\Go)$ and we can locally\footnote{Since $\Go$ is split, $\G/\Go$ is a principal bundle and thus admits a local section (cf.\ Remark \ref{rmk:split}).} write $\Acal/\Go \simeq_\loc \Acal/\G \times \G/\Go \simeq\Acal/\G \times \Gred$, identifying $\mathcal{E}_\pp\simeq \fGred$ (as a vector space) we have
    \[
    \uCo \simeq_\loc T(\Acal/\G \times \Gred) \simeq T(\Acal/\G) \times T\Gred \simeq T(\Acal/\G) \times \mathcal{E}_\pp\times \Gred 
    \]
    Consequently: $T\Gred/\Gred \simeq (\fGred\times \Gred)/\Gred \simeq (\mathcal{E}_\pp\times\Gred)/\Gred \simeq \mathcal{E}_\pp$. Indeed, from the above formula follows:
    \[
    \uuC=\uCo/\Gred \simeq_\loc  T(\Acal/\G) \times (\mathcal{E}_\pp\times \Gred )/\Gred \simeq T(\Acal/\G) \times \mathcal{E}_\pp.\qedhere
    \]
    \end{proof}

    \begin{remark}[Symplectic structure on the model]\label{rmk:sympl-model}
    To construct the symplectic structures on these model spaces, we specialize Proposition \ref{prop:omegadecomp} to $\mathfrak{H}=\fG$, e.g. by chosing $\UpsilonH=\UpsilonCou$ (Definition \ref{def:Coulombconnection}). Then, $\iota^*_\C\omega = \omega^H + \omega^\pp$ is basic with respect to $\C\to \uCo$. 
    In particular, the first term $\omega^H$ is a basic 2-form that descends to a symplectic 2-form on $T(\Acal/\G)$, the first factor in the above local model of $\uCo$:\footnote{Note that it is $E_{\rad,A}$ (Equations \eqref{eq:CtoP-YM} and \eqref{eq:horizontalYMcotangentsplit}) that appears on the right hand side, and not $E$: the remaining Coulombic contribution to $E$ simplifies with the boundary term in Equation \eqref{eq:omegaHpp} involving the curvature of $\Upsilon_\mathrm{Cou}$. Showing this requires a straightforward computation that uses Equations \eqref{eq:CtoP-YM} and \eqref{eq:horizontalYMcotangentsplit} as well as Lemma \ref{lemma:curv-Coul-conn}. Note that identifying $E_{\rad,A}$ with a vector at $\wt\Pi\in T_A\Acal$, Equation \eqref{eq:horizontalYMcotangentsplit} is equivalent to $\UpsilonCou|_A(\wt\Pi) = 0$.}
    \[
    \omega^H \doteq \int_\Sigma \tr(\bd_{\mathrm{Cou}} E_{\rad,A}\ \bd_{\mathrm{Cou}} A). 
    \]
    The second term, $\omega^\pp$, is also basic with respect to the action of $\Go$ however it is not manifestly the pullback of a symplectic structure on $T\Gred$. The issue is that the formulas of Proposition \ref{prop:omegadecomp} are global, whereas the model $T(\Acal/\G) \times T\Gred$ is valid only locally. 
    To see how the natural\footnote{Given the metric on $\G$ defined by $\tr(\cdot\ \cdot)$.} symplectic structure on $T\Gred$ emerges from Proposition \ref{prop:omegadecomp}, consider a local\footnote{We denote isomorphisms and equalities induced by the local section $\sigma : \uuU{}_\sigma \to \Acal$, for some $\uuU{}_\sigma\subset\Acal/\G$, through the symbol $\stackrel{\sigma}{\approx}$. In the main text, we abuse notation by replacing $\uuU{}_\sigma$ with $\Acal/\G$. We keep in mind that a priori all formulas are true only locally. For the existence of local sections (or, slices) see the discussion at the end of Section \ref{sec:runex-secondstage}. On the (non)existence of global sections, see e.g. \cite{Gribov,Singer1978,NarasimhanRamadas79}.)} section $\sigma: \Acal/\G \to \Acal$. This local section (or gauge fixing) provides us with a trivialization $\uuC \stackrel{\sigma}{\approx}T(\Acal/\G) \times \Gred \ni (\Pi, [A], k) $, and with an adapted (flat) principal Ehresmann connection $\Upsilon = \Upsilon_\sigma = k^{-1}\bd k$ on $\X\to\uuC$. 
    Using this trivialization, denoting $k|_{(A,E)}=k(A)$ simply by $k$ we introduce ``dressed'' fields $(A^k,E^k) \doteq (\Ad(k)\cdot A -  d k k^{-1}, \Ad^*(k^{-1})\cdot E)$. \cite{Francois2012,Gomes:2018shn,RielloGomesHopfmuller,Francois2021}. Then,\footnote{For a more detailed discussion of $\Upsilon$ and dressings, see \cite[Section 9]{RielloGomesHopfmuller} as well as \cite{Francois2021}. (\cite[Corollary 3.3]{RielloGomes} contains a similar formula to the one below, but with a few typos.)}
    \[
    \bd A^k = \Ad(k)\cdot \bd_{\Upsilon_\sigma } A 
    \quad\text{and}\quad
    \bd E^k = \Ad^*(k)\cdot \bd_{\Upsilon_\sigma } E,
    \]
    whence it is easy to see that Equation \eqref{eq:omegaHpp} for $\pi_\circ^*\uomegao = \iota_\C^*\omega = \omega^H + \omega^\pp$ simplifies to 
    \[
    \omega^H \stackrel{\sigma}{\approx} \iota_\C^*\int_\Sigma \tr(\bd E^k_\sigma \wedge \bd A^k ), 
    \]
    where $E_\sigma$ is the horizontal component of $E$ w.r.t.\ $\Upsilon_\sigma$ (as $E_\rad$ is w.r.t.\ $\Upsilon_\mathrm{Cou}$), and
    \[
    \omega^\pp  \stackrel{\sigma}{\approx} \int_{\pp\Sigma} \tr\big( \tfrac12 E_\pp^k[ \bd k k^{-1}, \bd k k^{-1}] - (-1)^{\dim\Sigma} \bd E_\pp^k \wedge \bd k k^{-1}\big).
    \]

    In fact, by looking at this formula we recognize that $(E_\sigma^k, A^k)$ and $\omega^H$ define (the $\Upsilon_\sigma$ horizontal lift of) a symplectic form on $T(\Acal/\G)$, whereas the variables $(E_\pp, \iota_{\pp\Sigma}^*k)$ together with $\omega^\pp$ give us the symplectic space
    \[
    \mathcal{E}_\pp\times \Gred \simeq T\Gred
    \]
    (note the close analogy with the KKS symplectic structure).
    The latter symplectic space is formally analogous to the ``edge mode'' phase space of \cite{DonnellyFreidel16}.\footnote{See \cite[Section 5.8]{RielloSciPost} for a more detailed comparison to edge modes.}
\end{remark}

\begin{remark}
    This proves the local model for $\uuC$ claimed in Theorem \ref{mainthm:YMon}, upon identifying $\mathcal{E}_\pp \simeq \F$, Equation \eqref{eq:YM-fluxdensity}.
\end{remark}

\medskip

We now turn to the model of the superselection sector $\uuS_{[f]}$. Fix a reference flux $f = \int_{\pp\Sigma}\tr(e_\pp\cdot)$ and its associated orbit $\mathcal{O}_f$, once and for all. We use the small cap notation to indicate that the corner electric field is fixed i.e.\ $\bd e_\pp =0$.

Observe that $\Go$ is a normal subgroup of $\G$ as well as of the stabiliser $\G_f$. Then,
\[
\mathcal{O}_f \simeq \G_f\backslash \G \simeq \underline{\G}_f\backslash \underline{\G},
\]
where $\underline{\G}_f \doteq \G_f / \Go$ and the backslash stands for a left quotient.
Consider now 
\[
\S_{[f]} \doteq 
\h_\smbullet^{-1}(\mathcal{O}_f) \cap \C
=\{ (A,E)\in\X \, |\, d_A E = 0,\, \text{ and }\, \exists u \in \underline{\G}\, \text{ s.t. } \, \iota_{\pp\Sigma}^*E = u^{-1} e_\pp u \}
\]
where for simplicity we used $\tr$ to identify $\fg^*$ and $\fg$, and implicitly identified $\underline{\G}$ with the group of ``corner gauge transformations'' $\Gamma(\pp\Sigma,\iota_{\pp\Sigma}^* \AD P)$ (see Equation \eqref{eq:uGYM}).\footnote{A more conceptual, and fully general, take on this bulk-boundary identifications is provided in Section \ref{sec:cornerdata} and in particular in the next running example, Section \ref{sec:runex-cornerdata}.}

Let us recall the construction of Section \ref{sec:runex-fluxannihilators}, where we proved that $\C$ is a fibration onto $\Acal\times \mathcal{E}_\pp$ with fibre $\mathcal{H}_A$ over $A\in\Acal$. We start from (\textit{i}) $\X = T^\vee \Acal = \mathcal{E} \times \Acal$, and observe that (\textit{ii}) the electric field decomposes into a radiative and a Coulombic part, i.e. $\mathcal{E}_A \doteq T_A^\vee \Acal = \mathcal{H}_A \oplus \mathcal{V}_A$ as per Equation \eqref{eq:YMcotangentsplit}. Owing to the Gauss constraint we have (\textit{iii}) $d_A E= 0$, and thus on $\C$ the Coulombic part of the electric field is fully determined by the boundary value $\mathsf{t}E \in \mathcal{E}_\pp$ as per Equation \eqref{eq:LaplaceEq} (cf.\ also Equation \eqref{eq:CtoP-YM}). Then, it is easy to see that $\S_{[f]}$ is similarly a fibration over $\Acal \times \mathcal{O}_f$ with fibre $\mathcal{H}_A$: i.e.\ points in $\S_{[f]}$ are uniquely identified by a triple $(A, [u], E_\text{rad})$, with $[u]\in\G_f\backslash \G\simeq \mathcal{O}_f$ and $E_\text{rad} \in \mathcal{H}_A$, so that 
\begin{equation}
    \label{eq:StoP-YM}
\iota_{[f]}:\S_{[f]}\hookrightarrow \C\subset\X, \quad \big(A, [u], E_\text{rad}\big) \mapsto \big(A, E_\text{rad} + \star d_A (\varphi(u^{-1} e_\pp u)) \big)
\end{equation}
where $\varphi(E_\pp)$ denotes the unique solution to the boundary value problem \eqref{eq:LaplaceEq}, and for simplicity we dropped the label $\bullet_A$ from $E_{\rad,A}$ (cf.\ Equation \eqref{eq:CtoP-YM}).
Note that by construction the flux $f'\in\mathcal{O}_f$ associated to the triple $(A,[u], E_\text{rad})$ is only a function of $[u]$, since $f'=\Ad^*(u)\cdot f$ or, more pedantically 
\[
f' = h_\smbullet \circ \iota_{[f]}(A,[u],E_\text{rad}) = \int_{\pp\Sigma}\tr\left((u^{-1} e_\pp u)\,\cdot\,\right)  \in \F \subset \fG^*_\loc.
\]

In this description of $\S_{[f]}$, the action of $\G$ on $\S_{[f]}\subset \X$ is given by\footnote{In writing $[ug]$ with $u\in\underline{\G}$ and $g\in \G$, we are abusing notation. What we mean is $[u\underline{g}]$ where $\underline{g}$ is the element of $\underline{\G} = \G/\Go$ associated to $g\in\G$.  Analogous abuses of notation will be committed below.}
\begin{equation}
    \label{eq:gaugetrC}
(A, \, [u],\, E_\text{rad}) \triangleleft g = (g^{-1} A g + g^{-1} d g , \, [ug], \, g^{-1} E_\text{rad} g)
\end{equation}
which is well defined since the space of the $[u]$'s, i.e.\ $\G_f\backslash \G\simeq \mathcal{O}_f$, is a \emph{left} quotient, whereas $\G$ acts from the right: i.e.\ if we set $E_\pp \equiv u^{-1} e_\pp u$, then $E_\pp \triangleleft g   =  (ug)^{-1} e_\pp (ug) =  (g_fug)^{-1} e_\pp (g_fug)$ for any $g_f \in \G_f$.

Then, let $\iota_{[f]}:\S_{[f]}\hookrightarrow\X$, and consider on $\S_{[f]}$ the 2-form 
\begin{equation}
    \label{eq:omegafYM}
\omega_{[f]} \doteq \iota_{[f]}^*\omega - (h_\smbullet\circ\iota_{[f]})^*\Omega_{[f]}.
\end{equation}
where the form $\iota_{[f]}^*\omega = (\iota_\C\circ\iota_{[f]}^\C)^*\omega$ can be read off of Equation \eqref{eq:omegaHpp} (specialised for $\mathfrak{H}=\fG$) for $\iota_\C^*\omega$ by setting $E_\pp = u^{-1}e_\pp u$ there, and off of Equation \eqref{eq:KKSYM} for $\Omega_{[f]}$ by noting that
\begin{equation}
    \label{eq:OmegaKKSYM}
(h_\smbullet\circ\iota_{[f]})^*\Omega_{[f]} = -\tfrac12 \int_{\pp\Sigma}\tr(e_\pp [\bd u u^{-1},\bd u u^{-1}]).
\end{equation}

Now, using that $\bd e_\pp =0$ and $\bi_{\rho(\xi)}u^{-1} \bd  u  = \xi$, one can check that not only $\L_{\rho(\xi)}\omega_{[f]}=0$, but also $\bi_{\rho(\xi)}\omega_{[f]}=0$ for \emph{all} $\xi\in\fG$.
Therefore, $\omega_{[f]}$ is $\G$-basic and can be projected down to the flux superselection sector $\uuS_{[f]}\doteq \S_{[f]}/\G$ where it defines the 2-form $\uuomegao_{[f]}$.

At the corner, one can use the flux degrees of freedom $u$, from $E_\pp = u^{-1}e_\pp u$, to introduce a \emph{dressed} versions of the Coulomb connection form $\UpsilonCou$ (for the group $\G$, Definition \ref{def:Coulombconnection}):\footnote{Technically, $\Upsilon^u_{\mathrm{Cou},\pp}$ is only defined up to a conjugation by an element that stabilized $e_\pp$. Also note that all we say about dressings can be extended to any principal connnection $\Upsilon$. \label{fnt:Upsilonu}}
\[
\Upsilon^u_{\mathrm{Cou},\pp} \doteq \Ad(u)\cdot\UpsilonCou\vert_{\pp\Sigma} - \bd u u^{-1}.
\]
Although the dressing of $\UpsilonCou$ looks like a gauge transformation, it isn't one. This is thoroughly discussed in\footnote{\cite{Francois2012} studies the dressing of connections over spacetime, rather than connections over field space. However, since $\UpsilonCou$ is equivariant by construction, one has $\UpsilonCou^u\vert_A = \UpsilonCou\vert_{A^u}$ where $A^u \doteq \Ad(u)\cdot A - d u u^{-1}$, even if $u$ is not constant over $\X$: indeed $\Upsilon_\mathrm{Cou}^u = \Delta_{A^u}^{-1} d_{A^u}^\star\bd A^u$. To prove this formula, one needs equivariance of $\UpsilonCou$ under field-\emph{dependent} gauge transformations $\sigma\in\Gamma(\X,\A)$.  For these, Equation \eqref{eq:varpi} implies  $\L_{\rho(\sigma)} \UpsilonCou = \bd \sigma + [\UpsilonCou,\sigma]$ (cf.\ \cite{RielloGomesHopfmuller}). } \cite{Francois2012} but as far as we are concerned it is enough to observe that it cannot be one because $\Upsilon^u_{\mathrm{Cou},\pp}$ fails to satisfy the basic properties of a connection form (Equation \eqref{eq:varpi}), and is instead both invariant and horizontal, i.e.\ basic: 
\[
\L_{\rho(\xi)}\Upsilon^u_{\mathrm{Cou},\pp} = 0 
\quad\text{and}\quad
\bi_{\rho(\xi)} \Upsilon^u_{\mathrm{Cou},\pp}=0.
\]

After some algebra, $\omega_{[f]}$ can be cast into the sum of two manifestly basic terms (this is obtained combining Equations (\ref{eq:omegaHpp}-\ref{eq:omegaHSdW}) and (\ref{eq:omegafYM}-\ref{eq:OmegaKKSYM})):
\begin{equation}
    \label{eq:YM-omega[f]}
\omega_{[f]} = \iota_{[f]}^*\int_\Sigma \tr\left(\bd_\mathrm{Cou} E_{\text{rad},A}\wedge \bd_\mathrm{Cou} A\right) + \bd \int_{\pp\Sigma} \tr(e_\pp \, \Upsilon^u_{\mathrm{Cou},\pp}) .
\end{equation}
Since all the terms in this formula are basic, they are each a lift to $\S_{[f]}$ of an object defined on the super-selection sector $\uuS_{[f]}$. 
Remarkably this is a \emph{global} formula that holds over the \emph{entire} phase space of (irreducible) YM configurations. 

\begin{remark}
A local result similar to the one discussed in Remark \ref{rmk:sympl-model} can be obtained for $\omega_{[f]} = \pi_{[f]}^*\uuomegao$ by replacing the connection $\Upsilon_\mathrm{Cou}$ in Equation \eqref{eq:YM-omega[f]} with the adapted connection $\Upsilon_\sigma$ used in Remark \ref{rmk:sympl-model}:
\begin{equation*}
\omega_{[f]} \stackrel{\sigma}{\approx} \int_\Sigma \tr\left(\bd E_\sigma^k\wedge \bd A^k \right) + \bd \int_{\pp\Sigma} \tr( e_\pp \bd u^k (u^k)^{-1}) ),
\end{equation*}
where $u^k = u k^{-1}$. 
This formula in particular shows that\footnote{See also Equation \eqref{eq:S-fibration}.}
\[
\uuS_{[f]} \simeq_\loc T(\Acal/\G) \times \mathcal{O}_f.
\qedhere\]
\end{remark}

The drawback to the global nature of Equation \eqref{eq:YM-omega[f]} is that the 1-forms $\bd_\mathrm{Cou} E_\rad$, $\bd_\mathrm{Cou} A$, and $\Upsilon^u_{\mathrm{Cou},\pp}$ cannot in general be exact, a fact that hinders a clear identification of the degrees of freedom of YM in $\uuS_{[f]}$.
This difficulty can be circumvented in two ways: either by introducing a local trivialization as done in the previous Remark, or by focusing on the tangent space structure of $T_{\phi}\uuS_{[f]}$ as described by Equation \eqref{eq:YM-omega[f]}. In physical parlance, the latter option corresponds to focusing on the ``linearized'' degrees of freedom over a reference background.\

Thus, from Equation \eqref{eq:YM-omega[f]}, we recognise the linearized degrees of freedom of YM in $\uuS_{[f]}$ to be \cite{RielloGomes,RielloSciPost}: (\emph{i}) the $\Upsilon_\mathrm{Cou}$-horizontal, ``radiative'', degrees of freedom $(\bd_\mathrm{Cou} E_\text{rad},\bd_\mathrm{Cou}A)$, together with (\emph{ii}) the Coulombic degrees of freedom (C.d.o.f.), captured by the corner dressed connection $\Upsilon_{\mathrm{Cou},\pp}^u$.  Note that although the C.d.o.f. are parametrized by quantities over the corner $\pp\Sigma$, they \emph{nonlocally} encode the Coulombic electric field \emph{throughout} $\Sigma$. 

Whereas the radiative degrees of freedom can be specified independently of the flux $f$, the C.d.o.f. depend on the choice of $\S_{[f]}$ as testified by their dependence on $u$. To better understand them, it is convenient to first write the dressed connection as
\[
\Upsilon^u_{\mathrm{Cou},\pp} = \Ad(u)\cdot( \UpsilonCou\vert_{\pp\Sigma} - u^{-1} \bd u).
\]
It is easy to see that the dressed connection is basic (in particular horizontal), and hence vanishes on any vertical vector field. This gives physical meaning to the \emph{difference}\footnote{In the literature, observables that arise from the fact that only the \emph{difference} between two co-varying quantities has physical meaning are sometimes referred to as ``relational'' \cite{GomesPhD}.} of the field space (principal) connection $\UpsilonCou\vert_{\pp\Sigma}$ (associated to the action of $\G$ on $\X$) and the flux orbit connection $u^{-1} \bd u$ associated to the transitive action of $\G$ on its co-adjoint orbits $\mathcal{O}_f$.

To test this point of view it is convenient to analyze the effects of a transformation of $(A,[u],E_{\text{rad}})$ which looks like a gauge transformation $g$ on the fluxes, but does \emph{not} transform the remanining d.o.f. accordingly:
\[
(A, \, [u], \, E_\text{rad}) \ \mapsto \ (A ,\, [ug], \, E_\text{rad} )\qquad g\in \G, \; g \not\in \G_f.
\]
These transformations, and their rather subtle Hamiltonian properties, were introduced and studied in \cite[Sections 4.3, 4.4, and 4.12]{RielloSciPost}, where they are called \emph{flux rotations}.
Comparison with Equation \eqref{eq:gaugetrC} readily shows that this map in general \emph{fails} to be a gauge transformation: indeed, by equivariance, a gauge transformation cannot act on $E_\pp = u^{-1}e_\pp u$ while keeping $(A;E_{\rad})$ fixed.
Rather, flux rotations implement \emph{physical} changes of the configuration that result in a modified Coulombic component $\star d_A\varphi$ of the electric field as computed by \eqref{eq:StoP-YM}. This change can be  quantified in a gauge invariant manner e.g. by computing the energy carried by the (Coulombic) electric fields associated to the two above configurations.  

\begin{remark}[A flat $\Upsilon$ and dressed electric fluxes]
\textit{If} the connection is flat, $\Upsilon = k^{-1}\bd k$, then one can as well perform the opposite dressing i.e.\ introduce $(u^{-1}\bd u)^k \doteq k (u^{-1}\bd u) k^{-1} - \bd k k^{-1}$ as well as the dressed, gauge invariant, corner electric fields $E_\pp^k = k^{-1} u^{-1} e_\pp u k$---i.e.\ introducing $[u^k] = [uk^{-1}]$. This ``inverse'' dressing, which is not always available, is however formally similar to the one mentioned in \cite{DonnellyFreidel16}. See the end of Section \ref{sec:runex-cornerdata} for more on this.
\end{remark}

~

\subsubsection{Abelian Case}
Whenever $G$ is Abelian, $\Upsilon_\mathrm{Cou}$ (Definition \ref{def:Coulombconnection}) is obtained by solving a Laplace equation with Neumann boundary conditions,
\begin{equation}
 \label{eq:Coulombgauge}
\Upsilon_\mathrm{Cou} = \bd \varsigma, \qquad
\begin{cases}
\Delta \varsigma = d^\star A &\text{in }\Sigma\\
\mathsf{n}d \varsigma = \mathsf{n} A & \text{at }\pp\Sigma.
\end{cases}
\end{equation}
The equation for $\varsigma\in C^\infty(\Sigma,\fg)$ shows that $A$ itself can be decomposed orthogonally w.r.t.\ $\langle\langle\cdot,\cdot\rangle\rangle$ (Appendix \ref{app:Hodge}) into 
\[
A = A_\text{rad} + d \varsigma
\]
in a manner fully analogous to $E = E_\text{rad} + \star d \varphi$.

From the above definition of $\varsigma=\varsigma(A)$ one sees that, similarly to $E_\text{rad}$, $A_\text{rad}$ is also in the kernel of $d^\star$ and satisfies a corner condition. That is, when written as order-1 tensors $X^i$, both $E_\text{rad}$ and $A_\text{rad}$ satisfy the following: $\nabla_i X^i = 0 = n_i X^i\vert_{\pp\Sigma}$. In other words, $A_\text{rad}$ is $A$ expressed in Coulomb gauge (generalized to the case with corners). It is then easy to verify that, in the Abelian case: 
\[
\bd_\mathrm{Cou} A = \bd A_\text{rad}, \qquad \bd_\mathrm{Cou} E_\text{rad} = \bd E_\text{rad}, \quad\text{and}\quad \mathbb{F}_\mathrm{Cou} = 0.
\]

Hence, in Maxwell theory, using the parametrization of $\C\subset\X$ introduced in Equation \eqref{eq:CtoP-YM} we find (cf.\ Proposition \ref{prop:omegadecomp})
\begin{equation}
 \label{eq:AbomegaH}
\iota_\C^*\omega = \int_\Sigma  \tr(\bd E_\text{rad} \wedge \bd A_\text{rad} )
-(-1)^{\dim\Sigma}
\int_{\pp\Sigma}\tr(\bd E_\pp \wedge \bd \varsigma_\pp), \qquad \text{Abelian}
\end{equation}
where $E_\pp\in\mathcal{E}_\pp^0$, and $\varsigma_\pp = \iota_{\pp\Sigma}^*\varsigma$ is most directly interpreted as an element of $\fGred \simeq C^\infty(\pp\Sigma,\fg)$.

Looking at the second stage reduction, the Abelian case is once again much simpler. In this case the orbits $\mathcal{O}_f$ are zero-dimensional (points) and  the Coulombic electric field is completely fixed by the choice of a superselected flux. In other words,  configurations in the superselection sector $\S_{[f]}$ are fully determined by $(A, E_\text{rad})\in\Acal\times\mathcal{H}$ since then, using Equation \eqref{eq:YM-fluxdensity},\footnote{In this case although $\varphi(e_\pp)$ is not uniquely determined by $e_\pp$---it is only up to an element of the stabiliser i.e.\ up to constant over $\Sigma$---the Coulombic electric field $\star d\varphi(e_\pp)$ \textit{is} uniquely determined by $e_\pp$.}
\[
\iota_{[f]}:\S_{[f]} \hookrightarrow \X, \quad (A, E_\text{rad}) \mapsto (A,E) = (A , E_\text{rad} + \star d\varphi(e_\pp))
\ \text{with} \ f = \int_{\pp\Sigma} \tr(e_\pp \ \cdot).
\]

Assuming, as in the previous section, that $P\to\Sigma$ is a trivial principal bundle for $\fg = (\mathbb R,>0)$, we have that the Abelian Coulomb connection is flat and exact and therefore, in this case, $\omega^{\pp,\mathrm{Cou}}_{[f]}=0$ and thus (cf.\ Equation \eqref{eq:AbomegaH} as well as Equation \eqref{eq:YM-omega[f]})\footnote{In the Abelian case, the orbit being a point, the dressing is trivial and the Coulombic contribution to $\omega_{[f]}$ in Equation \eqref{eq:YM-omega[f]} becomes:
$$
\bd \int_{\pp\Sigma} \tr(e_\pp \Upsilon^u_{\mathrm{Cou},\pp}) \stackrel{\text{Ab.}}{=} \int_{\pp\Sigma} e_\pp \bd \Upsilon_\mathrm{Cou} = \int_{\pp\Sigma} e_\pp  \mathbb{F}_\mathrm{Cou} = 0.
$$}
\[
\omega_{[f]} = \int_\Sigma \tr(\bd E_\text{rad} \wedge \bd A_\text{rad}),
\]
which is basic w.r.t.\ $\S_{[f]}\to\uuS_{[f]}$ and $\uuS_{[f]}$ is parametrized by the radiative fields. (See \cite{RielloEdge} for a pedagogical discussion.)

Furthermore, recalling that in the Abelian case over a trivial bundle the Coulomb connection $\Upsilon_\mathrm{Cou}=\bd\varsigma$ corresponds to the Coulomb gauge-fixing (suitably generalized to the presence of corners, Section \ref{sec:runex-firststage}), one can thus provide an explicit expression for the Poisson structure on $\uuS_{[f]}$ in terms of a Dirac bracket \cite{henneaux1992quantization}. On $\X=T^\vee\Acal$, this bracket is defined by
\[
\{ \hat E^i(x) , A_j(y) \}_D =  H(x,y)^i{}_j
\]
where $H(x,y)^i{}_j$ is the ``radiative'' projection kernel associated to $\mathcal{E}\to\mathcal{H}$. By duality this also projects $A$ onto the corresponding $A_\text{rad}$. Formally, one could write $H(x,y)^i{}_j = \delta(x,y)\delta^i{}_j - \nabla^i \Delta^{-1}(x,y) \nabla_j$---provided that $\Delta^{-1}(x,y)$ is understood as the Green's function appropriate to the boundary conditions of \eqref{eq:Coulombgauge}.
It is immediate to check that this bracket projects down from $\X$ to $\uCo=\C/\G$.

Of course, the bivector in $\mathfrak{X}^{\wedge2}(\X)$ associated to $\{\cdot,\cdot\}_D$ is
\[
\Pi_D = \int_\Sigma dx dy \, H(x,y)^i{}_j \frac{\delta}{\delta \hat E^i(x)} \wedge \frac{\delta}{\delta A_j(y)}.
\]

\medskip

To conclude, we notice that the final discussion of Section \ref{sec:runex-firststage} on the smoothness and symplecticity of $(\uCo,\uomegao)$, can be adapted step by step to discuss the smoothness and symplecticity of $(\uuS_{[f]},\uuomegao_{[f]})$. It is enough to replace $\fGo$ and $V_\circ = \rho(\C \times \fGo)$ by $\fG$ and $V_{[f]} = \rho(\S_{[f]}\times \fG)$, and to swap Dirichlet for Neumann boundary conditions when the equivariant Hodge--de Rham decomposition of $\Omega^1(\Sigma,\fg)$ is required (cf.\ Proposition \ref{prop:HodgeEqui} and Remark \ref{rmk:HodgeInYM}).

\section{Corner data}\label{sec:cornerdata}

So far we have looked at the fully-reduced phase space, and found that it is a collection of symplectic spaces that we called ``superselection sectors'' (SSS). Recall that these SSS can be identified with the symplectic leaves of the Poisson structure on $\uuC$.
The condition $\bHo=0$ restricts the set of fluxes that are compatible with the bulk reduction, and we have seen that generally only the inclusion $\F\subset \Foff$ holds.

In this section we look for a universal geometric structure associated to the corner that 
generalises the notion of superselection sector off shell.

In order to do this, we first use the flux map $\h$ to construct an auxiliary symplectic space of corner data $(\A_\pp,\varpi_\pp)$ from the restriction to the corner of the sections comprising the bulk algebroid $\A$.
Then, under a suitable regularity hypothesis (Assumption \ref{ass:AnnIm}), we show that the corner data $\A_\pp$ is itself a vector bundle $p_\pp:\A_\pp\to\X_\pp$ over a space of corner field configurations $\X_\pp$ with fibre $\fGpoff$ isomorphic to $\fGredoff$ (Definition \ref{def:fGp}).
Furthermore, we show that this vector bundle inherits an action Lie algebroid structure $\rho_\pp: \A_\pp \to T\X_\pp$ from the bulk anchor map $\rho:\A \to T\X$, and we prove that this induces a unique (partial) Poisson structure $\Pi_\pp$ on the space $\X_\pp$. We call the space $(\X_\pp,\Pi_\pp)$ the space of off-shell superselection sectors.

There are at least two natural procedures to induce field-theoretic data on the corner. One of these procedures, detailed in Section \ref{sec:offshellcorner1}, focuses on the properties of the flux map $h$. The other procedure---which is more general as it does only require the existence of a constraint form, and not that of a flux map\footnote{Note, the constraint form $\bHo$ is a meaningful object also in Lagrangian field theories with local symmetries which fail to be Hamiltonian and thus fail to possess a well-defined momentum map, see Appendix \ref{app:covariant}.}---introduces the corner symplectic data by treating the constraint form $\bHo$ as a Lagrangian density on the field space $\A = \X \times \fG$ over $\Sigma$.\footnote{This procedure is essentially a degree $0$-version of the BV-BFV construction of the corner phase space (in degree 1), with $\fG$ interpreted here as the zero-degree ``evaluation'' of the ghosts \cite{CMR1,RejznerSchiavina}, cf.\ Definition \ref{def:cornershift}. More generally, this relates to the ideas of Kijowski and Tulczijew \cite{KT1979}, applied one codimension higher.}

The flux-based procedure naturally connects to our reduction results from the previous sections. Although the two procedures need not coincide a priori, in Section \ref{sec:ultralocality} we show that in the case of ultralocal data $(\X,\bom)$ on $\Sigma$, they do.

\subsection{Off-shell corner data} \label{sec:offshellcorner1}
In Sections \ref{sec:flux} and \ref{sec:BulkCornerGrps} we characterized the constraint gauge ideal as the flux-annihilating ideal $\fGo = \fN$, and the flux gauge algebra as the corresponding quotient $\fGred = \fG/\fGo$. In order to obtain a direct relationship between $\fGo$ and the constraint space $\C$ (cf.\ Lemma \ref{lemma:coisoC}, and Propositions \ref{prop:symplclosVo} and \ref{prop:symplecticclosurenogo}), it was important to work on-shell. However, since the conditions defining $\C$ are typically non-local, this hinders a local description (see Remark \ref{rmk:nonlocality} and Section \ref{sec:runex-firststage} for an explicit example). In particular, $\fGred$ has so far been defined as the quotient between 
$\fG$ and $\fGo$, and---beside the observation that $\fGred$ is trivial whenever $\pp\Sigma=\emptyset$---it is not clear a priori whether and in which sense $\fGred$ is a Lie algebra genuinely supported on $\pp\Sigma$ (see Definition \ref{def:LocLieAlg}).
The goal of the forthcoming discussion is to fill this gap.

We start by revisiting the definitions of the off-shell analogues of the flux annihilator $\fNoff$ (Definition \ref{def:frakN}) and flux gauge algebra $\fGredoff = \fG/\fNoff$ (Definition \ref{def:fGp}). The upshot is that
\begin{enumerate}[label=(\roman{*})]
    \item $\fGredoff$ is canonically isomorphic to a Lie algebra supported on $\partial \Sigma$, named $\fGpoff$, which in all examples we consider is also \emph{local} (Definition \ref{def:LocLieAlg}).
    \item $\fGred$ will be described as the quotient of $\fGredoff\simeq\fGpoff$ by a Lie ideal and is therefore also isomorphic to a Lie algebra supported on $\partial\Sigma$ (Definition \ref{def:LocLieAlg}).
\end{enumerate}

Our strategy will employ a presymplectic reduction procedure for an auxiliary structure over the restriction of fields to $\partial\Sigma$.
To make this statement precise, we observe that, when restricting sections of a fibre bundle to the boundary, it is sometimes necessary to retain information of transverse derivatives. Formally this means looking at jets of sections of the fibre bundle over a thin collar of $\partial \Sigma$, evaluated at $\partial\Sigma \times \{0\}$. (We follow the constructions of \cite{CMRCorfu} here.) In local field theories, we typically work with a finite jet order, which means that restriction simply requires one to supplement sections of the induced bundle with a finite number of transverse derivatives.

\begin{definition}[Off-shell corner data]\label{def:cornerconf}
The space of \emph{pre-corner data} $\tA_{\partial \Sigma}$ is given by transverse jets of sections\footnote{This means that we are looking at equivalence classes of sections whose \emph{transverse} derivative, i.e.\ along the collar direction, coincide at a point. The collar neighborhood of the boundary is needed to ensure the coordinate independence of the construction. Observe as well that, for a given vector bundle $W\to \Sigma$, the infinite jet bundle $J^\infty W$ is a Fr\'echet manifold. Its space of sections is again Fr\'echet \cite[Prop. 30.1]{kriegl1997convenient}.} in a collar neighborhood $\partial\Sigma \times [0,\epsilon)$, evaluated at $\partial\Sigma\times\{0\}$.

Then, let $\varpi \in \iloc^{1,1}(\X\times\fG)\subset\iloc^2(\A)$ be defined as
\[
\varpi \doteq \langle \cbd_H h, \cbd_V \xi\rangle,
\]
and further define:
\begin{enumerate}
    \item $\wtpi_{\partial}\colon\A\to \tA_{\partial }$ to be the surjective submersion given by restriction,
    \item the pre-corner 2-form $\widetilde{\varpi}_{\partial }$ on $\tA_{\partial \Sigma}$, such that $\wtpi_\pp^*\widetilde{\varpi}_{\partial } = \varpi$.
\end{enumerate}
Moreover, whenever $\ker(\widetilde{\varpi}^\flat_{\partial })$ is regular, we define the space of \emph{off-shell corner data} $(\A_{\partial},\varpi_{\partial})$ as the presymplectic reduction of $(\tA_{\partial},\widetilde{\varpi}_{\partial})$, i.e.
\[
    \A_{\partial } \doteq \widetilde\A_{\partial}/\ker(\widetilde{\varpi}^\flat_{\partial}),
    \qquad
    \widetilde{\varpi}_\partial \doteq \pi_\text{pre}^*{\varpi}_\partial
\]
where we denoted $\pi_\text{pre}\colon \widetilde\A_{\partial \Sigma} \to \A_{\partial \Sigma}$ the presymplectic reduction map.
Finally, we also introduce
\[
    \pi_{\partial}: \A \to \A_{\partial}
\]
as the composition of $\pi_{\partial} \doteq \pi_\text{pre}\circ\widetilde\pi_{\partial}$. Summarising:
\[
\xymatrix@C=1cm{
(\A ,\varpi) \ar@{->>}[r]^-{\wtpi_\partial} \ar@{->>}[rd]_-{\pi_\partial}
& (\tA_\partial, \widetilde{\varpi}_\partial)  \ar@{->>}[d]^-{\pi_\text{pre}}\\
& (\A_\partial,\varpi_\partial)
}
\]
\end{definition}

\begin{remark}
Observe that $\varpi$, $\widetilde{\varpi}$, and $\varpi_\pp$ are all manifestly closed.  In fact, they are exact as a consequence of the linearity of the local momentum map---e.g. $\varpi = \cbd \langle \cbd h, \xi\rangle$.
\end{remark}

\begin{remark}\label{rmk:projectionsplitting1}
Since we work with $p:\A=\X\times\fG\to\X$ an \emph{action} Lie algebroid, we have a finer decomposition of the (pre-)corner data. In fact, the restriction
splits as $\tA_{\partial}=\tX_\partial\times\tfGp$,
where $\tX_\partial$ (resp. $\tfGp$) denotes the restriction of sections in $\X$ (resp. $\fG$) and their transverse jets. This comes with a splitting of the projection $\wtpi_\partial=\wtpi_{\partial,\X} \times \wtpi_{\partial,\fG}$.
For future reference we also introduce the following notation. First we introduce the projection $\wt{p}:\tA_{\partial}=\tX_\partial\times\tfGp \to \tX$; then, considering the associated canonical flat connections, we introduce the following horizontal/vertical splits
\[
T\A = H\A \oplus V\A,
\quad
T\wt{\A}_\pp = H\wt{\A}_\pp \oplus V\wt{\A}_\pp,
\]
where $H_{(\phi,\xi)}\A\simeq T_\phi\X$ and $V_{(\phi,\xi)}\A \simeq T_\xi\fG \simeq \fG$---and similarly for the corresponding pre-corner data.
\end{remark}

In order to characterise $\A_\pp$, we will need the following notions.

\begin{definition}[Off-shell flux annihilating locus]\label{def:fluxannlocus}
Let $\bd \h \in \iloc^1(\A)$, with $\h$ the off-shell flux map of Assumption \ref{ass:setup}. The \emph{off-shell flux annihilating locus} in $\A$ is
\[
\mathsf{N}^\off \doteq \{ (\phi,\xi)\in\A : \bd \h(\phi,\xi) \equiv \langle \bd h(\phi), \xi\rangle =0 \}.
\qedhere
\]
\end{definition}

\begin{remark}\label{rmk:fNphi}
Recall $p:\A\to\X$. Since $\bd \h^{\text{co}}$ can be seen as a linear map  $\bd\h^{\text{co}}:\fG \to \iloc^1(\X)$, it is clear that
\[
\fN_\phi \doteq \mathsf{N}^\off \cap p^{-1}(\phi) \subset \fG
\]
is a vector space coinciding with the kernel of $\mathrm{ev}_\phi \circ \bd h : \fG \to T^*_\phi \X$. Recall the Definition \ref{def:frakN} of $\fNoff = \ker(\bd \h^{\text{co}} )$.
Then, from the previous observation one also deduces that
\[
\fNoff \equiv \bigcap_{\phi\in\X} \fN_\phi .
\]
This construction and the following Assumption should be compared with Assumption \ref{assA:isotropy} and Remark \ref{rmk:isotropybundle3}.
\end{remark}

\begin{assumption}\label{ass:AnnIm}
The flux annihilating locus $\mathsf{N}^\off$ is a trivial sub-vector bundle of the trivial bundle $p:\A\to\X$, $\A = \X \times \fG$. We denote it
\[
^hp : \mathsf{N}^\off \to \X.
\qedhere
\]
\end{assumption}

From this assumption and the previous remark, it follows that for all $\phi\in\X$, 
$\fN_\phi = \fNoff$,
and therefore
\[
\mathsf{N}^\off = \X \times \fNoff \subset \A.
\]
Summarising, thanks to Assumption \ref{ass:AnnIm}, we have the following embedding relation between (trivial) fibre bundles:\footnote{The top and bottom lines denote fibre bundles according to the usual notation $F\to E \xrightarrow{\pi} X$, for $F$ the fibre, $X$ the base, and $\pi$ the projection.}
\[
\xymatrix@R=.25cm{
\fG \ar[r] & \A  \ar[rd]^-{p} \\
&& \X\\
\fNoff \ar[r]\ar@{^(->}[uu] & \mathsf{N}^\off  \ar@{^(->}[uu] \ar[ru]_-{^hp}
}
\]

\begin{remark}
Assumption \ref{ass:AnnIm} is satisfied in ``standard'' scenarios. For example, whenever $\X$ is a linear/affine space and $h$ is affine---since then $\fN_\phi\subset\fG$ is the same for all $\phi$. (See, Section \ref{sec:runex-cornerdata} for an explicit example.)
\end{remark}

\begin{remark}
If $\A$ were a (not necessarily trivial) Lie algebroid $p:\A\to\X$ with connection $\check{\mathbb{D}}$, one could rephrase assumption \ref{ass:AnnIm} as saying that $\mathsf{N}^\off$ is a sub-vector bundle of $p:\A\to\X$ generated by $\check{\mathbb{D}}$-constant sections. (Cf. Remarks \ref{rmk:isotropybundle1}, \ref{rmk:isotropybundle2} and \ref{rmk:isotropybundle3} for closely related comments on the isotropy bundle/locus.)
\end{remark}

\begin{lemma}[Restriction kernel]\label{lemma:precornerkernel}
Let $\wtpi_{\partial,\fG}$ be defined as in Remark \ref{rmk:projectionsplitting1}. Then, the \emph{restriction kernel}
\[
\mathfrak{K} \doteq \ker(\wtpi_{\partial,\fG})\subset\fG
\]
is a \emph{just} constraining ideal (Definition \ref{def:constraintideal}) with $\fG_c \subset \mathfrak{K} \subset\fNoff \subset \fN = \fGo$.
\end{lemma}
\begin{proof}
By the definition of $\wtpi_{\pp,\fG}$, observe that an $\eta \in \fG = \Gamma(\Sigma,\Xi)$ is an element of $\fK$ iff $(j^\infty\eta)(x)=0$ for all $x\in\pp\Sigma$.\footnote{Note that restriction of sections commutes with jet evaluations for \emph{tangent} directions. Here, we first evaluate the jets and then look at their values at an $x\in\pp\Sigma\subset\Sigma$, so that $(j^\infty\xi)\vert_{\pp\Sigma} = 0$ iff the restriction of $\xi$ vanishes at $\pp\Sigma$ \emph{and so do all the transverse jets}.}

By the locality of $[\cdot,\cdot]:\fG\times\fG\to\fG$, for every $\xi\in\fG$, $[\xi,\eta]$ descends to a linear function of $j^k\eta$, for $k$ large enough. Then, if $\eta\in\fK$, by the previous observation it follows that $j^\infty\eta$ vanishes at $\partial\Sigma$ and therefore so does the restriction to $\partial\Sigma$ of $j^\infty[\xi,\eta]$. From this we conclude that $[\xi,\eta]$ also belongs to $\fK$, whence $\mathfrak{K}$ is an ideal. 

To prove that $\fK\subset \fG$ is also a just constraining subalgebra, according to Lemma \ref{lemma:fGc} and Corollary \ref{cor:justconstr} it suffices to show that $\fG_c \subset \fK \subset \fN=\fGo$.

Recall that $\fG_c$ is the ideal of $\fG$ of elements $\xi$ for which there exists a collar neighborhood $U_\xi\subset\Sigma$ of $\partial\Sigma$ that does not intersect the support of $\xi$. Therefore for such $\xi$'s $j^\infty\xi$ vanishes in the neighborhood $U_\xi$ of $\pp\Sigma$, and thus $\fG_c\subset \fK$.

Finally, to prove that $\fK\subset \fN=\fGo$ it is enough to show that $\fK\subset\fNoff= \ker(\h_\smbullet)$.
But this is obvious since for all $\phi$
\[
\langle \h_\smbullet(\phi), \xi\rangle = \int_{\Sigma}d\langle\bh_\smbullet(\phi),\xi\rangle
\]
which only depends on the restriction of $j^\infty\xi$ to $\pp\Sigma$ and therefore vanishes if $\xi\in\fK$.
\end{proof}

\begin{lemma}[Pre-corner algebroid]\label{lemma:precorneralg}
Let $\tA_\partial$ be the restriction of $\A$ to the corner $\partial\Sigma$, as above. Then, 
\[
\tA_\partial =\tX_\partial\times\tfGp
\]
and $\wtpi_\partial=\wtpi_{\partial,\X} \times \wtpi_{\partial,\fG}$ so that  $\tX_\partial = \wtpi_{\partial,\X}(\X)$ and $\tfGp = \wtpi_{\partial,\fG}(\fG)$.
Furthermore, $\wt{p}_\pp:\tA_\partial\to\wt{\X}_\pp$ carries a canonical action Lie-algebroid structure induced by that of $\A$ (Equation \eqref{eq:A}).
\end{lemma}
\begin{proof}

To prove that a canonical action Lie agebroid structure $([\cdot,\cdot]^\sim_\pp, \wt\rho_\pp)$ exists on $\tA_\partial \to \tX_\partial$, we need first to show that $\tfGp$ carries a canonical Lie algebra structure.
This is a consequence of Lemma \ref{lemma:precornerkernel}, since $\tfGp\simeq \fG/\mathfrak{K}$ and $\mathfrak{K}$ is an ideal, and thus the bracket
\[
[\cdot,\cdot]^\sim_\partial:\tfGp\times\tfGp \to \tfGp,
\qquad
[\cdot,\cdot]^\sim_\partial = \wtpi_\partial\circ [\wtpi_\partial^{-1}(\cdot), \wtpi_\partial^{-1}(\cdot) ]
\]
is well defined.  
(Whenever unambiguous, we will write $[\cdot,\cdot]^\sim_\partial$ simply as $[\cdot,\cdot]$.)

Similarly, locality of the action $\rho$ (Definition \ref{def:localaction}) means that $\rho: \A=\X\times\fG \to T\X$ is local as a map of sections of fibre bundles. Therefore, one has $\rho(\mathfrak{K})\subset \ker(\cbd \wtpi)$ and the restriction of $\rho$ at $\partial\Sigma$ is a well-defined map 
\begin{align*}
\widetilde\rho_\partial : \tA_\partial = \tX_\partial\times \tfGp & \to T\tX_\pp,
\qquad \widetilde\rho_\partial = \cbd\wtpi_\partial\circ\rho\circ\wtpi_\partial^{-1}.
\end{align*}

Notice that the data $[\cdot,\cdot]^\sim_\partial$ and $\widetilde{\rho}_\partial$ define a unique action Lie algebroid structure on $\tA_\partial = \tX_\partial \times \tfGp \to \tX_\partial$ (cf.\ the text below Equation \eqref{eq:A}). If one restricts their attention to \emph{local} sections of $\A\to\X$ and $\tA_\partial\to\widetilde{\X}_\partial$ only, then the above discussion shows that the two action Lie algebroid structures are one the restriction of the other at $\partial\Sigma$ (ibidem).
\end{proof}

Recall that $\fG$ is a local Lie algebra, obtained as the sections of the vector bundle of Lie algebras $\Xi\to \Sigma$. Quotienting this by those sections whose jets have a vanishing restriction at the boundary, returns the Lie algebra $\tfGp$ of Lemma \ref{lemma:precorneralg}. Observe that the quotient leaves us with sections of the induced bundle, together with infinitely many transverse jets, which can be seen as local information but w.r.t.\ an infinite vector bundle over $\Sigma$. Typically, dependency on an infinite number of jets is eliminated with the next step: pre-symplectic reduction. The price to pay is that locality might be lost in the process, i.e.\ it might not be possible to identify $\A_\pp$ with the space of section of some bundle over $\partial\Sigma$.\footnote{In many cases of interest, however, locality turns out not to be lost.}

The main result of this section is that $\A_\pp$ is a (weak) symplectic manifold (cf.\ Definitions \ref{def:locsymp} and \ref{def:symplecticmanifolds}) additionally endowed with the structure of an action Lie algebroid $\rho_\pp\colon\A_\pp\to T\X_\pp$. This is detailed in the following theorem, which in particular shows that $\fGredoff$ is isomorphic to an algebra $\fGpoff$ \emph{supported on the corner} $\partial\Sigma$, in the sense of Definition \ref{def:LocLieAlg}.

\begin{theorem}[Corner algebroid]\label{thm:corneralgd}
Let $(\A_{\partial},\varpi_\pp)$ and $\pi_\pp:\A\to\A_\pp$ be as in Definition \ref{def:cornerconf}, and let Assumption \ref{ass:AnnIm} hold.
Consider the ``kernel'' distribution 
\[
^hK\X \doteq \{ \mathbb Y \in T\X :  \bi_{\mathbb Y}\bd h =0\}\subset T\X.
\]
Then, $\pi_\pp \doteq \pi_\text{pre}\circ \wtpi_\pp$ factorises as 
\[
\pi_\pp = \pi_{\pp,\X}\times\pi_{\pp,\fG}
\]
where $\pi_{\pp,\X}$ and $\pi_{\pp,\fG}$ implement the following quotients of $\X$ and $\fG$:
\[
\X_\pp \doteq \pi_{\pp,\X}(\X) = \frac{\wt{\X}_\pp}{\cbd\wtpi_{\pp,\X}({}^hK\X)}
\quad\text{and}\quad
\fGpoff \doteq \pi_{\pp,\fG}(\fG) = \frac{\tfGp}{\wt{\fN}_\pp^\text{off}}
\simeq \fGredoff
\]
where $\wt{\fN}_\pp^\text{off} \doteq \wtpi_{\pp,\fG}(\fNoff)$.
(The last isomorphism is canonical, see Equation (\ref{eq:Gppiso})). Therefore, $\A_\pp$ has the product structure 
\[
    \A_{\partial} = \X_\pp \times \fGpoff. 
\]

Moreover, $p_\pp:\A_\pp\to \X_\pp$ can be endowed with an action Lie algebroid structure with anchor $\rho_\pp:\A_\pp \to T\X_\pp$ canonically induced by $\rho:\A\to T\X$ through the formula
\[
\rho_\pp = \bd \pi_\pp \circ \rho \circ \pi_\pp^{-1} = \bd \pi_\text{pre} \circ \wt\rho_\pp \circ \pi_\text{pre}^{-1}
\]

Finally, seen as a function over $\A$, the flux map $h_\smbullet$ is basic with respect to $\pi_{\pp}$ and therefore defines a function $h_\pp$ over $\A_\pp$ via\footnote{Although $h_\pp$ depends on the choice of reference configuration $\phi_\smbullet$, we omit appending $\cdot_\smbullet$ to $h_\pp$ for ease of notation.}
\[
\h_\smbullet \doteq \pi_\pp^* h^{}_\pp.
\]
Similarly $k_\smbullet$ is basic with respect to $\pi_{\pp,\fG}$ and defines a CE cocycle over of $\fGpoff$ which we also denote by $k_\smbullet$.
Then, the map $h_\pp$ is equivariant with respect to the action of $\rho_\pp$ up to the cocycle $k_\smbullet$, i.e.\ 
\[
    \langle\L_{\rho_\pp(\xi_\pp)}h_\pp,\eta_\pp\rangle = \langle h_\pp, [\xi_\pp,\eta_\pp]\rangle + k_\smbullet(\xi_\pp,\eta_\pp)
    \qquad
    \forall \xi_\pp,\eta_\pp \in \fGpoff,
\]
and allows to write the weak symplectic structure over $\A_\pp$ as
\[
\varpi_\pp = \langle \cbd_H h_\pp , \cbd_V \xi_\pp \rangle
\]
where $\cbd=\cbd_H+ \cbd_V$ is now understood as the differential over $\A_\pp$ decomposed according to the $p_\pp$-horizontal/vertical split.
\end{theorem}

Before turning to the proof of the theorem, let us summarise diagrammatically the construction of the corner Lie algebra $\fGpoff\simeq\fGredoff$ (all arrows are surjective maps, with the bottom one a canonical isomorphism). 
\begin{equation}
    \label{eq:diagramreductiongroups}
\xymatrix@C=1cm{
\fG 
\ar@{->}[r]^-{\wtpi_{\partial,\fG}}
\ar@{->}[d]_-{ \cdot / \fNoff}
\ar@{->}[dr]^{\pi_{\partial,\fG}}
& \tfGp 
\ar@{->}[d]^-{\pi_{\text{pre},\fG}}\\
\fGredoff\; 
\ar@{->}_{\simeq}[r]
& \fGpoff
}
\end{equation}
Note: the algebras appearing on the left of this diagram are all defined over the bulk $\Sigma$, whereas the algebras on the right are defined over the corner $\pp\Sigma$.

\begin{proof}
Recall the  $p$-horizontal/vertical split (cf.\ Section \ref{sec:LieAlgPict} and Remark \ref{rmk:projectionsplitting1}),
\[
T\A = H\A \oplus V\A,
\]
induced by the trivial connection on $p:\A\to\X$, $\A = \X\times \fG$.
Accordingly, we denote the split of a vector $\mathbb Y\in T\A$ into its $p$-horizontal and $p$-vertical components as
\[
\mathbb Y= \mathbb Y_H + \mathbb Y_V \in H\A\oplus V\A.
\]

For $\A$ the action Lie algebroid $\X\times\fG\to\X$, one further has the identifications
\[
H\A = T\X\times \fG,\quad
V\A= \X\times T\fG \simeq \A \times \fG,
\]
thanks to which $\mathbb Y_V\vert_{(\phi,\xi)}$ can either be thought as an element of $T_\xi\fG$ or as an element of $\fG$.
In fact the latter statement is a general feature since, even for non-trivial vector bundles $p:\A\to\X$, the vertical distribution is $ V\A \simeq \A \times_\X \A$.

Thanks to Remark \ref{rmk:projectionsplitting1} and Lemma \ref{lemma:precorneralg} analogous considerations hold for $T\wt{\A}_\pp = H\wt{\A}_\pp \oplus V\wt{\A}_\pp$---in which case we will adopt a tilde notation, as in
\[
\wt{\mathbb Y}= \wt{\mathbb Y}_H + \wt{\mathbb Y}_V \in H\wt{\A}_\pp\oplus V\wt{\A}_\pp.
\]

In what follows, given $\wt{\mathbb Y}\in T\wt{\A}_\pp$ (resp. $\wt{\mathbb Y}_H$, and $\wt{\mathbb Y}_V$) we will denote by $\mathbb Y\in T\A$ (resp. $\mathbb Y_H$, and $\mathbb Y_V$) one of its lifts w.r.t.\ to $\cbd\wtpi_\pp$.
These always exists in view of the surjectivity of $\cbd\wtpi_\pp$, which also implies that
\begin{equation}\label{eq:kerneleq}
    0=\bi_{\wt{\mathbb{Y}}}\wt\varpi_\pp  \iff 0  =\bi_\mathbb{Y}\varpi =  \langle \mathbb{Y}_H \h,\cbd_V\xi\rangle +  \langle \cbd_H \h,\mathbb{Y}_V\rangle
\end{equation}
for any lift $\mathbb Y$ of $\wt{\mathbb Y}$ (in the rightmost term we used the identification of $\mathbb{Y}_V\vert_{(\phi,\xi)} \in V_{(\phi,\xi)}\A \simeq T_\xi\fG \simeq\fG$).
Linear independence of the splitting of forms over $\A$ into $p$-vertical and $p$-horizontal implies that the two terms must vanish independently. 

Now, the first term above vanishes iff $\mathbb{Y}_H \h=0$. 
Observing that the condition $(\mathbb{Y}_H \h)(\phi,\xi)=0$ does not depend on $\xi$, we introduce the distribution $^hK\X\subset T\X$ defined by
\begin{equation}
 \label{eq:hKP}
^hK_\phi\X \doteq \{ \mathbb{Y} \in T_\phi\X:  (\mathbb{Y}\h)(\phi) = 0\} ,   
\end{equation}
and find
\begin{equation}\label{eq:hKAfacto}
\mathbb{Y}_H\in \ker(\varpi^\flat) \iff
\mathbb{Y}_H\in {}^hK\A \doteq {}^hK\X \times \fG \subset H\A.
\end{equation}

Analogously, focusing on the second term in $\bi_\mathbb{Y}\varpi$, we are led to introducing the distribution $^hN\A\subset V\A \subset T\A$, as
\[
^hN_{(\phi,\xi)}\A \doteq \{ \mathbb{Y}_V \in V_{(\phi,\xi)}\A: \langle \bd h(\phi),\mathbb Y_V\rangle = 0 \},
\]
so that 
\[
\mathbb{Y}_V\in\ker(\varpi^\flat)
\iff
\mathbb{Y}_V\in {}^hN\A .
\]

Comparing with Definition \ref{def:fluxannlocus}, we recognise that 
\[
^hN\A \simeq \A\times_\X {}(\mathsf{N}^\off)\subset V\A,
\]
as a subset of $V\A\simeq\A\times_\X \A$.
Using the triviality of $\A=\X\times \fG$, together with Assumption \ref{ass:AnnIm}, we have ${}^hN\A  \simeq \X \times \fG \times \fNoff$.
Summarising:
\begin{equation}\label{eq:hNAfacto}
    \mathbb{Y}_V\in\ker(\varpi^\flat)
\iff
\mathbb{Y}_V\in{}^hN\A \simeq \X \times \fG \times \fNoff \subset V\A.
\end{equation}

For later convenience we introduce the notation
${}^hN\A \simeq \X \times {}^hN\fG $, where  $^hN \fG \subset T\fG$ is
\begin{equation}\label{eq:laterconveniencehNG}
 {}^hN\fG \simeq \fG \times \fNoff.
\end{equation}
Using these results, Equation \eqref{eq:kerneleq} can be rephrased as 
\[
\wt{\mathbb{Y}}\in \ker(\wt\varpi_\pp^\flat) 
\iff
\mathbb{Y} \in {}^hK\A\oplus {}^hN\A.
\]
for any lift $\mathbb Y$ of $\wt{\mathbb Y}$. 
Thus, projecting along $\cbd\wtpi_\pp$, we obtain that the kernel of $\widetilde{\varpi}^\flat_{\partial}$ factorises as:
\[
\ker(\widetilde{\varpi}^\flat_{\partial}) = \ker_H(\widetilde{\varpi}^\flat_{\partial}) \oplus\ker_V\widetilde{\varpi}^\flat_{\partial},
\qquad
\begin{cases}
 \ker_H(\widetilde{\varpi}^\flat_{\partial})= \cbd\wtpi_{\pp}( {}^hK\A)\\
 \ker_{V}(\widetilde{\varpi}^\flat_{\partial})= \cbd\wtpi_{\pp}( {}^hN\A)
\end{cases}.
\]

Before quotienting $\wt\A_\pp$ by $\ker(\wt\varpi_\pp^\flat)$, which is involutive as a distribution owing to the closure of $\wt{\varpi}_\pp$, we note that its vertical and horizontal subdistributions also are, independently.\footnote{Note: since the trivial connection is manifestly flat, $H\wt\A_\pp$ and $V\wt\A_\pp$ are both involutive; the conclusion follows from $\ker_H(\wt\varpi_\pp^\flat)\subset H\wt\A_\pp$ and $\ker_V(\wt\varpi_\pp^\flat)\subset V\wt\A_\pp$. }

Therefore, when quotienting $\wt{\A}_\pp$ to obtain $\A_\pp$, thanks in particular to Equations \eqref{eq:hKAfacto} and \eqref{eq:hNAfacto} (which crucially relies on Assumption \ref{ass:AnnIm}), we find:
\begin{multline}\label{eq:Appquotients}
\A_\pp \doteq \frac{\wt\A_\pp}{ \ker(\wt\varpi_\pp^\flat) }
= \frac{\tX_\pp\times\tfGp}{\cbd\wtpi_{\pp}({}^hK\A)\times\cbd\wtpi_{\pp}({}^hN\A)}
\\= \frac{\tX_\pp}{\cbd\wtpi_{\pp,\X}({}^hK\X)}\times\frac{\tfGp}{ \cbd\wtpi_{\pp,\fG}({}^hN\fG)}
\doteq \X_\pp\times\fGredoff.
\end{multline}

Let us define the (factorized) reduction maps $\pi_{\text{pre},\X}$ (resp. $\pi_{\text{pre},\fG}$) as the maps that identify elements of $\tX_\pp$ (resp. $\tfGp$) which are connected by the flow of the vectors in $\bd\wtpi_{\pp,\X}({}^hK\X)\subset T\tX_\pp$ (resp. $\cbd\wtpi_{\pp,\fG}({}^hN\fG)\subset T\tfGp$), i.e.
\[
\X_\pp \doteq \frac{\tX_\pp}{\cbd\wtpi_{\pp,\X}({}^hK\X)} \doteq \pi_{\text{pre},\X}(\wt{\X}_\pp),
\qquad
\fGpoff \doteq \frac{\tfGp}{\cbd\wtpi_{\pp,\fG}({}^hN\fG)} 
\doteq \pi_{\text{pre},\fG}(\tfG_\pp),
\]
Thus, also $\pi_\partial \doteq \pi_\text{pre}\circ\wtpi_\partial$ splits,
\[
{\pi}_\partial={\pi}_{\partial,\X} \times {\pi}_{\partial,\fG}.
\]
Hence, $\A_\pp$ is a symplectic manifold with product structure $\A_\pp=\X_\pp\times \fGpoff$. 

To prove the canonical isomorphism of Lie algebras $\fGpoff\simeq\fGredoff$ it is enough to observe that---once again in view of the canonical isomorphism $^hN\fG \simeq \fG \times \fNoff$---the quotient of $\tfGp$ by $\bd\wtpi_{\pp,\fG}({}^hN\fG)$ is equivalent to the quotient of $\tfGp$ by 
\[
\wt\fN^\off_\pp \doteq \wtpi_{\pp,\fG}(\fNoff) \simeq \fNoff/\mathfrak{K}. 
\]
Indeed, using this observation, we compute
\begin{equation}\label{eq:Gppiso}
\fGpoff \doteq \frac{\wt\pi_{\pp,\fG}(\fG)}{\cbd\wtpi_{\pp,\fG}({}^hN\fG)}
=  \frac{\tfGp}{\wt{\fN}^{\off}_\pp }
\simeq \frac{ \fG/ \mathfrak{K}}{\fNoff/\mathfrak{K}}
\simeq \frac{\fG}{\fNoff} \doteq \fGredoff,
\end{equation}
where all the isomorphisms are canonical isomorphisms of Lie algebras (cf.\ Lemma \ref{lemma:precorneralg} and Lemma \ref{lemma:precornerkernel}).

Since $\varpi$ is linear in $\xi$ and since by the above construction $h_\smbullet$---seen as a function over $\A=\X\times\fG$---is basic w.r.t.\ $\pi_{\pp}$, there exists a $\h_\pp \in \Omega^0(\X_\pp\times \fGpoff)$ which is linear in the $\fGpoff$ entry, such that $h_\smbullet = \pi_\pp^*h^{}_\pp$.
Then, we have
\[
\varpi_\pp=\langle \cbd_H h_\pp,\cbd_V \xi_\pp \rangle,
\] 
where $\xi_\pp\in\fGpoff$ and $\cbd=\cbd_H+ \cbd_V$ is now understood as the differential over $\A_\pp$ decomposed according to the $p_\pp$-horizontal/vertical split. 

We now turn to show that the anchor map $\wt\rho_\pp : \wt\A_\pp \to T\wt{\X}_\pp$ (Lemma \ref{lemma:precorneralg}) descends through $\pi_\text{pre}$ to an anchor $\rho_\pp : \A_\pp \to T\X_\pp$. 
It is most convenient to adopt the equivalent viewpoint in which the anchor map on an action Lie algebroid is seen as an action of the Lie algebra on the base space, e.g. $\rho_\pp : \fGpoff \to \Gamma(T\X_\pp)$.
Then, since $\fGpoff = \tfGp /\wt\fN^\off_\pp$, we have that $\wt\rho_\pp$ descends to an action $\rho_\pp$  iff 
\[
\wt\rho(\wt\fN^\off_\pp) \subset \Gamma( \ker(\wt\varpi_\pp^\flat)).
\]
Since $\wt\rho(\wt\fN^\off_\pp)\subset \Gamma(H\wt{\A})$, using the surjectivity of $\cbd\wtpi_{\pp}$ and the definition of $\wt\rho$ (Lemma \ref{lemma:precorneralg}), it is enough to prove that $\rho(\fNoff) \subset {}^hK\X$ (Equations (\ref{eq:kerneleq}--\ref{eq:hKP}))---i.e.\ that for all $\eta\in\fNoff$, $\rho(\eta)\h = 0$.
In fact, this equation follows from the equivariance properties of $\h$ (Proposition \ref{prop:equi}), Theorem \ref{thm:frakN} ($\fNoff = \ker(\bd \h)$), and Proposition \ref{prop:adjequi} ($\fNoff\subset\ker(k_\smbullet)$). From this it is immediate to compute that for all $\eta\in\fNoff$ and $\xi\in\fG$: 
\[
\langle\L_{\rho(\eta)} h_\smbullet,\xi\rangle 
= \langle h, [\xi,\eta]\rangle + k_\smbullet(\xi,\eta)  
=-\langle\L_{\rho(\xi)} h_\smbullet,\eta\rangle 
=  - \bi_{\rho(\xi)} \langle \bd h_\smbullet,\eta\rangle =0,
\]
and thus conclude that $\A_\pp$ is a Lie algebroid.  
This result also shows that the cocycle $k_\smbullet$ descends to $\fGpoff\simeq\fGredoff$, so that
\[
    \langle\L_{\rho_\pp(\xi_\pp)}h_\pp,\eta_\pp\rangle = \langle h_\pp, [\xi_\pp,\eta_\pp]\rangle + k_\smbullet(\xi_\pp,\eta_\pp).
\qedhere
\]
\end{proof}

The following straightforward corollary gives us a corner-suported description\footnote{Observe that all Lie algebras appearing below are naturally supported on $\pp\Sigma$ when marked by the subscript $\pp$, while the other ones are defined as either subalgebras or quotients of $\fG$ and are therefore supported on $\Sigma$, according to our definition.} of the on-shell corner Lie algebra $\fGred$:
\begin{corollary}[of Theorem \ref{thm:corneralgd}]\label{cor:Fsubset}
The on-shell flux gauge algebra $\fGred\doteq\fG/\fN$ is canonically isomorphic to the following quotients of $\fGredoff\doteq \fG/\fNoff$ and $\fGpoff \doteq \pi_{\pp,\fG}(\fG)$: 
\[
\fGred \simeq \frac{\fGredoff}{\fN/\fNoff} \simeq \frac{\fGpoff}{\wt{\fN}_\pp/\wt{\fN}_\pp^{\off}}.
\]
where $\wt\fN^{(\off)}_\pp \doteq \wtpi_{\pp,\fG}(\fN^{(\off)})$. Hence, $\F=\mathrm{Im}(\iota_\C^*\uh)$ can be thought of as a subset of $(\fGpoff)^*$.
\end{corollary}
\begin{proof}
The first isomorphism follows from the fact that $\fNoff\subset\fN$ are both ideals, together with the definition $\fGred \doteq \fG/\fN$ (Theorem \ref{thm:frakN}, Definition \ref{def:fGp}). For the second isomorphism, start by noticing that the quotient on the rhs makes sense in view of Equation \eqref{eq:Gppiso} ($\fGpoff \doteq \tfGp / \wt\fN^\off_\pp$) and the fact that $\wt{\fN}_\pp^\off\subset\wt\fN_\pp \subset \tfGp$, which is nothing else than the image along $\wt\pi_{\pp,\fG}$ of the relation $\fNoff\subset \fN\subset\fG$. The isomorphism then follows from the diagram \eqref{eq:diagramreductiongroups}.
\end{proof}

The following gives us a sufficient condition (satisfied in many examples of interest) for the algebra $\fGpoff$ (which is generally only \emph{supported} on $\pp\Sigma$) to actually be $\pp\Sigma$-\emph{local} (cf.\ Definition \ref{def:LocLieAlg}):

\begin{proposition}\label{prop:localGredoff}
If $\bH$ is order $1$, then $\fGpoff$ is a local Lie algebra over $\partial\Sigma$.
\end{proposition}
\begin{proof}

By assumption, $\bH$ is order $1$ as a local momentum form. From the identity $\bH = \bHo + d\bh$ and the fact that $\bHo$ is order 0 by definition, it follows that $d\bh$ is order $1$. Therefore, up to $d$-closed terms whose integral over $\pp\Sigma$ would vanish, $\bh$ can be taken to be order 0 as a local map $\fG\to \oloc^{\text{top}-1,1}(\Sigma\times\X)$.

From this it follows that $\wt\varpi_\pp=\int_{\partial\Sigma}\langle\cbd_H \bh,\cbd_V\xi\rangle$ is $C^\infty(\pp\Sigma)$-linear in $\cbd_V\xi$, and therefore that the equations defining the vertical part of the kernel,
\[
    \bi_{\wt{\mathbb{Y}}_{V}}\wt\varpi_\pp = \int_{\partial \Sigma}\langle\cbd\bh,\wt{\mathbb{Y}}_{V}\rangle = 0 ,
\]
are \emph{ultralocal} equations for $\wt{\mathbb{Y}}_{\xi}$ that do not depend on $\xi\in\fG$.

Now, $\wt\fG$ is the space of sections of the induced bundle\footnote{Notice that in principle $\wt{\fG}$ contains all higher transverse jets. However, since $\bh$ is order $0$, all the transverse jet direction will be trivially in the (vertical) kernel of $\wt{\varpi}$, so we can discard them directly.}  $\iota^*_{\pp\Sigma}\Xi$ on $\partial \Sigma$. The vertical part of the kernel equations then define a subbundle $\Xi_\pp\to\pp\Sigma$ of the induced bundle, whose sections are isomorphic to $\wt\fN{}^{\off}_\pp \doteq \wt\pi_{\pp,\fG}(\fNoff)$ (Theorem \ref{thm:corneralgd}). Hence, the reduction 
\[
\fGpoff = \frac{\wt\fG}{\cbd\pi_\text{pre}({}^hN\wt\fG ) }\simeq \frac{\wt\fG}{\wt{\fN}{}^{\off}_\pp} \simeq \Gamma(\pp\Sigma,\Xi_\pp)
\]
is a local Lie algebra.
\end{proof}

In all theories of interest $\rho$ is order 1 as a local map, i.e.\ it descends to a smooth map on the first jet bundle. Then from the local Hamiltonian flow equation 
$
    \bi_{\rho(\cdot)} \bom = \langle\bd \bH, \cdot\rangle
$
for an ultralocal $\bom$ (Definition \ref{def:ultralocal}), we conclude that if $\rho$ is order 1, so is $\bH$. This observation gives us the following:

\begin{corollary}[of Proposition \ref{prop:localGredoff}]\label{cor:localGredoff}
If $\bom$ is ultralocal and $\rho$ is order 1 as a local map, then $\fGpoff$ is a local Lie algebra on $\pp\Sigma$ in the sense of Definition \ref{def:LocLieAlg}.
\end{corollary}

\begin{remark}
Proposition \ref{prop:localGredoff} and its Corollary \ref{cor:localGredoff} prove Theorem \ref{mainthm:Poisson} item \ref{thmitem:main2-ultralocal}.
\end{remark}

\subsection{The Poisson manifold of off-shell corner data and superselections}
In the previous section we analyzed the structure of the space of off-shell corner data $(\A_\pp,\varpi_\pp)$ where $\A_\pp = \X_\pp \times \fGpoff$ and $\varpi= \langle \cbd_H h_\pp , \cbd_V\xi_\pp\rangle$. In particular, we showed that the bulk action $\rho:\A \to T\X$, induces a natural action Lie algebroid structure on $\A_\pp$. We are now going to show that, taken together, the symplectic and Lie algebroid structures uniquely fix a partial Poisson structure on $\X_\pp$ (Definition \ref{def:partialpoisson}) described by a bivector $\Pi_\pp$.

\subsubsection{Preliminary constructions}

To start with, recall that the pre-symplectic reduction procedure we outlined in Theorem \ref{thm:corneralgd} always returns a weak-symplectic manifold by construction. To the weak-symplectic manifold $(\A_\pp,\varpi_\pp)$ one associates the partial cotangent bundle $T^{\mathsf{p}}\A_\pp\subset T^*\A_\pp$ defined as the image of $\varpi_\pp^\flat\colon T\A_\pp\to T^*\A_\pp$ (See diagram \ref{eq:partialdiagram} and Definition \ref{def:Hamfunctions}). 
This endows the set of Hamiltonian functions $C^\infty_{\mathsf{p}}(\A_\pp)\subset C^\infty(\A_\pp)$ with a Poisson bracket directly issued from $\varpi_\pp$ (see Lemma \ref{lem:weaksymplecticpoisson}), which in turn endows $\A_\pp$ with a partial Poisson structure (Definition \ref{def:partialpoisson}). In what follows we will show that $\X_\pp$ is also endowed with a partial Poisson structure, this time induced by the Lie algebroid structure $\rho_\pp\colon\A_\pp\to T\X_\pp$. (cf.\ Remark \ref{rmk:PoissonLieAlg}.)

The splitting $\A_\pp=\X_\pp \times \fGpoff$ induces the splitting of the tangent bundle into the horizontal and vertical bundles $T\A_\pp = H\A_\pp \oplus V\A_\pp$ w.r.t.\ the projection $p_\pp:\A_\pp \to \X_\pp$. 
Dually, one obtains the splitting of the partial subbundle $T^{\mathsf{p}}\A_\pp \doteq V^{\mathsf{p}}\A_\pp \oplus H^{\mathsf{p}}\A_\pp$: where $V^{\mathsf{p}}\A_\pp = \mathrm{Im}(\varpi_\pp^\flat\vert_{H})$ and viceversa $H^{\mathsf{p}}\A_\pp = \mathrm{Im}(\varpi_\pp^\flat\vert_{V})$ (this is because $\varpi_\pp^\flat$ induces the isomorphisms $H\A_\pp\xrightarrow{\simeq\;} V^{\mathsf{p}}\A_\pp$ and $V\A_\pp\xrightarrow{\simeq\;} H^{\mathsf{p}}\A_\pp$).

\begin{remark}\label{rmk:horizontalgcotangent}
The space of horizontal cotangent vectors of $p_\pp:\A_\pp\to\X_\pp$ is defined as $H^*\A_\pp \doteq \Ann( V\A_\pp)$ or, equivalently, as $H^*\A_\pp \doteq \Im(\cbd p_\pp^* : T^*\X_\pp \to T^*\A_\pp)$, and the map
\[
\ell\colon \A_\pp  \times_{\X_\pp}T^*\X_\pp \to H^*\A_\pp, \quad \left(a_\pp, \alpha \right) \mapsto \left( a_\pp , \cbd p_\pp^*\alpha\vert_{a_\pp} \right)
\]
is 1-to-1. This allows us to define $\ell^{-1}(H^{\mathsf{p}}\A_\pp) \subset T^*\X_\pp\times_{\X_\pp}\A_\pp$. Examining the definition of $H^\mathsf{p}\A_\pp$ it is easy to see that it does not depend on $\xi_\pp$ and therefore there exists a subbundle $T^\mathsf{p}\X_\pp \subset T^*\X_\pp$ such that
\[
\ell^{-1}(H^\mathsf{p}\A_\pp) = T^\mathsf{p}\X_\pp \times_{\X_\pp} \A_\pp.\qedhere
\]
\end{remark}

\begin{definition}\label{def:T'P}
The partial cotangent bundle of $\X_\pp$ is the (convenient, weak) subbundle  $T^{\mathsf{p}}\X_\pp \xhookrightarrow{} T^*\X_\pp$ defined by $\ell^{-1}(H^\mathsf{p}\A_\pp)$ as per Remark \ref{rmk:horizontalgcotangent}.
We denote the algebra of $\mathsf{p}$-functions with respect to the subbundle $T^{\mathsf{p}}\X_\pp$ by $C^\infty_{\mathsf{p}}(\X_\pp)\subset C^\infty(\X_\pp)$ (see Definition \ref{def:partialpoisson}).
\end{definition}

In particular, this means that\footnote{In the trivial algebroid case under consideration this can be rewritten as $H^{\mathsf{p}}\A_\pp = T^\mathsf{p}\X_\pp \times \fGpoff$, in the sense that $H^{\mathsf{p}}_{(\phi_\pp,\xi_\pp)}\A_\pp \simeq T^{\mathsf{p}}_{\phi_\pp}\X_\pp \times \{\xi_\pp\}$.} $V \A_\pp \simeq H^{\mathsf{p}}\A_\pp = T^{\mathsf{p}}\X_\pp\times_{\X_\pp}\A_\pp$, and we have the fibre product
\[
\xymatrix{T^{\mathsf{p}}\X_\pp\times_{\X_\pp}\A_\pp \ar[r]\ar[d] & \A_\pp \ar[d]\\ 
T^{\mathsf{p}}\X_\pp \ar[r] & \X_\pp}
\]
from which we conclude the following:

\begin{lemma}\label{lemma:T'P=A}
$T^{\mathsf{p}}\X_\pp \simeq \A_\pp$.
\end{lemma}
\begin{proof}
Fix an element $\eta_\pp\in\fGpoff$, then the previous construction provides us with the following sequence of canonical isomorphisms:
\[
T^{\mathsf{p}}_{\phi_\pp}\X_\pp \simeq H^{\mathsf{p}}_{(\phi_\pp,\eta_\pp)}\A_\pp \simeq V_{(\phi_\pp,\eta_\pp)}\A_\pp  \simeq \fGpoff.
\]
where the first isomorphism is the one used to define $T^{\mathsf{p}}\X_\pp$, the second one is given by $\varpi_\pp^\flat$, whereas the last one is given by $V_{(\phi_\pp,\eta_\pp)}\A_\pp  = T_{\eta_\pp} \fGpoff \simeq \fGpoff$.
It is easy to check that the isomorphism between the first and last term in the above formula does not depend on the choice of $\eta_\pp$.
\end{proof}

\begin{remark}[Functional derivative]\label{rmk:functionalder}
This lemma allows us to define the functional derivative of $g\in C^\infty_{{\mathsf{p}}}(\X_\pp)$ with respect to $h_\pp$ as the corresponding section $\zeta_\pp[g] \in \Gamma(\X_\pp,\A_\pp)$---that is with a $\fGpoff$-valued function over $\X_\pp$:
\[
\bd g \doteq \langle \cbd_H h_\pp, \zeta_{\pp}[g] \rangle,
\qquad 
\frac{\delta g}{\delta h_\pp} \doteq \zeta_{\pp}[g] \in \Gamma(\X_\pp,\A_\pp).
\qedhere
\]
\end{remark}

\begin{remark}
When $\fGpoff$ is local over $\pp\Sigma$ and $\X_\pp$ is itself the fibrewise dual of $\fGpoff$---i.e.\ $\fGpoff = \Gamma(\pp\Sigma,\Xi_\pp)$ and $\X_\pp = \Gamma(\pp\Sigma,\Xi_\pp^*)$ with $h_\pp$ the identity---the above definition of functional derivative coincides with the standard, pointwise, definition. 
\end{remark}

\subsubsection{The partial Poisson structure of $\X_\pp$}

Next, we are going to show that $\X_\pp$ carries a natural partial Poisson structure $\Pi_\pp$. 
Using Remark \ref{rmk:functionalder}, this Poisson structure can easily be guessed by analogy with the  KKS Poisson structure on $(\fGpoff)^*$:
\begin{equation}
    \label{eq:Poissonpp}
\Pi_\pp = \frac12\langle h_\pp , \left[\frac{\delta}{\delta h_\pp} , \frac{\delta}{\delta h_\pp}\right] \rangle + 
\frac12 k_\smbullet\left(\frac{\delta}{\delta h_\pp} , \frac{\delta}{\delta h_\pp}\right).
\end{equation}
Thanks to that analogy, it is then rather obvious that \textit{if} $\Pi_\pp$ is well defined as a 2-derivation on $C^\infty_{{\mathsf{p}}}(\X_\pp)$ then it will satisfy the Jacobi identity. However, at this level it is still quite unclear why this is well-defined at all, i.e.\ why given two $\mathsf{p}$-functions functions $f,g\in C^\infty_{{\mathsf{p}}}(\X_\pp)$, $\Pi_\pp(f,g)$ should again be in $C_{\mathsf{p}}^\infty(\X_\pp)$.

Dualising the isomorphism of Lemma \ref{lemma:T'P=A}, and forgetting momentarily about the difference between $T^{\mathsf{p}}\X_\pp$ and $T^*\X_\pp$, we see that there is an isomorphism between vector fields over $\X_\pp$ and sections of $\A_\pp^*$. More precisely, the isomorphism is between 1-derivations of $\mathsf{p}$-functions over $\X_\pp$ and sections of $\A_\pp^*$.
Since a partial Poisson structure on $\X_\pp$ is a 2-derivation of this kind, it makes sense to extend that isomorphisms to $k$-derivations.
The advantage of this approach is that one can directly use the Lie-algebroid structure of $\A_\pp$ to define a partial Poisson structure on $\X_\pp$.

A convenient way to extend the isomorphism to $k$-derivations is by introducing the shifted vector bundle $\A_\pp[1]$---where the fibre of $\A_\pp$ is turned into a graded-vector space. 

\begin{definition}\label{def:cornershift}
Let $p_\pp^{[1]}:\A_\pp[1]\to \X_\pp$ be the degree-1 shift of the vector bundle $p_\pp:\A_\pp\to\X_\pp$. The space of smooth function on $\A_\pp[1]$ is then the graded-commutative algebra\footnote{We define $\wedge^k\A_\pp\to \X_\pp$ to be the vector bundle with fibre $(\fGpoff)^{\wedge k}$ and we understand $\X_\pp\times\mathbb{R}$ as a trivial line bundle over $\X_\pp$. Then, $L(\wedge^k\X_\pp, \X_\pp \times\mathbb{R})$ is the space of fibrewise linear bundle morphisms that are bounded \cite{kriegl1997convenient}. Cf. Remark \ref{rmk:diffForms}.} (see Definition \ref{def:algebroidforms} and Remark \ref{rmk:diffForms}) 
\[
    C^\infty(\A_\pp[1])\doteq L(\wedge^\bullet \A_\pp , \X_\pp \times \mathbb{R}) \equiv \Omega^\bullet_{\A_\pp}(\X_\pp).
\]
We denote
\begin{enumerate}[label=(\roman*)]
    \item by $c$ the degree $1$ fibre-variable, which we call the \emph{(corner) ghost variable}; consequently, we refer to the internal degree of this graded algebra as the \emph{ghost degree};\footnote{Observe that $c$ is defined as a vector-valued ``tautological'' map $\Gamma(\A_\pp[1])\to\Gamma(\A_\pp)$; as such it has degree $1$ with the convention that, given a vector space $W$, the shifted graded vector space $W[1]$ is concentrated in degree $-1$.\label{fnt:V[1]}}
    \item $C_{k}^\infty(\A_\pp[1]) \subset C^\infty(\A_\pp[1]) $ the set of functions on $\A_\pp[1]$ of ghost degree $k$;
    \item $\varpi_\pp^{[1]}=\langle \cbd_H h_\pp, \cbd_V c\rangle$ the shift of the weak symplectic structure $\varpi_\pp$;
    \item $C_{\mathsf{p}}^\infty(\A_\pp[1])$  the space of Hamiltonian functions with respect to $\varpi_\pp^{[1]}$;
    \item $\{\cdot,\cdot\}^{[1]}: C_{\mathsf{p}}^\infty(\A_\pp[1])^{\otimes 2} \to C_{\mathsf{p}}^\infty(\A_\pp[1])$ the degree $-1$ partial Poisson structure on $\A_\pp[1]$ associated to $\varpi_\pp^{[1]}$;
    \item $\rho_\pp^{[1]}: \A_\pp[1] \to T[1]\X_\pp$ the anchor map $\rho_\pp$ on the shifted algebroid.\footnote{The degree-1 property of $\rho_\pp^{[1]}$ corresponds the linearity of $\rho_\pp$.}\qedhere
\end{enumerate}
\end{definition}

\begin{remark}\label{rmk:offPoissonalg}
For future reference let us recall the well-known fact that $\{\cdot,\cdot\}^{[1]}$ defines a (odd) Poisson \emph{algebra} over $C^\infty_{\mathsf{p}}(\A_\pp[1])$, i.e.\ that the bracket of two Hamiltonian functions is also a Hamiltonian function---that is, denoting for all $F\in C^\infty_{\mathsf{p}}(\A_\pp[1])$ by $\mathbb{X}_F \in \mathfrak{X}(\A_\pp[1])$ the corresponding Hamiltonian vector field, $\bi_{\mathbb{X}_F}\varpi_\pp^{[1]} = \cbd F$:\footnote{Indeed: $\bi_{[\mathbb{X}_F, \mathbb{X}_G]} \varpi_\pp^{[1]} = [ \L_{\mathbb{X}_F},\bi_{\mathbb{X}_G}] \varpi_\pp^{[1]} = \L_{\mathbb{X}_f} \cbd G = \cbd \bi_{\mathbb{X}_F} \bi_{\mathbb{X}_G} \varpi_\pp^{[1]}$ and $ \bi_{\mathbb{X}_F} \bi_{\mathbb{X}_G} \varpi_\pp^{[1]}\doteq \{ F, G\}^{[1]}$.}
\begin{equation}
\mathbb{X}_{\{F,G\}^{[1]}} = [\mathbb{X}_F,\mathbb{X}_G].
\label{eq:Poissonalgrep}    
\end{equation}
Furthermore, $\{\cdot,\cdot\}^{[1]}$ lowers the homogeneous degree of the functions by $1$, so that if $|F^{(k)}|=k$ and $|F^{(j)}|=j$ we have $|\{F^{(k)},F^{(j)}\}^{[1]}|=k+j-1$.
In particular $\{F_1^{(0)}, F_2^{(0)}\}^{[1]} \equiv 0$.
\end{remark}

We now recall a classic result linking cohomological vector fields on $\A_\pp[1]$ and the associated Chevalley--Eilenberg algebroid differential, adapted to the case at hand.
\begin{theorem}[BRST]\label{thm:BRST}
Let $(\Omega^\bullet_{\A_\pp}(\X_\pp),\delta_{\A_\pp})$ be the Lie algebroid complex defined in Appendix \ref{def:algebroidforms}. The definition
\[
C^\infty(\A_\pp[1]) \equiv \Omega^\bullet_{\A_\pp}(\X_\pp)
\]
induces a cochain complex structure $(C^\infty(\A_\pp[1]) , \mathbb{Q}_\pp)$,
\[
    (C^\infty(\A_\pp[1]) , \mathbb{Q}_\pp) \simeq (\Omega^\bullet_{\A_\pp}(\X_\pp),\delta_{\A_\pp}),
\]
where $\mathbb{Q}_\pp\in\mathfrak{X}[1](\A_\pp[1])$ is a nilpotent degree-1 vector field,
\[
    [\mathbb{Q}_\pp,\mathbb{Q}_\pp]=0,
\]
which decomposes as the following sum of a $p_\pp^{[1]}$-horizontal and a $p_\pp^{[1]}$-vertical vector fields\footnote{More explicitly: $\rho_\pp^{[1]}(c)\vert_{(\phi_\pp,c)}$ is here understood as the horizontal lift to $T[1]_{(\phi_\pp,c)}\A_\pp[1]$ of $\rho_\pp^{[1]}(c)\in T[1]_{\phi_\pp}\X_\pp$, whereas $\ad(c)\vert_{(\phi_\pp,c)}$ is a vertical vector in $V[1]_{(\phi_\pp,c)}\A_\pp[1] \simeq T[1]_{c} \fGpoff[1]$.}
\[
\mathbb{Q}_\pp  \doteq 
\rho_\pp^{[1]}(c) - \frac12 \ad(c).
\qedhere
\]
\end{theorem}

In the following two lemmas, we establish the canonical isomorphisms between the space of Hamiltonian functions on $\A_\pp[1]$ of degree 0 and the space of $\mathsf{p}$-functions over $\X_\pp$, as well as the sought (canonical) isomorphism between the space of Hamiltonian functions over $\A_\pp[1]$ of degree $k>0$ and the space of multi-derivations over $\X_\pp$. The latter is in fact an isomorphism of Gesternhaber algebras.

\begin{lemma}\label{lemma:deg0}
$C_{\mathsf{p}}^\infty(\X_\pp) \simeq C_{\mathsf{p}, 0}^\infty(\A_\pp[1])$.
\end{lemma}
\begin{proof}
First, notice that the set of $p_\pp^{[1]}$-basic functions on $\A_\pp[1]$ is the set of functions in $C^\infty(\A_\pp[1])$ which are constant along the vertical $\fGpoff[1]$-direction. These can be equivalently characterized either as the set of degree-zero function over $\A_\pp[1]$ or as the set of functions over $\X_\pp$, i.e.\ $C^\infty(\X_\pp) \simeq C^\infty_{0}(\A_\pp[1])$.

The lemma follows from the Definition \ref{def:T'P} of $T^{\mathsf{p}}_{\phi_\pp}\X_\pp$ as 
\[
T^{\mathsf{p}}_{\phi_\pp}\X_\pp \simeq H^{\mathsf{p}}_{(\phi_\pp,\xi_\pp)} \A_\pp = \Im(\varpi_\pp^\flat(\phi_\pp,\xi_\pp)\vert_V).
\]
Indeed, this means that a $p_\pp^{[1]}$-basic function $g$ is a $\mathsf{p}$-function over $\X_\pp$ iff it is Hamiltonian in $(\A_\pp[1],\varpi_\pp^{[1]})$---i.e.\ if and only if there exists a vector field $\mathbb{X}_g\in \mathfrak{X}(\A_\pp[1])$ such that (cf.\ Remark \ref{rmk:functionalder})
\[
    \bi_{\mathbb{X}_g}\varpi_\pp^{[1]} = \cbd g
    \; \iff \; 
    \bd g = \langle \bd \h_\pp,\mathbb{X}_g(c)\rangle
    \;\iff\;
    \mathbb{X}_g(c) = \frac{\delta g}{\delta \h_\pp} \in \Gamma(\A_\pp).
\]
Notice that such an $\mathbb{X}_g\in \Gamma(V\A_\pp[1]) \subset \mathfrak{X}(\A_\pp[1])$ is $p_\pp^{[1]}$-vertical.
\end{proof}

Before proceeding to the second lemma, we need to introduce:
\begin{definition}
The space of $k$-derivations of the algebra $C^\infty_{\mathsf{p}}(\X_\pp)$ of $\mathsf{p}$-functions over $\X_\pp$ is ($k\geq 1$)
\[
\mathrm{Der}^k_{{\mathsf{p}}}(\X_\pp) \doteq \mathrm{Der}^{k}(C^\infty_{{\mathsf{p}}}(\X_\pp))
\]
i.e.\ the subset of multilinear maps $L(\wedge^k C^\infty_\mathsf{p}(\X_\pp), C^\infty_\mathsf{p}(\X_\pp))$ that satisfy the Leibniz rule in each of their entries.
\end{definition}

\begin{lemma}\label{lemma:Gestern}
Let $k\geq1$. The map
\[
C^\infty_{\mathsf{p},k} (\A_\pp[1]) \to \mathrm{Der}^k_{{\mathsf{p}}}(\X_\pp),
\quad
F^{(k)}\mapsto \Pi_{F^{(k)}}
\]
defined for all $g_1, \cdots,g_k \in C^\infty_{\mathsf{p}}(\X_\pp)$ as
\[
    \Pi_{F^{(k)}} (g_1, \cdots, g_k) \doteq \{\{  \cdots \{ F^{(k)},g_1\}^{[1]},  \cdots \}^{[1]},g_k\}^{[1]}
\]
is an isomorphism of Gesternhaber algebras:
\[
( C^\infty_{\mathsf{p},k} (\A_\pp[1]), \{\cdot,\cdot\}^{[1]}) \simeq ( \mathrm{Der}^k_{\mathsf{p}}(\X_\pp), [\cdot,\cdot]_\text{SN} )
\]
where $[\cdot,\cdot]_\text{SN} : \mathrm{Der}^j_{\mathsf{p}}(\X_\pp)\times\mathrm{Der}^k_{\mathsf{p}}(\X_\pp) \to \mathrm{Der}^{j+k-1}_{\mathsf{p}}(\X_\pp)$ is the Schouten–Nijenhuis bracket.
\end{lemma}
\begin{proof}
First we need to show that the definition of $\Pi_{F^{(k)}}$ is well-posed as a derivation on $C^\infty_{\mathsf{p}}(\X_\pp)$. Thanks to Lemma \ref{lemma:deg0} $g_1, \cdots,g_k$ can be identified with degree-0 Hamiltonian functions over $\A_\pp[1]$, so that $\Pi_{F^{(k)}}(g_1,  \cdots, g_k)$ is itself a degree-0 Hamiltonian function over $\A_\pp[1]$ (cf.\ Remark \ref{rmk:offPoissonalg}), which can finally be re-identified with an element of $C^\infty_{\mathsf{p}}(\X_\pp)$. Moreover, given functions $g_1, \cdots, g_k \in C^\infty_{\mathsf{p}}(\X_\pp)$, we need the following properties:
\begin{enumerate}[label=(\textit{\roman*})]
    \item \label{it:derivation_antisymm} $\Pi_{F^{(k)}}(g_1,  \cdots, g_k)$ is completely antisymmetric in its $k$ entries;
    \item \label{it:derivation_leibniz} it satisfies the Leibniz rule in each of its $k$ entries.\footnote{Viz. $\Pi_{F^{(k)}}(g_1g_1',  \cdots, g_k) = \Pi_{F^{(k)}}(g_1,  \cdots, g_k)g_1' + \Pi_{F^{(k)}}(g_1',  \cdots, g_k)f_1.$}
\end{enumerate}

Statement \ref{it:derivation_antisymm} is a consequence of $\varpi_\pp^{[1]} = \langle \cbd_H h_\pp, \cbd_V c\rangle$ being symplectic of degree 1 and $F^{(k)}$ being homogeneous of degree $k$ in the \emph{anticommuting} variable $c$. 
Finally, \ref{it:derivation_leibniz} follows from the Leibniz rule for $\{\cdot,\cdot\}^{[1]}$ and the fact that $\{g_i, g_j\}^{[1]}\equiv0$ (Remark \ref{rmk:offPoissonalg}). Therefore $\Pi_{F^{(k)}}$ is well defined as a multi-derivation on $C^\infty_\mathsf{p}(\X_\pp)$. 

Proving that $(C^\infty_{\mathsf{p}}(\A_\pp[1]),\{\cdot,\cdot\}^{[1]}) \hookrightarrow (\mathrm{Der}_{\mathsf{p}}(\X_\pp), [\cdot,\cdot]_\text{SN})$ as Gerstenhaber algebras is  straightforward---albeit tedious---and will be left to the reader. Here we simply notice that it relies on the following calculation:
\begin{align}
    [\Pi_{F^{(j)}}, \Pi_{G^{(k)}} ]_\text{SN}(g_1,  \cdots, g_{j+k-1}) &=  \cdots = \{\{\{\{ F^{(j)}, G^{(k)}\}^{[1]}, g_1\}^{[1]},   \cdots \}^{[1]}, g_{j+k-1}  \}^{[1]} \notag\\
    &= \Pi_{\{ F^{(j)}, G^{(k)}\}^{[1]}}(g_1,  \cdots, g_{j+k-1})
    \label{eq:gestern}
\end{align}
which generalises Equation \eqref{eq:Poissonalgrep} to $j,k\neq 1$ (observe: $\Pi_{F^{(1)}} \equiv \mathbb{X}_{F^{(1)}}$).

We now want to prove the opposite injection, that is $ \mathrm{Der}^k_{\mathsf{p}}(\X_\pp)\hookrightarrow C^\infty_{\mathsf{p},k}(\A_\pp[1])$. Formally, this is given by:
\[
\Pi^{(k)} \mapsto F^{(k)}\doteq \Pi^{(k)}( \langle h_\pp, c\rangle,...,\langle h_\pp,c\rangle).
\]
Below we discuss the meaning of this formula more carefully.

Let us start from the case $k=1$, and consider $\Pi^{(1)}\in \mathrm{Der}^1_{\mathsf{p}}(\X_\pp)$.
For a given $\xi_\pp\in\fGpoff$, define the function $q_{\xi_\pp}\in C^\infty(\X_\pp)$ by $q_{\xi_\pp}(\phi_\pp) \doteq \langle h_\pp(\phi_\pp),\xi_\pp\rangle$. Using Lemma \ref{lemma:deg0} it is immediate to see that $q_{\xi_\pp}$ is a $\mathsf{p}$-function, and therefore, since $\Pi^{(1)}$ is a derivation of $\mathsf{p}$-functions, the function $\Pi^{(1)}(q_{\xi_\pp})$ will also be a $\mathsf{p}$-function.
Then, define $Q_{\Pi^{(1)}}\in \Gamma(\X_\pp,\A^*_\pp)$ as
\[
Q_{\Pi^{(1)}}(\phi_\pp,\xi_\pp)\doteq\Pi^{(1)}(q_{\xi_\pp}(\phi_\pp)).
\]
Next, given the tautological, degree-1, function $c\colon \fG[1]\to \fG$, we can precompose $F$ with $c$ and obtain the degree $1$ function\footnote{More simply put, we replace $\xi \leadsto c$, as a degree-$1$ variable.} 
\[
F^{(1)}\doteq Q_{\Pi^{(1)}}\circ c\in C^\infty(\A[1]).
\]
It is immediate to check that for any $g\in  C^\infty_{\mathsf{p}}(\X_\pp)$, the expression $\{F^{(1)}, g\}$ is well-defined thanks to the fact that $\Pi^{(1)}(q_{\xi_\pp})$ is a $\mathsf{p}$-function, and is equal to
\[
\{ F^{(1)}, g\}^{[1]} = \Pi^{(1)} (g).
\]
This proves $\mathrm{Der}^1_{\mathsf{p}}(\X_\pp)\hookrightarrow C^\infty_{\mathsf{p},1}(\A_\pp[1])$. 

Higher derivations can be dealt similarly by considering $Q_{\Pi^{(k)}}\in L(\wedge^k\A_\pp,\X_\pp \times \mathbb{R})$, $Q_{\Pi^{(k)}}(\phi_\pp,\xi^1_\pp, \cdots, \xi^k_\pp) \doteq \Pi^{(k)}(q_{\xi_\pp^1}, \dots, q_{\xi_\pp^k})$. 
Precomposing with $c$, one can view an alternating linear function of $k$ vectors $(\xi_\pp^1,\dots \xi^k_\pp)$ as the evaluation of a degree $k$ function of the degree $1$ variable $c$ on $(\xi_\pp^1, \dots,\xi^k_\pp)$, and one finally obtains a degree $k$ function $F^{(k)} \in C^\infty_{\mathsf{p},k}(\A[1])$ with the sought property.
\end{proof}

\begin{corollary}[Lemma \ref{lemma:Gestern}]\label{rmk:CME}
To any degree-$2$ Hamiltonian function $F^{(2)}$ over $\A_\pp[1]$ there corresponds a 2-derivation $\Pi_{F^{(2)}}$ over $C^\infty_{\mathsf{p}}(\X_\pp)$. This 2-derivation satisfies the Jacobi identity if and only if $F^{(2)}$ satisfies the Classical Master Equation (CME), that is
\[
[\Pi_{F^{(2)}},\Pi_{F^{(2)}}]_\text{SN}=0
\quad\iff\quad
\{F^{(2)},F^{(2)}\}^{[1]} = 0.
\qedhere
\]
\end{corollary}

Then, in our quest for a partial Poisson structure on $\X_\pp$, the next goal is to find a natural candidate for $F^{(2)}$. So far we have only used the symplectic structure on $\A_\pp$.
We can now leverage the Lie algebroid structure of $\A_\pp$ and in particular its anchor map $\rho_\pp: \A_\pp \to T\X_\pp$. 
In fact, as the anchor map is given by the Lie algebra action $\fGpoff\to\mathfrak{X}(\X_\pp)$,  the shifted anchor $\rho_\pp^{[1]}: \A_\pp[1] \to T[1]\X_\pp$ can itself be re-interpreted as the map
\[
\rho_\pp^{[1]}: \fGpoff[1] \to \mathfrak{X}[1](\X_\pp),\quad c\mapsto \rho_\pp^{[1]}(c).
\]

This observation can be used to show that the anchor $\rho_\pp$ fixes a partial Poisson structure on $\X_\pp$. 

\begin{lemma}
The BRST vector field $\mathbb{Q}_\pp$ (Theorem \ref{thm:BRST}) is the unique Hamiltonian vector field on $(\A_\pp[1],\varpi_\pp^{[1]})$ of degree 1 whose horizontal part is $\rho_\pp^{[1]}(c)$. Its Hamiltonian is the degree-2 function $S_{\pp}  \in C^\infty_{\mathsf{p}}(\A_\pp[1])$,
\[
    \mathbb{Q}_\pp = \{S_{\pp}, \cdot \}^{[1]},
    \qquad
    S_{\pp} \doteq \tfrac12 \langle \h_\pp, [c,c]\rangle + \tfrac12 k_\smbullet(c,c).
\]
It satisfies the CME:
\[
    \{ S_{\pp}, S_{\pp} \}^{[1]}=0.
\]
\end{lemma}
\begin{proof}
We start by observing that the only \emph{constant} homogeneous functions of degree $k>0$ is the vanishing function. 

Now, let $\mathbb{X}'_\pp$ be a Hamiltonian vector field of  degree-1 such that its horizontal part is $\rho_\pp^{[1]}$. Then $\mathbb{X}_\pp \doteq \mathbb{X}'_\pp-\mathbb{Q}_\pp$ is of degree 1, purely vertical and also Hamiltonian. We want to show that $\mathbb{X}_\pp$ must be zero, i.e.\ that $\mathbb{X}_\pp(c)=0$. Indeed, from its Hamiltonian property, $0 = \L_{\mathbb{X}_\pp}\varpi_\pp = - \langle \cbd h_\pp, \cbd \mathbb{X}_\pp(c) \rangle$, which is possible iff $\cbd \mathbb{X}_\pp(c)=0$ that is iff $\mathbb{X}_\pp(c)$ is constant. Since $\mathbb{X}_\pp(c)$ is degree 2, by the initial observation, we deduce $\mathbb{X}_\pp(c)=0$.

The CME follows from the fact that $\mathbb{Q}_\pp$ is the Hamiltonian vector field of $S_\pp$, i.e.\ $\bi_{[\mathbb{Q}_\pp,\mathbb{Q}_\pp]}\varpi_\pp^{[1]} = \cbd\{S_\pp,S_\pp\}^{[1]}$ (cf.\ Equation \eqref{eq:Poissonalgrep}) and again the fact that there are no nonvanishing constant homogeneous functions of degree $k>0$.
\end{proof}

Piecing all the elements together we arrive at the following statement, proving Theorem \ref{mainthm:Poisson}, item \ref{thmitem:main2-algebroids}:
\begin{theorem}[Corner Poisson structure]\label{thm:Poissontheorem}
The corner configuration space $\X_\pp$ equipped with the 2-derivation of the algebra of $\mathsf{p}$-functions $C_{\mathsf{p}}^\infty(\X_\pp)$
\[
    \Pi_\pp(f,g) \doteq \{\{S_{\pp},f\}^{[1]},g\}^{[1]}
\]
is a partial Poisson manifold. Moreover, $\Pi_\pp^\sharp\colon T^{\mathsf{p}}\X_\pp \to T\X_\pp$ defines an algebroid structure, and $\Pi_\pp^\sharp = \rho_\pp$, making $\A_\pp\simeq T^\mathsf{p}\X_\pp$ into a symplectic Lie algebroid (see Remark \ref{rmk:sympliealgbd}).
\end{theorem}

\begin{remark}[Symplectic Lie algebroid]\label{rmk:sympliealgbd}
With reference to Appendix \ref{app:CEcoh}, we observe that through the isomorphism of Theorem \ref{thm:BRST}, the degree-2 Hamiltonian function $S_{\pp}$ is identified with the $\A_\pp$-2-form $\sigma_{\A_\pp}\in\Omega^{2}_{\A_\pp}(\X_\pp)$, such that
\[
    \sigma_{\A_\pp}(s_1,s_2) = \langle h_\partial,\llbracket s_1,s_2\rrbracket\rangle  + k_\smbullet(s_1,s_2)
\]
for any two sections $s_i\in\Gamma(\X_\pp,\A_\pp)$. 
It is then immediate to verify that $\sigma_{\A_\pp}$ is in fact $\delta_{\A_\pp}$-exact,
\[
        \sigma_{\A_\pp} = \delta_{\A_\pp} \tau_{\A_\pp} \quad\text{for}\quad \tau_{\A_\pp}\doteq\frac12\langle[h_\smbullet]_\pp,\cdot\,\rangle,
\]
and therefore to conclude that $S_{\pp}$ satisfies the CME:
\[
    \delta_{\A_\pp} \sigma_{\A_\pp} = 0
    \quad\iff\quad
    \mathbb{Q}_\pp S_{\pp} = 0
    \quad\iff\quad
    \{ S_{\pp}, S_{\pp} \}^{[1]} = 0.
\]
A Lie algebroid $\A_\pp$ equipped with a $\delta_{\A_\pp}$-closed $\A_\pp$-2-form $\sigma_{\A_\pp}\in\Omega^2_{\A_\pp}$ is called a \emph{symplectic Lie algebroid}.\footnote{Notice that this nomenclature does \emph{not} directly refer to the symplectic structure $\varpi_\pp$ existing on $\A_\pp$. This symplectic structure, however, is utilized in the construction of the bracket of the closely related algebra $(C^\infty_\mathsf{p}(\A_\pp[1]),\{\cdot,\cdot\}^{[1]})$.}
\end{remark}

\begin{remark}
From \cite[Prop. 2.4.5]{PelletierCabau} we have that the collection of $\Delta_\phi=\Pi_\pp^\sharp(T^{\mathsf{p}}_\phi\X_\pp)\subset T_{\phi}\X_\pp$ defines an involutive distribution whenever $\X_\pp$ is partial Poisson, $\Delta \simeq T^{\mathsf{p}}\X_\pp/\mathrm{ker}(\Pi_\pp^\sharp)$. (Involutivity comes from the fact that $\Pi^\sharp_\pp\equiv \rho_\pp$ is an algebroid anchor.) 

Moreover, the Poisson structure $\Pi_\pp$ defines a skew-symmetric bilinear form $\Omega_\phi:(T^{\mathsf{p}}_\phi\X)^{\wedge 2} \to \mathbb{R}$:
\[
\Omega_\phi(\alpha,\beta) \doteq \Pi_\pp(\alpha,\beta) \equiv \{f,g\}, \qquad \alpha=\bd f, \quad \beta=\bd g,
\]
for $f,g\in C^\infty_\mathsf{p}(\X_\pp)$. Then, globally, for any $\alpha,\beta\in\Gamma(T^\mathsf{p}\X_\pp)$, one has $\Omega(\alpha,\beta) = \langle \alpha, \Pi_\pp^\sharp(\beta)\rangle$. Since $\Pi_\pp^\sharp$ identifies $\Delta$ with $T^{\mathsf{p}}\X/\ker(\Pi^\sharp_\pp)$, $\Omega_\phi$ induces a skew-symmetric, (weakly) nondegenerate bilinear form on $\Delta$ (see \cite[Definition 2.4.3]{PelletierCabau} and subsequent discussion). We then assume that integral manifolds exist to thus obtain the symplectic leaves of the partial Poisson structure $\Pi_\pp$.
\end{remark}

\begin{corollary}\label{cor:sympleaves}
The symplectic leaves of $\Pi_\pp$ are in 1-1 correspondence with the connected components of the preimage of a coadjoint orbit $\mathcal{O}_{f}\subset (\fGpoff)^*$ along $\h_\pp$.
\end{corollary}
\begin{proof}
Observe that every symplectic leaf, i.e.\ every integral manifold of the characteristic distribution of the (partial) Poisson structure, $\Delta = \Pi_\pp^\sharp(T^{\mathsf{p}}\X_\pp)$, coincides with an orbit of the action of the Lie group $\Gpoff$ on $\X_\pp$, owing to $\rho_\pp = \Pi_\pp^\sharp$ and $T^{\mathsf{p}}\X_\pp\simeq \A_\pp$ (Lemma \ref{lemma:T'P=A} and Theorem \ref{thm:Poissontheorem}). However, by equivariance of $\h_\pp$, each symplectic leaf maps surjectively onto $\mathcal{O}_f$, with discrete kernel counting connected components, i.e.\ every $\Gpoff$-orbit in $\X_\pp$ is diffeomorphic to a coadjoint orbit in $(\fGpoff)^*$
\end{proof}

\subsection{Equivalence with on-shell superselection sectors}\label{sec:onshellequivalence}
In this section we tie the two-stage reduction of Sections \ref{sec:bulkreduction} and \ref{sec:cornerred} with the off-shell construction we just outlined. Notice that the observations leading to Corollary \ref{cor:sympleaves} in particular imply that every symplectic leaf, being an orbit of the exponential group $\Gpoff$ in $\X_\pp$, has a transitive action of $\Gpoff$. This justifies the following definition:

\begin{definition}\label{def:offshellSSS}
Let $f\in \calF^\off$ and denote by $\mathcal{O}_f$ its co-adjoint/affine orbit. The \emph{q-th off-shell corner superselections} $\S^\pp_{[f],q}$ are the symplectic leaves of $\Pi_\pp$ or, equivalently, the $\Gpoff$-orbits in $\X_\pp$.
\end{definition}

\begin{remark}
According to Corollary \ref{cor:sympleaves}, the off-shell corner superselection sectors can be equivalently presented as the $q$-th path-connected components of the preimage of a coadjoint orbit $\mathcal{O}_f\subset (\fGpoff)^*$ along the map $\h_\pp\colon \X_\pp \to (\fGpoff)^*$.
\end{remark}

We can look at the subset $\C_\pp=\pi_{\pp,\X}(\C)$ in $\X_\pp$ obtained by projecting $\C$ along $\pi_{\pp,\X}\colon \X \to \X_\pp$ (see Theorem \ref{thm:corneralgd}). Recall that $\X_\pp$ is a union of $q$-th off-shell corner superselections because they are its symplectic leaves (Definition \ref{def:offshellSSS} and Theorem \ref{thm:Poissontheorem}). Then, the quotient of $\X_\pp$ by the action of $\fGpoff$ is covered by the Poisson manifold $\uuC$, and the covering map is given by $\pi_{\pp,\X}\vert_\C$. Indeed, one has:

\begin{proposition}\label{prop:onshellsuperselections}
Assume that the set $\C_\pp\doteq \pi_{\pp,\X}(\C)\subset \X_\pp$ is a smooth submanifold. Then $\C_\pp$ is a $\Gpoff$-invariant submanifold in $\X_\pp$, and a union of symplectic leaves for the partial Poisson structure $\Pi_\pp$.
\end{proposition}
\begin{proof}

Let us look at the diagram
\[
    \xymatrix{
            & \ar[dl]_{\pi_\circ} \C \ar[dd]^{\iota_\C^*\pi_{\pp,\X}}\ \ar@{^(->}[r]^{\iota_\C}  & \X \ar[dd]^{\pi_{\pp,\X}} \ar[r]^{\h_\smbullet} & \dual{\fG}\\
    \uCo\ar[dr]_{\pi_{\pp,\uCo}} &&\\
            & \C_\pp\ \ar@{^(->}[r]^{\iota_{\C_\pp}}    & \X_\pp \ar[r]^{\h_\pp\;\;}& (\fGpoff)^*
    }
\]
It commutes because $\rho(\fGo)\subset \mathrm{ker}(\bd\pi_{\pp,\X}\vert_\C)$. Indeed, from the fact that the on-shell flux annihilator $\fN = \fGo$ is an ideal in $\fG$ (Theorems \ref{thm:frakN} and \ref{thm:fGo}), we have
\[
    \bi_{\rho(\xio)}(\bd \h)\vert_{\phi\in\C} = 0,
\]
which means that $\rho(\fGo)$ is contained in the restriction of the horizontal kernel distribution ${}^hK\X$ to $\C$ (see Theorem \ref{thm:corneralgd}), and $\iota^*_\C\pi_{\pp,\X}$ descends to a map $\pi_{\pp,\uCo}\colon \uCo \to \C_\pp$. Hence, we have the commuting diagram:
\begin{equation}\label{e:commutingdiagramoff-onshell}\xymatrix{
    \uCo \ar[d]^{\pi_{\pp,\uCo}} \ar[r]^{\uh} &\,\fGred^*\,\ar@{^(->}[r]& (\fGredoff)^* \ar[d]^{\simeq} \\
    \C_\pp \ar[rr]^{\h_\pp\vert_{\C_\pp}} && (\fGpoff)^* 
    }
\end{equation}
Take $\phi\in \uS_{[f],q}\subset \uCo$, i.e.\ $\phi\in\uCo$ and $\uh(\phi)\in\mathcal{O}_f$ (Definition \ref{def:fluxSSS}). The commutativity of the diagram \eqref{e:commutingdiagramoff-onshell} implies that $\pi_{\pp,\uCo}(\phi)\in \S^\pp_{[f],n} \subset \h_\pp^{-1}(\mathcal{O}_f)$ for some index-value $n$ labelling the relevant connected component. 
Similarly, owing to the equivariance of $\uh$ and connectedness of $\Gred$, for every $\underline{g}\in\Gred$ we have $\underline{g}\cdot \phi \in \uS_{[f],q}$ and $\pi_{\pp,\uCo}(\underline{g}\cdot\phi)\in \S^\pp_{[f],n}$ for $n$ as above.

Now, each connected component $\S^\pp_{[f],n} \subset \h_\pp^{-1}(\mathcal{O}_f)$ of the preimage along $\h_\pp$ of a coadjoint orbit $\mathcal{O}_f$ is by construction a symplectic leaf of $\Pi_\pp$. Hence the action of $\Gpoff$ on $\S^\pp_{[f],n}$ is transitive and there exists $g_\pp\in\Gpoff$ such that $\pi_{\pp,\uCo}(\underline{g}\cdot\phi)=g_\pp\cdot \pi_{\pp,\uCo}(\phi)$. Indeed, it is easy to check that, in virtue of the commuting diagram \eqref{e:commutingdiagramoff-onshell} and of Theorem \ref{thm:GloecknerNeeb}, this equation is satisfied if one chooses $g_\pp = \underline{g}$ through the isomorphism $\fGpoff\simeq \fGred$.
\end{proof}

We can summarise a few results on superselection sectors and flux orbits through the following diagram:
\[
\xymatrix{
\S_{[f]} \ar@{^(->}[r] \ar[d] & \C \;\ar@{^(->}[r] \ar[d]^{{\pi_\circ}} & \X \ar[dd]^{{\pi_{\pp,\X}}}\\
\uS_{[f]} \ar@{^(->}[r] \ar[d] & \uCo \ar[d]^{{\pi_{\pp,\uCo}}}\\
\pi_{\pp,\uCo}(\uS_{[f]}) \ar@{^(->}[r] &\C_\pp \,\ar@{^(->}[r]&\X_\pp
}
\]

\begin{remark}
We interpret Proposition \ref{prop:onshellsuperselections} as saying that the on-shell superselection sectors can be recovered from an off-shell analysis. At the same time this indicates that---generally speaking---not all off-shell superselections, which only know about $\partial\Sigma$, will be compatible with the specifics of the solutions of the constraint $\bHo$ on the interpolating manifold $\Sigma$.
\end{remark}

We conclude our discussion of corner perspective on superselections with the following result, which proves Theorem \ref{mainthm:Poisson}, item \ref{thmitem:main2-diagram}.

\begin{theorem}[Space of on-shell superselections]\label{thm:spaceofsuperselection}
There exists a commuting diagram of Poisson manifolds
\[
\xymatrix{
 \uCo \ar[d]_{\pi_{\pp,\uCo}}\ar[r]^{\underline{\pi}} & \uuC \ar[d]^{\underline{q}} \\
 \C_\pp \ar[r]^-{q_\pp} & \mathcal{B}
}
\]
where $\mathcal{B}\doteq\C_\pp/\Gpoff$ is the space of leaves of the symplectic foliation of $\C_\pp$ with trivial Poisson structure, and the maps $\underline{\pi}\colon \uCo \to \uuC$ and $q_\pp \colon \C_\pp \to \mathcal{B}$ are the respective group-action quotients by $\Gred$ and $\Gpoff$ respectively. 
We call $\mathcal{B}$ the \emph{space of on-shell superselections}. 
\end{theorem}

\begin{proof}
In the proof of Proposition \ref{prop:onshellsuperselections} we argued that $\pi_{\pp,\uCo}\colon \uCo \to \C_\pp$ is an equivariant map. This means that, if we denote elements of the equivalence class $[\phi]\in\uuC$ by $\underline{g}\cdot [\phi]_\circ$ for  $[\phi]_\circ\in \uCo$ and $\underline{g}\in\Gred$ we have
\[
\pi_{\pp,\uCo}(\underline{g}\cdot [\phi]_\circ) = g_\pp\cdot\pi_{\pp,\uCo}( [\phi]_\circ)
\]
for $g_\pp=\underline{g}$ through the isomorphism $\fGpoff\simeq \fGred$. 
Then, the map $\pi_{\pp,\uCo}$ is well defined on the equivalence class $[\phi]$ as
\[
\underline{q}([\phi]) \doteq [\pi_{\pp,\uCo}(\underline{g}\cdot [\phi]_\circ)]_\pp = [{g_\pp}\cdot \pi_{\pp,\uCo}([\phi]_\circ)]_\pp = [\pi_{\pp,\uCo}([\phi]_\circ)]_\pp
\]
where $q_\pp([\phi_\pp])\equiv[\phi_\pp]_\pp\in\mathcal{B}$ is the equivalence class of a point $\phi_\pp\in \X_\pp$ defined by the action of $\Gpoff$.

Whenever the fibre product $\C_\pp \times_{\mathcal{B}} \uuC$ exists as a smooth manifold, it is the pullback of the diagram given by the two morphisms $\C_\pp\rightarrow \mathcal{B} \leftarrow \uuC$. Hence it satisfies the universal property, whence we conclude that there exists, unique, a surjective map $\uCo \to \C_\pp \times_{\mathcal{B}} \uuC$.
\end{proof}

\subsection{An alternative set of pre-corner data}\label{sec:ultralocality}

We conclude this section with a brief exploration of an alternative construction of pre-corner and corner data, which, at least in the ultralocal case, ends up coinciding with the one studied so far. The advantage of this alternative corner structure is that it can be introduced even in the absence of a flux-map, purely out of the theory's constraint form. We leave a general analysis of these alternative corner data set to future work.

In Definition \ref{def:momentmapSourcedecomp} and Lemma \ref{lemma:checkC} we have interepreted $\bH$ as a local Lagrangian density on the triple product $\oloc^{\text{top},0}(\Sigma\times\A)$ for $\A=\X\times \fG$. Using the decomposition of Theorem \ref{thm:LocFormDec}, we do the same for $\bHo$ and write the total variation 
\[
\cbd \bHo = (\cbd\bHo)_\text{src} + (\cbd\bHo)_\text{bdry}
\]
where $(\cbd\bHo)_\text{bdry}\in d\oloc^{\text{top}-1,0}(\Sigma\times\X,\dual{\fG})$. Therefore, upon integrating over $\Sigma$, we intuitively see that $(\cbd \Ho)_\text{bdry} \doteq \int_\Sigma (\cbd\bHo)_\text{bdry}$ will descend to the corner $\partial \Sigma$. This statement is made precise by the following

\begin{definition}[Constraint corner data]\label{def:precorner2}
Let $\tA_\pp$ be the space of pre-corner data of Definition \ref{def:cornerconf}, with the restriction map $\wtpi\colon \A \to \tA_\pp$.
\begin{enumerate}
    \item The \emph{constraint} pre-corner 1-form on $\widetilde{\A}_\pp$ is  $\wtpi^*_\pp{\wt\theta}^{\circ}_\pp\doteq (\cbd \Ho)_{\text{bdry}}$.
    \item The \emph{constraint} pre-corner 2-form on $\widetilde{\A}_{\partial}$ is $\wt\varpi_\pp^\circ \doteq \cbd \wt{\theta}^\circ_\pp$.
\end{enumerate}
Finally, define the \emph{constraint} corner data $(\A_\pp^\circ,\varpi^\circ_\pp)$ as the pre-symplectic reduction $\tA_\pp/\ker(\wt{\varpi}_\pp^\circ)$, and $\pi^\circ_\pp \colon \A \to \A_\pp^\circ$.
\end{definition}

\begin{remark}
This definition shows that, a priori, there are at least two distinct natural pre-corner structures: one induced by the flux map $\h$ (Definition \ref{def:cornerconf}), the other by treating the constraint\footnote{Or, equivalently, the momentum form $\bH$, cf.\ Equation \eqref{eq:HoHpp}.} form $\bHo$ as a Lagrangian density over $\Sigma\times\A$ (Definition \ref{def:precorner2}).

Note, in particular, that the constraint form $\bHo$, and thus the associated constraint corner data, are meaningfully defined even within Lagrangian field theories with local symmetries which \emph{fail} to conform to the local Hamiltonian framework studied in this article (cf.\ Appendix \ref{app:covariant}).

However, the next proposition shows that at least for local Hamiltonian spaces admitting an \emph{ultralocal} symplectic form, the two pre-corner structures actually coincide, thus allowing the definition of an \emph{intrinsic} pre-corner datum.
\end{remark}

\begin{proposition}\label{prop:fullcornerdata}
If $\bom$ is ultralocal (Definition \ref{def:ultralocal}), then $\wt\varpi_\pp^{\circ}= \wt\varpi_\pp$.
\end{proposition}

\begin{proof}
As a preliminary step, define on $\tA_\pp$ the pre-corner 1-forms $\wt\theta_\pp$ and $\wt\theta_\pp^\text{tot}$, by the following formulas\footnote{Here, as above, we denote $(\cbd H)_\text{bdry} \doteq \int_{\Sigma} (\cbd \bH)_\text{bdry}$ and similarly below for $(\cbd_H\omega)_\text{bdry}$.}
\[
\wtpi_\pp^*\wt\theta_\pp \doteq \cbd_V h \equiv \langle h , \cbd_V \xi\rangle
\quad\text{and}\quad
\wtpi_\pp^*\wt\theta^\text{tot}_\pp \doteq (\cbd H)_\text{bdry}.
\]
Notice that, by the linearity of $\h$ in $\fG$, the pre-corner 1-form $\wt\varpi_\pp$ of Definition \ref{def:cornerconf} is equal to 
\[
\wt\varpi_\pp = \cbd_H \wt\theta_\pp = \cbd \wt\theta_\pp.
\]
Also, since $\bH = \bHo + d \bh$, one has $(\cbd\bH)_\text{bdry} = (\cbd\bHo)_\text{bdry} + d\cbd \bh$ and therefore $\wt\theta_\pp^\text{tot} = \wt\theta_\pp^\circ + \cbd h$---which means that
\begin{equation}
    \label{eq:HoHpp}
\wt\varpi_\pp^\circ = \cbd \wt\theta_\pp^\circ = \cbd \wt\theta_\pp^\text{tot}.
\end{equation}
Thus, to prove the proposition, it is enough to show that, when $\bom$ is ultralocal $\wt\theta_\pp = \wt\theta_\pp^\text{tot}$ or, equivalently, that $(\cbd H)_\text{bdry} = \cbd_V h$.

In order to prove this statement we proceed in two steps. First we show that $(\cbd_V H)_\text{bdry} = \cbd_V h$ and then that $(\cbd_H H)_\text{bdry} = 0$. 

The first step is fully general and readily follows from the decomposition 
\[
(\cbd_V \bH)_\text{bdry} = (\cbd_V \bHo)_\text{bdry} + d\cbd_V\bh
\]
and the fact that $\Ho$ is order 0 (i.e.\ $C^\infty(\Sigma)$-linear in $\xi$, Definition \ref{def:order-kduals}).

The second step, on the other hand, relies on the Hamiltonian flow equation and the of ultralocality of $\bom$:
\[
\cbd_H H=\bi_{\rho(\cdot)}\omega \implies (\cbd_H H)_\text{bdry} = (\bi_{\rho(\cdot)}\omega)_\text{bdry} = 0.
\]
This concludes the proof.
\end{proof}

\begin{corollary}
If $\bom$ is ultralocal, then $\A_\pp^{\circ}\simeq \A_\pp$.
\end{corollary}

\begin{remark}
We note that the construction of what we call ``contraint corner data'' is directly related to the bulk-to-boundary induction procedure used in the BV-BFV framework \cite{CMR1}, inspired by the constructions of Kijowski and Tulczijew \cite{KT1979}. It is unclear to us what this discrepancy of induction procedures might signify for those approaches to field theory on manifolds with boundary.
\end{remark}

\subsection{Yang--Mills theory: corner Poisson structure and superselections}\label{sec:runex-cornerdata}

In this section we are interested in the \emph{off-shell} corner data of YM theory, and off-shell results can be formulated in a unified manner independently of whether $G$ is Abelian or semisimple.

Our off-shell analysis relies on Assumption \ref{ass:AnnIm}, to which we now turn.
Recall that in YM, the flux map is $\langle h,\xi\rangle (E,A) = - \int_{\Sigma} d \tr( E \xi )$, which is obviously basic with respect to the projection $p:T^\vee\Acal\to\mathcal{E}_\pp$, i.e.\ $h(A,E)=p^*h_\pp(E^\pp)$ with $h_\pp(E_\pp)=-\int_{\pp{\Sigma}}\tr(E_\pp\cdot)$. As a consequence,  $h$ does not depend on $A$, and it is linear in $\mathcal{E}_\pp$ by direct inspection. Then $\bd h(A,E)$ is independent of the base point $(A,E)$. This allows us to conclude that the space $\fN_{(A,E)}$ introduced in Remark \ref{rmk:fNphi} is independent of the base point $(A,E)$, and it coincides with the off-shell flux annihilator:
\[
\fN_{(A,E)} \doteq \{ \xi\in\fG \,|\, \langle\bd h(A,E),\xi\rangle = 0 \} = \fNoff.
\]
Hence, the off-shell flux annihilating locus $\mathsf{N}^\off = \{ (A,E,\xi)\in\A\,\colon\, \langle\bd h(A,E),\xi\rangle=0 \}\subset \A$ turns out to be a trivial bundle over $\X$ with fibre $\fNoff$---in compliance with Assumption \ref{ass:AnnIm}.

Next, we introduce the 2-form on $\A = \X \times \fG$
\[
\varpi \doteq \langle \cbd_H h, \cbd_V \xi \rangle = - \int_{\pp\Sigma} \iota^*_{\pp\Sigma}\tr(\cbd_H E \wedge \cbd_V \xi ).
\]
From the previous observations, it is easy to see that the symplectic space of off-shell corner data $\A_\pp$---which is defined through a presymplectic reduction procedure starting from $\varpi$---is given by
\[
\A_\pp = \X_\pp \times \fGpoff
\]
for
\[
\X_\pp \doteq \mathcal{E}_\pp
\qquad\text{and}\qquad
\fGpoff \doteq \Gamma(\pp\Sigma, \iota_{\pp\Sigma}^*\Ad P)
\]
equipped with the symplectic form
\[
\varpi_\pp = -\int_{\pp\Sigma} \tr(\cbd E_\pp \wedge \cbd \xi_\pp).
\]
It is easy to see that in YM theory, $\varpi_\pp = \varpi^\circ_\pp$. Indeed, the latter 2-form is defined by the presymplectic reduction of $\wt\varpi^\circ_\pp=\cbd (\cbd \bHo)_\text{bdry}$, while the identity $(\cbd \bHo)_\text{bdry} =  d \cbd \bh$ can be read off of
\[
\cbd \bHo = \cbd \tr\big(  d_A E \, \xi\big) = \underbrace{ \tr\big( (\cbd E) d_A \xi +  \ad^*(\xi) E\,  \cbd A + (-1)^{\dim\Sigma} d_A E \cbd \xi\big)}_{\text{src}} \underbrace{-  d \tr\big( ( \cbd E)\xi \big)}_\text{bdry}.
\]

These results exemplify Theorem \ref{thm:corneralgd} and Proposition \ref{prop:fullcornerdata}. They also shed light on the isomorphism of Equation \eqref{eq:uGYM} of Section \ref{sec:runex-fluxgaugegroup} which can now be recognized to be a special instance of the general result $\fGred^\text{off} \simeq \fGpoff$---where $\fGpoff$ is a Lie algebra supported on the corner.

The action Lie algebroid structure on $\A_\pp$ is induced from that on $\A$ owing to the fact that
\[
\rho(\xi) E = \ad^*(\xi)\cdot E 
\]
naturally restricts to the corner in the form
\[
\rho_\pp(\xi_\pp) E_\pp = \ad^*(\xi_\pp)\cdot E_\pp.
\]
Notice that $\fGpoff$ is a local Lie algebra with a (ultra)local action on $\X_\pp$.

Also, in this case, Corollary \ref{cor:Fsubset} yields
\[
\F \subset \F^\text{off} \simeq (\fGpoff)^\vee
\]
We also have $\F^\text{off} \simeq \X_\pp$. Since the canonical Poisson structure on $\dual{(\fGpoff)}$ is local, it readily restricts to $(\fGpoff)^\vee$, and, from $\Foff \simeq \X_\pp$, one concludes that $\X_\pp$ is Poisson. Nonetheless, it is instructive to provide a more detailed discussion.

First, we start by noticing that the isomorphism $\X_\pp \simeq (\fGpoff)^\vee$ is a statement of Lemma \ref{lemma:T'P=A} which says that $T^\mathsf{p}\X_\pp \simeq \A_\pp$ since in the present case the partial cotangent bundle $T^\mathsf{p}\X_\pp$ (Definition \ref{def:T'P}) is given by $T^\mathsf{p}_{E_\pp} \X_\pp \simeq \fGpoff$. Summarising:
\begin{equation}
    \label{eq:AppYM}
\A_\pp \simeq \mathcal{E}_\pp \times \fGpoff \simeq T^\vee \mathcal{E}_\pp \simeq T^\mathsf{p}\mathcal{E}_\pp , \qquad \X_\pp \simeq \mathcal{E}_\pp \simeq (\fGpoff)^\vee.
\end{equation}

Recalling Definition \ref{def:partialpoisson} of the space of $\mathsf{p}$-functions $C^\infty_\mathsf{p}(\X_\pp)$ associated to a partial cotangent bundle $T^\mathsf{p}\X_\pp$, 
one sees that in the present case these are the functions $g\in C^\infty(\X_\pp)$ whose functional derivative can be naturally identified with a section of $\A_\pp\to \X_\pp$ (Remark \ref{rmk:functionalder}), i.e.\footnote{Formally, this is a rather natural statement, since $E_\pp\in(\fGpoff)^\vee$ is (locally) given by a $\fg^*$-valued densities.}
\begin{equation}
    \label{eq:pfunctionsYM}
\bd g = \int_{\pp\Sigma} \tr\left( (\cbd E_\pp) \frac{\delta g}{\delta E_\pp}\right), 
\quad\text{for}\quad 
\frac{\delta g}{\delta E_\pp} : \X_\pp \to \fGpoff.
\end{equation}

Conversely, on $\A_\pp$ one defines the $\mathsf{p}$-functions $C_\mathsf{p}^\infty(\A_\pp)$ as the space of Hamiltonian functions with respect to $\varpi_\pp$, which in light of \eqref{eq:AppYM}, can be seen as those functions $F$ whose functional derivatives can be naturally identified with elements of $\fGpoff$ and $(\fGpoff)^\vee$ respectively:
\[
\frac{\delta F}{\delta E_\pp} : \A_\pp \to \fGpoff, \quad\text{and}\quad \frac{\delta F}{\delta \xi_\pp} : \A_\pp \to (\fGpoff)^\vee.
\]
Replacing $\xi_\pp$ with a ghost variable $c$ of degree 1, one maps a function $F\in C^\infty_\mathsf{p}(\A_\pp)$ into a function (denoted by the same symbol) in $C^\infty_\mathsf{p}(\A_\pp[1])$. Any such functions can be developed as a sum of functions of given nonnegative degree $k$, i.e.\ $C^\infty_{\mathsf{p},k}(\A_\pp[1])$---the first three terms in this development are of the form:
\[
F(E_\pp, c) =  F_0(E_\pp)  + \tr( F_1(E_\pp) c) + \tfrac12 \tr( F_2(E_\pp) [c,c]) + \dots
\]
Clearly, $C^\infty_{\mathsf{p},0}(\A_\pp[1])\simeq C^\infty_\mathsf{p}(\X_\pp)$ (Lemma \ref{lemma:deg0}).

Leveraging this construction we can define the shifted symplectic structure on $T^\mathsf{p}[1]\X_\pp \simeq \A_\pp[1]$, \[
\varpi_\pp^{[1]} = \int_{\pp\Sigma} \tr( \cbd_H E_\pp \wedge \cbd_V c)
\]
and thus the associated (shifted) partial Poisson bracket:
\[
\{ E_\pp(x), c(y) \}^{[1]} = \delta(x-y) \quad\text{and}\quad \{ E_\pp(x), E_\pp(y) \}^{[1]} = 0 = \{ c(x),c(y) \}^{[1]}. 
\]

As per Theorem \ref{thm:BRST}, on $C^\infty_\mathsf{p}(\A_\pp[1])$ one has the degree-1 vector field 
\[
\mathbb{Q}_\pp : C^\infty_{\mathsf{p},k}(\A_\pp[1])\to C^\infty_{\mathsf{p},k+1}(\A_\pp[1]), 
\]
\[
\mathbb{Q}_\pp(E_\pp) = \ad^*(c)\cdot E_\pp 
\quad\text{and}\quad
\mathbb{Q}_\pp(c) = -\frac12  [c,c].
\]
We recognise $\mathbb{Q}_\pp$ as being the (corner) BRST operator.
Its nilpotency, $[\mathbb Q_\pp,\mathbb Q_\pp]=0$, is therefore standard.

We now explicitly verify the general result of  Theorem \ref{thm:BRST} that $\mathbb{Q}_\pp$ is Hamiltonian with generator ($\pm$) the degree-2 function:
\[
S_\pp = \frac12\langle h_\pp(E_\pp), [c,c]\rangle = -\frac12 \int_{\pp\Sigma} \tr\left( E_\pp [c,c] \right).
\]
Indeed:
\begin{align*}
\cbd S_\pp 
& = - (-1)^{\pp\Sigma} \int_{\pp\Sigma}  \tr( \cbd E_\pp \tfrac12 [c,c] ) + (-1)^{\partial\Sigma}\tr( \ad^*(c)\cdot E_\pp \cbd c ) \\
& = -\int_{\pp\Sigma} (-1)^{\pp\Sigma+1}  \tr( \cbd E_\pp  \mathbb Q_\pp(c)  ) + \tr( \mathbb{Q}_\pp(E_\pp) \cbd c )   \\
& = -(-1)^{\pp\Sigma}\bi_{\mathbb Q_\pp} \varpi_\pp^{[1]}
\end{align*}
which is equivalent to $\{S_\pp, \cdot\}^{[1]} = \mathbb{Q}_\pp$.

Finally, as per Theorem \ref{thm:Poissontheorem}, from $S_\pp$ we can define the partial Poisson structure $\Pi_\pp$ acting on the space of $\mathsf{p}$-functions $C^\infty_\mathsf{p}(\X_\pp)$, by setting:
\begin{align*}
\Pi_\pp(f,g) & \doteq \{ \{ S_\pp , f\}^{[1]}, g\}^{[1]} 
\end{align*}
From this, recalling $\mathsf{p}$-functions' property \eqref{eq:pfunctionsYM}, one can easily compute
\begin{align*}
\Pi_\pp(f,g) 
& = \{\rho_\pp(c) f , g\}^{[1]}  
= \rho_\pp\left( \frac{\delta g}{\delta E_\pp} \right) f \\
& = \int_{\pp\Sigma} \tr\left( \ad^*\left(\frac{\delta g}{\delta E_\pp}\right) \cdot E_\pp \ \frac{\delta f}{\delta E_\pp}\right)  
= - \int_{\pp\Sigma} \tr\left( E_\pp \left[ \frac{\delta f}{\delta E_\pp}, \frac{\delta g}{\delta E_\pp} \right] \right),
\end{align*}
where the formula at the end of the first line substantiates the general expression $\Pi_\pp^\sharp = \rho_\pp$ via $\A_\pp\simeq T^\mathsf{p} \X_\pp$, whereas at the end of the second line we recognise an expression in that exemplifies Equation \eqref{eq:Poissonpp}. (This construction should be compared with, e.g.\ \cite[Section 2.2.2]{CanepaCattaneo}.)

To conclude we notice that from the above formula it is clear that the superselections $\S^\pp_{[f]}$, defined as the symplectic leaves of $\Pi_\pp$, are the coadjoint orbits of $E_\pp \in \mathcal{E}_\pp \simeq (\fGpoff)^\vee \simeq \Foff$. 

On-shell, only those orbits in $\mathcal{E}_\pp$ which can be identified with elements of $\F\subset\Foff$ satisfy the compatibility requirements imposed by the cobordism $(\Sigma,\pp\Sigma)$ and the constraint equation $\bHo = d_A E = 0$. This distinction is not relevant in semisimple YM theory at irreducible configurations, but becomes important e.g. in the Abelian theory, where the (integrated) {{Gauss}} law must be satisfied by all on-shell configurations. This restricts the available superselections to those lying in $\mathcal{E}^\text{ab}_\pp$, whose properties depend on the cobordism $\Sigma$, as emphasized in Section \ref{sec:runex-fluxannihilators} and in particular Equation \eqref{eq:EppAb}.

\medskip

We now turn our attention onto Theorem \ref{thm:spaceofsuperselection}, and thus go back to describing the constraint- and fully-reduced degrees of freedom of YM theory by means of local trivializations (gauge fixings) of $\uCo \to \uuC$.

We investigate this question under the hypothesis that the on-shell isotropy locus is a trivial vector bundle, $\mathsf{I}_\rho = \C \times \mathfrak{I}$ (i.e.\ either in the Abelian, or the irreducible semisimple cases).
Accoding to Proposition \ref{prop:hamactionGp}, the triviality of $\mathsf{I}_\rho$ guarantees that the action of $\Gred$ is \emph{free} on $\uCo$, which can be thought of as the principal bundle
\[
\Gred \hookrightarrow \uCo \to \uuC.
\]

Now, in the irreducible semisimple case (resp. Abelian),  $\mathcal{E}_\pp^{(\mathrm{ab})}$ corresponds precisely to the space of on-shell corner fields $\C_\pp$ (Proposition \ref{prop:onshellsuperselections}) which, through $h_\pp$, can here be also identified with the space of \emph{on-shell} fluxes:
\[
 \C_\pp = \mathcal{E}_\pp^{(\mathrm{ab})}\simeq \F = \Im(\iota_\C^* \h_\smbullet).
\]
Thanks to this identification, Equations \eqref{eq:CtoP-YM} and \eqref{eq:CfibrationAb} can both be written as
\[
\C \simeq \X_\rad \times \F, 
\]
where $\X_\rad \to \Acal$ is a fibration with fibre $\mathcal{H}_A$ (resp. $\mathcal{H}$ in the Abelian case).
Using the fact that the constraint gauge algebra coincides with the on-shell flux annihilator ($\fGo = \fN$, Theorem \ref{thm:fGo}), we showed that $\mathcal{E}_\pp \simeq \F$ is invariant under the action of $\Go$ (Proposition \ref{prop:hamactionGp}). Therefore, from the previous equation we deduce that
\[
\uCo \simeq \underline{\X}{}_\rad \times \F, 
\qquad
\underline{\X}{}_\rad \doteq \X_\rad / \Go.
\]

In the cases we are investigating, it is easy to see that $\Gred$ acts freely on $\Acal/\Go$: in the irreducible semisimple case because the action $\G\circlearrowright\Acal$ is free, and in the Abelian case because the only stabilising elements are constant gauge transformations---which are already part of $\Go$ (Proposition \ref{prop:Gausslaw}, Section \ref{sec:runex-fluxgaugegroup}). Therefore, a fortiori, $\Gred$ must act freely also on $\underline{\X}{}_\rad$. From this, we deduce that $\underline{\X}_\rad$ is the principal bundle
\begin{equation}
    \label{eq:uCo-fibration}
\Gred \hookrightarrow \underline{\X}{}_\rad \to \uuC{}_\rad, 
\qquad
\uuC{}_\rad \simeq \X_\rad /\G.
\end{equation}
Therefore, from this and the previous decomposition of $\uCo$, we find that $\uCo$ is a fibre bundle over $\uuC{}_\rad$:
\[
\Gred \times \F \hookrightarrow \uCo  \to \uuC{}_\rad,
\]
which means that, locally, we have the decomposition $\uCo \approx \uuC{}_\rad \times \Gred \times \F$.

To understand the fibration structure of $\uuC = \uCo /\Gred$, we observe that in quotienting $\uCo$, $\Gred$ acts solely on the fibre $\Gred\times\F$ and not on the base $\uuC{}_\rad = \underline{\X}{}_\rad/\Gred$. Moreover, the action on the fibre is by right translations on the $\Gred$ factor and by coadjoint actions on the $\F$ one:
\begin{equation}
    \label{eq:radCoul-uCo}
\Gred\times(\Gred \times \F)\to (\Gred \times \F),
\qquad 
\big(\underline{g},  (\underline{k},f) \big) \mapsto (\underline{k}\underline{g} , \Ad^*(\underline{g})\cdot f).
\end{equation}
Therefore, $\uuC$ can be written as a fibration over $\uuC{}_\rad$:
\[
(\Gred \times \F)/\Gred \simeq
\F \hookrightarrow \uuC \to \uuC{}_\rad.
\]
where the isomorphism on the left is given by 
\[
\big[ (\underline{k}\, \underline{g},\ \Ad(\underline{g})^*\cdot f ) \big]_{\underline{g}\in\Gred} \mapsto  f^{\underline{k}} \doteq \Ad(\underline{k}^{-1})\cdot f.
\]

Focusing on a single superselection sector, we get
\begin{equation}
    \label{eq:S-fibration}
\uuS_{[f]}\simeq_{\mathrm{loc}} \uuC_\rad \times \mathcal{O}_f.
\end{equation}
Now, notice that---owing to the isomorphism $\C_\pp \simeq \F$---the space of on-shell superselections, $\mathcal{B} \doteq \C_\pp / \Gpoff \simeq \F / \Gred$ (Theorem \ref{thm:spaceofsuperselection}), is isomorphic to the space of coadjoint orbits of the on-shell fluxes:
\[
 \mathcal{O}_f \hookrightarrow \C_\pp \simeq \F \to \mathcal{B} = \{ \mathcal{O}_f \ | \ f\in\F\}
\]
where the leftmost $\mathcal{O}_f$ should be understood as the set $\{ f'\in\mathcal{O}_f \}\subset \F$.
Thus, since $\uuS_{[f]}\hookrightarrow \uuC \to \mathcal{B}$,  we conclude that
$\uuC$ can be written as 
\[
\uuC \simeq_\loc \uuC_\rad \times \F
\]

We now have all the necessary ingredients to investigate the full content of Theorem \ref{thm:spaceofsuperselection} in irreducible semisimple, and Abelian, YM theory. Putting everything together, we obtain the following commuting diagram:
\[
\xymatrix@C-=0.3cm@R-=0.3cm{
\uCo \ar[rr] \ar[dd] & & \uuC \ar[dd]\\
&  &\\
\C_\pp \ar[rr] && \mathcal{B} 
}
\qquad
\xymatrix@C-=0.1cm@R-=0.2cm{
(\uuphi, f, \underline{k}) \ar[rr]  \ar[dd] & & (\uuphi, f^{\underline{k}}) \ar[dd]\\
&  &\\
\;f\; \ar[rr] && \; \mathcal{O}_f = \mathcal{O}_{f^{\underline{k}}}\;
}
\]

\begin{remark}
    Summarising, in this section we have proven Theorem \ref{mainthm:YMPoisson}.
\end{remark}


\section{Other applications}\label{sec:applications}

In this section we consider other applications of our framework that go beyond what we discussed in the running examples on YM theory, that is:  Chern--Simons theory, a discussion of a 1-parameter family of corner ambiguities in YM theory, and $BF$ theory.

\subsection{Chern--Simons theory}\label{sec:ex-ChernSimons}

We start from the application of our general formalism to Chern--Simons theory, and recover the construction of \cite{MeinrenkenWoodward} to which we refer for many more details specific to the study of the moduli space of flat connections on Riemann surfaces.\footnote{We thank P.\ Mnev and A.\ Alekseev for pointing out this reference to us.} 

Consider Chern--Simons theory \cite{FroelichKingCS89,WittenCS89,AdPW91,Freed93classicalchern,MeinrenkenWoodward} for a real semisimple Lie group $G$ with Lie algebra $\fg$, and let $\Sigma$ be a 2d Riemann surface with boundaries.
The geometric phase space $\X$ is the space of principal connections over the principal $G$-bundle $P\to\Sigma$, i.e.\ $\X = \Acal\doteq \mathrm{Conn}( P \to \Sigma)$. The gauge group is defined as for YM theory (Section \ref{sec:runex-setup}), so that the gauge algebra is given by $\fG \simeq \Gamma(\Sigma, \Ad P)$ and
\[
\rho(\xi)A = d_A \xi.
\]

The above action of $\fG$ on $\X$ is not free unless one restricts to the subspace of irreducible connections. However, this restriction is too strong in Chern--Simons theory, since only reducible connections survive the imposition of the constraint $\bHo=0$, and we will therefore work with generally non-free actions, as encoded in the isotropy locus $\mathsf{I}_\rho$ of Definition \ref{def:isotropylocus}.

Chern--Simons theory equips $\X$ with the following symplectic density
\[
\bom = \tr( \bd A \wedge \bd A),
\]
which integrates to the Atiyah--Bott 2-form $\omega = \int_\Sigma\bom$ \cite{AtiyahBott}.
Here, $\tr$ is an $\ad$-invariant nondegenerate bilinear form on $\fg$.
Then, the action of $\G$ on $(\X,\bom)$ is \emph{locally} Hamiltonian  and the following is a  local momentum map $\bH = \bHo+ d\bh$ compatible with Assumption \ref{ass:setup}:
\[
\langle\bHo, \xi\rangle= \tr( F_A \xi ),
\qquad
\langle d\bh,\xi\rangle = -  d\tr((A-A_\smbullet)\xi),
\]
where $\bHo$ is Chern--Simons theory's flatness constraint:
\[
\C = \bHo^{-1}(0) = \{ A\, | \, F_A=0\}.
\]
Moreover, the local momentum map $\bH$ is only weakly equivariant:
\[
\langle g^*\bHo,\xi\rangle = \langle\bHo, \Ad(g)\cdot\xi\rangle
\qquad
\langle g^* d \bh, \xi\rangle = \langle d \bh, \Ad(g)\cdot \xi\rangle - d \tr( g^{-1}d g \xi )
\]
by which we identify the CE group and algebra cocycles
\[
\langle d\bc(g),\xi\rangle = - d \tr( g^{-1}d g \xi ),\qquad
d \bk(\eta,\xi) = d \tr(\xi d \eta).
\]

\begin{remark}[The reference connection $A_\smbullet$]
Since $\X = \mathcal A$, we stress the importance of choosing a reference connection $A_\smbullet$ in order to obtain a geometrically  sensible definition of $\bH$ and $d\bh$ across an atlas of charts over $\Sigma$.\footnote{Indeed, contrary to $A$, but just like $\xi$, $(A-A_\smbullet)$ is a section of $\Ad P$.}
This reference connection should be understood as a ``purely kinematical'' origin of the affine space $\Acal$. In the expressions above this is encoded in a fixed $\fg$-valued 1-form $A_\smbullet$ (i.e. $\bd A_\smbullet \equiv 0$) which therefore does \emph{not} transform under (active) gauge transformations, $\rho(\xi) A_\smbullet = \bi_{\rho(\xi)}\bd A_\smbullet \equiv 0$. (See Remark \ref{rmk:HamShift2}).  
\end{remark}

Looking now at the flux map
\[
\langle\bd h, \xi\rangle  = - \int_{\pp\Sigma} \iota^*_{\pp\Sigma}\tr((\bd A) \xi),
\]
we readily see that the \emph{off-shell} flux annihilator ideal and flux gauge algebra (and groups) are
\[
\fNoff = \{ \xi \in\fG | \iota_{\pp\Sigma}^* \xi = 0 \} ,
\qquad
\fGredoff \simeq \Gamma(\pp\Sigma, \iota_{\pp\Sigma}^*\Ad P).
\]
However, things are a priori different on-shell since the flatness condition in $\Sigma$ strongly restricts which 1d connection forms are available at the corner $\partial \Sigma$.

For definiteness let us now focus on the simplest example, where $\Sigma = B^2$ is the 2d disk and $P\to\Sigma$ is the trivial principal $G$-bundle over $\Sigma$, $P=B^2\times G$.
In this case, $\X  \simeq \Omega^1(B^2,\fg)$, and we can choose $A_\smbullet=0$; moreover, $\G \simeq C^\infty_0(B^2,G)$. With this, the constraint set, given by the space of flat connections, reads
\[
\C = \{ A\in \X \ |\ \exists u\in\G,\ A= u^{-1} d u \} \simeq G \backslash \G.
\]
In the last expression, we are quotienting out the constant elements $v\in\G$, $dv=0$, acting on $\G$ from the left: 
i.e\  $(v,u(x))\mapsto v u(x)$.

Notice that gauge transformations $g\in\G$ act instead from the \emph{right} on the fields $u$ parametrising $\C \simeq G\backslash \G$ and thus commute with the above quotient:
\[
Gu(x) \in \C, \quad Gu(x) \mapsto Gu(x) g(x).
\]

With this parametrization, and denoting $S^1=\pp B^2$ the corner, the on-shell flux map reads:
\[
\langle \iota_\C^* \bd h , \xi\rangle
= - \int_{S^1 } \iota_{S^1}^* \tr( \bd (u^{-1}du)  \xi ) 
= \int_{S^1}\iota_{S^1}^* \tr\big(   (u^{-1}\bd u)\ d_{u^{-1} d u} \xi \big).
\]

Now, the on-shell flux annihilator $\fN\doteq\Ann(\F)$ is defined as the subset of $\xi\in\fG$ such that $ \langle \iota_\C^* \bd h(A) , \xi\rangle $ = 0 for \emph{all} $A\in \C$. Therefore, $\fN$ is given by those $\xi\in\fG$ such that $\iota_{S^1}^* d_{u^{-1}d u}\xi = 0  $  for all $u\in\G$. 
Recalling that $P\to\Sigma$ is a trivial bundle and choosing $u$ such that $\iota_{S^1}^*u = e$ implies $\iota_{S^1}^*\xi = \text{const}$, hence---since at a fixed $x$, the map $\C\to \fg$ sending $u\mapsto(u^{-1}du)(x)$ is surjective---the previous condition says that $\iota_{S^1}^*\xi$ is a constant with value in the center of $\fg$.

If $G$ is semisimple and therefore centerless, then  $\iota_{S^1}\xi = 0$ and hence $\fN = \fNoff$; whereas if $G$ is Abelian, $\fN^\text{ab} = \{ \xi\in\fG | \iota_{S^1}^*\xi =\text{const}\}\subsetneq \fNoff$.

Recall now Theorem \ref{thm:fGo}, which establishes that the constraint gauge group $\fGo$ coincide with the on-shell flux annihilator $\fN$, and the definition of the flux guage group $\fGred \doteq \fG/\fGo$. Then, from the above we conclude that if $G$ is semisimple then $\fGred \doteq \fG/\fGo \simeq \fGredoff$, whereas if $G$ is Abelian then $\fGred^\text{ab} \simeq \fGredoff/\fg$ where $\fg\hookrightarrow \fGredoff$ as the subalgebra of constant functions on $\pp\Sigma$. 

To compute the space of on-shell fluxes $\F$, recall that we picked  the vanishing 1-form $A_\smbullet = 0$ as the reference on-shell connection; thus we write (notice the underline on the rightmost term)\footnote{Warning: the condition $ \langle \iota_\C^*h_\smbullet,\xi \rangle\vert_{A=u^{-1}du} = 0 $ for all $u$ does \emph{not} imply $\iota^*_{S^1}\xi=0$, although this is the correct result for semisimple Lie groups. In particular, in the Abelian case it only implies $\iota^*_{S^1}\xi=\text{const}$. This should be contrasted with the above treatment of the formula $\langle \iota_\C^* \bd h , \xi\rangle=0$ which can be thought of as the associated variational problem.}
 \[
 \langle \iota_\C^*h_\smbullet,\xi \rangle\vert_{A=u^{-1}du} = \int_{S^1} \tr\big( (u^{-1} d u) \xi \big) \equiv \langle \uh,\underline{\xi}\rangle.
 \]
where on the right-most side $\underline\xi = (\xi \,\mathrm{mod}\, \fGo) \in \fGred$. From this formula one finds that $\F =\Im(\iota_\C^*h_\smbullet)\simeq\Im(\uh)$ (Lemma \ref{lemma:Imh=flux}) is given by those maps from the circle into $G$ that are connected to the identity---modulo the constant maps:
\[
\F \simeq G\backslash C^\infty_0(S^1, G).
\]
Only the maps connected to the identity enter the above formula because the relevant $u(x):S^1 \to G$ are restriction to the boundary of smooth functions from the entire disk into $G$ and the fact that the disk is contractible.\footnote{Recall that in view of the fact that $P\to B^2$ is trivial, $\G\simeq C_0^\infty(B^2,G)$.}

In other words, the space of on-shell fluxes associated to the disk $B^2$ is given by the subgroup of elements of the based loop group $\Omega G \doteq G\backslash L G$ which are homotopic to the identity:
\[
\F \simeq (\Omega G)_0 \simeq \Omega \fg.
\]

From $\fGo = \fN = \Ann(\F)$ we find that in the semisimple and Abelian case one has respectively:
\[
\Go = \{ g\in\G | \iota_{S^1}^*g = e\}
\quad\text{and}\quad
\Go^\text{ab} = \{ g\in\G | \iota_{S^1}^*g = \text{const}\} \simeq G\cdot \Go,
\] 
and
\[
\Gred \simeq \Gredoff \simeq C^\infty(S^1,G)
\quad\text{and}\quad
\Gred^\text{ab} \simeq G\backslash \Gred.
\]
In both cases:
\[
\uCo \doteq \C/\Go \simeq \F,
\]
since for every $Gu(x)\in \C$ and every $g_\circ^\text{ab}(x)\in \Go^{ab}$, which can be written as $g_\circ^\text{ab}(x)=Gg_\circ(x)$ with $g_\circ(x)\in \Go$ we have
\[
    \C/\Go^\text{ab}\ni Gu(x) g^\text{ab}_\circ(x) = G u(x) G g_\circ(x) = G u(x) g_\circ(x) \in \C/\Go. 
\]

The action of $\Gred$ on $\uCo \simeq \F$ is given by point-wise right translation $u(x) \mapsto u(x) g(x)$ where $x\in S^1$. Clearly, on the disk $B^2$ there is one single gauge orbit in the on-shell space $\F$. In other words on $B^2$ there exists only one superslection sector---in agreement with the fact that only one flat connection exists there up to gauge.

\begin{remark}\label{rmk:CSpunctures}
This construction can be adapted to the case in $\Sigma$ has punctures and/or multiple boundary components. See for example \cite{jeffrey1994, AlekseevMalkin94,AlekseevMalkinSymp95,MeinrenkenWoodward,MeusburgerSchroers}. The simplest generalization is that of $\Sigma$ a cylinder: in this case there are two $S^1$ boundary components, and superselections simply require that the holonomies around them coincide, without fixing a preferred value. ``Pinching'' one boundary component to a point-like defect labelled by a conjugacy class of $\fg$ (which can be thought as the charge associated to the piercing of $\Sigma$ by a Wilson line) yields a punctured disk; in this case, the puncture's decoration fixes---via the flatness constraint---the holonomy around the remaining boundary component, returning us to a situation akin to the case analyzed in the main text where only one superselection sector is present.
\end{remark}

The corner off-shell structure is characterized by the Poisson manifold $\X_\pp = \mathrm{Conn}(\iota^*_{S^1}P\to \Sigma)$ endowed with the bivector $\Pi_\pp$, which at $a\in\X_\pp$ reads
\[
    \Pi_\pp\vert_a = \left\langle \h_\pp(a), \left[\frac{\delta}{\delta a},\frac{\delta}{\delta a}\right]\right\rangle + k\left(\frac{\delta}{\delta a},\frac{\delta}{\delta a}\right) =  \int_{S^1} \tr\left( a\left[\frac{\delta}{\delta a},\frac{\delta}{\delta a}\right]\right) + \tr\left(\frac{\delta}{\delta a} d \frac{\delta}{\delta a}\right)
\]
and for which the following expression holds:\footnote{We are using $\tr$ to implicitly identify $\fg^*$ and $\fg$ and thus the coadjoint and adjoint actions.}
\[
\Pi^\sharp(\langle \bd \h_\pp, \xi\rangle))=\int_{S^1} \tr\left(\ad_K^*(\xi)\cdot a \frac{\delta}{\delta a} \right).
\]
We can, in fact, gather that the solution of the classical master equation on $\A_\pp[1]$ reads:
\[
S_\pp = \int_{S^1} \tr(a[c,c] + cdc) = \int_{S^1} \tr(c d_a c ),
\]
where $c\in C^\infty(S^1,\fg)[1]=L\mathfrak{g}[1]$.

The previous expressions comes from the fact that $\A_\pp$ is an action algebroid, and that the off shell flux map $\h_\pp$ is equivariant up to the cocycle $k(\xi,\eta) = \int\tr(\xi d\eta)$. Indeed, rephrasing a well-known result \cite{Freed93classicalchern,MeinrenkenWoodward}, it is a standard calculation to show that the (off-shell) superselections are not given by the coadjoint orbits of $\Gpoff\simeq LG$, but rather by those of its central extension $\widehat{LG}$. Each of these orbits is labelled by a single gauge invariant quantity: the holonomy of $A$ around the circle boundary---which is the only gauge invariant quantity one can build out of $A$ on $S^1$. Only those holonomies that lie in the projection $\pi_{\pp,\X}(\C)=\C_\pp$ are \emph{compatible} with a given two dimensional surface (possibly with nontrivial genus) (see Remark \ref{rmk:CSpunctures} and references therein).

\begin{remark}[CS vs. YM]
While in Chern--Simons theory on a (decorated punctured) disk the space of superselections (called $\mathcal{B}$ in Theorem \ref{thm:spaceofsuperselection}) is zero-dimensional, and on a general (punctured) Riemann surface is finite dimensional, in YM theory the space of superselections is infinite-dimensional. 

From the perspective of the corner Poisson space $\mathcal{P}_\pp$ this is due to the fact that in CS $\mathcal{P}_\pp$ is an affine space and $\Pi_\pp$ encodes a centrally-extended (co)adjoint action.

From a more physical viewpoint, whereas in Chern--Simons theory the holonomies around the circle boundary components of $\pp\Sigma$ are the only gauge-invariant quantities in $\X_\pp$---and thus are the only candidate labels of points in $\mathcal{B}$---in YM theory, there are infinitely many such invariant quantities, e.g. the Killing norm of the electric flux at every point of the boundary, $|f|^2(x) = \tr(f(x)f(x))$.

This distinction is clear quantum-mechanically: the Chern--Simons punctures are labelled by the quantum irreps of the gauge group and that is the only nontrivial gauge invariant quantity one can associate to a flat connection on a punctured disk.
On the other hand, in the spin-network basis of lattice gauge theory, there is a precise sense in which the value of the irreps attached to each one of the links piercing the boundary can be understood as a discretized version of the continuous Casimirs labelling the orbit $\mathcal O_f$ \cite{DonnellyEntanglement}. In the continuum limit there are infinitely many such labels. (For a perfect\footnote{``Perfect'' here means that the relevant discretization captures \emph{all} the continuum degrees of freedom. This is possible because the theory is topological.} discretization of 3d $BF$ and Chern--Simons theory coupled to punctures, see also \cite{RielloDittrichFusion}.)
\end{remark}

\subsection{A 1-parameter corner ambiguity in Yang--Mills theory}\label{sec:thetaQCD}

To exemplify the role of the corner ambiguity in deriving $\bom$ from a  Lagrangian field theory, one can consider the following variation on the YM construction.

First, recall that $\bom =  \tr(\bd E \wedge \bd A)$ is the only \emph{ultralocal} symplectic density for YM theory as obtained from the covariant phase space method (Appendix \ref{app:covariant}, Proposition \ref{prop:bomuniqueness}). 
Consider then the following modification by addition of a ``Chern--Simons corner term'':
\[
\bom_\theta = \tr(\bd E  \wedge \bd A) - \frac{\theta}{2} d \tr( \bd A \wedge \bd A),
\]
for $\theta$ a real parameter. In the following we refer to this theory as $\theta$-YM.

Consider the curvature of $A$, $F_A = d A + A \wedge A$. Since $\bd F_A = d_A \bd A$ and  $d \tr(\bd A\wedge \bd A) = 2 \tr(d_A \bd A \wedge \bd A)$, it is immediate to see that we can rewrite 
\[
\bom_\theta = \tr\left( \bd (E + \theta F_A) \wedge \bd A\right).
\]

Since $E$ and $F_A$ have the same equivariance properties, we redefine the momentum to be\footnote{Throughout we implicitly identify $\fg$ and $\fg^*$ via $\tr$.} 
\[
E^\theta(E,A) \doteq E + \theta F_A
\]
and then the mathematical theory developed in the YM running examples proceeds unaltered upon the replacement of $E \leadsto E^\theta$. 
The physics, however, is affected by this modification: let us see how.

The key ingredient in our construction is the momentum form and its decomposition into a constraint and a flux form, viz. $\bH = \bHo + d \bh$. In the case of $\theta$-YM we find:
\[
\langle\bH^\theta,\xi\rangle = (-1)^{\dim\Sigma} \tr( E^\theta d_A \xi) 
\]
so that
\[
\langle\bHo^\theta,\xi\rangle = \tr( d_A E^\theta)
\quad \text{and}\quad
\langle d\bh^\theta ,\xi\rangle = - d \tr( E^\theta \xi)
\]

Crucially, the constraint form coincides with that of standard YM theory, since by the algebraic Bianchi identity $d_A F_A \equiv 0$,
$d_A E^\theta \equiv d_A E$---and therefore:
\[
\C^\theta = \C = \{ (A,E) \in \X \ | \ d_A E = 0 \} \subset \X.
\]
Therefore, in a sense, the first stage reduction---i.e.\ constraint reduction---preserves the same physical meaning as in standard YM. This was expected from the ``bulk'' nature of constraint reduction and the fact that $\bom_\theta$ and $\bom$ only differ by a corner term. 
Moreover, the space of off-shell corner configurations of YM theory (Equation \eqref{e:cornerconfYM}) coincides with that of $\theta$-YM theory, in virtue of the projection map
\[
    \pi_\pp\colon \X \to \X_\pp^\theta\simeq\X_\pp\simeq\mathcal{E}_\pp \; \qquad (A,E)\mapsto E_\pp = \iota_{\pp\Sigma}^* E^\theta ,
\]
and the fact that $E$ and $E^\theta$ have the same equivariance properties---which means that the canonical identification $\X_\pp^\theta\simeq\X_\pp$ preserves $\Pi_\pp$ (recall: $\Pi_\pp^\sharp =\rho_\pp$).

The algebraic treatment of flux superselection is also unchanged, as the unreduced superselections $\S^\theta_{[f]}$ are embedded as $\S^\theta_{[f]} \hookrightarrow \C $:\footnote{This formula is valid in the semisimple case away from reducible configurations, where there is no difference between the off- and on-shell treatment of fluxes.}
\[
\S^\theta_{[f]} = \left\{ (A, E)\in \C \ \left| \ \textstyle{\int}_{\pp\Sigma}\tr\big( E^\theta \cdot \big) \in \mathcal{O}_f\right. \right\}
\]
where the flux spaces $\F$ (and the coadjoints action on them) are just as in standard YM, cf.\ Section \ref{sec:runex-fluxannihilators}.

However, the \emph{physical} interpretation of superselection sectors \emph{is} different, since now it concerns the $\theta$-fluxes which mix electric and magnetic fields.
Therefore, we conclude that although the constraint reduction is unaffected by the corner ambiguity, the flux analysis can vary---as showcased by the fact that the physical content of the superselection sectors carry an imprint of the corner $\theta$-term appearing in $\bom_\theta$. 

\medskip

It is suggestive that a similar $\theta$-ambiguity affects the quantization of YM theory. 

A first point of contact is the fact that $\bom_\theta$ can be heuristically understood as the ``natural'' symplectic density associated to the Lagrangian for $\theta$-QCD.
In $\theta$-QCD, to the usual YM term $\frac12 \tr(\bar F\wedge \star \bar F)$ one adds the ``theta-term'' $\frac\theta{2}\tr( \bar F\wedge \bar F)$.\footnote{Here $\bar A$ is a spacetime principal connection defined from a principal bundle over $M$, $\Sigma \subset \pp M$. Cf. Appendix \ref{app:covariant}.} On a closed manifold the contribution of this term to the action is, up to a numerical normalization, the second Pontryagin number of $P$. This is a topological invariant which is not affected by infinitesimal variation and which therefore cannot itself affect the dynamics.
Despite this classical equivalence, corroborated by the $\theta$-independence of the equations of motion, and in particular of the constraint condition $\C^\theta = \C$, it is well known that an ambiguity arises at the quantum level---even in the absence of corners---due to the existence of unitary-inequivalent representations of $\pi_0(\G)$.

Semiclassically the $\theta$-ambiguity can be understood as due to the contribution of instantons \cite{JackiwRebbitheta}, whereas in a fully non-perturbative framework  the ambiguity arises due to the  way the (central) action of the fundamental group of the (non-simply connected) gauge group $\G\simeq \Gamma(\Sigma,\AD P)$ is represented on the physical Hilbert space \cite{MorchioStrocchiTheta,Strocchi13,strocchi19}.

It is therefore intriguing to observe that an analogous ambiguity can be detected already at the classical level, in the superselection structure, \emph{provided corners are present}. 
Of course, a much deeper analysis is needed to prove that a direct relationship with the quantum ambiguity of $\theta$-QCD exists, especially with respect to its non-perturbative description.


\subsection{BF theory}\label{sec:BFtheory}
We discuss here the example of $BF$ theory in spacetime dimension larger than 4, since for the 3-dimensional case one can simply adapt the constructions outlined in Section \ref{sec:ex-ChernSimons} for Chern--Simons theory, to include non-semisimple Lie groups. 

For later convenience, let us set:
\[
n \doteq \dim(\Sigma) \geq 3.
\]

For simplicity of exposition, we will focus on the case of a trivial principal bundle (see the running example on YM for how to handle the general case).
Then the symplectic space of fields for $BF$ theory with Lie algebra $\fg$ is given by 
\[
\X \doteq T^\vee \Acal \simeq \Omega^{n-2}(\Sigma,\fg^*) \times \Omega^1(\Sigma,\fg) \ni (B,A)
\]
with local symplectic density
\[
\bom = \tr(\bd B \wedge \bd A)
\]
with $\tr$ denoting the pairing between $\fg$ and its dual.

This space admits a Hamiltonian action of the Lie algebra:
\[
\fG \doteq \Omega^0(\Sigma,\fg) \ltimes \Omega^{n-3}(\Sigma,\fg^*) \ni \xi = (\lambda,\tau),
\]
so that, for a (constant) section $\xi\colon \X\to \X\times \fG$, the anchor map reads
\[
\rho(\xi)(A)= d_A \lambda; \qquad \rho(\xi)(B) = d_A\tau + [\lambda,B]
\]
yielding
\begin{align*}
\iota_{\rho(\xi)}\bom 
&= \tr(d_A \lambda \wedge \bd B) + \tr( \bd A \wedge d_A\tau) + \tr(\bd A \wedge [\lambda,B]) \doteq \bd\langle\bH,\xi\rangle \\ 
&= \bd\left( \lambda d_AB + \tau F_A\right) + d\bd \left(\lambda B + \tau A\right)
= \langle \bd \bHo,\xi\rangle + d\langle \bd \bh,\xi\rangle 
\end{align*}
and we denote\footnote{For simplicity of exposition, in this section we also forgo explicitly choosing a reference configuration $\phi_\smbullet=(B_\smbullet,A_\smbullet)$. Since we are considering trivial bundles $\phi_\smbullet=0$ would be a viable choice. More generally this can, and should, be done in complete analogy with Section \ref{sec:ex-ChernSimons}.}
\[
\begin{cases}
\langle\bHo,\xi\rangle = \tr(F_A\wedge\tau) +\tr(d_AB \wedge\lambda)\\
\langle\bh,\xi\rangle = \tr(B \lambda) + \tr( A \wedge \tau ).
\end{cases}
\]

From $\bHo$ we can read off the constraint,
\[
\C = \{ (B,A) \in \X \ | \ d_A B = 0 \text{ and } F_A =0 \},
\]
as well as the \emph{off}-shell flux annihilator,
\[
\fNoff = \left\{\xi \in \fG\ |\ \langle\bd\h ,\xi\rangle = 0 \right\} = \Omega^0_D(\Sigma,\fg^*) \times \Omega^{n-3}_D(\Sigma,\fg^*),
\]
by which we mean those forms whose pullback to $\pp\Sigma$ vanishes $\bullet_D$ is for ``Dirichlet'', with reference to the type of boundary condition. Cf.\ Appendix \ref{app:Hodge}).
Hence, we conclude that 
\[
\fGredoff \doteq \fG/\fNoff \simeq\Omega^0(\pp\Sigma,\fg) \times \Omega^{n-3}(\pp\Sigma,\fg^*).
\]
The space of off-shell corner data $\A_\pp$ is obtained from pre-symplectic reduction of the two form\footnote{We are abusing notation by employing the same symbol for $(\lambda,\tau)$ in $\Sigma$ and their restrictions to $\pp\Sigma$.} $\langle \cbd\bh,\cbd\xi\rangle = \int_{\pp\Sigma} \tr( \cbd B \wedge \cbd\lambda + \cbd A \wedge \cbd\tau)$:
(see Definition \ref{def:cornerconf} for a precise characterization), which in this case yields the trivial vector bundle
\[
\A_\pp=\X_\pp\times \fGpoff \to \X_\pp, \quad \X_\pp\doteq \Omega^{n-1}(\pp\Sigma,\fg^*)\times\Omega^1(\pp\Sigma,\fg),
\]
This space is turned into a Lie algebroid (Theorem \ref{thm:corneralgd}) when equipped with the obvious restriction of the gauge action to the boundary, $\rho_\pp: \A_\pp \to  T\X_\pp$.

Hence, $\A_\pp$ is a symplectic Lie algebroid over a (partial) Poisson manifold $\X_\pp$, with Poisson structure given by $\Pi_\pp=\rho_\pp^\sharp$ (Theorem \ref{thm:Poissontheorem}). It can be written as the bivector
\[
    \Pi_\pp = \int_{\pp\Sigma} \tr\left( \frac12 B \left[\frac{\delta}{\delta B},\frac{\delta}{\delta B}\right]  + \frac{\delta}{\delta B} d_A \frac{\delta}{\delta A} \right).
\]
This Poisson structure (together with an on-shell version we will describe shortly) has been recovered by means of the Batalin--Vilkovisky method for a field theory on manifolds with corners in \cite[Section 2.2.4]{CanepaCattaneo}.

On-shell some interesting things happen. It is easy to see that the ($A$-dependent) equivariant differential $\tau \mapsto d_A\tau$ has a kernel whenever $A$ is flat: $d_A\gamma \mapsto [F_A,\gamma] = 0 \mod F_A$, and thus everywhere on $\C$. (Observe that this is true regardless of whether $\tau$ and $A$ denote fields on $\Sigma$ or $\pp\Sigma$.) This means in particular that the algebroid anchor map \emph{fails} to be injective on-shell, and this failure is encoded by the isotropy locus $\mathsf{I}_\rho$. (Another way to say this is that the gauge symmetry is on-shell reducible---although, if $\fg$ is Abelian, this fact is true even off-shell.)

\medskip

In the Abelian case we easily recover (structural) information about the theory by looking at the defining condition for the on-shell flux annihilator $\fN$:
\begin{align*}
\fN 
&\doteq \{\xi\in \fG\ |\ \langle\iota^*_\C\bd\h(\phi),\xi\rangle = 0 \ \forall \phi\in\C \}\\
&\simeq \left(d\Omega^{n-3}(\Sigma,\fg^*)\times \Omega^{n-2}_D(\Sigma,\fg^*)\right) \times\Omega^0_D(\Sigma,\fg).
\end{align*}

Observe that $\mathfrak{I}\subset\fN\not=\fNoff$, where $\mathfrak{I}$ is the isotropy subalgebra, as expected from Proposition \ref{prop:Gausslaw}. We conclude that, in the \emph{Abelian} case, the flux gauge algebra is
\begin{align*}
\fGred \doteq \fG/\fGo &\simeq \fG/\fN \\
&\simeq \frac{\Omega^0(\Sigma,\fg)}{\Omega^0_D(\Sigma,\fg)}\times \frac{\Omega^{n-2}(\Sigma,\fg^*)}{d\Omega^{n-3}(\Sigma,\fg^*)\times \Omega^{n-2}_D(\Sigma,\fg^*)}\\
&\simeq \Omega^0(\pp\Sigma,\fg)\times \frac{\Omega^{n-2}(\pp\Sigma,\fg^*)}{d \Omega^{n-3}(\pp\Sigma,\fg^*)}.
\end{align*}
It is worth noting that, in contrast to its off-shell analogue $\fGredoff$, the flux gauge algebra $\fGred$ \emph{fails} to be local (even though it is still supported on $\pp\Sigma$, Definition \ref{def:LocLieAlg}). The associated algebroid is isomorphic to the restriction of the image of $\rho(\fG)$ to $\C$ (the anchor now being the inclusion map in $T\C$, and thus injective).

In the non-Abelian case, however, the $A$-dependent condition $\tau=d_A\gamma$ complicates matters significantly as it introduces a point-dependent condition for the isotropy locus $\mathsf{I}_\rho$. In fact, this means that $\mathsf{I}_\rho$ is a \emph{nontrivial} vector bundle\footnote{This is because all flat connections have isomorphic (maximal) reducibility algebras.} and therefore the on-shell reducibility of the symmetry cannot be registered in the definition of $\fN$ which is instead defined by a \emph{global} condition as the set of those sections $\xi$ such that $\langle\iota_\C^*\bd \h(\phi),\xi\rangle=0$ \emph{for all $\phi\in\C$}. Therefore, in the non-Abelian case $\fN = \fNoff$, and we seemingly lose track of the on-shell reducibility information encoded in $\mathsf{I}_\rho$.
As we shall discuss in a moment, although this information is not registered by $\fN$ (as compared to $\fNoff$), it is not completely lost either, since it is encoded in the geometry of the on- vs. off-shell superselections.
However, to explicitly describe this geometry at the level of the Lie algebroid itself, an extension of the formalism is needed, to handle more general, non-trivial, Lie algebroids. See Remark \ref{rmk:isotropybundle3}.

In both the Abelian and non-Abelian cases we can observe that the shell condition $F_A=0$ restricts to the corner submanifold,\footnote{Observe that this phenomenon is not present in either YM or Chern--Simons theories, whose shell conditions are automatically satisfied when restricted to the corner submanifold.} whereas $d_AB$ is an $n$ (i.e.\ top) form which identically vanishes on a codimension-one submanifold.

It is easy to gather that on the submanifold $\iota^*_{\pp\Sigma}F_A = 0\subset \X_\pp$ the ``order'' of the symplectic leaves of $\Pi_\pp$ jumps. Indeed, consider the (on-shell vanishing) $\mathsf{p}$-function 
\[
g_\gamma(B,A) = \int_{\pp\Sigma}\tr( F_A \wedge \gamma) \in C^\infty_\mathsf{p}(\X_\pp),
\quad
\bd g_\gamma(B,A) = \int_{\pp\Sigma} \tr(\bd A \wedge d_A \gamma);
\]
it is immediate to see that the corresponding vector fields 
\[
\Pi^\sharp_\pp(\bd g_\gamma) = \Pi_\pp^\sharp\left(\int_{\pp\Sigma}\tr( \bd A \wedge d_A\gamma)\right) = \int_{\pp\Sigma} [F_A,\gamma] \frac{\delta}{\delta B} \equiv 0 \mod F_A
\]
span (nontrivial) tangent directions to the symplectic leaf only off-shell.

Denoting by $\mathcal{I}_{\pi_{\pp}(\C)}$ the \emph{vanishing ideal} of the submanifold
\[
\C_\pp  \doteq \pi_{\pp,\X}(\C) = \{(A,B)\in \X_\pp \ |\ \exists (\wt{A},\wt{B})\in \C, \iota^*_{\pp\Sigma}\wt{A}=A, \iota^*_{\pp\Sigma}\wt{B}=B\},
\]
we see that it is generated by the functions $g_\gamma$ (since $\iota_{\pp\Sigma}^*d_{\wt{A}}\wt{B} \equiv0$ is automatically satisfied for degree reasons). 
Therefore, since the $\mathsf{p}$-functions $g_\gamma$ generate a Poisson subalgebra of $C^\infty_{\mathsf{p}}(\X_\pp)$ (and an associative ideal), the quotient $C^\infty_{\mathsf{p}}(\X_\pp)/\mathcal{I}_{\pi_{\pp}(\C)}\simeq C^\infty_{\mathsf{p}}(\C_\pp)$ is itself a Poisson algebra.

Summarising, within $\X_\pp$ there exists a closed Poisson submanifold $\C_\pp$ defined by the corner shell condition $\iota^*_{\pp\Sigma}F_A=0$. This can be rephrased as the statement that the off shell corner algebroid 
\[
\A_\pp = \X_\pp \times \fGpoff = \underbrace{\Omega^1(\pp\Sigma,\fg)\times \Omega^{n-2}(\pp\Sigma,\fg^*)}_{\X_\pp} \times \underbrace{\Omega^0(\pp\Sigma,\fg) \times \Omega^{n-3}(\pp\Sigma,\fg^*)}_{\fGpoff}
\]
restricts to an algebroid $\iota_{\C_\pp}^*\A_\pp \to \C_\pp$, within which one can identify a \emph{corner isotropy bundle} $\mathsf{I}_{\rho_\pp}\to \C_\pp$.

In order to simultaneously describe on- and off-shell behavior, a handy tool is given by the Batalin Vilkovisky framework. The result of that procedure for this case has been recently presented in \cite[Section 2.2.4]{CanepaCattaneo}, and gives rise to a Poisson-$\infty$ structure. This is a  generalization of the corner Poisson structure we discussed in this paper whose emergence can be traced back to the presence of the on-shell corner isotropy bundle $\mathsf{I}_{\rho_\pp}$---a generalization that our strictly-Hamiltonian framework cannot (presently) capture. We refer to Remark \ref{rmk:isotropybundle3} for a path towards a generalization which should be capable of doing just that.


\appendix
\section{Remarks on infinite-dimensional manifolds}\label{app:infinitegeometry}

We collect here a few observations and outline central constructions underpinning the infinite-dimensional smooth structures of the spaces used throughout. 

As we mentioned in the main text, the convenient setting endows locally convex vector spaces with the $c^\infty$ topology, which is generally finer than the associated locally convex topology, but it coincides with it in many cases of interest, e.g.\ for all bornological, metrizable vector spaces, as shown in \cite[Section 4 and Theorem 4.11]{kriegl1997convenient}. 
In particular, it coincides with the locally convex topology when the vector space is either Fr\'echet, or the strong dual of a Fr\'echet space---which we now define---provided that these space enjoy some additional properties:

\begin{remark}[Strong vs. bornological dual]\label{rmk:strong/bornduals}
Following \cite{kriegl1997convenient}, but changing the notation, given a locally convex vector space $\calW$ we can consider two (generally distinct) notions of dual. The strong dual $\strong{\calW}$ is the locally convex vector space of all \emph{continuous linear functionals} on $\calW$ endowed with the strong topology\footnote{The strong topology is the topology of uniform convergence on bounded subsets generated by the seminorms $\mu_B(f)=\mathrm{sup}_{x\in B}(f(x))$.}. On the other hand we have the \emph{bornological dual}, defined as the space $\dual{\calW}\doteq L(\calW,\mathbb{R})$ of \emph{bounded linear functionals}\footnote{Bounded linear maps map smooth curves to smooth curves, and viceversa \cite[Cor. 2.11]{kriegl1997convenient}.} on a locally convex vector space $\calW$ to be the space of linear maps that map bounded sets to bounded sets. 
For the important class of nuclear Fr\'echet vector spaces, to which all of our examples belong (see Remark \ref{rmk:nuclearFrechet} and Remark \ref{rmk:NFVS}, below) these two notions agree, and $\strong{\calW}=\dual{\calW}$ as topological vector spaces. (This follows from \cite[Lemma 5.4.3]{FrolicherKriegl} and the fact that the strong dual of nuclear Fr\'echet space is bornological and the $c^\infty$ topology coincides with the given locally convex topology\footnote{In fact one can relax nuclearity and require that $\calW$ be Fr\'echet Montel. However, among those, only Fr\'echet Schwartz spaces will generally enjoy the additional fact that their strong duals will be $c^\infty$-complete \cite[Theorem 4.11]{kriegl1997convenient}. (Recall that all nuclear Fr\'echet spaces are also Schwartz, and all Schwartz Fr\'echet spaces are also Montel.) \label{fnt:Montel}}). Finally, we recall that the strong dual of a Fr\'echet vector space $\calW$ is not necessarily Fr\'echet, unless $\calW$ is also Banach.

Hence, we will not distinguish between the $c^\infty$ topology used in the convenient setting and the Fr\'echet topology used in the locally-convex setting, nor between strong and bornological duals, and we  will henceforth use the unified notation $\dual{\calW}$. However, we recall in Remark \ref{rmk:cotangent} why the convenient and the locally convex settings generally differ.
\end{remark}

\begin{remark}[Based on 6.5(7)\cite{kriegl1997convenient}]\label{rmk:NFVS}
The following holds for Nuclear Fr\'echet (NF) vector spaces:
\begin{enumerate}
\item A NF vector space $\calW$ is necessarily reflexive\footnote{Note that reflexivity for locally convex spaces need not coincide with reflexivity in the $c^\infty$ sense, but they do in some cases, including Fr\'echet \cite{Jarchow}. See \cite[Cor. 5.4.7]{FrolicherKriegl}, and for a summary \cite[Definition 6.3 to Result 6.5]{kriegl1997convenient}.}, i.e.\ $\bidual{\calW}\simeq\calW$ canonically, where the bidual $\bidual{\calW}$ is considered either in the strong or the bornological sense.
\item The quotient of a NF vector space by a closed NF sub-vector space is also a NF vector space.
\item NF vector spaces admit partitions of unity (see \cite[Thm. 16.10]{kriegl1997convenient}  and related results).
\item The category of NF vector spaces is stable under duality and (projective) tensor product \cite[Thm. 7.5]{SchaeferWolff}.
\item Given two nuclear spaces $E,F$ the projective and injective tensor products coincide, and $E\hat\otimes_\pi F$ is nuclear. If  they are nuclear Fr\'echet, their duals are nuclear. Finally, the projective tensor products of nuclear spaces (and their completions) are nuclear \cite[Thm. 7.5]{SchaeferWolff}.\qedhere
\end{enumerate}
\end{remark}

\begin{remark}[Cotangent bundle]\label{rmk:cotangent}
Consider a locally convex vector space $\calW$ with its dual $\dual{\calW}$. The evaluation map $\mathrm{ev}\colon \calW\times \dual{\calW} \to \mathbb{R}$ is not continuous whenever $\calW$ is not normable. This is usually taken as an obstruction to the existence of a natural smooth cotangent bundle structure on locally convex manifolds (see \cite{NeebSahlmannThiemann,Neeb-locallyconvexgroups}).
However, since the evaluation map is always smooth, it is necessarily continuous when seen as a map\footnote{Here we denote by $c^\infty(\calW)$ the vector space $\calW$ with $c^\infty$ topology.} $c^\infty(\calW\times \dual{\calW}) \to \mathbb{R}$, (see \cite[Section  4.16 and Page 8]{kriegl1997convenient}). This is only possible because $\calW\times \dual{\calW}$ with the $c^\infty$ topology is not a locally convex vector space \cite[Proposition 4.20]{kriegl1997convenient}, and $c^\infty(\calW\times \dual{\calW})\not= c^\infty(\calW) \times c^\infty(\dual{\calW})$ \cite[Corollary 4.21 and Remark 4.22]{kriegl1997convenient}.
In the convenient setting, the cotangent bundle of a vector space $\calW$ is \emph{defined} as $\dual{T}\calW \doteq c^\infty(\calW\times\dual{\calW})$. This construction serves as a model for the definition of the cotangent bundle $\dual{T}\calPhi$ of any convenient smooth manifold $\calPhi$ as the $c^\infty$ vector bundle over $\calPhi$ whose fibre $\dual{T}_\phi\calPhi \doteq \dual{(T_\phi\calPhi)}$ is the dual of the tangent space\footnote{Here $T\calPhi$ denotes the ``kinematical'' tangent bundle, defined in terms of equivalence classes of curves through the base point \cite[Section  28.13 and Sect. 33.1]{kriegl1997convenient}. \label{fnt:TM}} $T_\phi\calPhi$.
As for $\dual{T}\calW = c^\infty(\calW\times\dual{\calW})$, also the topology of $T^*\calPhi$ fails to be the product topology of the topologies on the base and the fibre.\footnote{We refer to \cite[Appendices A and B]{DiezPhD} for a survey of an alternative method to define (noncanonically) a cotangent bundle within the locally convex setting, based on an appropriate choice of topology for dual pairs.}

Below, we will consider an instance where the cotangent bundle admits a (strong) symplectic structure, and is therefore a good model for certain purposes. However, as we will see below, in the \emph{local} field theoretic case another notion of cotangent bundle is available that is natural, smooth, and local. This will be our notion of choice for most applications.
\end{remark}

\begin{remark}[Reflexive cotangent bundle]\label{rmk:cotangentreflexive}
Recall Remark \ref{rmk:cotangent}.
The cotangent bundle $T^*\calV = c^\infty(\calV \times \dual{\calV})$ of a locally convex vector space $\calV$ admits a canonical  symplectic 2-form: 
\[
\omega_\text{can}((v_1,\eta_1),(v_2,\eta_2))= \eta_2(v_1) - \eta_1(v_2).
\] 
This 2-form is strongly symplectic iff $\calV$ is reflexive, i.e.\ $\bidual{\calV}=\calV$  \cite[Section 48.3]{kriegl1997convenient}. This is, in particular, the case for $\calV$ nuclear Fr\'echet (Remarks \ref{rmk:strong/bornduals}, \ref{rmk:cotangent}).

By using local charts, these results can be employed to equip the convenient cotangent bundle of any convenient (Fr\'echet) manifold $\calPhi = T^*\calN$ with a closed 2-form $\omega_{\text{can}}$ that coincides with the 2-form induced from the canonical Liouville 1-form on $T^*\calN$ (see \cite[Section 48.3]{kriegl1997convenient}). Once again, $\omega_\text{can}$ is strongly symplectic if and only if the local model $\calV$ of $\calN$ is a reflexive vector space, and in particular if it is nuclear Fr\'echet.
\end{remark}

\begin{remark}[Differential forms]\label{rmk:diffForms}
In the convenient setting, differential forms can be defined following \cite[Section 33]{kriegl1997convenient}.
Among the multiple (inequivalent) notions that one can naturally consider, only one enjoys all the desired properties generalising the finite dimensional scenario, and in particular all the machinery of Cartan's calculus. This is given by\footnote{Notice that the space of forms defined as sections $C^\infty(\calPhi,\wedge^\bullet T^*\calPhi)$ is only invariant along pullbacks, but not under exterior differentiation and Lie derivative \cite[Section 33.21]{kriegl1997convenient}.}
\[
    \Omega^k(\calPhi)\doteq C^{\infty}(\calPhi,L(\wedge^k T\calPhi,\calPhi\times \mathbb{R})), \quad \bd\colon \Omega^k(\calPhi) \to \Omega^{k+1}(\calPhi)
\]
where the space of multilinear alternating maps\footnote{By $L(T\calPhi,\calPhi\times\mathbb{R})$ we denote fibrewise linear bundle morphisms that are bounded, following \cite{kriegl1997convenient}.} on $\wedge^k T\calPhi$ is viewed as a vector bundle over $\calPhi$, and $\bd$ is defined as the skew-symmetrization of the derivative.

In the particular case of a Fr\'echet (or diffeological) space $\calPhi$, beside the defintion of forms as smooth multilinear maps on tangent spaces, there is another natural definition in terms of the extension of the de-Rham functor (see \cite{Iglesias-Zemmour:Diffeology}).
For spaces of fields $\calPhi=\Gamma(\Sigma,F)$ these two notions coincide \cite{waldorfI,BlohmannLFT}. We can therefore unambiguously denote that complex of  differential forms with $\Omega^\bullet(\calPhi)$.

A similar divergence of definitions and characterizations, marking a departure from the finite dimensional scenario, holds for vector fields. Indeed, the kinematical vector bundle $T\calPhi$ (see \cite[Section  28.13 and Sect. 33.1]{kriegl1997convenient}) in general only yields a subset of the  derivations of the algebra of smooth functions \cite[Lemma 32.3]{kriegl1997convenient}.
\end{remark}

\section{Local Forms}\label{app:localforms}
In this section we will assume to be working on $\calPhi=\Gamma(\Sigma,F)$ for $F\to\Sigma$ a fibre bundle. The manifold $\calPhi$ is then nuclear Fr\'echet.  

\begin{definition}[Local maps]\label{def:localmaps}
Let $F_i\to \Sigma$ denote fibre bundles for $i=1,2$, and let $\calPhi_i$ denote the respective spaces of smooth sections. A map $f \colon \calPhi_1 \to \calPhi_2$ is said to be local of jet order $k$ if $k$ is the smallest integer for which there is a smooth map $f_k\colon J^k F_1\to F_2$ such that the following diagram commutes: 
\begin{equation}
    \xymatrix{
    \Sigma\times \calPhi_1 \ar[r]^-{\;f\times \mathrm{id}_\Sigma\;} \ar[d]^-{j^k} & \Sigma\times \calPhi_2 \ar[d]^{j^0}\\
    J^k F_1 \ar[r]^-{f_k} & F_2
    }
\end{equation}

We will need the following standard results on local maps \cite[Section6]{BlohmannLFT}:
\begin{proposition}
The following properties hold.
\begin{enumerate}
    \item The composition of local maps is local;
    \item The direct product of local maps is local;
    \item Let $F_i\to M$ be fibre bundles, for $i\in\{1,2,3\}$, with sections $\calPhi_i$, and let $f\colon \calPhi_1\times\calPhi_2\to\calPhi_3$ be a local map. Then the maps
    \[
        f(\bullet, \varphi_2)\colon \calPhi_1 \to \calPhi_3, \qquad f(\varphi_1,\bullet) \colon \calPhi_2 \to \calPhi_3
    \]
    are local for all $\varphi_i\in\calPhi_1$.\qedhere
\end{enumerate}
\end{proposition}

\end{definition}

We will make extensive use of the notion of local forms on $\Sigma \times \calPhi$, which are defined as a bicomplex of differential forms (Remark \ref{rmk:diffForms})  which only depend on a finite jet of a section $\phi\in\calPhi$. More precisely:\footnote{For further details beyond what we will present here, see See \cite{Takens,Anderson:1989}}

\begin{definition}[Local forms and vector fields]\label{def:localforms}
A local form on a fibre bundle $F\to \Sigma$ is an element of
\[
    \left(\oloc^{\bullet,\bullet}(\Sigma \times \calPhi), d, \bd\right) \doteq (j^\infty)^*\left(\Omega^{\bullet,\bullet}(J^\infty F),d_H,d_V\right)
\]
where $J^\infty F$ denotes the infinite jet bundle of $F\to \Sigma$, the map $j^\infty\colon \Sigma \times \calPhi \to J^\infty F$ is the infinite jet prolongation, and $\left(\Omega^{\bullet,\bullet}(J^\infty F),d_H,d_V\right)$ is the variational bicomplex.
We denote the two differentials induced on $\oloc^{\bullet,\bullet}$ by $d_H,d_V$ as
\[
d(j^\infty)^*\balpha \doteq (j^\infty)^* d_H \balpha,
 \qquad
 \bd (j^\infty)^*\balpha \doteq (j^\infty)^* d_V \balpha.
\]

An \emph{integrated, local} form is an element in the image of the map
\[
    \int_{\Sigma} \colon \oloc^{\text{top},\bullet} \to \Omega^\bullet(\calPhi)
\]
obtained by integration along $\Sigma$. We denote the space of integrated, local, forms by $\iloc^{\bullet}(\calPhi)$. For notational purposes, we will denote local forms with boldface symbols, and integrated local forms with regular symbols. 

A vector field $\mathbb{X}$ on $\calPhi$ is called local if it is local as a map $\mathbb{X}\colon \calPhi \to T\calPhi=\Gamma(\Sigma, VF)$, where $VF$ denotes the vertical subbundle of $F\to \Sigma$. We denote local vector fields by $\Xloc(\calPhi)$.
\end{definition}

The image of the integral map lands into the complex of differential forms $\Omega^\bullet(\calPhi)$ (see Remark \ref{rmk:diffForms}), defined in the appropriate infinite-dimensional manifold setting. Notice that, if $\Sigma$ is compact and without boundary, we can identify $\iloc^{\bullet}(\calPhi)$ with
\[
    \iloc^{\bullet}(\calPhi) \simeq \oloc^{\text{top},\bullet}/d\oloc^{\text{top}-1,\bullet}.
\]

\begin{definition}[Forms of source type]\label{def:sourcetype}
Let us denote by $\{x^i,u^\alpha, u_I^\alpha\}$ a local chart in $J^{|I|}E$, where $I$ is a multiindex. Local forms of the \emph{source} type are defined as those local forms $\boldsymbol{s}\in\oloc^{\text{top},1}(\Sigma \times \calPhi) $ such that 
\[
\boldsymbol{s} = (j^\infty)^*(P_\beta(x^i, u^\alpha, u_I^\alpha) d_Vu^\beta \wedge dx^1 \wedge  \cdots \wedge dx^{\text{top}}).
\]
We denote the space of local forms of source type by $\osrc^{\text{top},1}(\Sigma \times \calPhi)$.
\end{definition}

In other words, we can characterise source forms $\boldsymbol{s}$ as depending only on a finite jet of $\varphi\in\calPhi$ and \emph{only on} $\bd\varphi$ in the vertical direction. That is, source forms do not contain terms of the form $\bd\partial\varphi$. An extension to more general local $(p,q)$-forms of source type is possible, although we will be mostly interested in $(\mathrm{top},1)$-forms of source type. We refer to \cite{DelgadoPhD,blohmann2018hamiltonian} for more details.

We recall now a well-known theorem of Zuckerman about local forms:

\begin{theorem}[Lemma 1 of \cite{Zuckerman}, see also \cite{Takens,Takens77}]
For $\calPhi$ as above, the space of $(\text{top},1)$-local forms decomposes as
\[
\oloc^{\text{top},1}(\Sigma \times \calPhi) = \osrc^{\text{top},1}(\Sigma \times \calPhi) \oplus d \oloc^{\text{top}-1,1}(\Sigma \times \calPhi).
\]
\end{theorem}

We conclude this appendix by defining a Cartan calculus for local forms.
\begin{proposition}[Local Cartan calculus, \cite{BlohmannLFT}]\label{prop:Cartancalc}
Let $\balpha$ be a local form and $\mathbb{X}$ be a local vector field. We denote by $\L_\mathbb{X}$ the operator on local forms defined by
\[
    \L_\mathbb{X} \balpha = (\bi_\mathbb{X} \bd + \bd \bi_\mathbb{X})\balpha.
\]
Then we have a the local Cartan calculus on $\oloc^{\bullet,\bullet}(\Sigma\times \calPhi)$, for every $\mathbb{X},\mathbb{Y}\in\Xloc(\calPhi)$ and $X,Y\in\mathfrak{X}(\Sigma)$:
\begin{align*}
    [\bd,\L_{\mathbb{X}}] &= [\iota_X,\L_{\mathbb{X}}] = 0 ; & \qquad [\L_{\mathbb{X}},\bi_{\mathbb{Y}}]&= \bi_{[\mathbb{X},\mathbb{Y}]}; & \qquad [\L_{\mathbb{X}},\L_{\mathbb{Y}}]& = \L_{[\mathbb{X},\mathbb{Y}]}\\
    [d,L_X] &= [\bi_{\mathbb{X}},L_X] = 0; & \qquad [L_X,\iota_Y]&=\iota_{[X,Y]}; & \qquad [L_X,L_Y]&=L_{[X,Y]}\\
    [d,\bi_\mathbb{X}] &= [\bd,\iota_X]=0; & [d,\bd]& =0  & &
\end{align*}
where, we recall, boldface symbols encode strictly vertical objects and standard symbols strictly horizontal ones, and $[\,\cdot\,,\cdot]$ denotes graded commutators.
\end{proposition}

\section{Lie theory and central extensions}\label{app:CEcoh}

In this appendix we recall a few basic notions in Lie group/algebra cohomology. Throughout, $\fG$ denotes the lie algebra of a Lie group $\G$. We refer to \cite{fuks2012cohomology,ChevalleyEilenberg}.

\subsection{Generalities}

Consider a $\G$-module $(\calW,\rho)$, then a $\calW$-valued $\G$-cocycle is a smooth map $\sigma\colon\G\to \calW$ such that
\[
    \sigma(gh)=\sigma(g) + \rho(g)\sigma(h).
\]
We note that $\sigma(e)=0$. The cocycle $\sigma$ is called a coboundary if there exists $w\in \calW$ such that $\sigma(g) = \rho(g)w - w$. Generalising this idea one gets:

\begin{definition}[Group/algebra Chevalley--Eilenberg complex]\label{def:CEcoh}
Let $(\calW,\rho)$ be a $\G$-module, and consider the space $C_\text{CE}^\bullet(\G,\calW)\doteq C^\infty_e(\G^\bullet,\calW)$ of smooth functions with values in $\calW$ that vanish when one of their arguments is the identity $e\in\G$. Then $C_\text{CE}^\bullet(\G,\calW)$ carries a natural cochain complex structure, with differential $\delta$:
\[
\delta_l : C_\text{CE}^l(\G,\calW) \to C_\text{CE}^{l+1}(\G,\calW)
\]
defined as follows:
\[
\delta_{0} w(g) \doteq \rho(g) w - w,
\]
for $l=0$, and for $l>0$:
\begin{align*}
\delta_{l} \sigma(g_1,  \cdots, g_{l+1}) 
= &{\rho}(g_1)( {\sigma}(g_2,  \cdots, g_{l+1} ))+\\
&+ \sum_{i} (-1)^i {\sigma}( g_1,  \cdots, g_i g_{i+1},  \cdots, g_{l+1})
+ (-1)^{l+1} {\sigma}(g_1,  \cdots, g_l).
\end{align*}

Similarly, if $(\calW,\rho)$ is a $\fG$-module, we have a cochain complex
\[
C_\text{CE}^\bullet(\fG,\calW)\doteq L(\wedge^\bullet \fG,\calW) 
\]
with differential $\delta'$ given by
\begin{align*}
    \delta_k'f(\xi_1\wedge  \cdots \wedge\xi_{k+1}) \doteq & \sum_{i} (-1)^i\rho(\xi_i) (f(\xi_1\wedge \cdots \hat{\xi}_i   \cdots \wedge\xi_{k+1}))\\
    & + \sum_{i<j} (-1)^{i+j-1}f([\xi_i,\xi_j]\wedge\xi_1\wedge  \cdots \hat{\xi}_i  \cdots \hat{\xi}_j  \cdots \wedge \xi_{k+1}).
\end{align*}

The cohomology groups of these two cochain complexes are called the group (resp.\ algebra) Chevalley--Eilenberg cohomology for the group (resp.\ algebra) module $\calW$, and are denoted by $H_{CE}^\bullet(\G,\calW)$ (resp. $H_{CE}^\bullet(\fG,\calW)$). We will denote the cocycles of the Chevalley-Eilenberg complex by $Z_{CE}^\bullet(\G,\calW)$ (resp. $Z_{CE}^\bullet(\fG,\calW)$).
\end{definition}

A special case of this definition is when the modules coincide as vector spaces $\calW=\dual{\fG}$, equippend respectively with the coadjoint action of the Lie group/algebra (see Definition \ref{def:dual*} for the relevant notions in the case of infinite-dimensional local lie groups/algebras). We have the following result (see \cite[Proposition 2.1.20]{RatiuOrtega03} for a proof):

\begin{proposition}[\cite{souriau1970structure}]\label{prop:algebracocyclefromgroupcocycle}
Let $\sigma\colon G\to \fG^*$ be a $\fG^*$ smooth group cocycle for the coadjoint action. Then, the map
\[
    \alpha\colon \fG\times\fG \to \mathbb{R}, \qquad \alpha(\xi,\eta)=\langle d_e\sigma(\xi),\eta\rangle
\]
is a Chevalley-Eilenberg two-cocycle for the coadjoint action of $\fG$ on $\fG^*$ if $\alpha$ is skew-symmetric.
\end{proposition}

\subsection{Infinite dimensional affine actions and orbits}\label{sec:dualPoisson}

Given a CE cocycle, one can define an affine action of $\G$ on $\dual{\fG}$ which generalises the coadjoint action:
\begin{definition}\label{def:affinecoad}
Let $C\in Z_{CE}^{1}(\G,\dual{\fG})$. 
The \emph{affine action} $\Ad_C^*$ of $\G$ on $\dual{\fG}$ is 
\[
    {\Ad_C^*: }\G\times \dual{\fG} \to \dual{\fG}, \qquad (g,\alpha) \mapsto {\Ad_C^*(g)\cdot\alpha \doteq} \Ad^*(g)\cdot\alpha + C(g).
\]
When $C=0$, the affine action is linear and coincides with the coadjoint action. The orbits under this action are denoted by
\[
    \mathcal{O}_\alpha \doteq \alpha\triangleleft_{{\Ad_C^*}}\G \subset \dual{\fG}.
\]
Infinitesimally, for $K\in {Z_{CE}^{2}(\fG,\mathbb{R})}$, we have the Lie algebra action,
\[
    \ad^*_K: \fG\times \dual{\fG} {\to \dual{\fG}}, \qquad (\alpha,\xi) \mapsto \ad^*_K(\xi)\cdot \alpha \doteq \ad^*(\xi) \cdot\alpha + K(\xi,\cdot).
    \qedhere
\]
\end{definition}

Nuclear Fr\'echet Lie algebras are reflexive (see Remark \ref{rmk:NFVS} and the discussion in Section \ref{sec:infinitedim}). One can endow their topological duals with a Poisson structure in the following way.

\begin{lemma}[Kirillov--Kostant--Souriau (KKS)]\label{lemma:KKS}
Let $\fG$ be a nuclear Fr\'echet Lie algebra. Then $\dual{\fG}$ is a (nuclear, reflexive) vector space for which $C^\infty_{\mathsf{p}}(\dual{\fG}) = C^\infty(\dual{\fG})$ admits a Poisson structure for any CE cocycle $K$ of $\fG$, represented by the bivector:
\[
    \Pi_{\dual{\fG}}^K\vert_{\alpha}(\xi\wedge\eta) \doteq  
    \langle\alpha,[\xi,\eta]\rangle + K(\xi,\eta) \equiv \langle  \ad^*_K(\xi)\cdot\alpha,\eta\rangle
\]
where $\alpha\in \dual{\fG}$ and $\xi,\eta\in T_\alpha^*\dual{\fG} \simeq \bidual{\fG} \simeq \fG$. \end{lemma}

\begin{proof}
This is a consequence of \cite[Remarks 4.3 and 4.4]{glockner2007applications}.

First observe that since $\fG$ is nuclear Fr\'echet, it is reflexive (Remark \ref{rmk:NFVS}), i.e.\ the inclusion $\iota: \fG\to \bidual\fG$ is an isomorphism. 
This equips $\bidual\fG$ with a Lie algebra structure via its identification with $\fG$, and identifies covectors in $\dual{T}_\alpha\dual{\fG}$ with elements of  $\bidual{\fG} \simeq \fG$.

With reference to the notion of partial Poisson structures (Definition \ref{def:partialpoisson}), this allows us to set $T^\mathsf{p}\dual{\fG}\equiv T^*\dual{\fG}$, and we have $C^\infty_{\mathsf{p}}(\dual{\fG}) = C^\infty(\dual{\fG})$.

Then, for every $K\in \mathrm{Z}_{CE}^2(\fG,\mathbb{R})$, one defines a bivector field $\Pi^K_{\dual\fG}$ over $\dual\fG$ by the formula in the statement of the lemma.
The proof that $\Pi_{\dual{\fG}}^K$ is Poisson is standard and follows from the Jacobi identity for $\fG$ and the cocycle condition for $K$.
\end{proof}

\begin{corollary}\label{cor:KKS}
The bivector $\Pi_{\dual{\fG}}^K$ admits symplectic leaves, which coincide with the connected components of the orbits $\mathcal{O}_\alpha$. We denote the corresponding symplectic structure by $\Omega_{[\alpha]}$.
\end{corollary}
\begin{proof}
The bivector $\Pi^K_{\dual{\fG}}\vert_{\mathcal{O}_\alpha}$ is (strongly) nondegenerate as a consequence of reflexivity by an argument analogous to that of Remark \ref{rmk:cotangentreflexive}. Then the claim follows from closedness of the associated symplectic form $\Omega_{[\alpha]}$ and \cite[Remark 2.4.6]{PelletierCabau}.
\end{proof}

\begin{remark}
An alternative proof for Lemma \ref{lemma:KKS} can be phrased within the framework of \cite{NeebSahlmannThiemann} for weak Poisson locally convex vector spaces. Consider a locally convex Lie algebra $\fG$, whose topological dual $\fG'$ is endowed with the strong topology: the adjoint map $\ad(\xi)\colon \fG\to\fG$ is continuous for all $\xi\in \fG$, and so is $\ad^*(\xi):\strong{\fG}\to\strong{\fG}$ (e.g.\ \cite[Proposition 19.5]{TransposeLinMap}). Then, employing the notation used by \cite[Example 2.14]{NeebSahlmannThiemann}, we can apply \cite[Corollary 2.11]{NeebSahlmannThiemann} to  $V\doteq \fG'$ and $V_\ast\doteq \fG$, so that $\fG'$ inherits a weak-poisson structure on $S(V_*)$ from the Lie algebra structure of $\fG$. When $\fG$ is reflexive  we have additionally the coincidence of dual and predual $V^*=\fG''=\fG=V_*$. This maximally enlarges the algebra on which one has a weak Poisson structure, as stated in \cite[Remark 2.13]{NeebSahlmannThiemann}.
\end{remark}

\begin{remark}\label{rmk:KKS-Hamiltonian}
Recall Definition \ref{def:dual*} of the local dual $\calW_\loc^*$ of the space of sections of a vector bundle, $\calW = \Gamma(\Sigma,W)$.
Since $\fG=\Gamma(\Sigma,\Xi)$ is a local Lie algebra, the affine/coadjoint action restricts to $\dual{\fG}_\loc \subset\dual{\fG}$, so that the orbit of $\alpha\in \dual{\fG}_\loc$ is automatically contained in $\dual{\fG}_\loc$, i.e.\ $\mathcal{O}_\alpha\subset \dual{\fG}_\loc\subset\dual{\fG}$. 
This means that $\mathcal{O}_\alpha$ can be seen a symplectic manifold in the local case as well.
However, the KKS structure introduced above restricts, in general, only to a partial Poisson structure on $\dual{\fG}_\loc\subset\dual{\fG}$, in the sense of \cite{PelletierCabau}, or weak-Poisson in the sense of \cite{NeebSahlmannThiemann}.

Let us recall here that the orbits $(\mathcal{O}_{\alpha},\Omega_{[\alpha]})$ carry a (transitive) Hamiltonian action by $\G$, inherited from the affine/coadjoint action on $\dual{\fG}$ (Definition \ref{def:affinecoad}). Since $\Im(\ad^*_K)$ is by construction tangent to $\mathcal{O}_\alpha$, we can use the same symbol to denote its restriction to $\mathcal{O}_\alpha\subset\dual{\fG}_\mathrm{loc}$, then for all $\alpha' \in \mathcal{O}_\alpha$,
\[
\bi_{\ad^*_K(\xi)}\Omega_{[\alpha]} \vert_{\alpha'} = \bd \langle i_{[\alpha]}(\alpha'), \xi\rangle
\]
and one sees that the momentum map is the embedding $i_{[\alpha]}:\mathcal{O}_\alpha \hookrightarrow \dual{\fG}_\mathrm{loc}$.
\end{remark}

\subsection{Central extension by a nontrivial cocycle}

In the presence of a CE cocycle, the coadjoint and affine orbits of Definition \ref{def:affinecoad} differ from each other. If we are given a Hamiltonian action $\rho$ of $\fG$ on a symplectic space $\X$ with a (local) momentum map $H$ that is not equivariant due to $K$, it is convenient to make consider a central extension $\wh{\fG}$ of $\fG$ governed by the cocycle itself, so that the extended action has an equivariant momentum map.

\begin{definition}[Central extension]\label{def:centralextension}
Let $\fG$ be a (possibly infinite-dimensional) Lie algebra and $K$ a CE cocycle thereof.  The \emph{centrally-extended} Lie algebra is the vector space $\wh{\fG}\doteq \fG\oplus \mathbb{R}$ together with the Lie bracket:
\[
    [(\xi,r),(\eta,s)] \doteq ([\xi,\eta],K(\xi,\eta)).
\]
Its dual is $\dual{\wh{\fG}}\simeq \dual{\fG}\oplus  \mathbb{R}$, with pairing $\langle (\alpha,a) , (\xi,r)\rangle \doteq \langle \alpha, \xi\rangle + ar$.

Furthermore, if $(\X,\omega,\fG,H)$ is Hamiltonian $\G$-space with action $\rho:\fG\to\mathfrak{X}(\X)$, introduce:
\begin{enumerate}[label=(\roman*)]
    \item the action $\wh{\rho}$ of $\wh{\fG}$ on $\X$, defined by $\wh{\rho}(\xi,r)\doteq\rho(\xi)$; 
    \item the map  $\wh{H}\doteq(H,1) \colon\X \to \dual{\wh{\fG}}$.\qedhere
\end{enumerate}
\end{definition}

From the above one deduces that the coadjoint action of $\wh{\fG}$ on $\dual{\wh{\fG}}$ is
    \[
    \wh{\ad}{}^*(\xi,r)\cdot(\alpha,a) =
    (\ad^*_K(\xi)\cdot\alpha, 0) ,
    \]
and therefore that the \emph{coadjoint} orbit of $(\alpha,1)$ in $\dual{\wh{\fG}}$ is given by the \emph{affine} orbit of $\alpha$ in $\dual{\fG}$: $\wh{\mathcal{O}}_{(\alpha,1)} =\mathcal{O}_\alpha \times \{1\}$. This construction restricts to local linear functionals in $\dual{\wh{\fG}}_\loc \simeq \dual{\fG}_\loc\oplus \mathbb{R}$.

\begin{theorem}[see e.g.\ \cite{RatiuOrtega03} Theorem 10.1.1.v]\label{thm:extendedreduction}
Let $(\X,\omega,\G,H)$ be a finite dimensional Hamiltonian $\G$-space, and let $\wh{H}$ and $\wh{\fG}$ as in Definition \ref{def:centralextension}. Then, $\wh{H}$
is a $\widehat\fG$-\emph{equivariant} momentum map for the action $\widehat\rho$ of $\wh{\fG}$ on $\X$ and the preimages of the affine and coadjoint orbits coincide, viz.
\[
{\wh{H}^{-1}(\wh{\mathcal{O}}_{(\alpha,1)})}
=
H^{-1}(\mathcal{O}_\alpha)  
\]
for every $\alpha\in \mathrm{Im}(H) \subset\dual{\fG}$. 

Additionally, if the quotient $\X/\wh{\G}$ is smooth and the reduction $H^{-1}(\mathcal{O}_\alpha)/\G$ is a smooth symplectic manifold for every $\alpha\in \dual{\fG}$, $\X/\wh{\G}$ admits a Poisson structure whose symplectic leaves are the reduction by $\G$ of the connected components of the fibres ${H}^{-1}(\mathcal{O}_\alpha)$.
\end{theorem}

\begin{remark}\label{rmk:extendedreduction}
This theorem holds for $(\X,\omega,\G,H)$ a finite dimensional Hamiltonian $\G$-space.
In finite dimensions, a proper and free action is sufficient to have a smooth orbit space $\X/\G$, and the regularity of the values of the momentum map ensures the existence of a regular foliation in $H^{-1}(\mathcal{O}_\alpha)$, which in turn guarantees that the reduced space is symplectic. The previous theorem is then presented here as a guideline, which will then need to be checked in concrete examples. For conditions on the existence of a centrally extended Lie \emph{group} $\widehat\G$ such that $\mathrm{Lie}(\widehat\G)=\widehat\fG$, we refer to \cite[Thm. VI.1.6 and Def. VI.1.3]{Neeb-locallyconvexgroups}.
\end{remark}

\subsection{Generalization to Lie algebroids}

The notion of differential forms and the associated de-Rham complex can be extended to the case of Lie algebroids as follows (see \cite{KosmannSchwarzbachSurvey}). For an introduction to algebroids of Fr\'echet type we refer to \cite{Baarsma}. Particularly relevant for us is the link with Poisson manifolds \cite{HuebschmannPoissonQuant1990}.

Let $\mathcal{V}$ be the fibre of a vector bundle $\A\to\calPhi$. By $\wedge^k \A$ we denote the completely anti-symmetrized $k$-fold tensor-product bundle over $\calPhi$ with fibre $\wedge^k\mathcal{V}$. We understand $\X_\pp\times\mathbb{R}$ as a trivial line bundle over $\X_\pp$. Then, $L(\wedge^k\X_\pp, \X_\pp \times\mathbb{R})$ is the space of fibrewise linear bundle morphisms that are bounded \cite{kriegl1997convenient}. Cf. Remark \ref{rmk:diffForms}.

\begin{definition}[Lie algebroid complex]\label{def:algebroidforms}
Let $\A\to \calPhi$ be a Lie algebroid, with anchor $\rho\colon \A\to T\calPhi$. The space of $\A$-forms, denoted by $\Omega^\bullet_\A(\calPhi)$, is the space of sections
\[
    \Omega^\bullet_\A(\calPhi)\doteq L(\wedge^k\A,\calPhi \times \mathbb{R})
\]
endowed with the $\A$-de Rham differential $\delta_\A : \Omega^\bullet_\A(\calPhi)\to \Omega^{\bullet+1}_\A(\calPhi) $

\begin{align*}
    \delta_\A\omega(a_1\wedge  \cdots \wedge a_{k+1}) \doteq &  \sum_{i=0}^{k+1}(-1)^i \rho(a_i)(\omega(a_1\wedge  \cdots \hat{a}_i   \cdots \wedge a_{k+1}))\\
      & + \sum_{i<j} (-1)^{i+j-1} \omega(\llbracket a_i,a_j\rrbracket\wedge a_1\wedge \cdots \hat{a}_i  \cdots \hat{a}_j  \cdots \wedge a_{k+1}),
\end{align*}

where the $a_i$'s are arbitrary sections of $\A\to\calPhi$.
The cohomology of the Lie algebroid complex $(\Omega_\A^\bullet(\calPhi),\delta_\A)$ (also known as $\A$-de Rham complex) is denoted by $H_\A^\bullet(\calPhi)$.
\end{definition}

\begin{remark}
Observe that the space of sections $\Gamma(\A)$ of a Lie algebroid comes equipped with a Lie algebra structure. Hence, one can construct the associated Chevalley--Eilenberg complex $C_\text{CE}^\bullet(\Gamma(\A),C^\infty(\calPhi))$ for the module $\calW = C^\infty(\calPhi)$. It is important to distinguish this notion from that of the $\A$-de Rham complex. 
Indeed, notice that $\Omega_\A^\bullet(\calPhi) = L(\wedge^\bullet\A,\calPhi \times \mathbb{R})$ is \emph{strictly} contained in $C_\text{CE}^\bullet(\Gamma(\A),C^\infty(\calPhi))$. An important example of this is the relation between the Chevalley--Eilenberg cochains $C_\text{CE}^\bullet(\mathfrak{X}(M),C^\infty(M))$ for the tangent algebroid (with identity anchor) $TM\to TM$, and its the de-Rham (sub)complex $\Omega^\bullet(M)$. (See \cite[Section 1.3]{fuks2012cohomology}.) Whenever $\A$ is an action Lie algebroid, though, one can understand the $\A$-de Rham complex as the Chevalley-Eilenberg complex for the \emph{fibre} Lie-algebra, as detailed below.
\end{remark}

\begin{theorem}
If $\rho:\A\to T \calPhi$ is an action Lie algebroid $\A = \calPhi \times \fG$, then
\[
(\Omega_\A^\bullet(\calPhi),\delta_\A) \simeq (C^\bullet_\text{CE}(\fG,C^\infty(\calPhi)), \delta')
\]
is an isomorphism of complexes.
\end{theorem}
\begin{proof}
This can be checked by direct inspection of the definitions, recalling that, in the case of a trivial bundle $\A = \fG \times \calPhi$, we have $L(\wedge^\bullet\A,\calPhi \times \mathbb{R}) = C^\infty(\calPhi,(\wedge^\bullet \fG)^*)$.
\end{proof}

\begin{remark}\label{rmk:PoissonLieAlg}Lemma \ref{lemma:KKS} can be seen as a special case of a more general relation between Poisson manifolds and Lie algebroids, which has been proven by Coste, Dazord and Weinstein in \cite{CosteDazordWeinstein} for finite dimensional manifolds. Here is a summary of their results. The cotangent bundle of a Poisson manifold $(Q,\Pi)$ is a Lie algebroid $A=T^*Q \to Q$ with anchor map $\varrho = \Pi^\#$ and bracket\footnote{For ease of comparison, we denote $\bd$ and $\L$ the differential and Lie derivative over $A$ even if this is finite dimensional in \cite{CosteDazordWeinstein}.} 
\[
\llbracket a,b\rrbracket_\Pi = \L_{\Pi^\#(a)}b  -  \L_{\Pi^\#(b)}a  - \bd \Pi(a\wedge b).
\]
Conversely, a Lie algebroid on $T^*Q\to Q$ with anchor $\varrho$ and bracket $\llbracket\cdot,\cdot\rrbracket$ defines a Poisson bivector $\Pi^\sharp = \varrho$ on $Q$ iff the brackets of two closed 1-forms, $\bd a = 0 = \bd b$, is also closed, $\bd \llbracket a,b\rrbracket = 0$. Such a Poisson bivector on $Q$ is the only one for which the two algebroids are isomorphic.

In Lemma \ref{lemma:KKS}, dimensionality aside, $Q = \dual{\fG}$ and one identifies $T_\alpha^*\dual{\fG} \simeq \fG$. Under this identification $\Pi_{\dual{\fG}}^{K\sharp} = \ad_K^*$, and $T^*\dual{\fG}$ has the structure of an action Lie algebroid with anchor $\varrho = \Pi_{\dual{\fG}}^{K\sharp}$. Indeed, we can view $T^*\dual{\fG}\simeq \fG\times \dual{\fG}$ as a vector bundle over $\dual{\fG}$ with a map $\Pi_{\dual{\fG}}^K\equiv \ad^*_K \equiv \varrho \colon \fG\times \dual{\fG} \to T\dual{\fG}$, and it is a matter of a straightforward calculation to check that the Poisson algebroid on $T^*\dual{\fG}$ is isomorphic to the action Lie algebroid $\varrho=\ad_K^*:\fG\times\dual{\fG} \to \dual{\fG}$
\[
 \llbracket \xi_a, \xi_b\rrbracket  = \L_{\ad^*_K(\xi_a)} \xi_b- \L_{\ad^*_K(\xi_b)} \xi_a + [\xi_a , \xi_b] = \xi_{\llbracket a,b \rrbracket_{\Pi^{K}_{\dual{\fG}}}}.
\]
where here we made explicit the mapping between \emph{1-forms} $a,\,b$ over $\dual{\fG}$ and $\fG$-valued \emph{functions} $\xi_a,\,\xi_b$ over $\dual{\fG}$ (induced by the isomorphism $\iota:\fG\to\bidual\fG$ introduced in the proof of Lemma \ref{lemma:KKS}). 
\end{remark}

\section{Covariant phase space method}\label{app:covariant}

We here provide a bird-eye review of Noether's theorem and the covariant phase space method.
The goal of this discussion is to clarify which of the Assumptions \ref{ass:setup}(1--3) are expected to automatically hold when $(\X,\bom)$ is the geometric phase space of a Lagrangian gauge theory, and which others require extra hypotheses.

Since the scope of this discussion is merely motivational, we will not attempt an exhaustive or self-contained treatment, for which we refer the reader to \cite{Zuckerman,LeeWald} and \cite{AshtekarBombelliKoul, KT1979,  CrnkovicWitten1987, Crnkovic_1988} (see also \cite{gotay1998momentum, gotay2004momentum} and, more recently, \cite{Khavkine2014} and \cite{MargalefCPS} where the relation with the Peierls and canonical approach is explained).

Consider a cylindrical ``spacetime'' manifold\footnote{We will not assume to have a Lorentzian metric, but we keep the terminology for the sake of familiarity.} $M\simeq{\Sigma}\times\mathbb{R}$, with ${\Sigma}\hookrightarrow M$ a closed (Cauchy) hypersurface, i.e.\ $\partial\Sigma=\emptyset$. 
Let $\overline{F}\to M$ be some fibre bundle specifying the field content, and $\overline{\mathcal{F}} \doteq \Gamma(M,\overline{F})$ be the \emph{off-shell covariant configuration space} over $M$.

A Lagrangian density is a local form $\boldsymbol{L}\in \oloc^{\text{top},0}(M\times\overline{\mathcal{F}})$. 
By means of Theorem \ref{thm:LocFormDec}, it defines the Euler--Lagrange form $\boldsymbol{EL} \doteq(\bd \boldsymbol{L})_\text{src}$, whose zero locus defines the equations of motions, and an equivalence class of symplectic potential densities $[\boldsymbol{\Theta}]$---where the brackets mean that $\boldsymbol{\Theta}$ is defined only up to $d$-closed (but in fact $d$-exact) terms:\footnote{To avoid clutter, we employ $(d,\bd$) to respectively denote the horizontal and vertical differential on both $\oloc^{\bullet,\bullet}(M\times\overline{\mathcal{F}})$ and $\oloc^{\bullet,\bullet}(\Sigma\times\X)$, see below.}
\begin{equation}
    \label{eq:deltaL}
\bd \boldsymbol{L} = \boldsymbol{EL} + d \boldsymbol\Theta.
\end{equation}
This is a consequence of Theorem \ref{thm:LocFormDec} and
\begin{theorem}[Taken's acyclicity \cite{Takens77}]
The complex $\left(\oloc^{p,q}(M\times \overline{\mathcal{F}}),d\right)$ is acyclic for $p<\mathrm{top}$ and $q\geq 1$.
\end{theorem}

By vertical differentiation, $[\boldsymbol{\Theta}]$ defines a \emph{covariant symplectic density} (where the terminology owes to the fact that $\boldsymbol\Omega$ is not a top-form over $M$, and it is thus not an actual symplectic density over $\overline{\mathcal{F}}$.) $[\boldsymbol{\Omega}] = [\bd\boldsymbol{\Theta}]$. 
In virtue of \eqref{eq:deltaL}, any representative $\boldsymbol{\Omega}\in[\boldsymbol{\Omega}]$ is \emph{conserved on-shell}, i.e.\ the form $\boldsymbol{\Omega}$ is $d$-closed modulo terms that vanish on the zero locus of $\boldsymbol{EL}$. We denote this symbolically by 
\[
d \boldsymbol{\Omega} \approx 0.
\]

From now on we will work with a choice of representatives of all horizontal cohomology classes, and investigate the ensuing ambiguities at the end.

\begin{remark}[Covariant Phase Space]\label{rmk:CPS}
Motivated by the on-shell conservation of $\boldsymbol\Omega$, one can then define the \emph{covariant phase space} $\overline{\mathcal{F}}_{\boldsymbol{L}}$ as the vanishing locus of the Euler--Lagrange form $\boldsymbol{EL}$ in $\overline{\mathcal{F}}$ equipped with (the pullback of) $\boldsymbol\Omega$.
If $M = \Sigma \times \mathbb{R} \ni (x,t)$ with $\pp\Sigma=\emptyset$, in view of the on-shell conservation of $\boldsymbol\Omega$, the integrated form $\Omega_\Sigma \doteq \int_\Sigma \boldsymbol\Omega$ does not depend on $t$. Then, one can argue that, barring technical complications, $(\overline{\mathcal{F}}_{\boldsymbol L}, \Omega_\Sigma)$ is a symplectic manifold isomorphic to $(\uCo,\uomegao)$. See e.g. \cite{LeeWald, KT1979, CMRCorfu,CattaneoSchiavinaPCH}.
However, since we are interested in making contact with the geometric phase space $(\X,\omega=\int_\Sigma \bom)$ as our starting point for reduction procedure, we shall henceforth focus on an off-shell analysis.
\end{remark}

Now, assume that the off-shell covariant configuration space $\overline{\mathcal{F}}$ is acted upon by some \emph{local} (gauge) Lie group $\overline\G$ with Lie algebra $\overline{\fG}\ni\xi$, whose infinitesimal (local) action we denote by $\overline{\rho}:\overline\fG\to\Xloc(\overline{\mathcal{F}})$.
Assume also that $\boldsymbol{L}$ is invariant (possibly up to boundary terms) under the infinitesimal action $(\overline{\G},\overline{\rho})$, i.e.\ ${\L}_{\overline{\rho}({\xi})} \boldsymbol{L} = d \langle\boldsymbol{R}, {\xi} \rangle$ for some $\boldsymbol{R}\in\oloc^{\text{top}-1,0}(\Sigma\times\overline{\mathcal{F}},\overline{\fG}{}^*)$.
Then, Noether's second theorem states that, for all $\xi\in\overline\fG$, the \emph{Noether current}
\begin{equation}
    \label{eq:Noether}
\langle\overline\bH,\xi\rangle\doteq \langle \boldsymbol{R},\xi\rangle-\bi_{\overline{\rho}({\xi})}\boldsymbol{\Theta},
\end{equation}
where $\overline{\bH}\in \oloc^{\text{top}-1,0}(\Sigma\times\overline{\mathcal{F}},\overline{\fG}{}^*)$ is a dual-valued local form (Definition \ref{def:dualvaluednew}), is not only closed on-shell of the equations of motion---which is the content of Noether's first theorem, $d\langle\overline{\bH},{{\xi}}\rangle=\bi_{\overline{\rho}({\xi})}\boldsymbol{EL}\approx0$---but it is also $d$-exact:
\begin{equation}
[\langle\overline\bH,\xi\rangle ] \approx 0.
\label{eq:Noether2}
\end{equation}

The full content of Noether's second theorem can be made explicit by noting that this equality must hold for \emph{all} $\xi\in\overline{\fG}$. It is possible to show \cite{LeeWald} that $\overline{\bH}$ splits into a \emph{order-0} part, $\overline{\bH}_\circ \in \oloc^{\text{top}-1,0}(M\times \overline{\mathcal{F}},\overline{\fG}^*_{(0)})$ which must vanish on-shell, and a total derivative $d \overline{\bh}$, i.e.
\begin{equation}
    \label{eq:covariant1}
\overline\bH = \overline{\bH}_\circ + d \overline{\bh}, \qquad \overline{\bH}_\circ \approx 0.
\end{equation}

Furthermore, by taking the Lie derivative of Equation \eqref{eq:deltaL} along $\overline{\rho}({\xi})$ and using the definition of the Noether current \eqref{eq:Noether}, one can show that there exist dual-valued forms $\boldsymbol{\mu}\in\oloc^{\text{top}-1,1}(M\times\overline{\mathcal{F}},\overline{\fG}{}^*)$ and $\overline{\boldsymbol{r}}\in\oloc^{\text{top}-2,1}(M\times\overline{\mathcal{F}},\overline{\fG}{}^*)$, such that\footnote{Since $\overline{\boldsymbol\mu}$ is a 1-form over $\overline{\mathcal{F}}$, the interpretation of the equation $\overline{\boldsymbol{\mu}}\approx0$ requires some care. Here, the symbol $\approx$ is used to mean that the expression on the lhs is evaluated at an on-shell configuration $\varphi_0 \in\overline{\mathcal{F}}_{\boldsymbol{L}}\doteq\{\varphi:\boldsymbol{EL}(\varphi)=0\}$, but the differential $\bd$ can still be ``off-shell''. That is $\overline{\boldsymbol{\mu}}(\varphi_0,\delta\varphi)$ is a valid expression for all $\delta\varphi\in T_{\varphi_0} \overline{\mathcal{F}}$, which must vanishes if $\delta\varphi$ is a ``Jacobi field'', i.e.\ an ``on-shell'' perturbation  $\delta\varphi\in T_{\varphi_0} \overline{\mathcal{F}}_{\boldsymbol{L}}\subset  T_{\varphi_0}\overline{\mathcal{F}}$ satisfying the linearized equations of motion. To ease comparison, let us note that the quantity $\L_{\overline{\rho}(\xi)} \boldsymbol{EL}\equiv d\overline{\boldsymbol{\mu}}$ is a rewriting of the rhs of Equation (3.26) of \cite{LeeWald}.}
\[
\L_{\overline{\rho}(\xi)} \boldsymbol{EL} = d \langle\overline{\boldsymbol{\mu}},\xi\rangle, \quad \overline{\boldsymbol{\mu}}   \approx 0
\]
and
\begin{equation}
    \label{eq:covariant2}
\bi_{\overline{\rho}({\xi})} \boldsymbol{\Omega} - \bd \langle\overline{\bH},\xi\rangle = \langle \overline{\boldsymbol{\mu}},\xi\rangle + d\langle \overline{\boldsymbol{r}},\xi\rangle.
\end{equation}

Now, restricting the (elements of the) off-shell covariant configuration space $\overline{\mathcal{F}}$ to a (Cauchy) hypersurface $\iota_\Sigma:\Sigma\hookrightarrow M$, together with integration of $\boldsymbol{\Omega}$ over $\Sigma$, produces---after presymplectic reduction---the \emph{off-shell} geometric (or ``canonical'') phase space $(\X,\omega=\int_\Sigma\bom)$, where $\bom$ is a (possibly trivial) reduction of $\boldsymbol{\Omega}$, thought of as a $(\text{top},2)$ local form on fields on $\Sigma$.

\begin{remark}
Schematically, this works as follows. There is a natural map from $\overline{\mathcal{F}}$ to $\overline{\mathcal{F}}_\Sigma$---the space of restrictions of fields and their transverse jets at $\Sigma$---w.r.t.\ which $\boldsymbol\Omega_\Sigma \doteq \iota^*_\Sigma \boldsymbol{\Omega}$ is basic. In typical applications, the space $(\overline{\mathcal{F}}_\Sigma,\Omega_\Sigma)$ is presymplectic, with reduction the locally symplectic space $(\X,\omega)$. If $\pi_{\Sigma}$ is the composition of the restriction map with the quotient by the pre-symplectic kernel, we have that $\pi_{\Sigma}:\overline{\mathcal{F}}\to \X$ and $\Omega_\Sigma=\pi^*_\Sigma \omega$. See e.g. \cite{LeeWald, KT1979, CMRCorfu,CattaneoSchiavinaPCH}. Finally, note that, contrary to the on-shell Covariant Phase Space construction of Remark \ref{rmk:CPS}, this off-shell construction cannot count on the on-shell conservation of $\boldsymbol\Omega$ and therefore does depend on the choice of surface $\Sigma$---even when $M=\Sigma\times \mathbb{R}$ and $\Sigma$ has no boundary.
\end{remark}

One can then ask whether the action $(\overline{\G},\overline{\rho})$ on the covariant phase space $(\overline{\mathcal{F}},\boldsymbol{\Omega})$ over $M$ descends to a local \emph{symplectic} action $(\G,\rho)$ on the geometric phase space $(\X,\omega)$ over $\Sigma$.  The answer to this question is controlled by the properties of the two terms on the rhs of Equation \eqref{eq:covariant2}.

\begin{remark}[$\overline{\boldsymbol\mu}\neq 0$ and $\overline{\boldsymbol r} \neq 0$]\label{rmk:obstructionsCPS}
A glimpse at Equation \eqref{eq:covariant2} shows that, whereas the last term $d\langle \overline{\boldsymbol{r}}, \xi\rangle$ vanishes upon integration over the closed hypersurface $\Sigma$, the term $\langle\overline{\boldsymbol\mu},\xi\rangle$ can turn into an obstruction to the Hamiltonian flow of $(\G,\rho)$ over $(\X,\omega)$.
E.g. in General Relativity $\iota^*_\Sigma\langle \overline{\boldsymbol{\mu}},\xi\rangle \neq 0$ for $\xi=\xi_\perp$ a vector field transverse to $\Sigma\hookrightarrow M$: it follows that transverse diffeomorphisms \emph{fail} to act symplectically on the geometric phase space of General Relativity \cite{LeeWald}.
Similarly, even considering only diffeomorphisms tangent to $\Sigma$, i.e.\ with $\xi_\perp = 0$, the term $d\langle\overline{\boldsymbol r},\xi\rangle$ is non-trivial and $\int_\Sigma d\langle\overline{\boldsymbol r},\xi\rangle$ vanishes if and only if $\xi$ is tangent to $\pp\Sigma$.

Typically, the relation between $\overline{\G}$ and $\G$ is given by $\G=\overline{\G}/\overline\G_{\overline{\Sigma}}$ where $\overline\G_{\overline{\Sigma}}= \{g\in\overline{\G}\ :\ g\vert_{\Sigma}=e\}$. But as we said, this is \emph{not} always the case. E.g.\ the problems arising with the diffeomorphisms transverse to $\Sigma$ are related to the fact that the subgroup of diffeomorphisms that stabilise a Cauchy hypersurface is not normal in $\mathrm{Diff}(M)$ \cite[Example 8.13]{blohmann2018hamiltonian}. Finally, note that going on-shell of the constraints \emph{partially} solves this problem, however other subtleties arise \cite{LeeWald}.
\end{remark}

In many cases of interest, the Euler--Lagrange form $\boldsymbol{EL}$ is gauge invariant even off-shell and therefore $\overline{\boldsymbol{\mu}} \equiv 0$. Then,  the action $(\overline{\G},\overline{\rho})$ on the covariant phase space $(\overline{\mathcal{F}},\boldsymbol{\Omega})$ does descend to a symplectic action $(\G,\rho)$ on the geometric phase space $(\X,\omega)$ whenever $\pp\Sigma = \emptyset$. 
However, this does not mean that the $\G$-space $(\X,\bom,\G)$ is necessarily \emph{locally Hamiltonian}  (Definition \ref{def:Ham+mommaps}), since the term $d\overline{\boldsymbol{r}}$ can still obstruct the Hamiltonian flow equation.

Let us express these statements in formulas. Passing from $(\overline{\mathcal{F}}, \boldsymbol{\Omega})$ to $(\X,\bom)$, one expects Equation \eqref{eq:covariant1} to descend to
\begin{equation}
    \label{eq:canonical1}
\bH = {\bH}_\circ + d {\bh},
\end{equation}
with $\bHo$ a \emph{order-0} dual-valued $(\text{top},0)$-form on $\Sigma\times \X$, which defines the \emph{canonical} constraint set,
\begin{equation}
    \label{eq:canonical2}
\C \doteq \bHo^{-1}(0) \subset\X,
\end{equation}
and, whenever $\overline{\boldsymbol{\mu}}=0$, Equation \eqref{eq:covariant2} to descend to
\begin{equation}
   \label{eq:canonical3}
\bi_{\rho({\xi})} \bom - \bd \langle{\bH},\xi\rangle = d\langle {\boldsymbol{r}},\xi\rangle.
\end{equation}

\begin{remark}[$\overline{\boldsymbol\mu}=0$ implies equivariance of $\bHo$]\label{rmk:mucovariance}
Noether's first theorem states that
\[
d \langle\overline{\bH},\xi\rangle = \bi_{\overline{\rho}(\xi)}\boldsymbol{EL}.
\]
Then, taking the Lie derivative of this expression along $\overline{\rho}(\eta)$ one obtains the relation:
\[
d( \langle\L_{\overline{\rho}(\eta)}\overline{\bH},\xi\rangle
- \langle\overline{\bH},[\eta,\xi]\rangle ) = \bi_{\overline{\rho}(\xi)} \L_{\overline{\rho}(\eta)}\boldsymbol{EL} \equiv d \bi_{\overline{\rho}(\xi)} \langle\overline{\boldsymbol\mu},\eta\rangle.
\]
Thus, we see that the condition $\overline{\boldsymbol\mu} = 0$ is equivalent to the (infinitesimal) equivariance of $\overline{\bH}$ up to $d$-exact terms. Therefore, in passing from $(\overline{\mathcal{F}},\boldsymbol{\Omega}, \overline{\bH})$ to $(\X,\bom,\bH)$, one sees that equation $\overline{\boldsymbol\mu} = 0$ implies the off-shell \emph{weak} equivariance of $\bH$---i.e.\ the equivariance of $\bHo$ (see Definition \ref{def:equivariance} and Proposition \ref{prop:equi}). This last statement is in turn equivalent to the canonical constraints being first class.
\end{remark}

Therefore, by inspection of \eqref{eq:canonical3}, we conclude that for the action of $\G$ on $(\X,\bom)$ to also be \emph{locally} Hamiltonian, we further need to assume that $\boldsymbol{r}$ can be reabsorbed away in the ambiguities we encountered in the above construction. Let us survey them.

The first ambiguity one encounters (one which we had not mentioned yet) is the coboundary ambiguity appearing in the choice of a Lagrangian density:
\[
\boldsymbol{L} \; \mapsto\; \boldsymbol{L} + d \boldsymbol{\ell}.
\]
This is because such a change in $\boldsymbol{L}$ does not affect the physical datum \emph{par-excellence}: the equations of motion, encoded in $\boldsymbol{EL}$.
This ambiguity, however, affects $[\boldsymbol{\Theta}]$ by a $\bd$-exact term. Thus, two ambiguities affect the choice of $\boldsymbol{\Theta}$:
\[
\boldsymbol{\Theta} \; \mapsto\; \boldsymbol{\Theta} + \bd \boldsymbol{\ell} + d \overline{\boldsymbol{\vartheta}}.
\]
Only the second ambiguity is reflected in $\boldsymbol{\Omega}$,
\[
\boldsymbol{\Omega} \doteq \bd  \boldsymbol{\Theta}\; \mapsto\; \boldsymbol{\Omega}  + \bd d \overline{\boldsymbol{\vartheta}},
\]
The definition of $\boldsymbol{R}$ suffers a similar ``corner'' ambiguity as $\boldsymbol{\Theta}$
\[
\boldsymbol{R}\; \mapsto\;\boldsymbol{R} + \L_{\overline{\rho}}\boldsymbol{\ell} + d \overline{\boldsymbol{\beta}}.
\]
Accordingly, the Noether current is defined up to
\[
\overline{\bH} \doteq \boldsymbol{R} - \bi_{\overline{\rho}}\boldsymbol\Theta\; \mapsto\; \overline{\bH} + d( \overline{\boldsymbol{\beta}} - \bi_{\overline{\rho}}\overline{\boldsymbol{\vartheta}}).
\]
Combining these expressions, we see that Equation \eqref{eq:covariant2} becomes:
\[
\bi_{\overline{\rho}} \boldsymbol{\Omega} - \bd \overline{\bH} =  \overline{\boldsymbol{\mu}} + d\big( \overline{\boldsymbol{r}} -  \L_{\overline{\rho}} \overline{\boldsymbol{\vartheta}} + \bd\overline{\boldsymbol{\beta}}\big),
\]

Thus, we see that to ensure that the action of $\G$ on $(\X,\bom)$ is \emph{locally} Hamiltonian, beyond $\overline{\boldsymbol\mu}=0$ we shall further assume that a $(\text{top}-2,1)$-local form $\overline{\boldsymbol{\vartheta}}$ and a dual-valued $(\text{top}-2,0)$-local form $\overline{\boldsymbol{\beta}}$ exist such that
\[
d(\overline{\boldsymbol{r}} - \L_{\overline{\rho}}\overline{\boldsymbol{\vartheta}} + \bd \overline{\boldsymbol{\beta}}) =  0.
\]

Indeed, if this is the case, then Equation \eqref{eq:covariant2} descends to a \emph{local Hamiltonian flow} equation on $(\X,\bom)$: 
\begin{equation}
    \label{eq:canonical4}
\bi_{\rho({\xi})} \bom - \bd \langle{\bH},\xi\rangle =0.
\end{equation}

Of course, such local forms $(\overline{\boldsymbol{\vartheta}}, \overline{\boldsymbol{\beta}})$ do not need to be unique nor do they need to exist. E.g.\ (spatial) diffeomorphisms over $\X=T^\vee\mathrm{Riem}(\Sigma)$ provide an example where no such choices are known.

\begin{remark}[Symplectic vs. Covariant]\label{rmk:covambiguity}
We conclude by noticing that a slight tension exists between the symplectic and Lagrangian interpretations of $\bH$. 
From a purely symplectic perspective, one thinks of $\bH$ as a local momentum map satisfying Equation \eqref{eq:canonical4}; then, from this perspective, $\bH$ is unique up to constant shifts $\balpha\in\oloc^{\text{top}}(\Sigma,\dual{\fG})$. But constant shifts this general are not allowed if we want to interpret $\bH$ as coming from a Noether current which satisfies Noether's second theorem (Equation \eqref{eq:Noether2}): for then only $d$-exact shifts are allowed, i.e.\ $\balpha=d\boldsymbol{\beta}$. 
This viewpoint is consistent with Remark \ref{rmk:mucovariance}, which states that the (local) momentum map $\overline{\bH}$ is equivariant possibly up to corner cocycles.
These observations can be seen as an explanation of Assumption \ref{assA:equiv} in the main text.
\end{remark}

Finally, let us emphasise once again that an action functional for a variational problem defines a form $\boldsymbol{\Omega}$, and thus a symplectic density $\bom$ only up to $d$-exact terms. This means that, given the data of a classical field theory, our results will depend on the \emph{choice} of an appropriate symplectic density $\bom$. 
One prominent example of corner ambiguity that we will analyze is the one (heuristically) related to the choice of $\theta$-vacua (see Section \ref{sec:thetaQCD}).

This said, natural candidates for such an $\bom$ usually exist.
Indeed, the arbitrary addition of $d$-exact terms will in general interfere with the existence of a local Hamitlonian action; moreover, since adding boundary terms introduces derivatives, one has:

\begin{proposition}\label{prop:bomuniqueness}
If an \emph{ultralocal} symplectic density $\bom$ exists, then it is uniquely determined by the action functional $\boldsymbol{L}\in\oloc^{\text{top},0}(M\times \mathcal{E})$.
\end{proposition}

\section{Equivariant Hodge--de Rham decompositions}\label{app:Hodge}

In this section we denote
\[
n \doteq \dim\Sigma,
\]
and assume $G$ compact and semisimple (or, $G=(\mathbb{R},+)$, in which case one recovers the usual Hodge--de Rham tehory \cite{SchwarzHodgeBook}).  We denote (minus) the Killing inner product in $\fg$ by $\tr( \cdot \; \cdot)$.

Let $P\to \Sigma$ be a principal $G$-bundle over a compact Riemannian manifold $(\Sigma,\gamma_{ij})$, possibly with boundary, and let $\Ad P \doteq P \times_{(G,\Ad)} \fg$ the associated adjoint bundle.

Consider the space of $\Ad$-equivariant differential forms:
\[
\Omega^\bullet(\Sigma,\Ad P) \doteq \Gamma(\wedge^{\bullet}T^*\Sigma \times_\Sigma \Ad P).
\]
Then, if $\star$ is the Hodge star operator associated to the positive definite metric $\gamma$ and $\balpha_i\in\Omega^{k_i}(\Sigma,\Ad P)$,
\[
\langle\langle \balpha_1, \balpha_2\rangle\rangle = \delta_{k_1 k_2}\int_\Sigma \tr(\balpha_1 \wedge \star \balpha_2)
\] 
defines a positive-definite inner product on $\Omega^\bullet(\Sigma,\Ad P)$.

On the other hand, a principal connection $A$ over $P$, defines the equivariant differential
\[
d_A : \Omega^k(\Sigma,\Ad P) \to \Omega^{k+1}(\Sigma,\Ad P), \quad \alpha \mapsto d_A \balpha \doteq d \balpha + [A,\balpha],
\]
which is nilpotent iff $A$ is flat (a property we will not assume).
Rather, we shall assume that $A$ is irreducible, i.e.\ that in degree-0
\[
(\ker d_A)_{k=0} \equiv 0,
\]
and comment on the reducible case at the end.

The formal adjoint of $d_A$ is then 
\[
d_A^\star \doteq (-1)^{n(k+1)+1}\star d_A \star.
\]
We denote $\Delta_A \doteq d_A^\star d_A + d_A d_A^\star$ the corresponding equivariant Laplace--Beltrami operator.

In the presence of corners $\pp\Sigma\neq\emptyset$, we follow \cite{SchwarzHodgeBook} and introduce the following notation for the tangential and normal component of $\balpha\in\Omega^k(\Sigma,\Ad P)$ at $\pp\Sigma$ (seen as embedded in $\Sigma$). Introducing the normal projection $p_\parallel: T_{\pp \Sigma}\Sigma \to T\pp\Sigma\subset T_{\pp\Sigma} \Sigma$:
\[
\mathsf{t}\balpha \doteq p_\parallel^*\balpha\vert_{\pp\Sigma}
\quad\text{and}\quad
\mathsf{n}\balpha \doteq \balpha\vert_{\pp\Sigma} - \mathsf{t}\balpha.
\]

Locally, these objects are understood as ($\fg$-valued) restrictions to $\pp\Sigma$ of sections of $T\Sigma$---and \emph{not} as sections of $T\pp\Sigma$. 

Noting that $\mathsf{t}\star{\balpha} = \star \mathsf{n} \balpha$ (and viceversa), one obtains the following equivariant version of Green's formula:
\[
\langle\langle d_A\balpha, \boldsymbol\beta \rangle\rangle = \langle\langle \balpha, d_A^\star\boldsymbol\beta\rangle\rangle + \int_{\pp\Sigma} \tr(\mathsf{t}\balpha \wedge \star \mathsf{n}\boldsymbol\beta).
\]

We then introduce the space of equivariant Dirichlet and Neumann forms,
\begin{align*}
\Omega_D^\bullet(\Sigma, \Ad P)& \doteq \{ \balpha \,|\, \mathsf{t}\balpha =0\},\\
\Omega_N^\bullet(\Sigma, \Ad P)& \doteq \{ \balpha \,|\, \mathsf{n}\balpha =0\}.
\end{align*}

The following holds:

\begin{proposition}[Equivariant Hodge--de Rham decomposition with boundary]\label{prop:HodgeEqui}
A choice of principal connection $A\in \mathrm{Conn}(P\to\Sigma)$ induces the following two decompositions, named after Dirichlet and Neumann, of the spaces of equivariant differential $0$- and $1$-forms, $\Omega^k(\Sigma, \Ad P)$, $k=0,1$:\footnote{The equality topped with ``irr.'' holds only if $A$ is irreducible, whereas the equivalence $\equiv$ holds in all cases.}
\[
\text{Neumann}: \quad
\begin{cases}
\Omega^0(\Sigma,\Ad P) = \Im (d_A^\star\vert_N) \oplus \ker d_A \stackrel{\text{irr.}}{=} \Im(d_A^\star\vert_N)\\
\Omega^1(\Sigma,\Ad P)  = \Im d_A \oplus (\ker d_A^\star)_N 
\end{cases}
\]
and
\[
\text{Dirichlet}: \quad
\begin{cases}
\Omega^0(\Sigma,\Ad P) = \Im d_A^\star \oplus (\ker d_A)_D \equiv \Im d_A^\star\\
\Omega^1(\Sigma,\Ad P)  = \Im(d_A\vert_D) \oplus \ker d_A^\star.
\end{cases}
\]
where for brevity we denoted $\Im(d_A^\star\vert_N) \doteq d_A^\star \Omega^1(\Sigma,\Ad P)$, $(\ker d_A^\star)_N \doteq \ker d_A^\star \cap \Omega^0_N(\Sigma,\Ad P),$ etc. with the form degrees always left implicit.
\end{proposition}

\begin{remark}[Open maps]\label{rmk:openmaps}
Denoting $L^2$-adjoints with a dagger, $\bullet^\dagger$, notice that
\[
(d_A\vert_D)^\dagger \doteq d_A^\star ,
\quad
(d_A^\star\vert_N)^\dagger = d_A
\quad\text{and}\quad
(d^\star_A\vert_N)^{\dagger\dagger} = d^\star_A\vert_N.
\]
From this one can conclude that both $d_A\vert_D$ and $d_A$ are weakly open maps as the image of their adjoints is closed by Proposition \ref{prop:HodgeEqui}.\footnote{We thank T.\ Diez for suggesting this argument.}
\end{remark}

This proposition generalises results from the usual Hodge--de Rham theory \cite{SchwarzHodgeBook} to the equivariant setting in degree 0 and 1 (in the context of the YM phase space, the analogous generalization in the absence of boundaries dates back at least to \cite{Arms1979, Arms1981, Garcia1980Hodge, MitterViallet1981}, see also \cite[Section 2.1.3]{DiezPhD}). Note that the nilpotency of $d_A$---i.e.\ the flatness of $A$---is \emph{not} required, so these decompositions are not cohomological in nature.

We now sketch a pedestrain argument for the above decompositions; an abstract argument is presented in \cite[Section 2.1.3]{DiezPhD}.
We shall focus on the case $k=1$ since the case $k=0$ can be dealt similarly. 
We want to show that $\balpha\in\Omega^1(\Sigma,\Ad P)$ admits---with obvious notation---the following decompositions
\[
\balpha = d_A \lambda + \bar\balpha_N = d_A \lambda_D + \bar\balpha.
\]
To check the uniqueness of these decompositions when $A$ is irreducible, one uses elliptic theory to argue that $\lambda$ and $\lambda_D$ are respectively (uniquely) determined by the following Neumann\footnote{In view of its boundary conditions, which are Robin rather than Neumann, we should maybe refer to this boundary value problem as ``equivariant Neumann''---but for simplicity we refrain from doing so.} and Dirichlet elliptic boundary-value problems:
\[
\begin{cases}
\Delta_A \lambda = d_A^\star \balpha & \text{in } \Sigma\\
\mathsf{n}d_A \lambda = \mathsf{n} \balpha &\text{at }\pp\Sigma
\end{cases}
\qquad\text{and}\qquad
\begin{cases}
\Delta_A \lambda_D = d_A^\star \balpha & \text{in } \Sigma\\
\mathsf{t}\lambda_D = 0 &\text{at }\pp\Sigma
\end{cases}
\]
It is then easy to verify that (\emph{i}) $\bar\balpha_N \doteq \balpha - d_A \lambda$ and $\bar\balpha\doteq \balpha - d_A \lambda_D$ are in the kernel of $d_A^\star$, with $\bar\balpha_N$ further satisfying Neumann boundary condition, and---using Green's formula---that (\emph{ii}) they are indeed orthogonal to $d_A\lambda$ and $d_A \lambda_D$ respectively.

\begin{remark}[Reducible connections]\label{rmk:EllipticReducible}
Technically, $\lambda$ and $\lambda_D$ are uniquely determined by the above boundary value problems only up to an element of $\ker d_A$ and $\ker d_A\vert_D$.
This is because, supposing that $\chi$ is in the kernel of either boundary value problem, then by Green's formula 
\[
0 = \langle\langle \chi, \Delta_A \chi \rangle\rangle - \int_{\pp\Sigma} \tr( \chi \star \mathsf{n} d_A \chi) =  \langle\langle d_A \chi, d_A\chi\rangle\rangle \geq 0,
\]
and the equality holds iff $d_A\chi$ vanishes.

Now, since $d_A\chi=0$ is an (overdetermined) first-order elliptic equation, $\chi \in \Omega^0(\Sigma,\Ad P)$ is (at best) completely determined by its value at a point. This means that $\ker d_A$ is always a finite dimensional space whose dimension is bounded by $\dim(G)$. For a generic $A$, $\ker d_A$ is trivial. Also, since $\chi_D$ has to vanish at $\pp\Sigma$, $(\ker d_A)_D$ is in fact \emph{always} trivial. 
In all cases the $\ker d_A$ ambiguity does not affect the element $d_A\lambda\in\Im d_A$, appearing in the decomposition. 

In general, if $\ker d_A$ is nontrivial and spanned by $\{\chi_I\}$, then the ``source'' and boundary terms appearing in the Neumann elliptic boundary value problem above must satisfy certain conditions for a solution to exist. These conditions are akin to the Abelian (integrated) {{Gauss}} law\footnote{Usually written as $4\pi Q_\text{el} = \int_{\Sigma} \mathrm{div} E = \int_{\pp\Sigma} n_i E^i$.}  of Maxwell theory, using that $d_A\chi_I=0$, they can be read off Green's formula for $\langle\langle d_A \chi_I, d_A \lambda\rangle\rangle \equiv 0$ (this is a special case of Proposition \ref{prop:Gausslaw}). However, for the Neumann problem above these conditions are always satisfied.
\end{remark}

A few further remarks.
First, we observe that by dualization via both $\star$ and $\tr$, one can adapt the above arguments to obtain a set of dual decompositions:

\begin{corollary}[of Proposition \ref{prop:HodgeEqui}]\label{cor:dualHodge}
Setting $n=\dim(\Sigma)$, one has the following decompositions:
\[
\text{Neumann}^\star: \quad
\begin{cases}
\Omega^n(\Sigma,\Ad^* P) = \Im (d_A\vert_D) \oplus \ker d_A^\star \stackrel{\text{irr.}}{=} \Im(d_A\vert_D).\\
\Omega^{n-1}(\Sigma,\Ad^* P)  = \Im d^\star_A \oplus (\ker d_A)_D 
\end{cases}
\]
and
\[
\text{Dirichlet}^\star: \quad
\begin{cases}
\Omega^n(\Sigma,\Ad^* P) = \Im d_A \oplus (\ker d^\star_A)_N \equiv \Im d_A\\
\Omega^{n-1}(\Sigma,\Ad^* P)  = \Im(d^\star_A\vert_N) \oplus \ker d_A 
\end{cases}
\]
where again we left the form degree implicit on the right hand side of these equations.
\end{corollary}

Second, we notice that fixing a reference principal connection $A_0$, we can define a family of Faddeev--Popov operators labelled by an (irreducible) principal connection $A$:
\[
\Delta_{A_0 A} : \Omega^0(\Sigma,\Ad P) \to \Omega^0(\Sigma, \Ad P), \quad \Delta_{A A_0} \doteq  d_{A_0}^\star d_A.
\]
As discussed above $\Delta_{A_0 A}$ is invertible at $A = A_0$ both for Neumann and Dirichlet boundary condition. Therefore, in view of its continuity in $A$, for each $A_0$ there must exist an open neighborhood $\mathcal{U}_0\subset\Omega^1(\Sigma,\Ad P)$ of $A_0$ in which $\Delta_{A_0 A}$ has inverses $(\Delta_{A_0 A})^{-1}_N$ and $(\Delta_{A_0 A})^{-1}_D$. Via arguments similar of those emplyed above, one can show that:

\begin{proposition}\label{prop:FaddeeevPopov}
Let $\mathcal{U}_0$ be a sufficiently small open neighborhood of $A_0\in\Omega^1(\Sigma,\Ad P)$ so that $\Delta_{A_0 A}$ admits Neumann and Dirichlet inverses for all $A\in\mathcal{U}_0$. Then, $\mathcal{U}_0 - A_0$ is a vector space which decomposes as\footnote{On the right hand side we are leaving the intersections with $\mathcal{U}_0$ understood, e.g. $\Im d_A$ stands for $\Im d_A \cap \mathcal{U}_0$. Notice that the sum fails to be orthogonal.}
\[
\mathcal{U}_0 - A_0 = \Im d_A \oplus (\ker d_{A_0}^\star)_N = \Im(d_A\vert_D) \oplus \ker(d_{A_0}).
\]
\end{proposition}
 
Third, we observe that in the context of YM theory, the objects introduced above acquire an immediate physical interpretation:

\begin{remark}[Application to YM theory]\label{rmk:HodgeInYM}
Adopting the notation of the YM running example Sections \ref{sec:runex-setup}, \ref{sec:runex-fluxannihilators}, and \ref{sec:runex-fluxgaugegroup}, where $\Acal\doteq \mathrm{Conn}(P\to\Sigma) \ni A$:
\begin{align*}
\fG &\simeq \Omega^0(\Sigma,\Ad P) & &\text{Gauge algebra;}\\
\fGo & \simeq  \Omega^0_D(\Sigma,\Ad P) & &\text{Constraint gauge ideal;}\\
\fG_A & \simeq  \Omega^0(\Sigma, \Ad P) \cap \ker d_A & &\text{Reducibility algebra of $A$;}\\
~\\
T_A\Acal & \simeq \Omega^1(\Sigma, \Ad P) & &\text{(Infinitesimal) ``variations'' of $A$;}\\\smallskip
\rho_A(\fG)& \simeq \Omega^1(\Sigma,\Ad P) \cap \Im d_A  & &\genfrac{}{}{0pt}{0}{\text{ 
Vertical distribution, or}}{\text{ ``pure gauge'' variations of $A$;}}\\
\ker(\varpi_\mathrm{Cou}) & \simeq \Omega^1_N(\Sigma,\Ad P) \cap \ker d_A^\star & &  \genfrac{}{}{0pt}{0}{\text{Coulomb horizontal variations, or}}{\text{ 
horizontal complement of $V$;}}\\\
\rho_A(\fGo)&\simeq \Omega^1(\Sigma,\Ad P)\cap \Im( d_A\vert_D) & &\text{Constraint vertical distribution;}\\
~\\
\mathcal{E}_A  \doteq T_A^\vee \Acal& \simeq \Omega^{n-1}(\Sigma, \Ad^* P) & &\text{Electric fields;}\\
\mathcal{V}_A & \simeq \Omega^{n-1}(\Sigma, \Ad^* P)\cap \Im d_A^\star & &\text{Coulombic electric fields;}\\
\mathcal{H}_A & \simeq \Omega^{n-1}(\Sigma, \Ad^* P)\cap(\ker d_A)_D & &\text{Radiative electric fields}.
\end{align*}
Furthermore, $\Delta_{A_0 A}$ is nothing else than the Faddeev--Popov operator \cite{Faddeev1969,BabelonViallet,BirminghamTQFT1991} and the decompositions of $\mathcal{U}_0$ represent generalizations of the (local\footnote{``Local'' is here meant over $\Acal$, i.e.\ within $\mathcal{U}_0$. For a discussion of the fact that in general $\mathcal{U}_0$ does not extend to the entirety of $\Acal$, and therefore the Coulomb gauge fixing fails to be global, see \cite{Gribov,Singer1978,Singer1981, NarasimhanRamadas79}.}) Coulomb gauge fixing respectively for the actions of $\fG$ and $\fGo$---generalizations that hold in the presence of corners and for different choices of background connection $A_0$. This is the statement of existence of slices at each $A_0 \in \Acal$ \cite{AbbatiCirelliMania}
\[
\mathcal{S}_{A_0} \doteq \mathcal{U}_{0}\cap \ker d^\star_{A_0} 
\quad\text{and}\quad
\mathcal{S}_{\circ,A_0} \doteq \mathcal{U}_{0}\cap (\ker d^\star_{A_0})_D
\]
transverse to $V$ and $V_\circ$, respectively.

Finally, we conclude by observing that, in form degree 1, $\langle\langle \cdot, \cdot\rangle\rangle$ defines a $\rho$-invariant (since constant) DeWitt supermetric $\mathbb{G}_A$ on $\Acal$.
\end{remark}

\section*{Declarations}
\subsubsection*{Conflict of interest statement} The authors have no competing interests to declare that are relevant to the content of this article.
\subsubsection*{Data availability statement} Data sharing not applicable to this article as no datasets were generated or analysed during the current study.

\medskip 

\begingroup
\sloppy
\printbibliography
\endgroup

\end{document}